\documentclass[reqno,11pt,letterpaper]{amsart}

\usepackage{lipsum}
\usepackage{amsmath}
\usepackage{amssymb}
\usepackage{amsthm}
\usepackage{mathrsfs}
\usepackage{accents}
\usepackage{calc}
\usepackage{arydshln}
\usepackage{upgreek}
\usepackage{slashed}
\usepackage{xifthen}
\usepackage{graphicx}
\usepackage{longtable}
\usepackage[inline]{enumitem}

\usepackage{xcolor}
\definecolor{winered}{rgb}{0.6,0,0}
\definecolor{lessblue}{rgb}{0,0,0.7}

\usepackage[pdftex,colorlinks=true,linkcolor=winered,citecolor=lessblue,urlcolor=lessblue,breaklinks=true,bookmarksopen=true]{hyperref}

\makeatletter
\newcommand{\myitem}[3]{\item[#2]\def\@currentlabel{#3}\label{#1}}
\makeatother

\usepackage{titletoc}

\makeatletter

\def\@tocline#1#2#3#4#5#6#7{
\begingroup
  \par
    \parindent\z@ \leftskip#3 \relax \advance\leftskip\@tempdima\relax
                  \rightskip\@pnumwidth plus 4em \parfillskip-\@pnumwidth
    \ifcase #1 
       \vskip 0.6em \hskip 0em 
       \or
       \or \hskip 0em 
       \or \hskip 1em 
    \fi%
    #6
    \nobreak\relax{\leavevmode\leaders\hbox{\,.}\hfill}
    \hbox to\@pnumwidth {\@tocpagenum{#7}}
  \par
\endgroup
}

 \def\l@section{\@tocline{0}{0pt}{0pc}{}{}}

\renewcommand{\tocsection}[3]{%
  \indentlabel{\@ifnotempty{#2}{ 
    \ignorespaces\bfseries{#2. #3}}}
  \indentlabel{\@ifempty{#2}{\ignorespaces\bfseries{#3}}{}} 
    \vspace{1.5pt}}

\renewcommand{\tocsubsection}[3]{%
  \indentlabel{\@ifnotempty{#2}{
    \ignorespaces#2. #3}}
  \indentlabel{\@ifempty{#2}{\ignorespaces #3}{}}
    \vspace{1.5pt}}

\renewcommand{\tocsubsubsection}[3]{%
  \indentlabel{\@ifnotempty{#2}{
    \ignorespaces#2. #3}}
  \indentlabel{\@ifempty{#2}{\ignorespaces #3}{}}
    \vspace{1.5pt}}

\makeatother

\makeatletter
\def\@nomenstarted{0}
\newlength{\@nomenoldtabcolsep}

\newcommand{\nomenstart}
  {%
    \def\@nomenstarted{1}%
    \setlength{\@nomenoldtabcolsep}{\tabcolsep}%
    \setlength{\tabcolsep}{3.5pt}%
    \begin{longtable}{p{0.11\textwidth} p{0.86\textwidth}}
  }

\newcommand{\nomenitem}[2]{%
    \ifcase\@nomenstarted%
      \or 
      \or \\ 
    \fi%
    #1\,{\leavevmode\leaders\hbox{\,.}\hfill} & #2%
    \def\@nomenstarted{2}%
  }%
\newcommand{\nomenend}
  {\\%
      \end{longtable}%
      \setlength{\tabcolsep}{\@nomenoldtabcolsep}%
      \def\@nomenstarted{0}%
  }
\makeatother

\makeatletter
\newcommand{\vast}{\bBigg@{4}}
\newcommand{\Vast}{\bBigg@{5}}
\newcommand{\VAST}[1]{\bBigg@{#1}}
\makeatother

\allowdisplaybreaks

\numberwithin{equation}{section}
\numberwithin{figure}{section}
\newtheorem{thm}{Theorem}[section]

\newtheorem{prop}[thm]{Proposition}
\newtheorem{lemma}[thm]{Lemma}
\newtheorem{cor}[thm]{Corollary}
\newtheorem*{thm*}{Theorem}
\newtheorem*{prop*}{Proposition}
\newtheorem*{cor*}{Corollary}
\newtheorem*{conj*}{Conjecture}

\theoremstyle{definition}
\newtheorem{definition}[thm]{Definition}
\newtheorem{notation}[thm]{Notation}

\theoremstyle{remark}
\newtheorem{rmk}[thm]{Remark}

\makeatletter
\newcommand{\fakephantomsection}{%
  \Hy@MakeCurrentHref{\@currenvir.\the\Hy@linkcounter}
  \Hy@raisedlink{\hyper@anchorstart{\@currentHref}\hyper@anchorend}%
  \Hy@GlobalStepCount\Hy@linkcounter%
}
\makeatother

\newcommand{\mc}{\mathcal}
\newcommand{\cA}{\mc A}

\newcommand{\cC}{\mc C}

\newcommand{\cH}{\mc H}

\newcommand{\cN}{\mc N}
\newcommand{\cO}{\mc O}

\newcommand{\cR}{\mc R}

\newcommand{\cV}{\mc V}

\newcommand{\C}{\mathbb{C}}
\newcommand{\N}{\mathbb{N}}
\newcommand{\R}{\mathbb{R}}
\newcommand{\Z}{\mathbb{Z}}

\newcommand{\Sph}{\mathbb{S}}

\newcommand{\sfa}{\mathsf{a}}
\newcommand{\sfb}{\mathsf{b}}

\newcommand{\sfr}{\mathsf{r}}

\newcommand{\fa}{\mathfrak{a}}

\newcommand{\fm}{\mathfrak{m}}

\newcommand{\slg}{\slashed{g}{}}

\renewcommand{\Re}{\operatorname{Re}}
\renewcommand{\Im}{\operatorname{Im}}
\newcommand{\Id}{\operatorname{Id}}

\newcommand{\supp}{\operatorname{supp}}

\newcommand{\tr}{\operatorname{tr}}

\newcommand{\diag}{\operatorname{diag}}
\newcommand{\Res}{\operatorname{Res}}

\newcommand{\QNM}{{\mathrm{QNM}}}

\newcommand{\dS}{{\mathrm{dS}}}

\newcommand{\eps}{\epsilon}

\newcommand{\hra}{\hookrightarrow}
\newcommand{\la}{\langle}

\newcommand{\ol}{\overline}
\newcommand{\pa}{\partial}
\newcommand{\dd}{{\mathrm d}}
\newcommand{\ra}{\rangle}

\newcommand{\tot}{\mathrm{tot}}

\newcommand{\ul}[1]{\underline{#1}{}}

\newcommand{\weakto}{\rightharpoonup}
\newcommand{\wh}{\widehat}
\newcommand{\wt}{\widetilde}
\newcommand{\xra}{\xrightarrow}

\newcommand{\bop}{{\mathrm{b}}}
\newcommand{\cop}{{\mathrm{c}}}
\newcommand{\Qop}{{\mathrm{Q}}}
\newcommand{\qop}{{\mathrm{q}}}

\newcommand{\scop}{{\mathrm{sc}}}
\newcommand{\chop}{{\mathrm{c}\semi}}

\newcommand{\cl}{{\mathrm{cl}}}

\newcommand{\scbtop}{{\mathrm{sc}\text{-}\mathrm{b}}}
\newcommand{\schop}{{\mathrm{sc},\semi}}

\newcommand{\semi}{\hbar}
\newcommand{\low}{{\mathrm{low}}}

\newcommand{\ff}{\mathrm{ff}}

\newcommand{\lb}{{\mathrm{lb}}}
\newcommand{\rb}{{\mathrm{rb}}}
\newcommand{\tlb}{{\mathrm{tlb}}}
\newcommand{\trb}{{\mathrm{trb}}}
\newcommand{\bface}{{\mathrm{bf}}}

\newcommand{\cface}{{\mathrm{cf}}}
\newcommand{\dface}{{\mathrm{df}}}

\newcommand{\iface}{{\mathrm{if}}}

\newcommand{\mface}{{\mathrm{mf}}}
\newcommand{\nface}{{\mathrm{nf}}}

\newcommand{\scface}{{\mathrm{scf}}}
\newcommand{\sface}{{\mathrm{sf}}}

\newcommand{\tface}{{\mathrm{tf}}}

\newcommand{\zface}{{\mathrm{zf}}}

\newcommand{\res}{{\mathrm{res}}}

\newcommand{\cp}{{\mathrm{c}}}

\newcommand{\Diff}{\mathrm{Diff}}

\DeclareMathOperator{\Op}{Op}

\newcommand{\Opb}{\Op_\bop}
\newcommand{\Opsc}{\Op_\scop}
\newcommand{\Vb}{\cV_\bop}
\newcommand{\Vscbt}{\cV_\scbtop}
\newcommand{\VQ}{\cV_\Qop}
\newcommand{\Vq}{\cV_\qop}
\newcommand{\Vch}{\cV_\chop}
\newcommand{\Diffch}{\Diff_\chop}
\newcommand{\Diffb}{\Diff_\bop}
\newcommand{\Diffscbt}{\Diff_\scbtop}
\newcommand{\DiffQ}{\Diff_\Qop}
\newcommand{\Diffq}{\Diff_\qop}

\newcommand{\Psib}{\Psi_\bop}
\newcommand{\Psiscbt}{\Psi_\scbtop}
\newcommand{\PsiQ}{\Psi_\Qop}
\newcommand{\Psiq}{\Psi_\qop}
\newcommand{\Psich}{\Psi_\chop}
\newcommand{\Psisc}{\Psi_\scop}
\newcommand{\Psisch}{\Psi_\schop}

\newcommand{\Tscbt}{{}^\scbtop T}

\newcommand{\Vsc}{\cV_\scop}
\newcommand{\Diffsc}{\Diff_\scop}
\newcommand{\Vsch}{\cV_\schop}
\newcommand{\Diffsch}{\Diff_\schop}

\newcommand{\Psih}{\Psi_\semi}

\newcommand{\Omegab}{{}^{\bop}\Omega}
\newcommand{\OmegaQ}{{}^{\Qop}\Omega}
\newcommand{\Omegaq}{{}^{\qop}\Omega}
\newcommand{\Omegasch}{{}^\schop\Omega}
\newcommand{\Omegascbt}{{}^\scbtop\Omega}
\newcommand{\Omegach}{{}^{\chop}\Omega}
\newcommand{\Omegasc}{{}^{\scop}\Omega}

\newcommand{\Tb}{{}^{\bop}T}
\newcommand{\Tch}{{}^\chop T}
\newcommand{\TQ}{{}^\Qop T}
\newcommand{\Tq}{{}^\qop T}
\newcommand{\Tsc}{{}^{\scop}T}
\newcommand{\Tsch}{{}^\schop T}

\newcommand{\SQ}{{}^\Qop S}
\newcommand{\Sch}{{}^\chop S}
\newcommand{\Ssc}{{}^{\scop}S}

\newcommand{\half}{{\tfrac{1}{2}}}

\newcommand{\sigmab}{{}^\bop\upsigma}
\newcommand{\sigmascbt}{{}^\scbtop\upsigma}
\newcommand{\sigmach}{{}^\chop\upsigma}
\newcommand{\sigmaQ}{{}^\Qop\upsigma}
\newcommand{\sigmaq}{{}^\qop\upsigma}
\newcommand{\sigmasc}{{}^\scop\upsigma}
\newcommand{\sigmasch}{{}^\schop\upsigma}

\newcommand{\loc}{{\mathrm{loc}}}
\newcommand{\CI}{\cC^\infty}
\newcommand{\CIdot}{\dot\cC^\infty}

\newcommand{\CIc}{\cC^\infty_\cp}

\newcommand{\CmI}{\cC^{-\infty}}

\newcommand{\Hb}{H_{\bop}}

\newcommand{\Hbext}{\bar H_{\bop}}

\newcommand{\Hbsupp}{\dot H_{\bop}}

\newcommand{\Hext}{\bar H}

\newcommand{\Hsupp}{\dot H}

\newcommand{\Hsc}{H_{\scop}}

\newcommand{\Ric}{\mathrm{Ric}}

\newcommand{\bhm}{\fm}
\newcommand{\bha}{\fa}

\newcommand{\openbigpmatrix}[1]
  {%
    \def\@bigpmatrixsize{#1}%
    \addtolength{\arraycolsep}{-#1}%
    \begin{pmatrix}%
  }
\newcommand{\closebigpmatrix}
  {%
    \end{pmatrix}%
    \addtolength{\arraycolsep}{\@bigpmatrixsize}%
  }

\newlength{\enummargin}

\newcommand{\usref}[1]{{\upshape\ref{#1}}}

\DeclareGraphicsExtensions{.mps}

\makeatletter
\newcommand*{\fwbw}[1]{\expandafter\@fwbw\csname c@#1\endcsname}
\newcommand*{\@fwbw}[1]{\ifcase #1 \or {\rm fw}\or {\rm bw}\fi}
\AddEnumerateCounter{\fwbw}{\@fwbw}
\makeatother

\begin{document}

\title[Mode stability of Kerr--de~Sitter]{Mode stability and shallow quasinormal modes of Kerr--de~Sitter black holes away from extremality}

\date{\today}

\begin{abstract}
  A Kerr--de~Sitter black hole is a solution $(M,g_{\Lambda,\bhm,\bha})$ of the Einstein vacuum equations with cosmological constant $\Lambda>0$. It describes a black hole with mass $\bhm>0$ and specific angular momentum $\bha\in\R$. We show that for any $\eps>0$ there exists $\delta>0$ so that mode stability holds for the linear scalar wave equation $\Box_{g_{\Lambda,\bhm,\bha}}\phi=0$ when $|\bha/\bhm|\in[0,1-\eps]$ and $\Lambda\bhm^2<\delta$. In fact, we show that all quasinormal modes $\sigma$ in any fixed half space $\Im\sigma>-C\sqrt\Lambda$ are equal to $0$ or $-i\sqrt{\Lambda/3}(n+o(1))$, $n\in\N$, as $\Lambda\bhm^2\searrow 0$. We give an analogous description of quasinormal modes for the Klein--Gordon equation.

  We regard a Kerr--de~Sitter black hole with small $\Lambda\bhm^2$ as a singular perturbation either of a Kerr black hole with the same angular momentum-to-mass ratio, or of de~Sitter spacetime without any black hole present. We use the mode stability of subextremal Kerr black holes, proved by Whiting and Shlapentokh-Rothman, as a black box; the quasinormal modes described by our main result are perturbations of those of de~Sitter space. Our proof is based on careful uniform a priori estimates, in a variety of asymptotic regimes, for the spectral family and its de~Sitter and Kerr model problems in the singular limit $\Lambda\bhm^2\searrow 0$.
\end{abstract}


\author{Peter Hintz}
\address{Department of Mathematics, ETH Z\"urich, R\"amistrasse 101, 8092 Z\"urich, Switzerland}
\email{peter.hintz@math.ethz.ch}

\maketitle


\section{Introduction}
\label{SI}

The metric of a subextremal Kerr--de~Sitter (KdS) spacetime depends on the parameters $\Lambda>0$ (cosmological constant), $\bhm>0$ (mass of the black hole), and $\bha\in\R$ (specific angular momentum). It involves the quartic polynomial
\begin{equation}
\label{EqImu}
  \mu_{\Lambda,\bhm,\bha}(r) = (r^2+\bha^2)\Bigl(1-\frac{\Lambda r^2}{3}\Bigr) - 2\bhm r.
\end{equation}
The spacetime, or the set of parameters $(\Lambda,\bhm,\bha)$, is called \emph{subextremal} if $\mu_{\Lambda,\bhm,\bha}$ has four distinct real roots
\[
  r_{\Lambda,\bhm,\bha}^- < r_{\Lambda,\bhm,\bha}^C < r_{\Lambda,\bhm,\bha}^e < r_{\Lambda,\bhm,\bha}^c.
\]
For subextremal parameters, the KdS metric is given on the \emph{domain of outer communications}
\begin{equation}
\label{EqIKdSDOC}
  M_{\Lambda,\bhm,\bha}^{\rm DOC} = \R_t \times (r_{\Lambda,\bhm,\bha}^e, r_{\Lambda,\bhm,\bha}^c)_r \times (0,\pi)_\theta \times (0,2\pi)_\phi 
\end{equation}
in Boyer--Lindquist coordinates (introduced in the special case $\Lambda=0$ in \cite{BoyerLindquistKerr}) by
\begin{align}
\label{EqIMetric}
\begin{split}
  g_{\Lambda,\bhm,\bha} &:= -\frac{\mu_{\Lambda,\bhm,\bha}(r)}{b_{\Lambda,\bhm,\bha}^2\varrho_{\Lambda,\bhm,\bha}^2(r,\theta)}(\dd t-\bha\,\sin^2\theta\,\dd\phi)^2 + \varrho_{\Lambda,\bhm,\bha}^2(r,\theta)\Bigl(\frac{\dd r^2}{\mu_{\Lambda,\bhm,\bha}(r)} + \frac{\dd\theta^2}{c_{\Lambda,\bhm,\bha}(\theta)}\Bigr) \\
    &\qquad + \frac{c_{\Lambda,\bhm,\bha}(\theta)\sin^2\theta}{b_{\Lambda,\bhm,\bha}^2\varrho_{\Lambda,\bhm,\bha}^2(r,\theta)}\bigl( (r^2+\bha^2)\dd\phi - \bha\,\dd t\bigr)^2,
\end{split} \\
\label{EqIMetric2}
  &\hspace{-1em}b_{\Lambda,\bhm,\bha}:=1+\frac{\Lambda\bha^2}{3},\quad
    c_{\Lambda,\bhm,\bha}(\theta):=1+\frac{\Lambda\bha^2}{3}\cos^2\theta,\quad
    \varrho_{\Lambda,\bhm,\bha}^2(r,\theta):=r^2+\bha^2\cos^2\theta.
\end{align}
Its physical relevance stems from the fact that it is a solution of the Einstein vacuum equation $\Ric(g_{\Lambda,\bhm,\bha})-\Lambda g_{\Lambda,\bhm,\bha}=0$. It was discovered by Carter \cite{CarterHamiltonJacobiEinstein}, following the earlier discovery \cite{KerrKerr} of the Kerr metric, which is obtained by formally setting $\Lambda=0$:
\[
  g_{\bhm,\bha} := g_{0,\bhm,\bha};\qquad
  \Ric(g_{\bhm,\bha})=0\ \ \text{on}\ \ \R_t\times(r_{\bhm,\bha}^e,\infty)_r\times(0,\pi)_\theta\times(0,2\pi)_\phi.
\]
For $\Lambda=0$, the condition for subextremality is that $r^2+\bha^2-2\bhm r$ have two distinct real roots $r_{\bhm,\bha}^C<r_{\bhm,\bha}^e$; these roots are $\bhm\mp\sqrt{\bhm^2-\bha^2}$, and thus a Kerr spacetime is subextremal if and only if $|\bha/\bhm|<1$. When $\Lambda\bhm^2>0$ is sufficiently small, this is also a sufficient condition for the subextremality of the KdS spacetime; see Figure~\ref{FigI} below, and Lemma~\ref{LemmaKRoots} for a weaker---but sufficient for our purposes---statement.

The above expression for the metric becomes singular at $r=r_{\Lambda,\bhm,\bha}^e$ and $r=r_{\Lambda,\bhm,\bha}^c$. This is merely a coordinate singularity, as can be seen by passing to the coordinates
\begin{equation}
\label{EqIKerrStar}
\begin{alignedat}{2}
  t_* &:= t - T_{\Lambda,\bhm,\bha}(r),&\qquad T_{\Lambda,\bhm,\bha}'(r)&=(r^2+\bha^2)\frac{b_{\Lambda,\bhm,\bha}}{\mu_{\Lambda,\bhm,\bha}(r)}F_{\Lambda,\bhm,\bha}(r), \\
  \phi_* &:= \phi - \Phi_{\Lambda,\bhm,\bha}(r),&\qquad \Phi_{\Lambda,\bhm,\bha}'(r)&=\bha \frac{b_{\Lambda,\bhm,\bha}}{\mu_{\Lambda,\bhm,\bha}(r)}F_{\Lambda,\bhm,\bha}(r),
\end{alignedat}
\end{equation}
where $F_{\Lambda,\bhm,\bha}(r)=2\frac{r-r_{\Lambda,\bhm,\bha}^e}{r_{\Lambda,\bhm,\bha}^c-r_{\Lambda,\bhm,\bha}^e}-1$. Expressed in the coordinates $(t_*,r,\theta,\phi_*)$, the metric $g_{\Lambda,\bhm,\bha}$ extends real analytically to
\begin{equation}
\label{EqIKdSExt}
  \wt M_{\Lambda,\bhm,\bha} = \R_{t_*} \times \wt X_{\Lambda,\bhm,\bha},\qquad
  \wt X_{\Lambda,\bhm,\bha} := (r_{\Lambda,\bhm,\bha}^C,\infty)_r \times \Sph^2_{\theta,\phi_*}.
\end{equation}
See~\cite[Equation~(5)]{PetersenVasySubextremal} for the explicit expression.\footnote{In the main part of the paper, we will make a different choice of $F_{\Lambda,\bhm,\bha}(r)$ which has better properties in the limit $\Lambda\bhm^2\searrow 0$; see~\S\ref{SsKL}.} The two null hypersurfaces
\[
  \cH^+_{\Lambda,\bhm,\bha} = \R_{t_*}\times\{r_{\Lambda,\bhm,\bha}^e\}\times\Sph^2_{\theta,\phi_*}, \qquad
  \ol\cH{}^+_{\Lambda,\bhm,\bha} = \R_{t_*}\times\{r_{\Lambda,\bhm,\bha}^C\}\times\Sph^2_{\theta,\phi_*}
\]
are called the \emph{(future) event horizon} and \emph{(future) cosmological horizon}, respectively.

The object of main interest in this paper is the set
\[
  \QNM(\Lambda,\bhm,\bha)\subset\C
\]
of \emph{resonances}, or \emph{quasinormal modes}, of the scalar wave operator $\Box_{g_{\Lambda,\bhm,\bha}}$. Here, $\sigma\in\QNM(\Lambda,\bhm,\bha)$ if and only if there exists a \emph{resonant state} $u_0(r,\theta,\phi_*)\in\CI(\wt X_{\Lambda,\bhm,\bha})$ so that $u(t_*,r,\theta,\phi_*)=e^{-i\sigma t_*}u_0(r,\theta,\phi_*)\in\CI(M_{\Lambda,\bhm,\bha})$ is a \emph{mode solution} of the wave equation $\Box_{g_{\Lambda,\bhm,\bha}}u=0$. (For an equivalent definition in terms of Boyer--Lindquist coordinates, see e.g.\ \cite[Theorem~3]{DyatlovQNM} or \cite[Definition~2.4]{CasalsTeixeiradCModes}.)

\begin{thm}[Quasinormal modes of Kerr--de~Sitter black holes away from extremality: massless scalar fields]
\label{ThmI}
  Fix $C>0$, and let $\eps>0$. Then there exists $\delta>0$ so that for\footnote{The quantities $\bha/\bhm$, $\Lambda\bhm^2$, and $\Lambda^{-\frac12}\sigma$ are dimensionless; see~\S\ref{SsIS}.} $|\bha/\bhm|\leq 1-\eps$ and $\Lambda\bhm^2\in(0,\delta)$, every
  \[
    \sigma\in\QNM(\Lambda,\bhm,\bha),\qquad \Im(\Lambda^{-1/2}\sigma)>-C
  \]
  either satisfies $\sigma=0$ or $\sigma=-i\sqrt{\Lambda/3}(n+o(1))$ for some $n\in\N$ as $\Lambda\bhm^2\searrow 0$. Moreover, the only mode solutions with $\sigma=0$ are constant functions. Conversely, for any $n\in\N$ and $\eta>0$ there exists, for sufficiently small $\Lambda\bhm^2>0$ and for any $\bha/\bhm\in[-1+\eps,1-\eps]$, an element $\sigma\in\QNM(\Lambda,\bhm,\bha)$ with $|\sigma\sqrt{3/\Lambda}+i n|<\eta$.
\end{thm}

Thus, the set $(\Lambda/3)^{-\frac12}\QNM(\Lambda,\bhm,\bha)$ converges in any half space $\Im\sigma>-C$ to the set $-i\N_0$ as $\Lambda\bhm^2\searrow 0$ when $|\bha/\bhm|$ remains bounded away from $1$. The full result, Theorem~\ref{ThmK} (together with Lemma~\ref{LemmaKdSQNM}), is more precise: we show the convergence of resonances \emph{with multiplicity}, and we also prove the convergence of (generalized) resonant states, appropriately rescaled, to (generalized) resonant states on the static patch of de~Sitter space (see~\S\ref{SsIA}). (Petersen--Vasy \cite{PetersenVasyAnalytic}, based on earlier work by Galkowski--Zworski \cite{GalkowskiZworskiHypo}, showed that resonant states are \emph{analytic}, but our analysis does not make use of this fact.) We refer the reader to \cite[\S\S 1 and 4]{HintzXieSdS} for plots and numerics in the Schwarzschild--de~Sitter case $\bha=0$, and to Figure~\ref{FigI} below for a schematic illustration of Theorem~\ref{ThmI}.

Mode stability is an immediate consequence of Theorem~\ref{ThmI}:\footnote{The KdS parameter range covered by Corollary~\ref{CorIMS} has been confirmed to constitute a ``large'' range in the sense of \cite[{\includegraphics[width=0.8em]{MichelinStar.png}\,\includegraphics[width=0.8em]{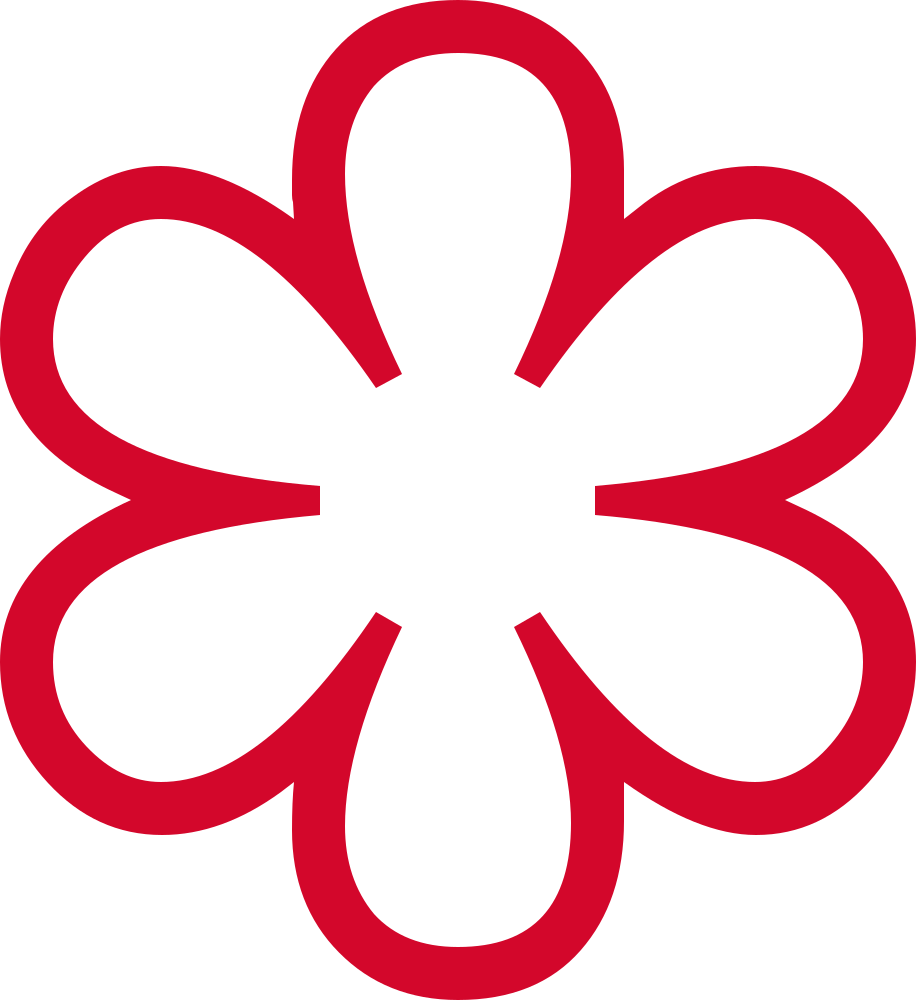}} Conjecture~4]{ZworskiResonanceReview}.}

\begin{cor}[Mode stability of Kerr--de~Sitter black holes away from extremality]
\label{CorIMS}
  For any $\eps>0$, there exists $\delta>0$ so that mode stability holds for the scalar wave equation on Kerr--de~Sitter black holes with parameters $\Lambda,\bhm,\bha$ satisfying $|\bha/\bhm|\leq 1-\eps$ and $\Lambda\bhm^2\in(0,\delta)$. That is, no $\sigma\in\C$ with $\Im\sigma\geq 0$ and $\sigma\neq 0$ is a quasinormal mode; equivalently, for $\sigma\in\QNM(\Lambda,\bhm,\bha)$, either $\Im\sigma<0$ or $\sigma=0$. Moreover, for $\sigma=0$, the only mode solutions are constants.
\end{cor}

In particular, when $\Lambda>0$ and the ratio $|\bha/\bhm|<1$ are fixed, this implies the mode stability of KdS when the black hole mass $\bhm$ is sufficiently small. Alternatively, when $\bhm$ and $|\bha/\bhm|<1$ are fixed, we conclude mode stability when $\Lambda>0$ is sufficiently small; this regime is of particular astrophysical interest since, according to the currently favored $\Lambda$CDM model, $\Lambda$ is indeed positive but very small.

\begin{figure}[!ht]
\centering
\includegraphics{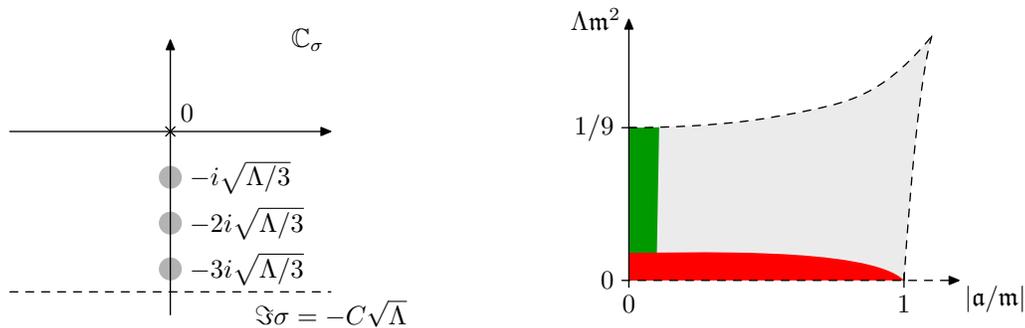}
\caption{\textit{On the left:} illustration of Theorem~\ref{ThmI}. The wave operator on Kerr--de~Sitter spacetimes with small $\Lambda\bhm^2$ has a resonance at $0$, and resonances near $-i n\sqrt{\Lambda/3}$, $n=1,2,3,\ldots$ \textit{On the right:} the dashed region is the parameter space of subextremal Kerr--de~Sitter black holes. The red region is a schematic depiction of the set of parameters to which Theorem~\ref{ThmI} applies. Mode stability is known to hold in the union of the red region (Corollary~\ref{CorIMS}) and the green region (see~\S\ref{SsIMS}).}
\label{FigI}
\end{figure}

The KdS black holes considered in Theorem~\ref{ThmI} fit into Vasy's framework \cite[\S6]{VasyMicroKerrdS}, recently extended to the full subextremal range of KdS black holes by Petersen--Vasy \cite{PetersenVasySubextremal}. This implies resonance expansions for solutions of the wave equation up to exponentially decaying remainders.\footnote{In the present context, the dynamical assumptions required by Vasy's framework follow already by combining the $r$-normal hyperbolicity for every $r$ of the trapped set of subextremal Kerr black holes, proved by Dyatlov \cite{DyatlovWaveAsymptotics}, with the structural stability of such trapped sets \cite{HirschPughShubInvariantManifolds}. In fact, however, in the course of our proof of Theorem~\ref{ThmI}, we directly prove the meromorphicity of, and high energy estimates for, the inverse of the spectral family of $\Box_{g_{\Lambda,\bhm,\bha}}$ in $\Im(\Lambda^{-\frac12}\sigma)>-C$, which imply such resonance expansions.} We state this in the simplest form, and only record the terms corresponding to the quasinormal modes captured by Theorem~\ref{ThmI}:

\begin{cor}[Resonance expansions for waves]
\label{CorI}
  Put $x=(r,\theta,\phi_*)$. For $C>0$ and $\eps>0$, let $\delta>0$ be as in Theorem~\usref{ThmI}, and suppose $|\bha/\bhm|\leq 1-\eps$ and $\Lambda\bhm^2\in(0,\delta)$. Let $X:=[r_-,r_+]\times\Sph^2_{\theta,\phi_*}$, where $r_-\in(r_{\Lambda,\bhm,\bha}^C,r_{\Lambda,\bhm,\bha}^e)$ and $r_+\in(r_{\Lambda,\bhm,\bha}^c,\infty)$. Let $u=u(t_*,x)$ denote the solution of the initial value problem
  \[
    \Box_{g_{\Lambda,\bhm,\bha}}u=0,\qquad (u,\pa_{t_*}u)|_{t_*=0} = (u_0,u_1) \in \CI(X)\oplus\CI(X).
  \]
  Then $u$ has an asymptotic expansion
  \[
    \biggl|u(t_*,x) - u_0 - \sum_{j=1}^N\biggl(\sum_{k=0}^{k_j} t_*^k e^{-i\sigma_j t_*}u_{j k}(x)\biggr)\biggr| \leq C_1 e^{-C\sqrt{\Lambda}\,t_*},
  \]
  where $u_0\in\C$, and where $\sigma_1,\ldots,\sigma_N\in\C$ (possibly with repetitions) are the quasinormal modes with $0>\Im(\Lambda^{-\frac12}\sigma_j)\geq-C$, and the $\sum_{k=0}^{k_j}t_*^k e^{-i\sigma_j t_*}u_{j k}$ are (generalized) mode solutions\footnote{We do not rule out the possibility that some of the resonances controlled by Theorem~\ref{ThmI} are not simple; hence the need to allow for $k_j\geq 1$.} of the wave equation. In particular,
  \[
    |u(t_*,x)-u_0| \leq C_2 \exp\biggl(-\Bigl((1+\eta)\sqrt{\frac{\Lambda}{3}}\,\Bigr)t_*\biggr),
  \]
  where $\eta=\eta(\Lambda\bhm^2,\bha/\bhm)\to 0$ as $\Lambda\bhm^2\searrow 0$. Above, $C_1,C_2$ are constants depending only on $\Lambda,\bhm,\bha$, and on the initial data $u_0,u_1$.
\end{cor}

See \cite[Theorem~1.5]{PetersenVasySubextremal} (based on \cite[Theorem~1.4]{VasyMicroKerrdS}) for a more precise statement which has weaker regularity requirements and allows for the presence of forcing. See moreover \cite{HintzVasyKdSStability} and \cite[Theorem~1.6]{PetersenVasySubextremal} (based on \cite{HintzVasyQuasilinearKdS}) for applications of such resonance expansions to quasilinear equations.

\begin{rmk}[Spacetime degeneration]
  A uniform description of the singular limit of (waves on) KdS \emph{spacetimes} as $\bhm\searrow 0$ is beyond the scope of this paper.
\end{rmk}

As an illustration of the flexibility of our method of proof, we also show:

\begin{thm}[Quasinormal modes of Kerr--de~Sitter black holes away from extremality: massive scalar fields]
\label{ThmIKG}
  Let $\nu\in\C$. Denote by $\QNM(\nu;\Lambda,\bhm,\bha)$ the set of resonances for the Klein--Gordon operator $\Box_{g_{\Lambda,\bhm,\bha}}-\frac{\Lambda}{3}\nu$.\footnote{The normalization of the zeroth order term is chosen so that $\nu$ is dimensionless; see~\S\ref{SsIS}.} Put $\lambda_\pm=\frac32\pm\sqrt{\frac94-\nu}$. Let $C>0$. Then for any $\eps>0$, there exists $\delta>0$ so that for $|\bha/\bhm|\leq 1-\eps$ and $\Lambda\bhm^2\in(0,\delta)$, every $\sigma\in\QNM(\nu;\Lambda,\bhm,\bha)$ with $\Im(\Lambda^{-\frac12}\sigma)>-C$ satisfies
  \[
    \sigma = -i\sqrt{\Lambda/3}(\lambda_\pm+n+o(1))
  \]
  for some $n\in\N_0$ as $\Lambda\bhm^2\searrow 0$. Conversely, there does exist a quasinormal mode near each $-i\sqrt{\Lambda/3}(\lambda_\pm+n)$.
\end{thm}

Solutions of the Klein--Gordon equation admit resonance expansions in a manner analogous to Corollary~\ref{CorI}. \emph{For the remainder of this introduction, we restrict attention to the massless case (Theorem~\usref{ThmI}) unless explicitly stated otherwise.}

\subsection{Prior work on quasinormal modes and resonance expansions on de~Sitter black hole spacetimes}
\label{SsIQ}

In the special case $\bha=0$ of \emph{Schwarzschild--de~Sitter} black holes (in which case the subextremality condition becomes $0<9\Lambda\bhm^2<1$), the discreteness of $\QNM(\Lambda,\bhm,0)$ was shown by S\'a Barreto--Zworski \cite{SaBarretoZworskiResonances}, relying in particular on \cite{MazzeoMelroseHyp}. For fixed parameters $(\Lambda,\bhm)$, they also characterized resonances in the high frequency regime $|\Re\sigma|\gg 1$, and showed that in conic sectors $\Im\sigma>-\theta|\Re\sigma|$ (with $\theta>0$ sufficiently small) they are given by
\begin{equation}
\label{EqISBZw}
  \Bigl(\pm l\pm\frac12-\frac{i}{2}\Bigl(n+\frac12\Bigr)\Bigr)\frac{(1-9\Lambda\bhm^2)^{\frac12}}{3\Lambda^{\frac12}\bhm}\sqrt{\Lambda/3} + o(1),\qquad l\to\infty.\quad (n=0,1,2,\ldots)
\end{equation}
When $\Lambda\bhm^2\searrow 0$, these approximate resonances thus leave any fixed half space $\Im(\Lambda^{-\frac12}\sigma)>-C$; in this sense, Theorem~\ref{ThmI} concerns an altogether different regime of resonances than \cite{SaBarretoZworskiResonances}. One is moreover led to conjecture that (at least away from the negative imaginary axis) Theorem~\ref{ThmI} continues to hold in the larger $\bhm$-dependent range $\Im\sigma\geq-c\bhm^{-1}$ for any $c<\frac{1}{12\sqrt 3}$.

Still for $\bha=0$, the author and Xie \cite{HintzXieSdS} proved a version of Theorem~\ref{ThmI} which only provides uniform control of resonances in any fixed ball $|\Lambda^{-\frac12}\sigma|\leq C$ provided they are associated with mode solutions which moreover have a \emph{fixed} angular momentum $l\in\N_0$ (i.e.\ their dependence on the angular variables is given by a degree $l$ spherical harmonic). The proof proceeded via uniform estimates for a degenerating family of ordinary differential equations, whereas the proof of Theorem~\ref{ThmI} requires more sophisticated tools (see~\S\ref{SsIA}).

On Schwarzschild--de~Sitter spacetimes, high energy resolvent estimates and resonance expansions similar to Corollary~\ref{CorI} were established in \cite{MelroseSaBarretoVasySdS,MelroseSaBarretoVasyResolvent} (exponential decay to constants on $\wt M_{\Lambda,\bhm,0}$) and previously in \cite{BonyHaefnerDecay}: in the latter paper, Bony--H\"afner showed that on $M_{\Lambda,\bhm,0}^{\rm DOC}$, waves are convergent sums over possibly infinitely many resonances, up to an error term which has any desired amount of exponential decay. In recent work, Mavrogiannis \cite{MavrogiannisSdSMorawetz} gives a proof of exponential decay to constants (thus exponential energy decay) using vector field (`physical space') techniques; this improves on earlier work by Dafermos--Rodnianski \cite{DafermosRodnianskiSdS} which gave superpolynomial energy decay. An alternative definition of the quasinormal mode spectrum, as the set of eigenvalues of an appropriate evolution semigroup, and a proof of some of its salient properties (such as discreteness), was given by Warnick \cite{WarnickQNMs}, and extended to asymptotically flat settings by Gajic--Warnick \cite{GajicWarnickModel}; see also \cite{GalkowskiZworskiGevrey2}.

These results were generalized to the case of slowly rotating Kerr--de~Sitter black holes in a series of papers by Dyatlov. In~\cite{DyatlovQNM}, Dyatlov defined resonances by exploiting the separability of the wave equation and proved the discreteness of the set of resonances; he moreover showed exponential decay to constants of waves first in $M_{\Lambda,\bhm,\bha}^{\rm DOC}$, and then in $\wt M_{\Lambda,\bhm,\bha}$ in~\cite{DyatlovQNMExtended} using red-shift estimates of Dafermos--Rodnianski \cite{DafermosRodnianskiRedShift} near the horizons. The paper \cite{DyatlovAsymptoticDistribution} gives a description of the high energy resonances generalizing and significantly refining~\eqref{EqISBZw}, and proves resonance expansions for solutions of the wave equation, up to error terms with any desired amount of exponential decay. As in the case $\bha=0$, the semiclassical methods of \cite{DyatlovAsymptoticDistribution} are effective only in a high frequency regime, and all resonances captured by it leave the subset of the complex frequency plane described in Theorem~\ref{ThmI} when $\Lambda\bhm^2\searrow 0$.

A key ingredient in Dyatlov's works is a robust approach to the analysis of the spectral family at high frequencies near the trapped set. Wunsch--Zworski \cite{WunschZworskiNormHypResolvent,WunschZworskiNormHypResolventCorrection} showed that the trapped set of slowly rotating Kerr black holes is $k$-normally hyperbolic for every $k$ \cite{HirschPughShubInvariantManifolds}; this was later extended to the full subextremal range, and to KdS black holes with either small angular momentum or small cosmological constant by Dyatlov \cite{DyatlovWaveAsymptotics}. Moreover, \cite{WunschZworskiNormHypResolvent} provided microlocal semiclassical (i.e.\ high energy) estimates at the trapped set. Dyatlov subsequently devised a particularly elegant method \cite{DyatlovSpectralGaps} to prove semiclassical estimates at normally hyperbolic trapped sets; we will use \cite{DyatlovSpectralGaps} (rephrased as a propagation estimate as in \cite[Theorem~4.7]{HintzVasyQuasilinearKdS}) as a black box in the present paper. Dyatlov's method has since been extended to give estimates at the trapped set for waves on asymptotically Kerr(--de~Sitter) spacetimes \cite{HintzPolyTrap}.

\begin{rmk}[Further comments on trapping]
  An important conceptual feature of the analysis of the trapped set in \cite{DyatlovQNM} is that it is based solely on the dynamical structure of the trapped set (which is stable under perturbations \cite{HirschPughShubInvariantManifolds}), rather than on the separability of the wave equation. Using the separability, estimates at the trapped set of rotating Kerr spacetimes can be proved using rather explicit pseudodifferential multipliers, as shown by Tataru--Tohaneanu \cite{TataruTohaneanuKerrLocalEnergy}; see also \cite{DafermosRodnianskiKerrBoundedness,DafermosRodnianskiKerrDecaySmall} and the definitive \cite{DafermosRodnianskiShlapentokhRothmanDecay} for a very explicit approach of this nature. Andersson--Blue \cite{AnderssonBlueHiddenKerr} can avoid this issue altogether by exploiting a second order `hidden' symmetry operator which is closely related to the complete integrability of the geodesic flow on Kerr spacetimes.
\end{rmk}

Vasy's influential non-elliptic Fredholm theory \cite{VasyMicroKerrdS} provides a general framework for proving the discreteness of resonance spectra and for establishing resonance expansions of waves. This framework is fully microlocal, and makes use in particular of radial point estimates (originating in \cite{MelroseEuclideanSpectralTheory}) and real principal type propagation estimates \cite{DuistermaatHormanderFIO2}, together with high energy estimates in the presence of normally hyperbolic trapping. Without having to separate variables, \cite{VasyMicroKerrdS} recovers the results on exponential decay to constants proved in \cite{BonyHaefnerDecay,MelroseSaBarretoVasySdS,DyatlovQNMExtended}. A detailed account is given by Dyatlov--Zworski \cite{DyatlovZworskiBook}.

The absence of modes for the Klein--Gordon equation in $\Im\sigma\geq 0$ can be proved directly for all $\nu>0$ (in the notation of Theorem~\ref{ThmIKG}) in the case $\bha=0$. In the case of small $\bha/\bhm\neq 0$, it also follows for sufficiently small $\nu>0$ from a perturbative calculation off the massless KdS case (see \cite{DyatlovQNM} or \cite[Lemma~3.5]{HintzVasySemilinear}). We also note that Besset--H\"afner \cite{BessetHaefnerBomb} proved, by such perturbative means, the existence of exponentially growing modes for weakly charged and weakly massive scalar fields on slowly rotating Kerr--Newman--de~Sitter spacetimes.

\subsection{Prior work on mode stability}
\label{SsIMS}

We now turn to a discussion of the problem of mode stability for black hole spacetimes. Mode stability (for massless scalar waves) is a much weaker statement than Theorem~\ref{ThmI}, and by itself far from sufficient to obtain Corollary~\ref{CorI} (or even just boundedness of waves).\footnote{When combined with the Fredholm theory of \cite{PetersenVasySubextremal}, it does however imply the existence of a spectral gap, i.e.\ a small number $\alpha>0$ so that $0$ is the only resonance in $\Im(\Lambda^{-\frac12}\sigma)>-\alpha$; and this gives decay to constants, at the rate $e^{-\alpha t_*}$, of smooth linear waves.} It is, however, more amenable to direct investigations. Indeed, for $\bha=0$, mode stability can be proved via an integration by parts argument (when $\Im\sigma>0$) and a Wronskian (or boundary pairing) argument (when $\sigma\in\R\setminus\{0\}$); and the zero mode can be analyzed using an integration by parts argument as well. (Even for the linearized Einstein equation, an appropriate notion of mode stability for Schwarzschild--de~Sitter black holes can be proved with moderate effort, see e.g.\ \cite{KodamaIshibashiMaster} and \cite[\S7]{HintzVasyKdSStability}.) Given Vasy's general perturbation-stable framework \cite{VasyMicroKerrdS}, or using the arguments specific to the Kerr--de~Sitter metric by Dyatlov \cite[Theorem~4]{DyatlovQNM}, mode stability follows for the wave equation on KdS with parameters $(\Lambda,\bhm,\bha)$ provided $|\bha/\bhm|$ is sufficiently small.

For subextremal KdS black holes with $\bha\neq 0$, one can consider separated mode solutions
\[
  e^{-i\sigma t}e^{i m\phi}S(\theta)R(r),\qquad m\in\Z,
\]
where the angular function $S$ and the radial function $R$ solve decoupled ordinary differential equations (ODEs). Mode stability can then be proved for certain values $\sigma\in\R$ by means of Wronskian (or energy) arguments for the radial ODE. More precisely, this applies to $\sigma$ which are not superradiant (see \cite[\S1.6]{ShlapentokhRothmanModeStability} for this notion on Kerr spacetimes); the set of superradiant frequencies $\sigma\in\R$ is a nonempty (when $\bha\neq 0$ and $m\neq 0$) open interval centered roughly around $m\bha$. This argument also excludes resonances outside an appropriate subset of the upper half plane. (There are no superradiant modes when one restricts to axially symmetric mode solutions, i.e.\ $m=0$, so mode stability for \emph{axially symmetric} scalar perturbations holds true.) Casals--Teixeira da Costa \cite{CasalsTeixeiradCModes} exploit subtle discrete symmetries of the radial ODE, conjectured in \cite{AminovGrassiHatsudaQNM,HatsudaTeukolskyAlt}, to prove mode stability outside a smaller, but still always nonempty, subset of the closed upper half plane. Numerical evidence \cite{YoshidaUchikataFutamaseKdS,HatsudaKdSQNM} supports the conjecture that mode stability does hold in the full subextremal range.

By contrast with the Kerr--de~Sitter case, the mode stability of subextremal Kerr spacetimes is settled (and $0$ is not a resonance in this case). It was proved in $\Im\sigma>0$ by Whiting \cite{WhitingKerrModeStability} who used a carefully defined integral transform which maps the radial function $R$ to another function which satisfies an ODE for which Wronskian arguments can be applied successfully; Shlapentokh-Rothman \cite{ShlapentokhRothmanModeStability} showed that Whiting's transformation can be used to prove mode stability on the real axis.\footnote{The quantitative main result of \cite{ShlapentokhRothmanModeStability} was a key input in the proof of decay of solutions of the wave equation on subextremal Kerr spacetimes by Dafermos--Rodnianski--Shlapentokh-Rothman \cite{DafermosRodnianskiShlapentokhRothmanDecay}. The merely qualitative mode stability result is sufficient for this purpose as well if one uses it, in conjunction with strong (Fredholm and high energy) estimates for the spectral family, to exclude the presence of a nontrivial nullspace of the spectral family for $\Im\sigma\geq 0$; see \cite{HintzPrice} and also Propositions~\ref{PropKSyKerr}, \ref{PropKnfNz}, and \ref{PropKnfZ}.} Mode stability in $\Im\sigma\geq 0$ for the Teukolsky equation for other values of the spin $s\in\half\Z$ (with $s=0$ corresponding to the scalar wave equation) was subsequently proved by Andersson--Ma--Paganini--Whiting \cite{AnderssonMaPaganiniWhitingModeStab}. A different proof of these mode stability results, based on a discrete symmetry of the relevant confluent Heun equation, was given in \cite{CasalsTeixeiradCModes}.

\begin{thm}[Mode stability of subextremal Kerr black holes \cite{WhitingKerrModeStability,ShlapentokhRothmanModeStability}]
\label{ThmIKerr}
  Denote by $g_{\bhm,\bha}=g_{0,\bhm,\bha}$ the Kerr metric on $\R_{t_*}\times[r_{\bhm,\bha}^e,\infty)_r\times\Sph^2_{\theta,\phi_*}$, expressed in terms of the coordinates $t_*,\phi_*$ in~\eqref{EqIKerrStar} where we take $F_{0,\bhm,\bha}(r)$ to be equal to $-1$ near $r=r_{\bhm,\bha}^e$ and equal to $0$ for $r>2 r_{\bhm,\bha}^e$.\footnote{In these coordinates, $g_{\bhm,\bha}$ extends analytically down to, and across, the future event horizon $\cH^+_{\bhm,\bha}=\cH^+_{0,\bhm,\bha}$, with the level sets of $t_*$ being transversal to $\cH^+_{\bhm,\bha}$. See~\eqref{EqKKerr} for the explicit form of this metric when the black hole mass and (specific) angular momentum are $1$ and $\hat\bha$, respectively, and the function $F_{0,\bhm,\bha}$ is denoted $-\tilde\chi^e$.} Let $0\neq\sigma\in\C$, $\Im\sigma\geq 0$. Suppose $u(t_*,r,\theta,\phi_*)=e^{-i\sigma t_*}u_0(r,\theta,\phi_*)$ is a mode solution of $\Box_{g_{\bhm,\bha}}u=0$, where $u_0$ is smooth on $[r_{\bhm,\bha}^e,\infty)_r\times\Sph^2_{\theta,\phi_*}$, and so that $e^{-i\sigma r}r^{-2 i\bhm\sigma}u_0(r,\theta,\phi_*)=r^{-1}v_0(r^{-1},\theta,\phi_*)$ where $v_0=v_0(\rho,\theta,\phi_*)$ is smooth on $[0,1/r_{\bhm,\bha}^e)\times\Sph^2_{\theta,\phi_*}$. Then $u_0\equiv 0$ on $[r_{\bhm,\bha}^e,\infty)\times\Sph^2$.
\end{thm}

In the references, this is stated for fully separated mode solutions and in Boyer--Lindquist coordinates, with the incoming condition at the event horizon being equivalent to smoothness down to $r=r_{\bhm,\bha}^e$ in the $(t_*,r,\theta,\phi_*)$ coordinates. Theorem~\ref{ThmIKerr} as stated follows by applying \cite[Theorems~1.5 and 1.6]{ShlapentokhRothmanModeStability} to each individual separated component $e^{i m\phi}S(\theta)R(r)$ of $u_0$. The case of $\sigma=0$, in which the boundary condition at infinity becomes an appropriate decay requirement, is analyzed in Lemma~\ref{LemmaKz}.

Teixeira da Costa \cite{TeixeiradCModes} proved the mode stability of \emph{extremal} Kerr black holes, i.e.\ $|\bha|=\bhm$, using an appropriate integral transform---which due to the different character of the radial ODE, related to the presence of a degenerate event horizon, is substantially different from that introduced by Whiting. (The exceptional values $\sigma\in (2\bhm)^{-1}\N_0$ are not covered by this result.) See \cite[Theorem~1.2]{TeixeiradCModes}; see also Remark~\ref{RmkIExtr} for the relation between Teixeira da Costa's result and the topic of the present paper.

We remark that mode stability \emph{fails} for the Klein--Gordon equation on subextremal Kerr spacetimes for a large range of parameters, as shown by Shlapentokh-Rothman \cite{ShlapentokhRothmanBlackHoleBombs}. Moschidis \cite{MoschidisSuperradiant} proved a number of related mode instability results for deformations of the Kerr spacetime by means of potentials or metric deformations which either exhibit stable trapping or feature a non-Euclidean conic infinity. These results do not have a bearing on Theorem~\ref{ThmIKG} however, since the scalar field mass term vanishes in the appropriate Kerr limit. (In any case, depending on the value of $\nu$, Theorem~\ref{ThmIKG} implies mode stability or mode instability.)

A proof of mode stability for the scalar wave equation on Kerr--de~Sitter black holes (without restriction to axially symmetric modes), beyond the Schwarzschild--de~Sitter case and its small perturbations, has remained elusive, with all attempts so far having been based on integral transforms \cite{SuzukiTakasugiUmetsuKdS,UmetsuKdS} or discrete symmetries \cite{CasalsTeixeiradCModes}. The starting point for the present paper is the idea, substantiated in a simple special case in \cite{HintzXieSdS}, that subextremal KdS spacetimes with small $\Lambda\bhm^2$ can be regarded as singular perturbations of subextremal Kerr spacetimes and of de~Sitter space, and that one can extrapolate mode stability and the approximate values of quasinormal modes from these two singular limits. We explain this in some detail in~\S\ref{SsIA}.

\begin{rmk}[KdS mode stability in the full subextremal range]
  In the event that a direct proof (via an integral transform, discrete symmetries, or otherwise) of the conjectural mode stability of all subextremal KdS black holes should be found, the recent work by Petersen--Vasy \cite{PetersenVasySubextremal} would immediately imply exponential decay to constants of solutions of the wave equation. But even then, Theorem~\ref{ThmI} and Corollary~\ref{CorI} would give, in the regime in which they apply, significantly more precise information on the quasinormal mode spectrum which likely remains out of reach for any direct methods. We hope that the rather general singular perturbation perspective put forth in the present paper can be put to use in other settings involving spectral or resonance analysis in singular limits.
\end{rmk}

\subsection{Scaling}
\label{SsIS}

In order to reduce the number of parameters, we note:

\begin{lemma}[Scaling]
\label{LemmaIS}
  For $s>0$, let $M_s\colon(t_*,r,\theta,\phi_*)\mapsto(s t_*,s r,\theta,\phi_*)$. Then on the extended spacetime $\wt M_{\Lambda s^2,\bhm/s,\bha/s}$ (see~\eqref{EqIKdSExt}), we have
  \begin{equation}
  \label{EqISMetric}
    M_s^*g_{\Lambda,\bhm,\bha} = s^2 g_{\Lambda s^2,\bhm/s,\bha/s}
  \end{equation}
  In the notation of Theorems~\usref{ThmI} and \usref{ThmIKG}, we furthermore have
  \begin{equation}
  \label{EqISQNM}
  \begin{split}
    \QNM(\Lambda,\bhm,\bha) &= s^{-1}\QNM(\Lambda s^2,\bhm/s,\bha/s), \\
    \QNM(\nu;\Lambda,\bhm,\bha) &= s^{-1}\QNM(\nu;\Lambda s^2,\bhm/s,\bha/s).
  \end{split}
  \end{equation}
\end{lemma}
\begin{proof}
  The expressions~\eqref{EqImu} and \eqref{EqIMetric2} imply that
  \begin{equation}
  \label{EqISPf}
  \begin{alignedat}{2}
    &(M_s^*\mu_{\Lambda,\bhm,\bha})(r) = s^2\mu_{\Lambda s^2,\bhm/s,\bha/s}(r),&\qquad
    &(M_s^*\varrho^2_{\Lambda,\bhm,\bha})(r,\theta)=s^2\varrho^2_{\Lambda s^2,\bhm/s,\bha/s}(r,\theta), \\
    &M_s^*b_{\Lambda,\bhm,\bha}=b_{\Lambda,\bhm,\bha}=b_{\Lambda s^2,\bhm/s,\bha/s},&\qquad
    &(M_s^*c_{\Lambda,\bhm,\bha})(\theta)=c_{\Lambda,\bhm,\bha}(\theta)=c_{\Lambda s^2,\bhm/s,\bha/s}(\theta).
  \end{alignedat}
  \end{equation}
  Therefore, $r_{\Lambda,\bhm,\bha}^\bullet=s r_{\Lambda s^2,\bhm/s,\bha/s}^\bullet$ for $\bullet=-,C,e,c$. Plugged into~\eqref{EqIKerrStar} (with the choice of $F_{\Lambda,\bhm,\bha}$ made there), this gives
  \[
    (M_s^*(\pa_r T_{\Lambda,\bhm,\bha}))(r)=\pa_r T_{\Lambda s^2,\bhm/s,\bha/s}(r);
  \]
  since $M_s^*(s\pa_r)=\pa_r$, we can choose the constant of integration for $T_{\Lambda,\bhm,\bha}$ so that
  \[
    (M_s^* T_{\Lambda,\bhm,\bha})(r) = s T_{\Lambda s^2,\bhm/s,\bha/s}(r).
  \]
  We can similarly arrange $(M_s^*\Phi_{\Lambda,\bhm,\bha})(r)=s\Phi_{\Lambda s^2,\bhm/s,\bha/s}(r)$. We conclude that $M_s$ takes the form $(t,r,\theta,\phi)\mapsto(s t,s r,\theta,\phi)$ in Boyer--Lindquist coordinates. The claim~\eqref{EqISMetric} then follows on $M_{\Lambda,\bhm,\bha}^{\rm DOC}$ from~\eqref{EqISPf} and the explicit form~\eqref{EqIMetric} of $g_{\Lambda,\bhm,\bha}$. On the extended manifold $\wt M_{\Lambda s^2,\bhm/s,\bha/s}$, the equality~\eqref{EqISMetric} follows by analytic continuation, or directly by inspection of the explicit form~\eqref{EqKExt} of the metric in $(t_*,r,\theta,\phi_*)$ coordinates.

  As a consequence of~\eqref{EqISMetric}, pulling back along $M_s^*$ or $M_{1/s}^*$ proves the equivalence
  \[
    \Bigl(\Box_{g_{\Lambda,\bhm,\bha}}-\frac{\Lambda}{3}\nu\Bigr)\bigl(e^{-i\sigma t_*}u(r,\theta,\phi_*)\bigr)=0 \Leftrightarrow
    \Bigl(\Box_{g_{\Lambda s^2,\bhm/s,\bha/s}}-\frac{\Lambda s^2}{3}\nu\Bigr)\bigl(e^{-i(s \sigma)t_*}u(s r,\theta,\phi_*)\bigr)=0.
  \]
  Thus, $\sigma\in\QNM(\nu;\Lambda,\bhm,\bha)$ if and only if $s\sigma\in\QNM(\nu;\Lambda s^2,\bhm/s,\bha/s)$. This implies~\eqref{EqISQNM} and finishes the proof.
\end{proof}

It thus suffices to consider the first asymptotic regime mentioned after Corollary~\ref{CorIMS}. Concretely, we take $s=\sqrt{3/\Lambda}$ in Lemma~\ref{LemmaIS}, and henceforth work with
\[
  \Lambda=3.
\]

\subsection{Singular limits and asymptotic regimes}
\label{SsIA}

We now describe a few elements of the proof of Theorem~\ref{ThmI}. Let us fix $\Lambda=3$, and fix also the ratio $\bha/\bhm=\hat\bha\in(-1,1)$; thus, in this section we exclusively work with Kerr--de~Sitter metrics
\[
  g_{\Lambda,\bhm,\bha} = g_{3,\bhm,\hat\bha\bhm},
\]
and we are interested in the limit $\bhm\searrow 0$. For notational simplicity, we work with Boyer--Lindquist coordinates here, and we restrict our attention to frequencies $\sigma$ which lie in a strip rather than a half space; thus, $\Im\sigma$ is bounded, but $\Re\sigma$ is unbounded.

For fixed $r>0$, the Kerr--de~Sitter metric $g_{\Lambda,\bhm,\bha}=g_{3,\bhm,\hat\bha,\bhm}$ in~\eqref{EqIMetric} converges, as the black hole mass tends to $0$ (i.e.\ the black hole `disappears'), to the de~Sitter metric
\[
  g_\dS = -(1-r^2)\dd t^2 + \frac{1}{1-r^2}\dd r^2 + r^2\slg,\qquad \slg=\dd\theta^2+\sin^2\theta\,\dd\phi^2.
\]
This metric is singular at the cosmological horizon $r=1$, but a coordinate change similar to~\eqref{EqIKerrStar} shows that this is merely a coordinate singularity (see~\eqref{EqKdS}). Moreover, $g_\dS$ is the expression in polar coordinates $(r,\theta,\phi)$ of a metric on $\R_t\times B(0,1)$, where $B(0,1):=\{x\in\R^3\colon r=|x|<1\}$, which is smooth across $x=0$. One can then define resonances and mode solutions for $\Box_{g_\dS}$ as in the Kerr--de~Sitter setting explained before Theorem~\ref{ThmI}; the set of quasinormal modes of $\Box_{g_\dS}$ (which are known explicitly, see Lemma~\ref{LemmaKdSQNM}) is then precisely the limit of $\QNM(3,\bhm,\hat\bha\bhm)$ as $\bhm\searrow 0$ in Theorem~\ref{ThmI}.

Now, $g_{\Lambda,\bhm,\bha}$ does not converge smoothly to $g_\dS$. Rather, in rescaled coordinates
\[
  \hat t = \frac{t}{\bhm},\qquad
  \hat r = \frac{r}{\bhm},
\]
the rescaled metric $\bhm^{-2}g_{\Lambda,\bhm,\bha}$ converges, for fixed $\hat r>0$ and as $\bhm\searrow 0$, to the metric
\begin{align*}
  &\hat g = -\frac{\hat\mu(\hat r)}{\hat\varrho^2(r,\theta)}\bigl(\dd\hat t{-}\hat\bha\,\sin^2\theta\,\dd\phi\bigr)^2 + \frac{\hat\varrho^2(r,\theta)}{\hat\mu(r)}\dd\hat r^2 + \hat\varrho^2(\hat r,\theta)\dd\theta^2 + \frac{\sin^2\theta}{\hat\varrho^2(\hat r,\theta)}\bigl((\hat r^2{+}\hat\bha^2)\dd\phi{-}\hat\bha\,\dd\hat t\bigr)^2, \\
  &\qquad \hat\mu(\hat r):=\hat r^2-2\hat r+\hat\bha^2,\qquad
    \hat\varrho^2(\hat r,\theta):=\hat r^2+\hat\bha^2\cos^2\theta,
\end{align*}
of a Kerr black hole with mass $1$ and angular momentum $\hat\bha$. Note the relationship
\begin{equation}
\label{EqIASigma}
  e^{-i\sigma t}=e^{-i\tilde\sigma\hat t},\qquad \tilde\sigma=\bhm\sigma,
\end{equation}
between frequencies on the KdS spacetime and frequencies for the rescaled observer on the Kerr spacetime. Thus, $\tilde\sigma$ is small compared to $\sigma$ when $\bhm>0$ is small; but since $|\sigma|$ itself may be large, the rescaled frequency $\tilde\sigma$ may nonetheless be large too---or not, depending on the relative size of $|\sigma|$ and $\bhm^{-1}$.

\begin{rmk}[Simple model]
\label{RmkISimple}
  An operator on $(2\bhm,2)_r\times\Sph^1_\theta$ that the reader may keep in mind in the subsequent discussion is
  \[
    P_\bhm(\sigma):=\Bigl(1-\frac{2\bhm}{r}-r^2\Bigr)D_r^2+r^{-2}D_\theta^2-\sigma^2,\qquad D=\frac{1}{i}\pa.
  \]
  (This is a poor approximation of the spectral family of the Schwarzschild--de~Sitter wave operator.) The two singular limits as $\bhm\searrow 0$ are
  \begin{equation}
  \label{EqISimpleScale}
  \begin{alignedat}{3}
    P_\bhm(\sigma) &\to P_0(\sigma)&&=(1-r^2)D_r^2+r^{-2}D_\theta^2-\sigma^2,&\qquad& r\simeq 1, \\
    \bhm^2 P_\bhm(\sigma) &\to \hat P(\tilde\sigma)&&=\Bigl(1-\frac{2}{\hat r}\Bigr)D_{\hat r}^2+\hat r^{-2}D_\theta^2-\tilde\sigma^2,&\qquad& \hat r\simeq 1,\quad \tilde\sigma=\lim_{\bhm\searrow 0} \bhm\sigma.
  \end{alignedat}
  \end{equation}
  (In the second line, $\sigma$ may vary with $\bhm$.) Here, $P_0(\sigma)$ plays the role of the de~Sitter model, and $\hat P(\tilde\sigma)$ that of the Kerr model.
\end{rmk}

We now list the different frequency regimes for $\sigma$ and $\tilde\sigma$ as $\bhm\searrow 0$, together with a brief description of the two limiting problems that one needs to study in each regime.

\begin{enumerate}
\item \textit{Bounded frequencies.} $\sigma$ remains bounded as $\bhm\searrow 0$: the spectral theory for de~Sitter space for bounded frequencies enters---and thus the de~Sitter quasinormal mode spectrum---but the Kerr wave operator enters only at frequency $\tilde\sigma=0$ by~\eqref{EqIASigma}.
\item\label{ItILarge} \textit{Large frequencies.} $1\ll|\Re\sigma|\ll\bhm^{-1}$, i.e.\ $\sigma$ is large but remains small compared to $\bhm^{-1}$: this involves high energy (semiclassical) analysis on de~Sitter space---where there are no quasinormal modes---and low (i.e.\ near zero) frequency analysis for the Kerr wave operator. From this point onwards, we are in the high frequency regime from the perspective of the de~Sitter limit.
\item\label{ItIVLarge} \textit{Very large frequencies.} $|\Re\sigma|$ is comparable to $\bhm^{-1}$: in this case, $\tilde\sigma=\bhm\sigma$ is, in the limit $\bhm\searrow 0$, of unit size but real. Thus, we are in a bounded real frequency regime for the Kerr wave operator. Excluding the possibility of KdS resonances in this regime thus requires as an input the absence of modes on the real axis for the Kerr wave operator (Theorem~\ref{ThmIKerr}).
\item\label{ItIELarge} \textit{Extremely large frequencies.} Finally, when $|\Re\sigma|$ is large compared to $\bhm^{-1}$ as $\bhm\searrow 0$, then we are in a high (real) frequency regime ($|\tilde\sigma|=|\bhm\sigma|\gg 1$) also from the perspective of the Kerr model. In this case, the absence of Kerr modes follows directly using semiclassical methods.
\end{enumerate}

More concretely then, in the bounded frequency regime, the uniform analysis of the spectral family $\Box_{g_{\Lambda,\bhm,\bha}}(\sigma)=e^{i t\sigma}\Box_{g_{\Lambda,\bhm,\bha}}e^{-i t\sigma}$ (acting on functions of the spatial variables only) takes place on function spaces which incorporate the two different spatial limiting regimes: for $\hat r\simeq 1$, we measure regularity with respect to $\pa_{\hat r}$, $\pa_\omega$ (spherical derivatives), and for $r\simeq 1$ with respect to $\pa_r$, $\pa_\omega$; put differently, writing $x\in\R^3$ for spatial coordinates on de~Sitter space, we use $\pa_{\hat x}=\bhm\pa_x$ (where $\hat x=\frac{x}{\bhm}$) for bounded $|\hat x|$, and $\pa_x$ when $|x|\simeq 1$. (In the region $\hat r\gtrsim 1$, the vector fields $r\pa_r$, $\pa_\omega$ work in both regimes simultaneously.) This is conveniently phrased on a geometric resolution (blow-up) of the total space $[0,1]_\bhm\times B(0,1)_x$ in which one introduces polar coordinates around $(\bhm,x)=(0,0)$, see Figure~\ref{FigIqSingle}.

\begin{figure}[!ht]
\centering
\includegraphics{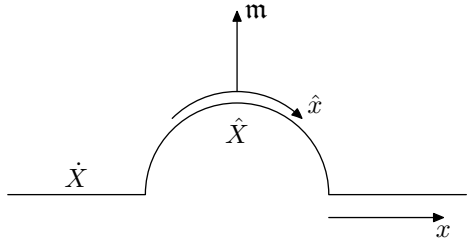}
\caption{The total space for analysis at bounded frequencies.}
\label{FigIqSingle}
\end{figure}

We call this total space the \emph{q-single space $X_\qop$ of $X=B(0,1)$}, and refer to the corresponding scale of function spaces as (weighted) \emph{q-Sobolev spaces} $H_{\qop,\bhm}^{s,l,\gamma}$: these are spaces of functions of the spatial variables, and indeed equal to $H^s$ as a set, but with norms that degenerate in a specific manner as $\bhm\searrow 0$. For functions supported in $\hat r\gtrsim 1$, the $\bhm$-dependent norm on $H_{\qop,\bhm}^{s,l,\gamma}$ for integer $s$ is given by
\[
  \|u\|_{H_{\qop,\bhm}^{s,l,\gamma}}^2 = \sum_{j+|\alpha|\leq s} \Bigl\| r^{-l}\Bigl(\frac{\bhm}{r}\Bigr)^{-\gamma} (r D_r)^j D_\omega^\alpha u \Bigr\|_{L^2}^2,
\]
where $L^2$ is the standard $L^2$-norm on $X$. The algebra of q-(pseudo)differential operators is described in detail in~\S\ref{Ssq}; it is a close relative of the surgery calculus of McDonald \cite{McDonaldThesis} and Mazzeo--Melrose \cite{MazzeoMelroseSurgery}, see Remark~\ref{RmkqComp}.

The proof of Theorem~\ref{ThmI} for bounded $\sigma$ uses a priori estimates on q-Sobolev spaces for $u$ in terms of $\Box_{g_{\Lambda,\bhm,\bha}}(\sigma)u$, with constants that are uniform as $\bhm\searrow 0$. These estimates are based on three ingredients.

\begin{enumerate}
\item \textit{Symbolic analysis: elliptic regularity, radial point estimates, microlocal propagation of regularity.} This is a direct translation to the q-calculus of the corresponding estimates introduced in the black hole setting by Vasy \cite{VasyMicroKerrdS}; by design, these q-estimates are uniform in $\bhm$. They take the form
  \begin{equation}
  \label{EqIEst0}
    \|u\|_{H_{\qop,\bhm}^{s,l,\gamma}} \leq C \bigl( \|\Box_{g_{\Lambda,\bhm,\bha}}(\sigma)u \|_{H_{\qop,\bhm}^{s-1,l-2,\gamma}} + \|u\|_{H_{\qop,\bhm}^{s_0,l,\gamma}} \bigr),\qquad s_0<s;
  \end{equation}
  that is, symbolic (or principal symbol) arguments control $u$ to leading order in the q-differentiability sense. The differential order $s-1$ on $\Box_{g_{\Lambda,\bhm,\bha}}(\sigma)u$ reflects the usual loss of one derivative in radial point or hyperbolic propagation estimates. The shift of $-2$ in the weight $l-2$ reflects the scaling near the Kerr regime $\hat r\simeq 1$, cf.\ \eqref{EqISimpleScale}.
\item \textit{Estimates for the Kerr model problem.} This is a quantitative estimate for a function $v$ on $\hat X$ (i.e.\ expressed in the rescaled coordinates $\hat x$) in terms of the zero energy operator $\Box_{\hat g}(0)$ applied to $v$. Apart from involving symbolic estimates as before, such an estimate involves analysis at spatial infinity, where the operator $\Box_{\hat g}(0)$ is an elliptic element of Melrose's b-algebra \cite{MelroseTransformation,MelroseAPS}. Applying this estimate to the error term $\|u\|_{H_{\qop,\bhm}^{s_0,l,\gamma}}$ in~\eqref{EqIEst0} (cut off to a neighborhood of $\hat X$ in Figure~\ref{FigIqSingle}) and noting that $\Box_{\hat g}(0)$ and $\bhm^2\Box_{g_{\Lambda,\bhm,\bha}}(\sigma)$ differ by an operator whose coefficients vanish to leading order at $\bhm=0$ for bounded $\hat r$, this gives the improved estimate
  \begin{equation}
  \label{EqIEst1}
    \|u\|_{H_{\qop,\bhm}^{s,l,\gamma}} \leq C \bigl( \|\Box_{g_{\Lambda,\bhm,\bha}}(\sigma)u \|_{H_{\qop,\bhm}^{s-1,l-2,\gamma}} + \|u\|_{H_{\qop,\bhm}^{s_0,l_0,\gamma}} \bigr),\qquad s_0<s,\ l_0<l.
  \end{equation}

\item \textit{Estimates for the de~Sitter model problem.} This is a quantitative estimate for a function $v$ on $\dot X$ (see Figure~\ref{FigIqSingle}) in terms of $\Box_{g_\dS}(\sigma)v$ where $\Box_{g_\dS}(\sigma)$ is the spectral family of de~Sitter space. The caveat here is that the singular limit $\bhm\searrow 0$ leaves a mark not just geometrically (as in Figure~\ref{FigIqSingle}) but also analytically, in that the point $x=0$ is blown up, and q-Sobolev spaces involve a choice of weight at $r=0$. Indeed, in the near-de Sitter region $\bhm\lesssim r$, q-Sobolev spaces are cone Sobolev spaces (i.e.\ weighted b-Sobolev spaces) with cone point at $r=0$, and for appropriate weights one has elliptic estimates at the cone point. (This issue was already addressed in a simple setting in \cite[\S2.1]{HintzXieSdS}.) Thus, if $\sigma$ is not a de Sitter quasinormal mode, one can apply this quantitative estimate to the error term in~\eqref{EqIEst1} and thereby weaken the error term to\footnote{While not apparent from this sketch, careful accounting of the orders required to apply the two model operator estimates, and of the q-regularity of the error term, shows that the symbolic analysis is indeed necessary in order to get an error term with differential order $\leq s$ here.} $C\|u\|_{H_{\qop,\bhm}^{s_0,l_0,\gamma_0}}$ where $\gamma_0<\gamma$. But this is bounded by $C\bhm^\delta\|u\|_{H_{\qop,\bhm}^{s,l_0+\delta,\gamma_0+\delta}}$ where $0<\delta\leq\min(l-l_0,\gamma-\gamma_0)$, and hence \emph{small} compared to $\|u\|_{H_{\qop,\bhm}^{s,l,\gamma}}$ when $\bhm$ is small. Therefore, we obtain a uniform estimate
  \[
    \|u\|_{H_{\qop,\bhm}^{s,l,\gamma}} \leq C\|\Box_{g_{\Lambda,\bhm,\bha}}(\sigma)u\|_{H_{\qop,\bhm}^{s-1,l-2,\gamma}}
  \]
  for all sufficiently small $\bhm$, and for bounded $\sigma$ which are at most at a fixed small distance away from de Sitter quasinormal modes. See Proposition~\ref{PropKBdNo}. A Grushin problem setup together with Rouch\'e's theorem takes care of KdS quasinormal modes \emph{near} de~Sitter quasinormal modes.
\end{enumerate}

\begin{rmk}[Comparison with \cite{HintzXieSdS}]
  The work \cite{HintzXieSdS} demonstrated how on the spherically symmetric Schwarzschild--de~Sitter spacetime, and after separation into spherical harmonics, uniform estimates for a degenerating family of ordinary differential equations in the radial variable imply Theorem~\ref{ThmI} for bounded spectral parameters and for fixed spherical harmonic degrees. In the present paper, we adopt a point of view based fully on the analysis of \emph{partial} differential operators; the part of the proof concerned with bounded frequencies is conceptually very similar to \cite[\S3]{HintzXieSdS}, except now the uniform estimates are proved using microlocal means, as described above. The remaining three frequency regimes~\eqref{ItILarge}--\eqref{ItIELarge} are not covered by \cite{HintzXieSdS}.
\end{rmk}

The large frequency regime~\eqref{ItILarge} is the most delicate one. From the perspective of de~Sitter space, uniform analysis away from the cone point utilizes semiclassical Sobolev spaces (i.e.\ measuring regularity with respect to $h\pa_x$ for $|x|\simeq 1$ where $h=|\sigma|^{-1}$), but there is now an artificial conic point at $r=0$ through which we need to propagate semiclassical estimates (along null-bicharacteristics which hit the cone point or emanate from it). We do this by adapting the semiclassical propagation estimates which were proved in \cite{HintzConicProp} by means of the semiclassical cone calculus introduced in \cite{HintzConicPowers}: this involves radial point estimates at incoming and outgoing radial sets over the cone point, and estimates for a model operator on an exact Euclidean cone which here is the spectral family of the Laplacian at frequency $1$ (i.e.\ on the spectrum). In terms of the model of Remark~\ref{RmkISimple}, we are considering $h^2 P_0(h^{-1})=(1-r^2)(h D_r)^2+r^{-2}(h D_\theta)^2-1$, and the model operator arises by passing to $\tilde r:=r/h$ and taking the limit $h\searrow 0$ for bounded $\tilde r$, giving $D_{\tilde r}^2+\tilde r^{-2}D_\theta^2-1$. (We refer the reader to \cite{MelroseWunschConic,MelroseVasyWunschEdge,MelroseVasyWunschDiffraction,XiConeParametrix,YangDiffraction} for further results on propagation through, and diffraction by, conic singularities.)

From the perspective of the rescaled Kerr model on the other hand, the large frequency regime~\eqref{ItILarge} puts us into a regime of \emph{low} frequencies $\tilde\sigma$, and we need to prove uniform estimates for the spectral family $\Box_{\hat g}(\tilde\sigma)$ for real $\tilde\sigma$ near $0$. Uniform estimates for low energy resolvents on asymptotically flat spaces or spacetimes have a long history going back to work by Jensen--Kato \cite{JensenKatoResolvent}, with recent contributions including \cite{GuillarmouHassellResI,GuillarmouHassellResII,GuillarmouHassellSikoraResIII,BonyHaefnerResolvent,DonningerSchlagSofferPrice,DonningerSchlagSofferSchwarzschild,TataruDecayAsympFlat,VasyLowEnergy,VasyLowEnergyLag,HintzPrice,StrohmaierWatersHodge,MorganDecay,MorganWunschPrice}. Here, we use an approach that matches up exactly with the semiclassical cone analysis on the de~Sitter side: we work with function spaces (and a corresponding ps.d.o.\ algebra which we call the \emph{scattering-b-transition algebra}---see~\S\ref{SsPscbt}---which is taken directly from \cite{GuillarmouHassellResI} except for different terminology) which resolve the transition from the (elliptic) b-analysis at zero frequency to (non-elliptic) scattering theory (in the spirit of \cite{MelroseEuclideanSpectralTheory}) at nonzero frequencies. The \emph{same} model operator as above (conic Laplacian at frequency $1$) now captures the transition from zero to nonzero energies for the low energy spectral family of the Kerr wave operator. This is less precise than, but technically simpler than the very precise second microlocal approach introduced recently by Vasy \cite{VasyLowEnergyLag}. In terms of the model of Remark~\ref{RmkISimple}, we pass to $\hat\rho=\hat r^{-1}$ in order to work at spatial infinity, so $\tilde\sigma^{-2}\hat P(\tilde\sigma)=(1-2\hat\rho)\tilde\sigma^{-2}(\hat\rho^2 D_{\hat\rho})^2+\hat\rho^2\tilde\sigma^{-2}D_\theta^2-1$, then introduce $\tilde\rho=\hat\rho/\tilde\sigma$, and pass to the limit $\tilde\sigma\searrow 0$ for bounded $\tilde\rho$; this produces $(\tilde\rho^2 D_{\tilde\rho})^2+\tilde\rho^2 D_\theta^2-1$. Upon identifying $\tilde\rho=\tilde r^{-1}$, this is the same operator as the one arising from the high frequency cone point perspective above.

On the level of estimates, we combine symbolic estimates and estimates for the two model spectral families by means of an appropriate family of $(\bhm,\sigma)$-dependent \emph{Q-Sobolev norms} which reduce to semiclassical cone Sobolev norms in the high energy de~Sitter regime, and to scattering-b-transition Sobolev norms in the low energy Kerr regime. Concretely, an integer order norm with these properties is
\begin{align*}
  &\|u\|_{H_{\Qop,\bhm,\sigma}^{s,(l,\gamma,l',\sfr)}}^2 = \sum_{j+|\alpha|\leq s} \Bigl\| r^{-l}\Bigl(\frac{\bhm}{r}\Bigr)^{-\gamma}(h+r)^{-l'+l}\Bigl(\frac{h}{h+r}\Bigr)^{-\sfr+\gamma}\Bigl(\frac{h}{h+r}r D_r\Bigr)^j\Bigl(\frac{h}{h+r}D_\omega\Bigr)^\alpha u \Bigr\|_{L^2}^2, \\
  &\qquad h:=|\sigma|^{-1}\in(\bhm,1],
\end{align*}
for $u$ with support in $r\gtrsim\bhm$. For fixed $\bhm>0$ and $\sigma$, this is equivalent to the $H^s$-norm, but it degenerates in the correct manner as $\bhm\searrow 0$. (In the main part of the paper, such \emph{weighted Q-Sobolev norms} have an extra order, denoted $b$, which however does not matter outside the extremely high frequency regime. Moreover, the order $\sfr$ will be variable to accommodate incoming and outgoing radial point estimates.)

Next, in the very large frequency regime~\eqref{ItIVLarge}, we are now, from the de~Sitter perspective, fully in a semiclassical regime. The symbolic propagation through the conic singularity again follows \cite{HintzConicProp}, but the model problem at the cone point is now the spectral family of the Kerr wave operator at bounded nonzero real frequencies. Estimates for the latter are limiting absorption principle type estimates; they are proved as in \cite{MelroseEuclideanSpectralTheory} up to compact error terms, and removing these error terms precisely requires the mode stability for the Kerr spacetime \cite{ShlapentokhRothmanModeStability}. (This is reminiscent of propagation results for $3$- or $N$-body scattering \cite{VasyThreeBody,VasyManyBody}, where microlocal propagation of decay through collision planes requires the invertibility of a spectral problem for a subsystem.)

In the extremely large frequency regime~\eqref{ItIELarge} finally, we can use semiclassical methods also for the spectral family on the Kerr spacetime (and therefore the absence of extremely large frequency quasinormal modes can be proved entirely using symbolic means). Here, the full null-geodesic dynamics of the Kerr spacetime enter; this is described in detail in \cite{DyatlovWaveAsymptotics}, and we can use this and the relevant microlocal propagation results, in particular at the trapped set \cite{DyatlovSpectralGaps}, as black boxes.

While the analysis of bounded frequencies is done separately (see~\S\ref{SsKBd}), the analysis of all three high frequency regimes is phrased in terms of the single aforementioned family of weighted Q-Sobolev spaces. These capture regularity with respect to a Lie algebra of vector fields adapted to each of the regimes discussed. We adopt a fully geometric microlocal point of view and describe this underlying Lie algebra of \emph{Q-vector fields} on a suitable total space (a resolution of $\ol{\R_\sigma}\times[0,1]_\bhm\times B(0,1)$ where $\ol\R=\R\cup\{-\infty,+\infty\}$); the full spectral family $(\sigma,\bhm)\mapsto\Box_{g_{3,\bhm,\hat\bha\bhm}}(\sigma)$ is then (for fixed $\Im\sigma$) a \emph{single} element of a corresponding space of Q-differential operators. Its microlocal analysis is accomplished by means of an algebra of Q-pseudodifferential operators. Q-geometry and Q-analysis are developed in detail in~\S\ref{SQ}.

\begin{rmk}[Separation of variables]
  It is conceivable that one can prove Theorem~\ref{ThmI} by starting with Carter's separation of variables \cite{CarterHamiltonJacobiEinstein} and extending the ODE techniques introduced in \cite{HintzXieSdS} to keep track of uniformity in half spaces $\Im\sigma>-C$ and also in the parameters $(\ell,m)$ of the spheroidal harmonics (generalizing the usual parameters $\ell\in\N_0$ and $m\in\Z\cap[-\ell,\ell]$ of spherical harmonics); we shall not pursue this possibility here. We merely note that this approach would introduce yet another large parameter ($|\ell|+|m|\to\infty$). Elements of the low frequency analysis for the Kerr model in the case $\hat\bha=0$ are developed from a separation of variables point of view in \cite{DonningerSchlagSofferPrice,DonningerSchlagSofferSchwarzschild}. 
\end{rmk}

\begin{rmk}[Mode stability in the full subextremal range]
\label{RmkIExtr}
  For simplicity of notation, fix the black hole mass to be $1$, and consider a sequence $(\Lambda_j,1,\bha_j)$ of subextremal KdS parameters with $\Lambda_j\searrow 0$, $|\bha_j|<1$. Then the limiting Kerr parameters $(1,\bha)$, $\bha=\lim\bha_j$, may be extremal. While the mode stability of extremal Kerr black holes is known \cite{TeixeiradCModes} (with the exceptional frequencies requiring separate treatment), there do not exist any estimates yet on the spectral family on an extremal Kerr spacetime (in any frequency regime) which could take the place of the estimates on the subextremal Kerr spectral family used above. \emph{If} such estimates were available, one could likely generalize Theorem~\ref{ThmI} to all subextremal KdS black holes (possibly even including the extremal case) when $\Lambda\bhm^2$ is sufficiently small; at present, this is out of reach however.
\end{rmk}

The analytic framework introduced in this paper is very flexible. In particular, it can be generalized in a straightforward manner to degenerating families of operators acting on sections of vector bundles. In particular, for the Teukolsky equation on Kerr--de~Sitter spacetimes, we expect an analogue of Theorem~\ref{ThmI} to hold; this would be an important step towards an unconditional proof of the nonlinear stability of Kerr--de~Sitter black holes without restriction to small angular momenta. (The case of small angular momenta was treated in \cite{HintzVasyKdSStability}.) Furthermore, other singular limits with similar scaling behavior can be analyzed using the same approach. As a simple (albeit contrived) example, the operator
\[
  \Box_{g_\dS}+\bhm^{-2}V(x/\bhm),
\]
where $V\in\CIc(\R^3)$ (or more generally with inverse cubic decay), fits into our framework: the analogue of the de Sitter model is now simply the spectral family of $\Box_{g_\dS}$, while the analogue of the Kerr model is $\Delta_{\hat x}-\sigma^2+V(\hat x)$, i.e.\ the spectral family of the Schr\"odinger operator $\Delta+V$ on $\R^3_{\hat x}$. Thus, if $\Delta+V$ has no resonances in the closed upper half plane, then the resonances of $\Box_{g_\dS}+\bhm^{-2}V(x/\bhm)$ have the same description as in Theorem~\ref{ThmI}. (Note that separation of variables is not available at all for this operator when $V$ has no symmetries.)
  
On the other hand, if the Kerr model of the equation under study has zero energy resonances or bound states---as is the case for the Maxwell equations \cite{SterbenzTataruMaxwellSchwarzschild,AnderssonBlueMaxwellKerr} or the equations of linearized gravity \cite{AnderssonBackdahlBlueMaKerr,HaefnerHintzVasyKerr}---the bounded frequency analysis sketched above fails. It is an interesting open problem to analyze the limiting behavior of KdS quasinormal modes in this case.

\subsection{Outline of the paper}
\label{SsIO}

The technical heart of the paper is~\S\ref{SQ}. We first discuss in detail the geometric and analytic tools (\emph{q-analysis}) which we will use for the uniform analysis at bounded frequencies---see~\S\ref{Ssq}---before describing the appropriate large frequency generalization (\emph{Q-analysis}) in~\S\S\ref{SsQS}--\ref{SsQH}. The main result of the paper, Theorem~\ref{ThmK}, is set up in~\S\S\ref{SsKL}--\ref{SsKMain}. After placing the full spectral family of a degenerating family of Kerr--de~Sitter spacetimes into the framework of Q-analysis in~\S\ref{SsKS}, the proof of Theorem~\ref{ThmK} occupies~\S\S\ref{SsKSy}--\ref{SsKU}, with \S\ref{SsKU} describing the modifications necessary to treat resonances in a full half space (rather than merely in strips, as described in~\S\ref{SsIA}). The proof of Theorem~\ref{ThmIKG} does not require any further work, and is given in~\S\ref{SsKG}.

Appendix~\ref{SP} reviews elements of geometric singular analysis and recalls the various pseudodifferential algebras (the b-, scattering, semiclassical scattering, semiclassical cone, and scattering-b-transition algebras) that are used in the analysis of the model problems discussed in~\S\ref{SsIA}. Appendix~\ref{SQSemi} contains supplementary material for \S\ref{SsQP}; this is included for conceptual completeness, but it is not used in the proofs of the main results.

\subsection*{Acknowledgments}

I am very grateful to Simone Ferraro for an inspiring conversation during our time as Miller Research Fellows at UC Berkeley which spawned the idea for the present work (and for the earlier \cite{HintzXieSdS}). I am grateful to Andr\'as Vasy for discussions about his work \cite{PetersenVasySubextremal} with Oliver Lindblad Petersen, which prompted the writing of this paper. Thanks are also due to Dietrich H\"afner and Andr\'as Vasy who shared with me their unpublished manuscript \cite{HaefnerVasyKerrUnfinished}, as well as to Maciej Zworski for encouragement and support. This research is supported by the U.S.\ National Science Foundation under Grant No.\ DMS-1955614, and by a Sloan Research Fellowship.

\section{Geometric and analytic setup of the singular limit}
\label{SQ}

Let us fix an $n$-dimensional manifold $X$ without boundary, and fix a point $0\in X$ and local coordinates $x\in B(0,2)=\{x\in\R^n\colon|x|<2\}$ so that $x=0$ at the point $0$. (All constructions presented below go through whether $X$ is compact or not. The main case of interest in this paper is when $X\subset\R^3$ is the spatial part of the de~Sitter manifold. For compact $X$ the discussion of function spaces is slightly simplified.)

We first describe somewhat briefly the geometric and analytic setup for the degenerate limit for fixed frequencies in~\S\ref{Ssq}; we call this q-analysis. The geometric setup for uniform analysis across all frequency regimes is then discussed in detail in~\S\S\ref{SsQS}--\ref{SsQP}; we call this Q-analysis. (The letters `q' and `Q' stand for `quasinormal modes'.) We freely make use of the material in Appendix~\ref{SP}.

\subsection{q-geometry and -analysis}
\label{Ssq}

When, in the context of Theorems~\ref{ThmI} and \ref{ThmIKG}, the frequency $\sigma$ is fixed, the following space captures the geometric degeneration of the spacetime as $\bhm\to 0$.

\begin{definition}[q-single space]
\label{DefqSingle}
  The \emph{q-single space} of $X$ is the resolution $X_\qop$ of $[0,1]_\bhm\times X$ defined as the blow-up
  \[
    X_\qop := \bigl[ [0,1]\times X; \{0\}\times\{0\} \bigr].
  \]
  We denote by $\zface_\qop$ the front face, and by $\mface_\qop$ the lift of $\{0\}\times X$. We write $\rho_{\zface_\qop},\rho_{\mface_\qop}\in\CI(X_\qop)$ for defining functions of these two boundary hypersurfaces.
\end{definition}

See Figure~\ref{FigqSingle}. Our interest will be in uniform analysis as $\bhm\searrow 0$; thus, one may as well replace $[0,1]$ by any other interval $[0,\bhm_0]$ with $\bhm_0>0$. We work with a closed interval of values of $\bhm$ since it will be convenient to keep all parameter spaces compact.

\begin{figure}[!ht]
\centering
\includegraphics{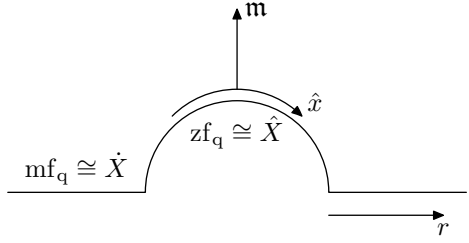}
\caption{The q-single space $X_\qop$ when $\dim X=1$.}
\label{FigqSingle}
\end{figure}

\begin{rmk}[q-analysis and analytic surgery]
\label{RmkqComp}
  In the case that $X$ is 1-dimensional, the set $\{0\}\subset X$ is a hypersurface, and $X_\qop$ is equal to the \emph{single surgery space} defined in \cite{MazzeoMelroseSurgery}; this was first introduced by McDonald \cite{McDonaldThesis}. For higher-dimensional $X$, the single surgery space is defined via blow-up of a hypersurface of $X$, rather than a point as in the q-single space above. However, much of the discussion of the geometry, Lie algebra of vector fields, and pseudodifferential calculus carries over from \cite[\S\S3--4]{MazzeoMelroseSurgery} to the q-setting with minor changes. We shall nonetheless give a self-contained account here to fix the notation and to facilitate the subsequent generalization to the Q-calculus.
\end{rmk}

We denote by $\bhm$ the lift of the first coordinate on $[0,1]\times X$ to $X_\qop$; we furthermore write
\begin{gather}
\label{EqqCoord1}
  x = r\omega,\qquad r\geq 0,\ \omega\in\Sph^{n-1}, \\
\label{EqqCoord2}
  \hat x := \frac{x}{\bhm},\qquad
  \hat r := \frac{r}{\bhm},\qquad
  \hat\rho := \hat r^{-1} = \frac{\bhm}{r}.
\end{gather}
We finally put
\begin{equation}
\label{EqqXs}
  \dot X := [X;\{0\}] = [0,2)_r\times\Sph^{n-1},\qquad
  \hat X := \ol{\R^3_{\hat x}}.
\end{equation}
Thus, $\pa\dot X=r^{-1}(0)\subset\dot X$ is the front face of $\dot X$. Moreover, $\hat X$ is the radial compactification $\ol{T_0}X$ of the tangent space $T_0 X$. We have natural diffeomorphisms
\[
  \zface_\qop \cong \hat X,\qquad
  \mface_\qop \cong \dot X,
\]
and we shall use both notations for these boundary hypersurfaces.

\begin{definition}[q-vector fields]
\label{DefqVF}
  The space of \emph{q-vector fields} on $X$ is defined as
  \[
    \Vq(X) := \{ V\in\Vb(X_\qop) \colon V\bhm=0 \}.
  \]
  For $m\in\N_0$, we denote by $\Diffq^m(X)$ the space of $m$-th order q-differential operators, consisting of locally finite sums of up to $m$-fold compositions of elements of $\Vq(X)$ (a $0$-fold composition being multiplication by an element of $\CI(X_\qop)$). For $\alpha=(\alpha_\zface,\alpha_\mface)\in\R^2$, put
  \[
    \Diffq^{m,\alpha}(X) = \rho_{\zface_\qop}^{-\alpha_\zface}\rho_{\mface_\qop}^{-\alpha_\mface}\Diffq^m(X) = \bigl\{ \rho_{\zface_\qop}^{-\alpha_\zface}\rho_{\mface_\qop}^{-\alpha_\mface}A \colon A\in\Diffq^m(X) \bigr\}.
  \]
\end{definition}

Since $X_\qop\cap\{\bhm>0\}=(0,1]\times X$, an element $V\in\Vq(X)$ is thus a smooth family $(0,1]\ni\bhm\mapsto V_\bhm\in\cV(X)$ of smooth vector fields on $X$ which degenerate in a particular fashion in the limit $r\to 0$, $\bhm\to 0$. Since $\Vb(X_\qop)$ is a Lie algebra, and since $[V,W]\bhm=V(W\bhm)-W(V\bhm)=0$ whenever $V\bhm=0$ and $W\bhm=0$, we conclude that also $\Vq(X)$ is a Lie algebra.

\begin{rmk}[Comparison with \cite{HintzXieSdS}]
\label{RmkqXie}
  The uniform ODE analysis of \cite{HintzXieSdS} was phrased in terms of horizontal b-vector fields on the subset of $[[0,1)_\bhm\times[0,1)_r;\{0\}\times\{0\}]$ where $\bhm\lesssim r\lesssim 1$; thus, the b-behavior at the lift of $r=0$ was excised. The q-single space and class of q-vector fields defined here, even in the ODE setting where $X$ is an open interval containing $0$, is more natural, as it does not introduce an artificial b-boundary at the lift of $r=0$.
\end{rmk}

In local coordinates $\bhm\geq 0$, $\hat x\in\R^3$ near the interior $\zface_\qop^\circ$ of $\zface_\qop$, the space $\Vq(X)$ is spanned by $\pa_{\hat x^j}$ ($j=1,\ldots,n$) over $\CI(X_\qop)$. Near the interior $\mface_\qop^\circ$, $\Vq(X)$ is spanned by $\pa_{x^j}$ ($j=1,\ldots,n$) or equivalently by $\pa_r$, $\pa_\omega$ (schematic notation for spherical vector fields). Near the corner $\zface_\qop\cap\mface_\qop$, where we have local coordinates $\hat\rho,r,\omega$, we can use $r\pa_r-\hat\rho\pa_{\hat\rho}$, $\pa_\omega$ as a spanning set. A global frame near $\zface_\qop$ is given by $\sqrt{\bhm^2+|x|^2}\pa_{x^j}$ ($j=1,\ldots,n$). In particular, if we regard $\cV(X)$ as the subset of $\bhm$-independent vector fields on $X_\qop$, then
\begin{equation}
\label{EqqVFX}
  \cV(X)\subset\rho_{\zface_\qop}^{-1}\Vq(X),\qquad
  \Diff^m(X) \subset \rho_{\zface_\qop}^{-m}\Diffq^m(X) = \Diffq^{m,(m,0)}(X).
\end{equation}

We denote by
\[
  \Tq X \to X_\qop
\]
the \emph{q-vector bundle} which has local frames given by the above collections of vector fields; thus there is a bundle map $\Tq X\to T X_\qop$ so that $\Vq(X)=\CI(X,\Tq X)$. From the above local coordinate descriptions, we can then also conclude that the restriction maps
\begin{equation}
\label{EqqVFNorm}
  N_{\zface_\qop} \colon \Vq(X) \to \Vb(\hat X),\qquad
  N_{\mface_\qop} \colon \Vq(X) \to \Vb(\dot X)
\end{equation}
are surjective, and their kernels are $\rho_{\zface_\qop}\Vq(X)$ and $\rho_{\mface_\qop}\Vq(X)$, respectively. These maps thus induce bundle isomorphisms
\begin{equation}
\label{EqqBundleIso}
  \Tq_{\zface_\qop}X \cong \Tb\hat X,\qquad
  \Tq_{\mface_\qop}X \cong \Tb\dot X,
\end{equation}
and corresponding isomorphisms of cotangent bundles. We can define the q-principal symbol for $V\in\Vq(X)$ as $\sigmaq^1(V)\colon\Tq^*X\ni\xi\mapsto i\xi(V)$, and by linearity and multiplicativity we can define $\sigmaq^m(A)\in P^m(\Tq^*X)$ for $A\in\Diffq^m(X)$; the principal symbol $\sigmaq^m(A)$ vanishes if and only if $A\in\Diffq^{m-1}(X)$. We also have surjective restriction maps
\begin{equation}
\label{EqqDiffNorm}
  N_{\zface_\qop} \colon \Diffq^m(X) \to \Diffb^m(\hat X),\qquad
  N_{\mface_\qop} \colon \Diffq^m(X) \to \Diffb^m(\dot X),
\end{equation}
and $\sigmab^m(N_H(A))=\sigmaq^m(A)|_{\Tq^*_H X}$ for $H=\zface_\qop,\mface_\qop$ under the above bundle isomorphisms. These maps can be defined completely analogously to restrictions of b-vector fields: that is, $N_{\zface_\qop}(A)u=(A\tilde u)|_{\zface_\qop}$ for $u\in\CIdot(\hat X)=\CI(\zface_\qop)$ where $\tilde u\in\CI(X_\qop)$ is any smooth extension of $u$; similarly for $N_{\mface_\qop}$.

\begin{definition}[Weighted q-Sobolev spaces]
\label{DefqSob}
  Suppose $X$ is compact, and fix a finite collection $V_1,\ldots,V_N\in\Vq(X)$ of q-vector fields which at any point of $X_\qop$ span the q-tangent space. Fix any weighted positive density $\nu=\rho_{\zface_\qop}^{\alpha_\zface}\rho_{\mface_\qop}^{\alpha_\mface}\nu_0$ where $0<\nu_0\in\CI(X_\qop,\Omegaq X)$. We then define, for $s\in\N_0$ and $l,\gamma\in\R$, the function space $H_{\qop,\bhm}^{s,l,\gamma}(X,\nu)$ to be equal to $H^s(X)$ as a set, but equipped with the squared norm
  \[
    \|u\|_{H_{\qop,\bhm}^{s,l,\gamma}(X,\nu)}^2 := \sum_{\alpha\in\N_0^N,\ |\alpha|\leq m} \|\rho_{\zface_\qop}^{-l}\rho_{\mface_\qop}^{-\gamma}V^\alpha u\|_{L^2(X,\nu_\bhm)}^2,\qquad V^\alpha=\prod_{j=1}^N V_j^{\alpha_j},
  \]
  where we write $0<\nu_{\bhm_0}\in\CI(X,\Omega X)$ for the restriction of $\nu$ to $\bhm^{-1}(\bhm_0)$.
\end{definition}

In particular, if $\nu=|\dd x|$, then for $u$ supported in $|\hat x|\lesssim 1$, resp.\ $r\gtrsim 1$, the norm $\|u\|_{H_{\qop,\bhm}^s(X)}$ is uniformly equivalent to $\bhm^{n/2}\|u\|_{\Hb^s(\hat X)}$ (since $|\dd x|=\bhm^n|\dd\hat x|$), resp.\ $\|u\|_{\Hb^s(\dot X)}$.

To analyze q-differential operators using microlocal techniques, we need to define a corresponding pseudodifferential algebra.

\begin{definition}[q-double space]
\label{DefqDouble}
  The \emph{q-double space} of $X$ is defined as the resolution of $[0,1]_\bhm\times X^2$ given by
  \[
    X^2_\qop := \bigl[ [0,1]\times X^2; \{0\}\times\{0\}\times\{0\}; \{0\}\times\{0\}\times X, \{0\}\times X\times\{0\} \bigr].
  \]
  We denote the front face of $X^2_\qop$ by $\zface_{\qop,2}$, the lift of $\{0\}\times X^2$ by $\mface_{\qop,2}$, and the lift of $[0,1]\times\diag_X$ (with $\diag_X\subset X^2$ denoting the diagonal) by $\diag_\qop$. Furthermore, $\lb_{\qop,2}$, resp.\ $\rb_{\qop,2}$ denotes the lift of $\{0\}\times\{0\}\times X$, resp.\ $\{0\}\times X\times\{0\}$. See Figure~\ref{FigqDouble}.
\end{definition}

\begin{figure}[!ht]
\centering
\includegraphics{FigqDouble}
\caption{The q-double space $X^2_\qop$.}
\label{FigqDouble}
\end{figure}

\begin{lemma}[b-fibrations from the q-double space]
\label{LemmaqbFib}
  The left projection $[0,1]\times X\times X\ni(\bhm,x,x')\mapsto(\bhm,x)$ and right projection $(\bhm,x,x')\mapsto(\bhm,x')$ lift to b-fibrations $\pi_L,\pi_R\colon X^2_\qop\to X_\qop$.
\end{lemma}
\begin{proof}
  We only consider the left projection. It lifts to a projection $[[0,1]\times X\times X;\{0\}\times\{0\}\times X]=X_\qop\times X\to X_\qop$ which is b-transversal to $\{0\}\times\{0\}\times\{0\}$, and hence lifts to a b-fibration
  \begin{equation}
  \label{EqqbFib}
    \bigl[ [0,1]\times X\times X; \{0\}\times\{0\}\times X; \{0\}\times\{0\}\times\{0\}\bigr] \to X_\qop.
  \end{equation}
  On the left, we can reverse the order of the two blow-ups since the second center is contained in the first. Since the map~\eqref{EqqbFib} is b-transversal to the lift of $\{0\}\times X\times\{0\}$, this lift can be blown up, and the map~\eqref{EqqbFib} lifts to the desired b-fibration.
\end{proof}

It is easy to check in local coordinates on $X^2_\qop$ that the lift of $\Vq(X)$ to $X_\qop^2$ along $\pi_L$ is transversal to $\diag_\qop$ (see also~\eqref{EqqQuant} below). (This can also be deduced from the analogous statement for b-double spaces by using Lemma~\ref{LemmaqBdy}, together with the analogous statement in $\bhm>0$.) The resulting isomorphism $\Tq X\cong N\diag_\qop$ induces a bundle isomorphism $N^*\diag_\qop\cong\Tq^*X$.

\begin{definition}[q-pseudodifferential operators]
\label{DefqPsdo}
  Let $s,l,\gamma\in\R$. Then $\Psiq^{s,l,\gamma}(X)$ is the space of all smooth families of bounded linear operators on $\CIc(X)$, parameterized by $\bhm\in(0,1]$, with Schwartz kernels $\kappa\in\rho_{\zface_{\qop,2}}^{-l}\rho_{\mface_{\qop,2}}^{-\gamma}I^{m-\frac14}(X^2_\qop,\diag_\qop;\pi_R^*\Omegaq X)$ which vanish to infinite order at $\lb_{\qop,2}$ and $\rb_{\qop,2}$, and which are conormal at $\zface_{\qop,2}$ and $\mface_{\qop,2}$. When $X$ is non-compact, we furthermore demand that $\kappa$ is properly supported, i.e.\ the projection maps $\pi_L,\pi_R\colon\supp\kappa\to X_\qop$ are proper.
\end{definition}

A typical element of $\Psi_\qop^{s,l,\gamma}(X)$ is given in coordinates $\bhm>0$ and $x,x'\in\R^n$ (the lift of coordinates on $X$ centered around $0$ to the left and right factor of $X^2$) as a quantization\footnote{In these local coordinates, we can take $\rho_{\zface_{\qop,2}}=\sqrt{\bhm^2+|x|^2+|x'|^2}$.}
\begin{equation}
\label{EqqQuant}
  (\Op_{\qop,\bhm}(a)u)(x) = (2\pi)^{-n}\iint \exp\Bigl(i\frac{x-x'}{\rho_{\zface_{\qop,2}}}\xi\Bigr)\chi\Bigl(\frac{|x-x'|}{\rho_{\zface_{\qop,2}}}\Bigr) a(\bhm,x,\xi) u(x')\,\frac{\dd x}{\rho_{\zface_{\qop,2}}^n}\,\dd\xi,
\end{equation}
where $\chi\in\CIc((-\half,\half))$ is identically $1$ near $0$, and $a$ is the local coordinate expression of an element of the symbol space $S^{s,l,\gamma}(\Tq^*X)$ consisting of conormal functions on $\ol{\Tq^*}X$ with weights $-s$, $-l$, and $-\gamma$ at fiber infinity, over $\zface_{\qop,2}$, and over $\mface_{\qop,2}$, respectively.

\begin{lemma}[Boundary hypersurfaces of $X^2_\qop$]
\label{LemmaqBdy}
  In the notation of~\S\usref{SsPbsc}, We have natural diffeomorphisms
  \begin{equation}
  \label{EqqBdy}
    \zface_{\qop,2} \cong \hat X^2_\bop,\qquad
    \mface_{\qop,2} \cong \dot X^2_\bop.
  \end{equation}
\end{lemma}
\begin{proof}
  The front face of $[[0,1]\times X^2;\{0\}\times\{0\}\times\{0\}]$ is the radial compactification $\ol{T_{(0,0)}}(X^2)$. The lift of $\{0\}\times\{0\}\times X$, resp.\ $\{0\}\times X\times\{0\}$ intersects this at $\{0\}\times\pa(\ol{T_0}X)$, resp.\ $\pa(\ol{T_0}X)\times\{0\}$. The first isomorphism in~\eqref{EqqBdy} is thus the same as the fact---which can be checked by direct computation---that the resolution of $\ol{\R^{2 n}}$ at $\{0\}\times\pa\ol{\R^n}$ and $\pa\ol{\R^n}\times\{0\}$ is naturally diffeomorphic to $(\ol{\R^n})^2_\bop$.

  For the second isomorphism in~\eqref{EqqBdy}, note that the lift of $\{0\}\times X^2$ to $[[0,1]\times X^2;\{0\}\times\{0\}\times\{0\}]$ is $[X^2;\{0\}\times\{0\}]$. In this manifold, we then further blow up the lift of $X\times\{0\}$---resulting in $[X\times\dot X;\{0\}\times\pa\dot X]$---and then we blow up the lift of $\{0\}\times\dot X$, which can in fact be done prior to blowing up $\{0\}\times\pa\dot X$ and thus results in $[X\times\dot X;\{0\}\times\dot X;\{0\}\times\pa\dot X]=[\dot X^2;(\pa\dot X)^2]=\dot X^2_\bop$, as claimed.
\end{proof}

The principal symbol map $\sigmaq^{s,l,\gamma}$ fits into the short exact sequence
\[
  0 \to \Psiq^{s-1,(l,\gamma)}(X) \hra \Psi^{s,l,\gamma}(X) \xra{\sigmaq^{s,l,\gamma}} S^{s,l,\gamma}(\Tq^*X)/S^{s-1,l,\gamma}(\Tq^*X) \to 0.
\]
Restricting to operators whose Schwartz kernels are classical (denoted by an added subscript `$\cl$') at $\zface_{\qop,2}$ and $\mface_{\qop,2}$ (thus smooth when the corresponding order vanishes), we obtain from Lemma~\ref{LemmaqBdy} surjective normal operator maps
\begin{equation}
\label{EqqPsdoNorm}
  N_{\zface_\qop} \colon \Psi_{\qop,\cl}^{s,0,\gamma}(X) \to \Psib^{s,\gamma}(\hat X),\qquad
  N_{\mface_\qop} \colon \Psi_{\qop,\cl}^{s,l,0}(X) \to \Psib^{s,l}(\dot X).
\end{equation}
As in the case of q-differential operators, the principal symbols of $N_H(A)$ are related to that of $A$ by restriction using~\eqref{EqqBundleIso}. Also, the normal operators can be defined via testing, and therefore are multiplicative once we know that $\Psiq(X)$ is closed under composition; we turn to this now.

Pushforward along $\pi_L$ maps the Schwartz kernel of elements of $\Psi_{\qop,\cl}^{s,l,\gamma}(X)$, resp.\ $\Psi_\qop^{s,l,\gamma}(X)$ into $\rho_{\zface_\qop}^{-l}\rho_{\mface_\qop}^{-\gamma}\CI(X_\qop)$, resp.\ $\cA^{(-l,-\gamma)}(X_\qop)$. Therefore, compositions of q-ps.d.o.s are well-defined as maps on conormal functions on $X_\qop$. One can prove that the composition is again a q-ps.d.o.\ using the explicit quantization map in local coordinates above and direct estimates for the residual remainders (in $\Psiq^{-\infty,l,\gamma}(X)$). A geometric proof proceeds via the construction of an appropriate triple space:

\begin{definition}[q-triple space]
\label{DefqTriple}
  Define the following submanifolds of $[0,1]_\bhm\times X^3$:
  \begin{alignat*}{3}
    &&C&=\{(0,0,0,0)\}, \\
    L_F&=\{0\}\times \{0\}\times\{0\}\times X, &\quad
    L_S&=\{0\}\times X\times\{0\}\times\{0\},&\quad
    L_C&=\{0\}\times \{0\}\times X\times\{0\}, \\
    P_F&=\{0\}\times X\times X\times\{0\},&\quad
    P_S&=\{0\}\times \{0\}\times X\times X,&\quad
    P_C&=\{0\}\times X\times\{0\}\times X.
  \end{alignat*}
  The \emph{q-triple space} of $X$ is then defined as
  \[
    X^3_\qop := \bigl[ [0,1]\times X^3; C; L_F,L_S,L_C; P_F,P_S,P_C \bigr].
  \]
  We denote by $\zface_{\qop,3}$ and $\mface_{\qop,3}$ the lifts of $C$ and $\{0\}\times X^3$, respectively. For $*=F,S,C$, we denote by $\bface_{\qop,*}$ and $\mface_{\qop,*}$ the lifts of $L_*$ and $P_*$, respectively; and $\diag_{\qop,*}$ denotes the lift of $[0,1]\times(\pi^X_*)^{-1}(\diag_\qop)$ where $\pi^X_*\colon X^3\to X^2$ are the projections $\pi^X_F\colon(x,x',x'')\mapsto(x,x')$, $\pi^X_S\colon(x,x',x'')\mapsto(x',x'')$, $\pi^X_C\colon(x,x',x'')\mapsto(x,x'')$. Finally, $\diag_{\qop,3}$ is the lift of $[0,1]\times\diag_3$ where $\diag_3=\{(x,x,x)\colon x\in X\}$ is the triple diagonal.
\end{definition}

\begin{lemma}[b-fibrations from the q-triple space]
  The projection map $[0,1]_\bhm\times X^3\ni(\bhm,x,x',x'')\mapsto(\bhm,x,x')\in[0,1]\times X^2$ to the first and second factor of $X^3$ lifts to a b-fibration $\pi_F\colon X^3_\qop\to X^2_\qop$, similarly for the lifts $\pi_S$, $\pi_C\colon X^3_\qop\to X^2_\qop$ of the projections to the second and third, resp.\ first and third factor of $X^3$.
\end{lemma}
\begin{proof}
  We only prove the result for $\pi_F$. Since the lifted projection $[[0,1]\times X^3;L_F]\to[[0,1]\times X^2;\{0\}\times\{0\}\times\{0\}]$ is b-transversal to the lift of $C\supset L_F$, it lifts to a b-fibration
  \[
    \bigl[[0,1]\times X^3;C;L_F] \to \bigl[ [0,1]\times X^2; \{0\}\times\{0\}\times\{0\} \bigr].
  \]
  The preimage of the lift of $\{0\}\times\{0\}\times X$, resp.\ $\{0\}\times X\times\{0\}$, is the lift of $P_S$, resp.\ $P_C$, and thus the lifted projection
  \[
    \bigl[ [0,1]\times X^3;C;L_F;P_S,P_C \bigr] \to X^2_\qop
  \]
  is a b-fibration still. It is b-transversal to the lift of $L_S$, and thus lifts to a b-fibration if we blow up $L_S$ in the domain; since $L_S$ and $P_S$ are transversal, and since $L_S\subset P_C$, \cite[Proposition~5.11.2]{MelroseDiffOnMwc} implies that we can commute the blow-up of $L_S$ through that of $P_S,P_C$. Arguing similarly for $L_C$, we thus have a b-fibration
  \[
    \bigl[ [0,1]\times X^3; C; L_F,L_S,L_C; P_S,P_C \bigr] \to X_\qop^2.
  \]
  This is b-transversal to the lift of $P_F$; blowing up $P_F$ in the domain thus gives the desired b-fibration $\pi_F$.
\end{proof}

For later use, we record
\begin{equation}
\label{EqqPreimages}
\begin{alignedat}{2}
  \pi_F^{-1}(\zface_{\qop,2}) &= \zface_{\qop,3}\cup\bface_{\qop,F}, &\qquad
  \pi_F^{-1}(\mface_{\qop,2}) &= \mface_{\qop,3}\cup\mface_{\qop,F}, \\
  \pi_F^{-1}(\lb_{\qop,2}) &= \bface_{\qop,C}\cup\mface_{\qop,S}, &\qquad
  \pi_F^{-1}(\rb_{\qop,2}) &= \bface_{\qop,S}\cup\mface_{\qop,C}, \\
  \pi_F^{-1}(\diag_\qop) &= \diag_{\qop,F}.
\end{alignedat}
\end{equation}
similarly for the preimages under $\pi_S$ and $\pi_C$.

\begin{prop}[Composition of q-ps.d.o.s]
\label{PropqComp}
  Let $A_j\in\Psiq^{s_j,l_j,\gamma_j}(X)$, $j=1,2$. Then $A_1\circ A_2\in\Psiq^{s_1+s_2,l_1+l_2,\gamma_1+\gamma_2}(X)$.
\end{prop}
\begin{proof}
  Since the space $\Psiq^s(X)$ is invariant under conjugation by powers of $\rho_{\zface_\qop}$ and $\rho_{\mface_\qop}$, it suffices to prove the result for $l_1=l_2=0$ and $\gamma_1=\gamma_2=0$. Write the Schwartz kernel $\kappa$ of $A_1\circ A_2$ in terms of the Schwartz kernels $\kappa_1,\kappa_2$ of $A_1,A_2$ as
  \[
    \kappa = (\nu_1\nu_2)^{-1}(\pi_C)_* \bigl( \pi_F^*\kappa_1 \cdot \pi_S^*\kappa_2 \cdot \pi_C^*\nu_1\cdot\pi^*\nu_2 \bigr)
  \]
  where $0<\nu_1\in\CI(X_\qop;\Omegaq X)$ is an arbitrary q-density, and $\nu_2=|\frac{\dd\bhm}{\bhm}|$ is a b-density on $[0,2)_\bhm$ with $\pi\colon X_\qop^3\to[0,1]$ denoting the lifted projection. The term in parentheses is then a bounded conormal section of $\pi_F^*\Omegaq X\otimes\pi_S^*\Omegaq X\otimes\pi_C^*\Omegaq X\otimes\pi^*\Omegab_{[0,1]}[0,2)\cong\Omegab X^3_\qop$ which vanishes to infinite order at the boundary hypersurfaces of $X^3_\qop$ which map to $\lb_{\qop,2}$ or $\rb_{\qop,2}$ under $\pi_C$. The conclusion then follows using pullback and pushforward results for conormal distributions, see \cite[\S4]{MelroseDiffOnMwc} and \cite{MelrosePushfwd}.
\end{proof}

A proof of the uniform (for $\bhm\in(0,1]$, the point being uniformity as $\bhm\searrow 0$) boundedness of elements of $\Psiq^0(X)$ on $L^2(X,\nu)$ for $0<\nu\in\CI(X_\qop,\Omegaq X)$ can be reduced, using H\"ormander's square root trick (see the proof of~\cite[Theorem~2.1.1]{HormanderFIO1}), to the uniform $L^2$-boundedness of elements of $\Psiq^{-\infty}(X)$. Such elements have Schwartz kernels $\kappa\in\CI(X_\qop^2,\pi_R^*\Omegaq X)$ which vanish to infinite order at $\lb_{\qop,2}$ and $\rb_{\qop,2}$. Pushforward of $\kappa$ along $\pi_L$ thus gives an element of $\CI(X_\qop)$. The Schur test implies the desired $L^2$-boundedness; since $\Psiq(X)$ is invariant under conjugation by weights, we deduce boundedness on $L^2(X,\nu)$ for any weighted q-density $\nu$. One can then define weighted Sobolev spaces $H_\qop^{s,l,\gamma}(X)$ also for real orders $s\in\R$ in the usual manner (cf.\ \S\ref{SsPH}), and any $A\in\Psiq^{s,l,\gamma}(X)$ defines a (uniformly as $\bhm\searrow 0$) bounded map $H_\qop^{\tilde s,\tilde l,\tilde\gamma}(X)\to H_\qop^{\tilde s-s,\tilde l-l,\tilde\gamma-\gamma}(X)$ for all $\tilde s,\tilde l,\tilde\gamma\in\R$.

The normal operator maps~\eqref{EqqVFNorm}--\eqref{EqqDiffNorm} for q-differential operators imply relationships between integer order q-Sobolev spaces on $X$ and families of b-Sobolev spaces on collar neighborhoods of $\hat X$ and $\dot X$. We immediately state the version for general orders, which rests on~\eqref{EqqPsdoNorm}; for brevity, we restrict to the class of densities which we will use in~\S\ref{SK}.

\begin{prop}[Relationships between Sobolev spaces]
\label{PropqHRel}
  Fix a density $\nu=\rho_{\zface_q}^{n/2}\nu_0$ where $0<\nu_0\in\CI(X,\Omegaq X)$.\footnote{This includes as a special case $\bhm$-independent smooth positive densities on $X$.}
  \begin{enumerate}
  \item\label{ItqHRelzf} Consider the (change of coordinates) map $\phi_{\zface_q}\colon(0,1]_\bhm\times\hat X^\circ\ni(\bhm,\hat x)\mapsto(\bhm,\bhm\hat x)\in X_\qop$, and let $\chi\in\CI(X_\qop)$ be identically $1$ near $\zface_\qop$ and supported in a collar neighborhood of $\zface_\qop\subset X_\qop$. Then we have a uniform equivalence of norms
    \begin{equation}
    \label{EqqHRelzf}
      \|\chi u\|_{H_{\qop,\bhm}^{s,l,\gamma}(X)} \sim \bhm^{\frac{n}{2}-l}\| \phi_{\zface_q}^*(\chi u)|_\bhm \|_{\Hb^{s,\gamma-l}(\hat X,|\dd\hat x|)},
    \end{equation}
    in the sense that there exists a constant $C>1$ independent of $\bhm\in(0,1]$ so that the left hand side is bounded by $C$ times the right hand side, and vice versa.
  \item\label{ItqHRelmf} Consider the inclusion map $\phi_{\mface_\qop}\colon(0,1]_\bhm\times\dot X^\circ\hra X_\qop$, and let $\chi\in\CI(X_\qop)$ be identically $1$ near $\mface_\qop$ and supported in a collar neighborhood of $\dot X\subset X_\qop$. Then we have a uniform equivalence of norms
    \begin{equation}
    \label{EqqHRelmf}
      \|\chi u\|_{H_{\qop,\bhm}^{s,l,\gamma}(X)} \sim \bhm^{-\gamma} \|\phi_{\mface_\qop}^*(\chi u)|_\bhm\|_{\Hb^{s,l-\gamma}(\dot X,\nu_\cop)},
    \end{equation}
    where $\nu_\cop$ is the lift of a fixed smooth positive density on $X$ to $\dot X$. (Thus, one can take $\nu_\cop=|\dd x|=r^{n-1}|\dd r\,\dd g_{\Sph^{n-1}}|$ near $r=0$.)
  \end{enumerate}
\end{prop}
\begin{proof}
  Via division by $\bhm^l$, we can reduce to the case $l=0$. Moreover, $\rho_{\mface_q}:=\bhm/\sqrt{|x|^2+\bhm^2}$ is a defining function of $\mface_q$, and its pullback along $\phi_{\zface_q}$ is $\la\hat r\ra^{-1}$ which is a defining function of $\pa\hat X$; therefore, we may also reduce to the case $\gamma=0$.

  For part~\eqref{ItqHRelzf}, the $L^2$-case $s=0$ now follows from the observation that $\phi_{\zface_q}^*(|\dd x|)=\bhm^n|\dd\hat x|$. For $s>0$, fix an elliptic operator $A_0\in\Psib^s(\hat X)$ (independent of $\bhm$) with Schwartz kernel $\kappa_0$, and fix also $\tilde\chi\in\CIc(X_\qop)$ to be identically $1$ near $\supp\chi$ but still with support in a collar neighborhood of $\zface_\qop$; define then $A\in\Psiq^s(\hat X)$ via its Schwartz kernel $\kappa$ as
  \begin{equation}
  \label{EqqHRelzfOp}
    \kappa=(\pi_L^*\tilde\chi)(\pi_R^*\tilde\chi)\cdot(\phi_{\zface_q}^{-1})^*\kappa_0
  \end{equation}
  where $\pi_L,\pi_R\colon X^2_\qop\to X_\qop$ denote the lifted left and right projections. (Thus, $\kappa$ is obtained from $\kappa_0$ via dilation-invariant extension off $\zface_{\qop,2}$, followed by cutting it off to a neighborhood of $\zface_{\qop,2}$.) In particular, $A$ is elliptic as a q-ps.d.o.\ near the q-cotangent bundle over $\supp\chi$. We then have a uniform equivalence of norms
  \begin{align*}
    \|\chi u\|_{H_{\qop,\bhm}^s(X)} &\sim \|(\chi u)|_\bhm\|_{L^2(X)} + \|A(\chi u)|_\bhm\|_{L^2(X)} \\
      &\sim \bhm^{\frac{n}{2}} \bigl( \|\phi_{\zface_q}^*(\chi u)|_\bhm\|_{L^2(\hat X,|\dd\hat x|)} + \| A_0(\phi_{\zface_q}^*(\chi u)|_\bhm)\|_{L^2(\hat X,|\dd\hat x|)}\bigr) \\
      &\sim \bhm^{\frac{n}{2}} \|\phi_{\zface_q}^*(\chi u)|_\bhm\|_{\Hb^s(\hat X,|\dd\hat x|)},
  \end{align*}
  as claimed. For $s<0$, the claim follows by duality.

  The proof of part~\eqref{ItqHRelmf} is completely analogous; one now takes an elliptic operator $A_0\in\Psib^s(\dot X)$ to measure $\Hb^s(\dot X)$-norms, and relates this to $H_\qop^s(X)$-norms by measuring the latter using a q-ps.d.o.\ $A$ defined analogously to~\eqref{EqqHRelzfOp}.
\end{proof}

\subsection{Q-single space}
\label{SsQS}

We shall control (solutions of) the degenerating spectral family for an infinite range of spectral parameters on the following space, which is a resolution of a parameter-dependent version of the q-single space $X_\qop$ from Definition~\ref{DefqSingle}: 

\begin{definition}[Q-single space]
\label{DefQSingle}
  The \emph{Q-single space} of $X$ is the resolution of $\ol{\R_\sigma}\times[0,1]_\bhm\times X$ defined as the iterated blow-up
  \begin{align}
  \label{EqQSingleXq}
    X_\Qop &:= \bigl[ \ol\R\times X_\qop; \pa\ol\R\times\zface_\qop; \pa\ol\R\times\mface_\qop \bigr] \\
  \label{EqQSingle}
      &=\bigl[ \ol\R\times[0,1]\times X; \ol\R\times\{0\}\times\{0\}; \pa\ol\R\times\{0\}\times\{0\}; \pa\ol\R\times\{0\}\times X \bigr].
  \end{align}
  We denote its boundary hypersurfaces as follows:
  \begin{enumerate}
  \item $\mface$ (the `main face') is the lift of $\ol\R\times\{0\}\times X$;
  \item $\zface$ (the `zero energy face') is the lift of $\ol\R\times\{0\}\times\{0\}$;
  \item $\nface$ (the `nonzero energy face') is the lift of $\pa\ol\R\times\{0\}\times\{0\}$;
  \item $\iface$ (the `intermediate semiclassical face') is the lift of $\pa\ol\R\times\{0\}\times X$;
  \item $\sface$ (the `semiclassical face') is the lift of $\pa\ol\R\times[0,1]\times X$.
  \end{enumerate}
  The hypersurfaces $\nface$, $\iface$, $\sface$ have two connected components each, denoted $\nface_\pm$, $\iface_\pm$, $\sface_\pm$, corresponding to whether $\sigma=+\infty$ or $-\infty$. For $H\subset X_\Qop$ equal to any one of these boundary hypersurfaces, we denote by $\rho_H\in\CI(X_\Qop)$ a defining function of $H$, i.e.\ $H=\rho_H^{-1}(0)$ and $\dd\rho_H\neq 0$ on $H$. For $H=\nface$, we denote by $\rho_H$ a total boundary defining function for $\nface^+\cup\nface^-$, likewise for $H=\iface,\sface$.
\end{definition}

We introduce a variety of functions defined on (subsets of) $X_\Qop$. We denote by $\sigma,\bhm$ the lifts of the first two coordinates on $\R\times[0,1]\times X$. We furthermore write $x=r\omega$ and $\hat x=\frac{x}{\bhm}$, $\hat r=\frac{r}{\bhm}$, $\hat\rho=\hat r^{-1}$ as in~\eqref{EqqCoord1}--\eqref{EqqCoord2}. We also set
\begin{equation}
\label{EqQCoord}
  h = |\sigma|^{-1},\qquad
  \tilde r := \frac{r}{h},\qquad
  \tilde\rho := \tilde r^{-1} = \frac{h}{r},\qquad
  \tilde\sigma := \bhm\sigma,\qquad
  \tilde h := |\tilde\sigma|^{-1} = \frac{h}{\bhm}.
\end{equation}
See Figure~\ref{FigQSingle}.

\begin{figure}[!ht]
\centering
\includegraphics{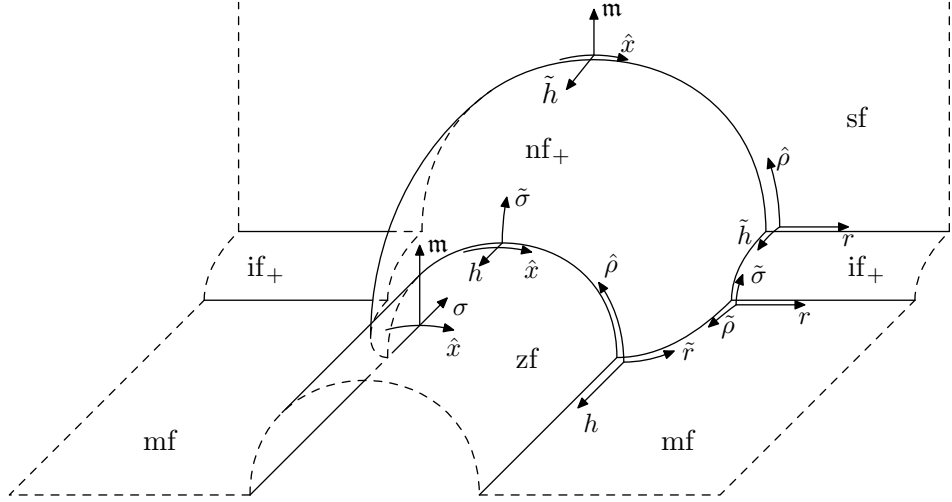}
\caption{The Q-single space $X_\Qop$ in the case $X=(-1,1)$, restricted to $\sigma>-C$, and the coordinates~\eqref{EqqCoord1}, \eqref{EqqCoord2}, and \eqref{EqQCoord}.}
\label{FigQSingle}
\end{figure}

\begin{figure}[!ht]
\centering
\includegraphics{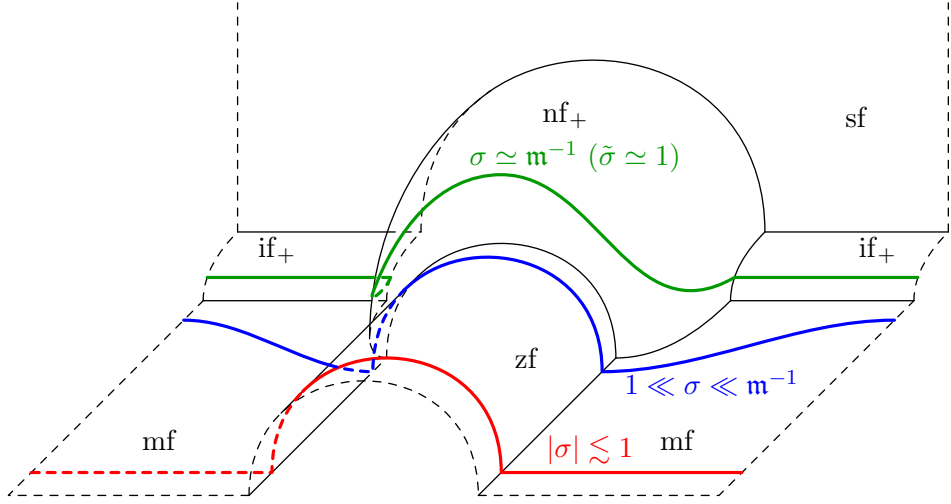}
\caption{The Q-single space $X_Q$ for $X=(-1,1)$. We show here the intersections of three level sets of the (rescaled) frequency variable $\sigma$ ($\tilde\sigma$) with $\bhm^{-1}(0)$: one level set $\sigma=\sigma_0$ where $|\sigma_0|\lesssim 1$ is bounded and thus the rescaled Kerr frequency $\tilde\sigma_0=\bhm\sigma_0=0$ vanishes; one level set $\sigma=\sigma_1$ where $\sigma_1$ is large but $\tilde\sigma_1:=\bhm\sigma_1$ still vanishes; and one level set $\tilde\sigma=\bhm\sigma=\tilde\sigma_0$ where the rescaled frequency $\tilde\sigma_0$ is of order $1$.}
\label{FigQSingleLvl}
\end{figure}

\begin{prop}[Structure of boundary hypersurfaces]
\fakephantomsection
\label{PropQStruct}
  \hspace{0in}\begin{enumerate}
  \item\label{ItQStructzf} The restriction of $(\sigma,\hat x)$ to the interior of $\zface$ induces a diffeomorphism
    \[
      \zface \cong \ol{\R_\sigma} \times \hat X
    \]
    Thus, $\zface$ is the total space of the (trivial) fibration $\hat X-\zface\to\ol{\R_\sigma}$.
  \item\label{ItQStructmf} The restriction of $(\sigma,(r,\omega))$ to the interior of $\mface$ induces a diffeomorphism
    \[
      \mface \cong \bigl[ \ol\R \times \dot X; \pa\ol\R\times\pa\dot X \bigr].
    \]
  \item\label{ItQStructnf} The restriction of $(\tilde\sigma,\hat x)$ to the interior of $\nface_\pm$ induces a diffeomorphism
    \[
      \nface_\pm \cong \bigl[ (\pm[0,\infty]) \times \hat X; \{0\}\times\hat X \bigr].
    \]
  \item\label{ItQStructif} The restriction of $(\tilde\sigma,x)$ to the interior of $\iface_\pm$ induces a diffeomorphism
    \[
      \iface_\pm \cong (\pm[0,\infty])\times\dot X.
    \]
  \end{enumerate}
\end{prop}
\begin{proof}
  The front face of $[\ol\R\times[0,1]\times X;\ol\R\times\{0\}\times\{0\}]=\ol\R\times[[0,1]\times X;\{0\}\times\{0\}]$ is diffeomorphic to $\ol\R\times\ol{T_0}X=\ol\R\times\hat X$ (with coordinates $\sigma$, $\hat x$ in the interior). The boundary hypersurface $\zface$ is obtained from this front face by blowing up $\sigma=\pm\infty$ which does not change the smooth structure. (Note that the lift of the final submanifold $\pa\ol\R\times\{0\}\times X$ in~\eqref{EqQSingle} is disjoint from this front face.) This proves part~\eqref{ItQStructzf}.

  For part~\eqref{ItQStructmf}, we note that the lift of $\ol\R\times\{0\}\times X$ to $X_\Qop$ is given by first resolving $\ol\R\times X$ at $\ol\R\times\{0\}$ (which produces $\ol\R\times\dot X$) followed by the resolution of $\pa\ol\R\times\pa\dot X$. Within this space then, the final resolution in~\eqref{EqQSingle} only blows up the lift of $\pa\ol\R\times\dot X$, which does not change the smooth structure.

  For part~\eqref{ItQStructnf}, we first note that the front face $\nface_\pm'$ of the blow-up of the lift of $\{\pm\infty\}\times\{0\}\times\{0\}$ to $X_\Qop':=[\ol\R\times[0,1]\times X;\ol\R\times\{0\}\times\{0\}]$ is diffeomorphic to $[0,\infty]_\mu\times\hat X$ where $\mu=\frac{\nu}{h}$ with $\nu=(r^2+\bhm^2)^{1/2}$ a defining function of the front face of $X_\Qop'$. The final blow-up in~\eqref{EqQSingle} restricts to $\nface_\pm'$ as the blow-up of $\{\infty\}\times\pa\hat X$, that is,
  \[
    \nface_\pm = \bigl[ [0,\infty]_\mu\times\hat X; \{\infty\}\times\pa\hat X\bigr].
  \]
  Upon restriction to a compact subset $K$ of the interior $\hat X^\circ$ in the second factor (thus $r\lesssim\bhm$), we can replace $\nu$ by $\bhm$, and thus $\mu$ by $\frac{\bhm}{h}=\pm\tilde\sigma$. (That is, $\mu/(\pm\tilde\sigma)$ is a positive smooth function on $[0,\infty]_\mu\times K$.) Near the boundary of $\hat X$ on the other hand, let us work in the collar neighborhood $[0,1)_{\hat\rho}\times\Sph^{n-1}_\omega$ of $\pa\hat X\subset\hat X$. Since there we can replace $\nu$ by $r$ and thus $\mu$ by $\tilde r$, the lift of $[0,\infty]_\mu\times[0,1)_{\hat\rho}\times\Sph^{n-1}_\omega$ to $\nface_\pm$ is
  \[
    \bigl[ [0,\infty]_{\tilde r} \times [0,1)_{\hat\rho}\times\Sph^{n-1}_\omega; \{\infty\}\times\{0\}\times\Sph^{n-1}_\omega \bigr] = \bigl[ [0,\infty]\times[0,1); \{\infty\}\times\{0\} \bigr] \times \Sph^{n-1}.
  \]
  Observe then that the map $(\tilde r,\hat\rho)\to(\tilde r\hat\rho,\hat\rho)$ induces a diffeomorphism
  \begin{equation}
  \label{EqQnfBlowup}
    \bigl[ [0,\infty]\times[0,1); \{\infty\}\times\{0\} \bigr] \cong \bigl[ [0,\infty] \times [0,1); \{0\}\times\{0\} \bigr].
  \end{equation}
  Since $\tilde r\hat\rho=\pm\tilde\sigma$, this proves part~\eqref{ItQStructnf}.

  Finally, for the proof of part~\eqref{ItQStructif}, we note that coordinates near the lift of $\{\infty\}\times\{0\}\times X$ to $[\ol\R\times[0,1]\times X;\pa\ol\R\times\{0\}\times\{0\}]$ are $r\geq 0$, $\omega\in\Sph^{n-1}$, $\hat\rho=\frac{\bhm}{r}\geq 0$, and $\tilde\rho=\frac{h}{r}\geq 0$, with the lift of $\{\infty\}\times\{0\}\times X$ given by $\tilde\rho=\hat\rho=0$. Therefore,
  \[
    \iface \cong [0,\infty]_{\hat\rho/\tilde\rho} \times X,
  \]
  and it remains to note that $\hat\rho/\tilde\rho=\bhm/h=\pm\tilde\sigma$.
\end{proof}

\begin{definition}[Pieces of $\zface$, $\mface$ and $\nface_\pm$]
\label{DefQPieces}
  We define
  \begin{alignat*}{2}
    \mface_{\pm,\semi} &:= \mface \cap \sigma^{-1}(\pm[1,\infty]), \\
    \nface_{\pm,\tilde\semi} &:= \nface_\pm \cap \tilde\sigma^{-1}(\pm[1,\infty]), &\qquad
    \nface_{\pm,\low} &:= \nface_\pm \cap \tilde\sigma^{-1}(\pm[0,1]).
  \end{alignat*}
  We furthermore set, for $\sigma_0\in\R$ and $\tilde\sigma_0\in\R\setminus\{0\}$,
  \[
    \zface_{\sigma_0} := \zface \cap \sigma^{-1}(\sigma_0),\qquad
    \mface_{\sigma_0} := \mface \cap \sigma^{-1}(\sigma_0),\qquad
    \nface_{\tilde\sigma_0} := \nface \cap \tilde\sigma^{-1}(\tilde\sigma_0).
  \]
\end{definition}

Thus, using the notation for the single spaces for semiclassical cone, $\scbtop$-transition, and semiclassical scattering analysis from~\S\S\ref{SsPch}, \ref{SsPbsc}, and \ref{SsPscbt}, respectively, Proposition~\ref{PropQStruct} provides diffeomorphisms
\begin{equation}
\label{EqQPieces}
\begin{alignedat}{2}
  \mface_{\pm,\semi} &\cong \dot X_\chop&\quad&\text{(with semiclassical parameter $h=|\sigma|^{-1}\in[0,1]$)}, \\
  \nface_{\pm,\tilde\semi} &\cong \hat X_{\scop,\tilde\semi}&\quad&\text{(with semiclassical parameter $\tilde h=|\tilde\sigma|^{-1}\in[0,1]$)}, \\
  \nface_{\pm,\low} &\cong \hat X_\scbtop&\quad&\text{(with spectral parameter $\tilde\sigma\in\pm[0,1]$)},
\end{alignedat}
\end{equation}
as well as
\[
  \zface_{\sigma_0} \cong \hat X,\qquad
  \mface_{\sigma_0} \cong \dot X,\qquad
  \nface_{\tilde\sigma_0} \cong \hat X.
\]

\subsection{Q-vector fields and differential operators}
\label{SsQV}

We next turn to the class of $\sigma$- and $\bhm$-dependent vector fields on $X$ on which our uniform analysis will be based.

\begin{definition}[Q-vector fields]
\label{DefQVF}
  The space of \emph{Q-vector fields} on $X$ is defined as
  \[
    \VQ(X) := \{ V\in\rho_\iface\rho_\sface\Vb(X_\Qop) \colon V\sigma=0,\ V\bhm=0 \}.
  \]
\end{definition}

Since $X_\Qop\cap\{\sigma\in\R,\ \bhm>0\}=\R_\sigma\times(0,1]_\bhm\times X$, an element $V\in\VQ(X)$ is thus a smooth family $\R\times(0,1]\ni(\sigma,\bhm)\mapsto V_{\sigma,\bhm}$ of smooth vector fields on $X$ which degenerate or become singular in a particular fashion in the limits $r\to 0$, $\bhm\to 0$, $|\sigma|\to\infty$, or any combination thereof.

\begin{lemma}[Properties of $\VQ(X)$]
\label{LemmaQVF}
  The space $\VQ(X)$ is Lie algebra, and in fact
  \begin{equation}
  \label{EqQVFComm}
    V,W\in\VQ(X) \implies [V,W] \in \rho_\sface\rho_\iface\VQ(X).
  \end{equation}
  Moreover, for any weight $w=\prod_H\rho_H^{\alpha_H}$ where $H\subset X_\Qop$ ranges over all boundary hypersurfaces and $\alpha_H\in\R$, we have $w^{-1}[V,w]\in\rho_\sface\rho_\iface\CI(X_\Qop)$ for any $V\in\VQ(X)$.
\end{lemma}
\begin{proof}
  The final claim follows from the fact that $w^{-1}[V_0,w]\in\CI(X_\Qop)$ for \emph{any} $V_0\in\Vb(X_\Qop)$. In order to prove~\eqref{EqQVFComm}, we observe that $[V,W]\sigma=V W\sigma-W V\sigma=0$, likewise $[V,W]\bhm=0$; moreover, we have, for $w:=\rho_\sface\rho_\iface$ and $V=w V_0$, $W=w W_0\in\VQ(X)$,
  \[
    [V,W] = w\bigl((w^{-1}[V_0,w])W_0 - (w^{-1}[W_0,w])w V_0\bigr) \in w\Vb(X_\Qop),
  \]
  which implies the claim.
\end{proof}

We make this explicit in various local coordinate systems; we use the notation from~\eqref{EqqCoord1}, \eqref{EqqCoord2}, and \eqref{EqQCoord}.
\begin{enumerate}
\item The intersection of $X_\Qop$ with $|\hat x|<C$ is $[[0,1]_\bhm\times\ol{\R_\sigma}\times B;\{0\}\times\pa\ol\R\times B]$ where $B=\{\hat x\in\hat X\colon|\hat x|<C\}$ is a ball. Thus, in the coordinates $\bhm,\sigma,\hat x$, a basis of Q-vector fields is given by $\pa_{\hat x^j}$ ($j=1,\ldots,n$) in the set where $\sigma$ is bounded or even where $|\sigma|\geq 1$ but $\tilde\sigma$ is bounded, and by $\tilde h\pa_{\hat x^j}$ when $|\tilde\sigma|\gtrsim 1$ (where $\tilde h$ is a defining function of $\sface$).
\item The intersection of $X_\Qop$ with $r>c>0$ is $[[0,1]_\bhm\times\ol{\R_\sigma}\times A;\{0\}\times\pa\ol\R\times A]$ where $A=\{x\in X\colon|x|>c\}$. A basis of Q-vector fields, in the coordinates $\bhm,\sigma,x$ (or $r,\omega$ instead of $x$) is then for bounded $\sigma$ given by $\pa_{x^j}$ ($j=1,\ldots,n$) (or equivalently $\pa_r$ and spherical vector fields, which we schematically write as $\pa_\omega$), and for large $|\sigma|$ by $h\pa_{x^j}$ (or $h\pa_r$, $h\pa_\omega$).
\end{enumerate}

It remains to consider the subset of $X_\Qop$ where $|\hat x|>C$ and $r<c$.
\begin{enumerate}
\setcounter{enumi}{2}
\item Near the interior of $\zface\cap\mface$, we have local coordinates $\sigma\in\R$, $\hat\rho\geq 0$, $r\geq 0$, $\omega\in\Sph^{n-1}$. Q-vector fields are spanned by $r\pa_r-\hat\rho\pa_{\hat\rho}$, $\pa_\omega$.
\item Near the corner $\zface\cap\mface\cap\nface_+$, local coordinates are $h\geq 0$, $\tilde r\geq 0$, $\hat\rho\geq 0$, $\omega\in\Sph^{n-1}$, and Q-vector fields are spanned by $\tilde r\pa_{\tilde r}-\hat\rho\pa_{\hat\rho}$, $\pa_\omega$.
\item Near the corner $\mface\cap\nface_+\cap\iface_+$, local coordinates are $r\geq 0$, $\tilde\rho\geq 0$, $\tilde\sigma\geq 0$, $\omega$, with  $\tilde\rho$ a defining function of $\iface_+$. Q-vector fields are thus spanned by $\tilde\rho(r\pa_r-\tilde\rho\pa_{\tilde\rho})$, $\tilde\rho\pa_\omega$.
\item Near the corner $\nface_+\cap\iface_+\cap\sface$ finally, local coordinates are $r\geq 0$, $\hat\rho\geq 0$, $\tilde h\geq 0$, $\omega$, with $\hat\rho$ and $\tilde h$ being local defining functions of $\iface_+$ and $\sface$, respectively. Thus, Q-vector fields are spanned by $\hat\rho\tilde h(r\pa_r-\hat\rho\pa_{\hat\rho})$, $\hat\rho\tilde h\pa_\omega$.
\end{enumerate}

One can also give a more global description: in $|\hat x|\lesssim 1$, resp.\ $|\hat x|\gtrsim 1$, Q-vector fields are spanned by
\begin{equation}
\label{EqQSpan}
  \frac{h}{h+\bhm}\pa_{\hat x^j}\ (j=1,\ldots,n),\qquad \text{resp.}\qquad \frac{h}{h+r}r\pa_r,\ \frac{h}{h+r}\pa_\omega.
\end{equation}

\begin{definition}[Q-bundles]
\label{DefQBundle}
  We denote by $\TQ X\to X_\Qop$ the \emph{Q-vector bundle}, which is the vector bundle equipped with a smooth bundle map $\TQ X\to T X_\Qop$ with the property that $\VQ(X)=\CI(X_\Qop,\TQ X)$. The dual bundle $\TQ^*X$ is the \emph{Q-cotangent bundle}.
\end{definition}

We next study restrictions of Q-vector fields to various boundary hypersurfaces of $X_\Qop$. We use the notation from Appendix~\ref{SP}. The following result, based on~\eqref{EqQPieces} is the reason for the appearance of the various model problems in uniform singular analysis in the Q-setting:

\begin{lemma}[Restriction to boundary hypersurfaces]
\fakephantomsection
\label{LemmaQVFRes}
  \hspace{0in}\begin{enumerate}
  \item\label{ItQVFReszf} Restriction to $\zface$ induces a surjective map $N_\zface\colon\VQ(X)\to\CI(\ol\R;\Vb(\hat X))$ with kernel $\rho_\zface\VQ(X)$.
  \item\label{ItQVFResmf} Restriction to $\mface_{\sigma_0}$ induces a surjective map $N_{\mface_{\sigma_0}}\colon\VQ(X)\to\Vb(\dot X)$. Restriction to $\mface_{\pm,\semi}$ induces a surjective map $N_{\mface_{\pm,\semi}}\colon\VQ(X)\to\Vch(\dot X)$ (see~\S\usref{SsPch}). The kernel of $\oplus_{\sigma_0\in\R}N_{\mface_{\sigma_0}}$ is $\rho_\mface\VQ(X)$.
  \item\label{ItQVFResnf} Restriction to $\nface_{\pm,\low}$, resp.\ $\nface_{\pm,\tilde\semi}$ induces a surjective map $N_{\nface_{\pm,\low}}\colon\VQ(X)\to\Vscbt(\hat X)$ (see~\S\usref{SsPscbt}), resp.\ $N_{\nface_{\pm,\tilde\semi}}\colon\VQ(X)\to\Vsch(\hat X)$ (see~\S\usref{SsPbsc}). The kernel of $(N_{\nface_{\pm,\low}},N_{\nface_{\pm,\tilde\semi}})$ is $\rho_{\nface_\pm}\VQ(X)$.
  \end{enumerate}
\end{lemma}

We could leave $\mface$ in one piece; then restriction to $\mface$ induces a map from $\VQ(X)$ onto the space of b-vector fields on $[\ol{\R_\sigma}\times\dot X;\pa\ol\R\times\pa\dot X]$ which annihilate $\sigma$ and vanish at the lift of $\pa\ol\R\times\dot X$. This target space consists of smooth families of b-vector fields which degenerate like semiclassical cone vector fields as $|\sigma|\to\infty$. (An analogous remark applies to $\nface_\pm$.) The reason for splitting $\mface$ (or $\nface_\pm$) is that the analysis at high energies $|\sigma|\to\infty$ (or $|\tilde\sigma|\to\infty$) will be conceptually different from the analysis at bounded frequencies $\sigma$ (or $\tilde\sigma$).

\begin{proof}[Proof of Lemma~\usref{LemmaQVFRes}]
  We prove this using the coordinate systems and local spanning sets of $\VQ(X)$ listed before the statement of Lemma~\ref{LemmaQVFRes}. Thus, part~\eqref{ItQVFReszf} follows from the observation that the map $N_\zface$, in the coordinates $\bhm,\sigma,\hat x$, resp.\ $\sigma,\hat\rho,r,\omega$, maps $\pa_{\hat x^j}$ to itself ($j=1,\ldots,n$), resp.\ $r\pa_r-\hat\rho\pa_{\hat\rho}$, $\pa_\omega$ to $-\hat\rho\pa_{\hat\rho}$, $\pa_\omega$, with coefficients that are smooth on $\ol{\R_\sigma}\times\hat X$.

  For part~\eqref{ItQVFResmf}, consider first the case of bounded $\sigma$. The conclusion is then clear in $r>c>0$, whereas near $\mface\cap\zface$ and in the coordinates $\sigma,\hat\rho,r,\omega$, the map $N_\mface$ maps $r\pa_r-\hat\rho\pa_{\hat\rho}\mapsto r\pa_r$ and $\pa_\omega\mapsto\pa_\omega$, thus has range equal to smooth families (in $\sigma$) of elements of $\Vb(\dot X)$. In the coordinates $h,\tilde r,\hat\rho,\omega$ near $\mface\cap\nface_\pm\cap\zface$, with $\hat\rho$ a defining function of $\mface$, the map $N_\mface$ takes $\tilde r\pa_{\tilde r}-\hat\rho\pa_{\hat\rho}\mapsto\tilde r\pa_{\tilde r}$, $\pa_\omega\mapsto\pa_\omega$, thus its range consists of $\chop$-vector fields indeed. This is true also in the coordinates $r,\tilde\rho,\pm\tilde\sigma,\omega$ near $\mface\cap\nface_\pm\cap\iface_\pm$ (with $\pm\tilde\sigma$ defining $\mface$), in which $N_\mface$ maps $\tilde\rho(r\pa_r-\tilde\rho\pa_{\tilde\rho})$ and $\tilde\rho\pa_\omega$ to the same expressions; since the semiclassical face of $\dot X_\chop$ is defined by $\tilde\rho=0$, this proves part~\eqref{ItQVFResmf}.

  For part~\eqref{ItQVFResnf}, the maps $N_{\nface_\pm,\low}$ and $N_{\nface_\pm,\tilde\semi}$ are given by the restriction of coefficients of Q-vector fields, with respect to the bases listed in the various coordinate systems prior to the statement of Lemma~\ref{LemmaQVF}, to $\nface_\pm$. These vector fields are indeed $\scbtop$-vector fields on $\nface_{\pm,\low}$ (with scattering behavior at $\tilde\rho=0$, cf.\ the coordinate system near $\nface_\pm\cap\mface\cap\iface_\pm$), and semiclassical scattering vector fields on $\nface_{\pm,\tilde\semi}$ (with $\tilde h$ the semiclassical parameter, and with scattering behavior at $\hat\rho=0$, cf.\ the coordinate system near $\nface_+\cap\iface_+\cap\sface$).
\end{proof}

\begin{cor}[Bundle identifications]
\label{CorQBundle}
  The restriction maps of Lemma~\usref{LemmaQVFRes} induce bundle isomorphisms
  \begin{alignat*}{2}
    \TQ_\zface X &\cong \ol\R \times \Tb\hat X\ \text{(as bundles over $\zface=\ol\R\times\hat X$)},\hspace{-20em}& \\
    \TQ_{\mface_{\sigma_0}}X &\cong \Tb\dot X, &\qquad
    \TQ_{\mface_{\pm,\semi}}X &\cong \Tch\dot X, \\
    \TQ_{\nface_{\pm,\low}}X &\cong \Tscbt\hat X, &\qquad
    \TQ_{\nface_{\pm,\tilde\semi}}X &\cong {}^{\scop\tilde\semi}T\hat X,
  \end{alignat*}
  and $\TQ_{\nface_{\tilde\sigma_0}}X \cong\Tsc\hat X$, where $\sigma_0\in\R$ and $\tilde\sigma_0\in\R\setminus\{0\}$.
\end{cor}

\begin{definition}[Q-differential operators]
\label{DefQDiff}
  For $m\in\N_0$, we denote by $\DiffQ^m(X)$ the space of locally finite sums of up to $m$-fold compositions of elements of $\VQ(X)$ (a $0$-fold composition is, by definition an element of $\CI(X_\Qop)$). Given a collection $\alpha=(\alpha_H)$ of weights $\alpha_H\in\R$ for $H=\zface,\mface,\nface,\iface,\sface$, we denote more generally
  \[
    \DiffQ^{m,\alpha}(X) = \biggl(\prod_H\rho_H^{-\alpha_H}\biggr)\DiffQ^m(X) = \biggl\{ \biggl(\prod_H\rho_H^{-\alpha_H}\biggr) A \colon A \in \DiffQ^m(X) \biggr\}.
  \]
\end{definition}

Analogously to Q-vector fields, Q-differential operators $A\in\DiffQ^m(X)$ are smooth families $(\bhm,\sigma)\mapsto A_{\bhm,\sigma}\in\Diff^m(X)$ of differential operators on $X$ which degenerate in a particular fashion as $\bhm\to 0$, $|\sigma|\to\infty$, and/or $r\to 0$. Note that elements of $\DiffQ(X)$ commute with multiplication by $\bhm$ and $\sigma$, with
\begin{equation}
\label{EqQDiffEx}
  \bhm \in \DiffQ^{0,(-1,-1,-1,-1,0)}(X),\quad
  \sigma \in \DiffQ^{0,(0,0,1,1,1)}(X).
\end{equation}
Thus, for instance, it suffices to restrict in Definition~\ref{DefQDiff} to the case $\alpha_\mface=\alpha_\nface=0$. We also remark that a $\sigma$-independent q-differential operator $A\in\Diffq^m(X)$ defines an element $A\in\DiffQ^{m,(0,0,0,m,m)}(X)$; this is a consequence of the fact that $V\in\Vq(X)$, regarded as a $\sigma$-independent vector field on $X_\Qop$, satisfies $V\in\rho_\iface^{-1}\rho_\sface^{-1}\VQ(X)$, as follows directly from the definition. Recalling~\eqref{EqqVFX}, this implies that, regarding an operator on $X$ as an $\bhm$- and $\sigma$-independent operator on $X_\qop$ and $X_\Qop$,
\begin{equation}
\label{EqQDiffEx2}
  \Diff^m(X)\subset\rho_{\zface_\qop}^{-m}\Diffq^m(X)\subset\DiffQ^{m,(m,0,m,m,m)}(X).
\end{equation}

The principal symbol $\sigmaQ^1(V)$ of $V\in\VQ(X)$, defined as mapping $\xi\in\TQ^*X$ to $i\xi(V)$, is a fiber-linear function. The property~\eqref{EqQVFComm} implies that the principal symbol extends to a multiplicative family of maps $\sigmaQ^m$ with the property that
\begin{equation}
\label{EqQDiffSym}
  0 \to \rho_\sface\rho_\iface\DiffQ^{m-1}(X) \hra \DiffQ^m(X) \xra{\sigmaQ^m} P^m(\TQ^*X)/\rho_\sface\rho_\iface P^{m-1}(\TQ^*X) \to 0
\end{equation}
is a short exact sequence. By Lemma~\ref{LemmaQVFRes}, we get multiplicative normal operator maps
\begin{equation}
\label{EqQNormOp}
\begin{alignedat}{2}
  N_\zface &\colon \DiffQ^m(X) \to \CI(\ol\R;\Diffb^m(\hat X)), \\
  N_{\mface_{\sigma_0}} &\colon \DiffQ^m(X) \to \Diffb^m(\dot X), &\qquad
  N_{\mface_{\pm,\semi}} &\colon \DiffQ^m(X) \to \Diffch^m(\dot X), \\
  N_{\nface_{\pm,\low}} &\colon \DiffQ^m(X) \to \Diffscbt^m(\hat X), &\qquad
  N_{\nface_{\pm,\tilde\semi}} &\colon \DiffQ^m(X) \to \Diff_{\scop\tilde\semi}^m(\hat X),
\end{alignedat}
\end{equation}
as well as similar maps on spaces of weighted operators (with the weight at $H$ required to be $0$ in the definition of $N_H$). Moreover, the principal symbol of $N_\zface(P)$ is given by the restriction of $\sigmaQ^m(P)$ to $\TQ^*_\zface X\cong\ol\R\times\Tb^*\hat X$ via Corollary~\ref{CorQBundle}, similarly for the principal symbols of the other normal operators. Note also that the vanishing of $N_\zface(P)$, resp.\ $N_{\mface_{\sigma_0}}(P)$ for all $\sigma_0$, resp.\ $N_{\nface_{\pm,\low}}(P)$ and $N_{\nface_{\pm,\tilde\semi}}(P)$ implies that $P$ vanishes to leading order the appropriate boundary hypersurface, i.e.\ $P\in\rho_\zface\DiffQ^m(X)$, resp.\ $P\in\rho_\mface\DiffQ^m(X)$, resp.\ $P\in\rho_{\nface_\pm}\DiffQ^m(X)$. Together with $\sigmaQ^m(P)$, these normal operators thus capture $P$ to leading order in all $6$ senses (corresponding to the $6$ orders in Definition~\ref{DefQDiff}).

Furthermore, we can restrict to level sets $\sigma^{-1}(\sigma_0)$ or $\tilde\sigma(\tilde\sigma_0)$ for $\sigma_0\in\R$ or $\tilde\sigma_0\in\R\setminus\{0\}$. This gives normal operator homomorphisms
\[
  N_{\sigma_0} \colon \DiffQ^m(X)\to\Diffq^m(X),\qquad
  N_{\nface_{\tilde\sigma_0}} \colon \DiffQ^m(X)\to\Diffsc^m(\hat X).
\]

See~\S\ref{SsKS} for the way in which the spectral family of interest in Theorem~\ref{ThmI} fits into this framework of Q-analysis.

\subsection{Q-pseudodifferential operators}
\label{SsQP}

The microlocal analysis of Q-differential operators relies on a corresponding Q-pseudodifferential algebra, which we proceed to define; analogously to Q-differential operators, a Q-ps.d.o.\ $A$ will be a smooth family $\R\times(0,1]\ni(\sigma,\bhm)\mapsto A_{\sigma,\bhm}$ of ordinary ps.d.o.s on a manifold $X$ without boundary.

\begin{definition}[Q-double space]
\label{DefQPDouble}
  Recall the q-double space $X_\qop^2$ of $X$ and its submanifolds $\zface_{\qop,2}\cong\hat X^2_\bop$, $\mface_{\qop,2}\cong\dot X^2_\bop$, $\lb_{\qop,2}$, $\rb_{\qop,2}$, and $\diag_\qop$ from Definition~\ref{DefqDouble}. The \emph{Q-double space} of $X$ is then defined as the resolution of $\ol{\R_\sigma}\times X^2_\qop$ given by
  \begin{equation}
  \label{EqQPDouble}
  \begin{split}
    X^2_\Qop &:= \bigl[ \ol\R \times X^2_\qop; \pa\ol\R\times\zface_{\qop,2}; \pa\ol\R\times\diag_\qop; \\
      &\quad\qquad \pa\ol\R\times(\diag_\qop\cap\,\mface_{\qop,2}), \pa\ol\R\times\lb_{\qop,2}, \pa\ol\R\times\rb_{\qop,2}; \pa\ol\R\times\mface_{\qop,2} \bigr].
  \end{split}
  \end{equation}
  We label its boundary hypersurfaces as follows:
  \begin{enumerate}
  \item $\zface_2$ is the lift of $\ol\R\times\zface_{\qop,2}$;
  \item $\mface_2$ is the lift of $\ol\R\times\mface_{\qop,2}$;
  \item $\nface_2$ is the lift of $\pa\ol\R\times\zface_{\qop,2}$;
  \item $\iface_2$ is the lift of $\pa\ol\R\times(\diag_\qop\cap\,\mface_{\qop,2})$, and $\iface'_2$ is the lift of $\pa\ol\R\times\mface_{\qop,2}$;
  \item $\sface_2$ is the lift of $\pa\ol\R\times\diag_\qop$, and $\sface'_2$ is the lift of $\pa\ol\R\times X^2_\qop$;
  \item $\lb_2$, resp.\ $\rb_2$ is the lift of $\ol\R\times\lb_{\qop,2}$, resp.\ $\ol\R\times\rb_{\qop,2}$;
  \item $\tlb_2$, resp.\ $\trb_2$ is the lift of $\pa\ol\R\times\lb_{\qop,2}$, resp.\ $\pa\ol\R\times\rb_{\qop,2}$.
  \end{enumerate}
  We denote by $\nface_{2,\pm}$ the connected components of $\nface_2$ corresponding to the value of $\sigma=\pm\infty$; similarly for $\iface_{2,\pm}$, $\iface_{2,\pm}'$, $\sface_{2,\pm}$, $\sface_{2,\pm}'$, $\tlb_{2,\pm}$, $\trb_{2,\pm}$. Furthermore, we write for $\sigma_0\in\R$ and $\tilde\sigma_0\in\R\setminus\{0\}$
  \begin{alignat*}{2}
    \mface_{2,\sigma_0} &:= \mface_2 \cap \sigma^{-1}(\sigma_0),&\qquad
    \mface_{2,\pm,\semi} &:= \mface_{2,\pm} \cap \sigma^{-1}(\pm[1,\infty]), \\
    \nface_{2,\pm,\low} &:= \nface_{2,\pm}\cap\tilde\sigma^{-1}(\pm[0,1]),&\qquad
    \nface_{2,\pm,\tilde\semi} &:= \nface_{2,\pm}\cap\tilde\sigma^{-1}(\pm[1,\infty]),
  \end{alignat*}
  and $\nface_{2,\tilde\sigma_0}:=\nface_2\cap\tilde\sigma^{-1}(\tilde\sigma_0)$. Finally, $\diag_\Qop$ denotes the lift of $\ol\R\times\diag_\qop$.
\end{definition}

In~\eqref{EqQPDouble}, note that $\pa\ol\R\times\lb_{\qop,2}$, $\pa\ol\R\times\rb_{\qop,2}$, and $\pa\ol\R\times(\diag_\qop\cap\,\mface_{\qop,2})$ are disjoint, and hence they can be blown up in any order.

\begin{lemma}[b-fibrations from the Q-double space]
\label{LemmaQPbFib}
  The left projection, resp.\ right projection $\R\times(0,1]\times X\times X\ni(\sigma,\bhm,x,x')\mapsto(\sigma,\bhm,x)\in\R\times(0,1]\times X$, resp.\ $(\sigma,\bhm,x')$ lifts to a b-fibration $\pi_L$, resp.\ $\pi_R\colon X^2_\Qop\to X_\Qop$.
\end{lemma}
\begin{proof}
  We only discuss the case of the left projection. Using Lemma~\ref{LemmaqbFib}, we start with the fact that the left projection lifts to a b-fibration $\tilde\pi_L\colon\ol\R\times X_\qop^2\to\ol\R\times X_\qop$; the preimages of the centers in~\eqref{EqQSingleXq} under it are
  \begin{equation}
  \label{EqQPbFibPre}
  \begin{split}
    \tilde\pi_L^{-1}(\pa\ol\R\times\zface_\qop) &= (\pa\ol\R\times\zface_{\qop,2}) \cup (\pa\ol\R\times\lb_{\qop,2}), \\
    \tilde\pi_L^{-1}(\pa\ol\R\times\mface_\qop) &= (\pa\ol\R\times\mface_{\qop,2}) \cup (\pa\ol\R\times\rb_{\qop,2}).
  \end{split}
  \end{equation}
  From the first line and \cite[Proposition~5.12.1]{MelroseDiffOnMwc}, we deduce that the lift of $\tilde\pi_L$ to
  \[
    \bigl[\ol\R\times X_\qop^2; \pa\ol\R\times\zface_{\qop,2}; \pa\ol\R\times\lb_{\qop,2} \bigr] \to \bigl[ \ol\R\times X_\qop; \pa\ol\R\times\zface_\qop \bigr]
  \]
  is a b-fibration. Since this is b-transversal to the lift of $\pa\ol\R\times\diag_\qop$ (which is mapped diffeomorphically to a copy of $X_\qop$), this lifts to a b-fibration
  \[
    \bigl[\ol\R\times X_\qop^2; \pa\ol\R\times\zface_{\qop,2}; \pa\ol\R\times\diag_\qop, \pa\ol\R\times\lb_{\qop,2} \bigr] \to \bigl[ \ol\R\times X_\qop; \pa\ol\R\times\zface_\qop \bigr].
  \]
  By~\eqref{EqQPbFibPre}, the preimage of the lift of $\pa\ol\R\times\mface_\qop$ under this map is the union of the lifts of $\pa\ol\R\times\mface_{\qop,2}$, $\pa\ol\R\times(\diag_\qop\cap\,\mface_{\qop,2})$, and $\pa\ol\R\times\rb_{\qop,2}$. By \cite[Proposition~5.11.2]{MelroseDiffOnMwc}, the lift
  \[
    \bigl[\ol\R\times X_\qop^2; \pa\ol\R\times\zface_{\qop,2}; \pa\ol\R\times\diag_\qop, \pa\ol\R\times\lb_{\qop,2}; \pa\ol\R\times\rb_{\qop,2}; \pa\ol\R\times(\diag_\qop\cap\mface_{\qop,2}); \pa\ol\R\times\mface_{\qop,2} \bigr] \to X_\Qop
  \]
  is therefore a b-fibration. This finishes the proof.
\end{proof}

\begin{definition}[Q-pseudodifferential operators]
\label{DefQP}
  Let $s\in\R$ and $\alpha=(\alpha_H)$ where $\alpha_H\in\R$ for $H=\mface,\zface,\nface,\iface,\sface$. Then $\PsiQ^{s,\alpha}(X)$ consists of all smooth families $A=(A_{\bhm,\sigma})_{\bhm\in(0,1],\sigma\in\R}$ of bounded linear operators on $\CIc(X)$ whose Schwartz kernels are elements of
  \begin{equation}
  \label{EqQP}
    \rho_{\zface_2}^{-\alpha_\zface}\rho_{\mface_2}^{-\alpha_\mface}\rho_{\nface_2}^{-\alpha_\nface}\rho_{\iface_2}^{-\alpha_\iface}\rho_{\sface_2}^{-\alpha_\sface}I^{m-\frac12}(X^2_\Qop,\diag_\Qop;\pi_R^*\OmegaQ X)
  \end{equation}
  which are conormal down to all boundary hypersurfaces of $X^2_\Qop$ and vanish to infinite order at all boundary hypersurfaces other than $\mface_2$, $\zface_2$, $\nface_2$, $\iface_2$, $\sface_2$ (and the lift of $\bhm^{-1}(1)$). The subspace of operators whose Schwartz kernels are classical conormal at $\zface_2,\mface_2,\nface_2$ is denoted $\Psi_{\Qop,\cl}^{s,\alpha}(X)$.
\end{definition}

\begin{rmk}[Defining functions]
\label{RmkQPDefFn}
  Note that $\pi_L^{-1}(\zface)=\zface_2\cup\lb_2$, and indeed the defining function $\zface$ lifts to $X_\Qop^2$ under $\pi_L$ to a product of defining functions of $\zface_2$ and $\lb_2$. In view of the infinite order of vanishing of Schwartz kernels of Q-ps.d.o.s at $\lb_2$, we can therefore replace the weight $\rho_{\zface_2}$ in~\eqref{EqQP} by (the left lift of) $\rho_\zface$. Similarly,
  \begin{alignat*}{2}
    \pi_L^{-1}(\mface)&=\mface_2\cup\rb_2,&\qquad
    \pi_L^{-1}(\nface)&=\nface_2\cup\tlb_2, \\
    \pi_L^{-1}(\iface)&=\iface_2\cup\iface'_2\cup\trb_2,&\qquad
    \pi_L^{-1}(\sface)&=\sface_2\cup\sface'_2.
  \end{alignat*}
  Similar statements hold for $\pi_R$ in place of $\pi_L$. Together, they imply that $\PsiQ^{s,\alpha}(X)$ is invariant under conjugation by weights $\prod\rho_H^{-\alpha_H}$ on $X_\Qop$.
\end{rmk}

For local coordinate descriptions, we shall use the smooth functions on $X_\Qop^2$ obtained by lifting coordinates on $X_\Qop$ to the left, resp.\ right factor; the left lift will be denoted by the same symbol, and the right lift with the primed symbol. For example, $\hat x$ and $\hat x'$ denote the left and right lift of the function on $X_\Qop$ denoted $\hat x$ in~\eqref{EqqCoord2}.

For bounded $\sigma$, Q-ps.d.o.s are smooth families (in $\sigma$) of q-ps.d.o.s, for which a local coordinate description was given in~\eqref{EqqQuant}. Consider next the region $|\hat x|,|\hat x'|\lesssim 1$ for $\sigma\gtrsim 1$. Near $\{\infty\}\times(\diag_\qop\cap\,\zface_{\qop,2}^\circ)\subset\ol\R\times X_\qop^2$, we can then use local coordinates $h\geq 0$, $\bhm\geq 0$, $\hat x'$, and $y:=\hat x-\hat x'$, with the diagonal defined by $y=0$. Upon blowing up $h=\bhm=0$, the lift of $h=0$ is defined by $\frac{h}{h+\bhm}=0$; upon passing to the subsequent blow-up of the lift of $\pa\ol\R\times\diag_\qop$, coordinates near the Q-diagonal are thus
\[
  y_\Qop:=\frac{y}{h/(h+\bhm)},
\]
and therefore a typical element of $\PsiQ^{s,\alpha}(X)$ is given by
\begin{equation}
\label{EqQPQuant1}
  (\Op_{\Qop,\bhm,h^{-1}}(a)u)(\hat x) = (2\pi)^{-n}\int \exp\Bigl(i\frac{\hat x-\hat x'}{h/(h+\bhm)}\cdot\xi\Bigr) \chi\bigl(|\hat x-\hat x'|\bigr) a(h,\bhm,\hat x,\xi)\,\dd\xi
\end{equation}
where $a$ is a symbol, or more precisely $a$ is conormal on $X_\Qop\times\ol{\R^n}$ with order $\alpha_H$ at $H\times\ol{\R^n}$ for $H=\zface,\nface,\sface$, and order $s$ at $X_\qop\times\pa\ol{\R^n}$; and $\chi\in\CIc((-\half,\half))$ is identically $1$ near $0$. Thus,~\eqref{EqQPQuant1} is essentially a semiclassical ps.d.o.\ with semiclassical parameter $\frac{h}{h+\bhm}$. We also note that the left lift of the basis $\frac{h}{h+\bhm}\pa_{\hat x^j}$ of $\VQ(X)$ in this coordinate system (see~\eqref{EqQSpan}) is given by $\pa_{y_\Qop^j}$, which is transversal to $\diag_\Qop=y_\Qop^{-1}(0)$.

Working in the region $|\hat x|,|\hat x'|\gtrsim 1$ for $\sigma\gtrsim 1$, we can use as smooth coordinates near $\{\infty\}\times\diag_\qop\subset\ol\R\times X^2_\qop$ the functions $h\geq 0$, $\frac{\bhm}{r'}\geq 0$, $r'\geq 0$, $\omega'\in\R^{n-1}$, $z=\frac{r-r'}{r'}$, $w=\omega-\omega'\in\R^{n-1}$ where we fix local coordinates on $\Sph^{n-1}$. Upon blowing up $\pa\ol\R\times\zface_{\qop,2}$ (given by $h=r'=0$), the lift of $h=0$ is given by $\frac{h}{h+r'}=0$; passing to the subsequent blow-up of the lift of $\pa\ol\R\times\diag_\qop$, coordinates transversal to the lifted diagonal are thus
\[
  (z_\Qop,w_\Qop) := \frac{(z,w)}{h/(h+r')}.
\]
These coordinates remain transversal to the lift of the diagonal to the subsequent blow-ups in~\eqref{EqQPDouble}. Thus, an element of $\PsiQ^{s,\alpha}(X)$ is given by
\begin{equation}
\label{EqQPQuant2}
\begin{split}
  (\Op_{\Qop,\bhm,h^{-1}}(a)u)(r,\omega) &= (2\pi)^{-n}\iint \exp\Bigl[i\Bigl(\frac{r-r'}{r'\frac{h}{h+r'}}\xi + \frac{\omega-\omega'}{h/(h+r')}\cdot\eta\Bigr)\Bigr] \\
    &\quad\hspace{6em} \times \chi\Bigl(\Bigl|\frac{r-r'}{r'}\Bigr|\Bigr)\chi(|\omega-\omega'|) a(h,\bhm,r,\xi,\eta)\,\dd\xi\,\dd\eta,
\end{split}
\end{equation}
where $a$ is conormal on $X_\Qop\times\ol{\R^n_{(\xi,\eta)}}$ with order $\alpha_H$ at $H\times\ol\R^n$ for all boundary hypersurfaces $H\subset X_\Qop$, and order $s$ at $X_\Qop\times\pa\ol{\R^n_{(\xi,\eta)}}$. Since the second spanning set of Q-vector fields in~\eqref{EqQSpan} lifts to the left factor of $X_\Qop^2$ as $\pa_{z_\Qop}$, $\pa_{w_\Qop}$, we conclude that also in this region the left lift of $\VQ(X)$ is transversal to $\diag_\Qop$.

As a consequence of the two transversality statements, we obtain a bundle isomorphism $\TQ X\cong T_{\diag_\Qop}X_\Qop^2/T\diag_\Qop=N\diag_\Qop$ given by the left lift; and therefore
\begin{equation}
\label{EqQPCon}
  N^*\diag_\Qop \cong \TQ^*X.
\end{equation}
Moreover, for $m\in\N_0$, we conclude that $\DiffQ^{m,\alpha}(X)\subset\PsiQ^{m,\alpha}(X)$ consists of those operators whose Schwartz kernels are Dirac distributions at $\diag_\Qop$. Generalizing~\eqref{EqQDiffSym}, the principal symbol map $\sigmaQ^{s,\alpha}$ on $\PsiQ^{s,\alpha}(X)$ fits into the short exact sequence
\[
  0 \to \rho_\iface\rho_\sface\PsiQ^{s,\alpha}(X) \hra \PsiQ^{s,\alpha}(X) \xra{\sigmaQ^{s,\alpha}} (S^{s,\alpha}/\rho_\iface\rho_\sface S^{s-1,\alpha})(\TQ^*X) \to 0.
\]
Finally, we conclude that pushforward along $\pi_L$ is a continuous map from $\Psi_\Qop^s(X)$, resp. $\Psi_{\Qop,\cl}^s(X)$ into $\cA^0(X_\Qop)$, resp.\ $\CI(X_\Qop)$; thus, Q-ps.d.o.s define bounded linear maps on $\cA^0(X_\Qop)$, or on $\CI(X_\Qop)$ for classical ps.d.o.s.

We may allow for the orders $s$, $\alpha_\iface$, $\alpha_\sface$ to be variable; in this paper we only need to consider the case that the $\iface$-order is variable,
\[
  \upalpha_\iface \in \CI(\ol{\TQ^*_\iface}X),
\]
while $s,\alpha_\sface$ are constant; for $\alpha=(\alpha_\zface,\alpha_\mface,\alpha_\nface,\upalpha_\iface,\alpha_\sface)$, the principal symbol map then takes values in $(S^{s,\alpha}/\rho_\iface^{1-2\delta}\rho_\sface S^{s-1,\alpha})(\TQ^*X)$ for any $\delta>0$.

In order to study the normal operators of Q-ps.d.o.s, we need the following result, which is the double space analogue of Lemma~\ref{LemmaQVFRes}:

\begin{prop}[Boundary hypersurfaces of $X_\Qop^2$]
\label{PropQPBdy}
  We have the following natural diffeomorphisms:
  \begin{enumerate}
  \item\label{ItQPBdyzf2} $\zface_2\cong\ol\R\times\hat X^2_\bop$;
  \item\label{ItQPBdymf2lo} $\mface_{2,\sigma_0}\cong\dot X^2_\bop$ (for $\sigma_0\in\R$);
  \item\label{ItQPBdymf2hi} $\mface_{2,\pm,\semi}\cong\dot X^2_\chop$ (see~\S\usref{SsPch}) with semiclassical parameter $h=|\sigma|^{-1}$;
  \item\label{ItQPBdynf2lo} $\nface_{2,\pm,\low}\cong\hat X^2_\scbtop$ (see~\S\usref{SsPscbt}) with spectral parameter $\tilde\sigma=\bhm\sigma$.
  \item\label{ItQPBdynf2hi} $\nface_{2,\pm,\tilde\semi}\cong\hat X^2_\schop$ (see~\S\usref{SsPbsc}) with semiclassical parameter $\tilde h=|\tilde\sigma|^{-1}$.
  \end{enumerate}
  That is, the local coordinates $\sigma,\hat x,\hat x'$ restrict to a map $\zface_2^\circ\to\R\times\R^n_{\hat x}\times\R^n_{\hat x'}$ which extends by continuity to the diffeomorphism in part~\eqref{ItQPBdyzf2}; similarly for the other diffeomorphisms.
\end{prop}
\begin{proof}
  We obtain $\zface_2$ by first blowing up $\pa\ol\R\times\zface_{\qop,2}\subset\ol\R\times\zface_{\qop,2}$, which thus does not change the smooth structure of $\ol\R\times\zface_{\qop,2}$; the lifts of the remaining submanifolds in~\eqref{EqQPDouble} to $[\ol\R\times X^2_\qop;\pa\ol\R\times\zface_{\qop,2}]$ are disjoint from the lift of $\ol\R\times\zface_{\qop,2}$. This proves part~\eqref{ItQPBdyzf2}.

  Next, $\mface_2$ arises from $\ol\R\times\mface_\qop=\ol\R\times\dot X^2_\bop$ (see Lemma~\ref{LemmaqBdy}) by first blowing up its intersection $\pa\ol\R\times\ff_\bop$ with $\pa\ol\R\times\zface_{\qop,2}$, where $\ff_\bop$ denotes the front face of $\dot X^2_\bop$; then one blows up the intersection with the lift of $\pa\ol\R\times\diag_\qop$, which is equal to the intersection with the lift of $\pa\ol\R\times(\diag_\qop\cap\,\mface_{\qop,2})$ and thus given by the lift of $\pa\ol\R\times\diag_\bop$ to $[\ol\R\times\dot X^2_\bop;\pa\ol\R\times\ff_\bop]$. This blow-up thus produces
  \[
    \bigl[ \ol\R\times\dot X^2_\bop; \pa\ol\R\times\ff_\bop; \pa\ol\R\times\diag_\bop \bigr].
  \]
  The intersection of this space with the lift of $\pa\ol\R\times\lb_{\qop,2}$ is $\pa\ol\R\times\lb_\bop$, similarly for the right boundary, and hence blowing up both of these lifts produces
  \begin{equation}
  \label{EqQPBdymf2}
    \bigl[ \ol\R\times\dot X^2_\bop; \pa\ol\R\times\ff_\bop, \pa\ol\R\times\lb_\bop, \pa\ol\R\times\rb_\bop, \pa\ol\R\times\diag_\bop \bigr].
  \end{equation}
  The intersections of this space with the lift of $\pa\ol\R\times(\diag_\qop\cap\,\mface_{\qop,2})$ or with the lift of $\pa\ol\R\times\mface_{\qop,2}$ are both boundary hypersurfaces, hence their blow-up does not affect the smooth structure. Upon intersecting the space~\eqref{EqQPBdymf2} with $\sigma^{-1}(\sigma_0)$ or $\sigma^{-1}(\pm[1,\infty])$, we thus obtain the isomorphisms stated in parts~\eqref{ItQPBdymf2lo} and \eqref{ItQPBdymf2hi}.

  Finally, we consider $\nface_{2,+}$. Let $[0,\eps)_{\rho_{\zface_{\qop,2}}}\times\zface_{\qop,2}$ be a collar neighborhood of $\zface_{\qop,2}\subset X^2_\qop$. We take $\rho_{\zface_{\qop,2}}=\sqrt{\bhm^2+|x|^2+|x'|^2}$ for concreteness. Then the front face of $[\ol\R\times X_\qop^2;\{\infty\}\times\zface_{\qop,2}]$ is that of $[[0,1]_h\times[0,\eps)_{\rho_{\zface_{\qop,2}}}\times\zface_{\qop,2};\{0\}\times\{0\}\times\zface_{\qop,2}]$, and thus equal to
  \[
    \nface_{2,+}' := [0,\infty]_{\tilde h'} \times \zface_{\qop,2} = [0,\infty]_{\tilde h'} \times \hat X^2_\bop,\qquad \tilde h':=\frac{h}{\rho_{\zface_{\qop,2}}}.
  \]
  Its intersections with the lifts of
  \[
    \{\infty\}\times\diag_\qop,\quad
    \{\infty\}\times(\diag_\qop\cap\,\mface_{\qop,2}), \quad
    \{\infty\}\times\lb_{\qop,2},\quad
    \{\infty\}\times\rb_{\qop,2},\quad
    \{\infty\}\times\mface_{\qop,2}
  \]
  with $\nface_{2,+}'$ are given by
  \[
    \{0\}\times\diag_\bop,\quad
    \{0\}\times\pa\diag_\bop, \quad
    \{0\}\times\lb_\bop,\quad
    \{0\}\times\rb_\bop,\quad
    \{0\}\times\ff_\bop,
  \]
  respectively; we need to blow these up in the listed order. In fact, the first two blow-ups can be performed after the third and fourth (since the first/second and third/fourth submanifolds are disjoint); then, since $\pa\diag_\bop=\diag_\bop\cap\,\ff_\bop$, one can blow up $\{0\}\times\diag_\bop$, $\{0\}\times\pa\diag_\bop$, and $\{0\}\times\ff_\bop$ in the order $\{0\}\times\ff_\bop$, $\{0\}\times\pa\diag_\bop$, $\{0\}\times\diag_\bop$. Thus,
  \begin{equation}
  \label{EqQPBdynf2funny}
    \nface_{2,+} = \bigl[ [0,\infty]_{\tilde h'}\times\hat X_\bop^2; \{0\}\times\lb_\bop, \{0\}\times\rb_\bop; \{0\}\times\ff_\bop; \{0\}\times\pa\diag_\bop; \{0\}\times\diag_\bop \bigr].
  \end{equation}

  To analyze this space, we introduce $\hat\rho_\tot:=(1+|\hat x|^2+|\hat x'|^2)^{-\frac12}=\hat\rho_{\lb_\bop}\hat\rho_{\ff_\bop}\hat\rho_{\rb_\bop}$, which is a total boundary defining function of $\hat X^2_\bop$. We then claim that the change of coordinates map $(\tilde h',\hat x,\hat x')\mapsto(\tilde\sigma,\hat x,\hat x')$ with $\tilde\sigma=(1+|\hat x|^2+|\hat x'|^2)^{-\frac12}/\tilde h'=\frac{\hat\rho_\tot}{\tilde h'}$ induces a diffeomorphism\footnote{This is the analogue, in the double space setting, of the isomorphism~\eqref{EqQnfBlowup}.} 
  \begin{equation}
  \label{EqQPBdynf2partial}
  \begin{split}
    &\bigl[ [0,\infty]_{\tilde h'}\times\hat X_\bop^2; \{0\}\times\lb_\bop, \{0\}\times\rb_\bop; \{0\}\times\ff_\bop \bigr] \\
    &\qquad \cong \bigl[ [0,\infty]_{\tilde\sigma} \times \hat X_\bop^2; \{0\}\times\ff_\bop; \{0\}\times\lb_\bop, \{0\}\times\rb_\bop \bigr].
  \end{split}
  \end{equation}
  (See Figure~\ref{FigQPBdynf2}.) This is clear over the interior $(\hat X^\circ)^2$ of $\hat X^2_\bop$. We have $\ff_\bop\cong[0,\infty]_s\times(\pa\hat X)^2$ where $s=\frac{\hat\rho}{\hat\rho'}$ with $\hat\rho=|\hat x|^{-1}=\hat\rho_{\lb_\bop}\hat\rho_{\ff_\bop}$ and $\hat\rho'=|\hat x'|^{-1}=\hat\rho_{\rb_\bop}\hat\rho_{\ff_\bop}$ for suitable defining functions $\hat\rho_{\lb_\bop}$, $\hat\rho_{\ff_\bop}$, $\hat\rho_{\rb_\bop}$ of $\lb_\bop$, $\ff_\bop$, $\rb_\bop\subset\hat X^2_\bop$, so
  \[
    s=\frac{\hat\rho_{\lb_\bop}}{\hat\rho_{\rb_\bop}}.
  \]
  Thus, a collar neighborhood of $\ff_\bop\subset\hat X^2_\bop$ is given by $[0,\eps)_{\hat\rho_{\ff_\bop}}\times[0,\infty]_s\times(\pa X)^2$. Upon dropping the factor $(\pa X)^2$, the claim~\eqref{EqQPBdynf2partial} thus reads
  \begin{equation}
  \label{EqQPBdynf2partial2}
  \begin{split}
    &\bigl[ [0,\infty]_{\tilde h'} \times [0,\eps)_{\hat\rho_{\ff_\bop}} \times [0,\infty]_s; \{0\}\times[0,\eps)\times\{0\}, \{0\}\times[0,\eps)\times\{\infty\}; \{0\} \times \{0\} \times [0,\infty] \bigr] \\
    &\quad \cong \bigl[ [0,\infty]_{\tilde\sigma} \times [0,\eps) \times [0,\infty]; \{0\}\times\{0\}\times[0,\infty]; \{0\}\times[0,\eps)\times\{0\}, \{0\}\times[0,\eps)\times\{\infty\} \bigr]
  \end{split}
  \end{equation}
  via the change of coordinates map $\kappa\colon(\tilde h',\hat\rho_{\ff_\bop},s)\mapsto(\frac{\hat\rho_{\lb_\bop}\hat\rho_{\ff_\bop}\hat\rho_{\rb_\bop}}{\tilde h'},\hat\rho_{\ff_\bop},s)$, where we put $\hat\rho_{\lb_\bop}=\frac{s}{s+1}$ and $\hat\rho_{\rb_\bop}=\frac{1}{s+1}$. The proof of~\eqref{EqQPBdynf2partial2} proceeds by explicit calculations in local coordinate systems, and is pictorially given in Figure~\ref{FigQPBdynf2}.
  \begin{figure}[!ht]
  \centering
  \includegraphics{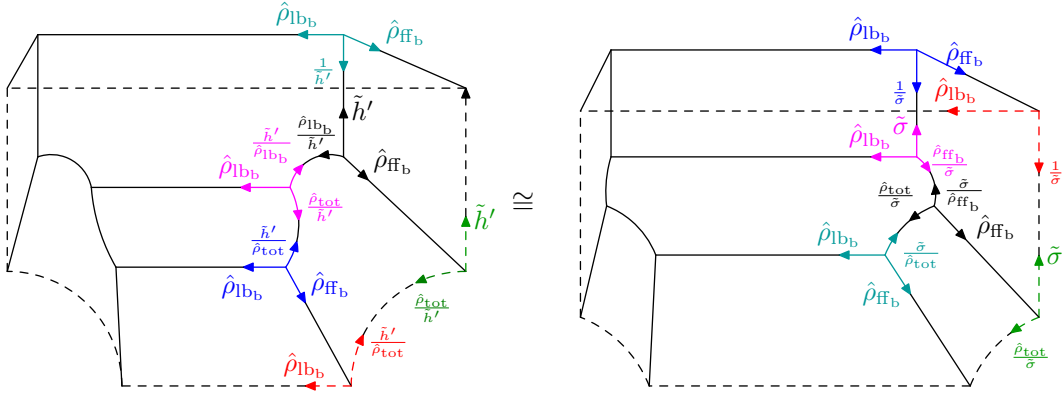}
  \caption{Illustration of (the proof of) the diffeomorphism~\eqref{EqQPBdynf2partial2}. \textit{On the left:} the space on the left in~\eqref{EqQPBdynf2partial2}. \textit{On the right:} the space on the right in~\eqref{EqQPBdynf2partial2}. Also shown are matching local coordinate systems near the various boundary faces; in the listed coordinates systems, we have $\hat\rho_\tot\sim\hat\rho_{\lb_\bop}\hat\rho_{\ff_\bop}$, and we also recall that $\tilde\sigma=\hat\rho_\tot/\tilde h'$.}
  \label{FigQPBdynf2}
  \end{figure}
  Using the diffeomorphism~\eqref{EqQPBdynf2partial} in~\eqref{EqQPBdynf2funny}, we then find that
  \[
    \nface_{2,+} = \bigl[ [0,\infty]_{\tilde\sigma}\times\hat X_\bop^2; \{0\}\times\ff_\bop; \{0\}\times\lb_\bop, \{0\}\times\rb_\bop; [0,\infty]\times\pa\diag_\bop; \{\infty\}\times\diag_\bop \bigr].
  \]
  This implies parts~\eqref{ItQPBdynf2lo} and \eqref{ItQPBdynf2hi}. The proof is complete.
\end{proof}

The relationship between the semiclassical, resp.\ doubly semiclassical cone algebras of \cite{HintzConicPowers} and the Q-algebra in the intermediate semiclassical regime $|\sigma|\sim\bhm^{-1}$ (mentioned in the discussion of the \emph{very large frequency regime} in~\S\ref{SsIA}), resp.\ fully semiclassical regime $|\sigma|\gg\bhm^{-1}$ is described in Appendix~\ref{SQSemi}.

We now switch to a less cumbersome notation for the weights, writing $l=\alpha_\zface$, $\gamma=\alpha_\mface$, $l'=\alpha_\nface$, $r=\alpha_\iface$, $b=\alpha_\sface$.

\begin{cor}[Normal operators]
\label{CorQPNormal}
  Restricting Schwartz kernels of classical Q-ps.d.o.s to the boundary hypersurfaces $\zface_2$, $\mface_{2,\pm,\semi}$, $\nface_{2,\pm,\low}$, and $\nface_{2,\pm,\tilde\semi}$ defines surjective normal operator maps
  \begin{alignat*}{2}
    N_\zface &\colon \Psi_{\Qop,\cl}^{s,(0,\gamma,l',r,b)}(X) &&\to \CI\bigl(\ol\R;\Psib^{s,\gamma}(\hat X)\bigr), \\
    N_{\mface_{\pm,\semi}} &\colon \Psi_{\Qop,\cl}^{s,(l,0,l',r,b)}(X) &&\to \Psich^{s,l,l',r}(\dot X), \\
    N_{\nface_{\pm,\low}} &\colon \Psi_{\Qop,\cl}^{s,(l,\gamma,0,r,b)}(X) &&\to \Psiscbt^{s,r,\gamma,l}(\hat X), \\
    N_{\nface_{\pm,\tilde\semi}} &\colon \Psi_{\Qop,\cl}^{s,(l,\gamma,0,r,b)}(X) &&\to \Psi_{\scop\tilde\semi}^{s,r,b}(\hat X).
  \end{alignat*}
  Moreover, for $\sigma_0\in\R$ and $\tilde\sigma_0\in\R\setminus\{0\}$, restriction to $\sigma^{-1}(\sigma_0)$, $\sigma^{-1}(\sigma_0)\cap\mface_2$, and $\nface_{2,\tilde\sigma_0}$ defines surjective maps
  \begin{alignat*}{2}
    N_{\sigma_0} &\colon \PsiQ^{s,(l,\gamma,l',r,b)}(X) &&\to \Psiq^{s,(l,\gamma)}(X), \\
    N_{\mface_{\sigma_0}} &\colon \PsiQ^{s,(l,0,l',r,b)}(X) &&\to \Psib^{s,l}(\dot X), \\
    N_{\nface_{\tilde\sigma_0}} &\colon \PsiQ^{s,(l,\gamma,0,r,b)}(X) &&\to \Psisc^{s,r}(\hat X),
  \end{alignat*}
  respectively. All statements hold also for variable $\iface$-orders $\sfr\in\CI(\ol{\TQ^*_\iface}X)$.
\end{cor}

Since $\Psi_{\Qop,\cl}^s(X)$ acts boundedly on $\CI(X_Q)$ and is invariant under conjugation by weights, these normal operators can be defined via testing. That is, for $A\in\Psi_{\Qop,\cl}^{s,(0,\gamma,l',r,b)}(X)$, the operator $N_\zface(A)$ can be defined via $N_\zface(A)u:=(A\tilde u)|_\zface$ where $\tilde u\in\CI(X_Q)$ is any smooth extension of $u\in\CIdot(\zface)$; likewise for the other normal operators. In particular, the above normal operator maps are homomorphisms under composition, where we compose Q-ps.d.o.s as operators between spaces of weighted smooth functions (i.e.\ classical conormal distributions) on $X_Q$.

We finally show that the spaces $\PsiQ(X)$ and $\Psi_{\Qop,\cl}(X)$ are closed under composition. This can be done in a straightforward but tedious manner using the local coordinate descriptions~\eqref{EqQPQuant1}--\eqref{EqQPQuant2} (while residual operators, i.e.\ those with orders $s$, $\alpha_\iface$, $\alpha_\sface=-\infty$ are handled directly on the level of Schwartz kernels). Keeping in line with the presentation thus far, we instead sketch the proof based on an appropriate triple space.

We use the notation for the q-triple space $X_\qop^3$ from Definition~\ref{DefqTriple}, and furthermore write
\[
  \pa\ol\R\times \mface_{\qop,S/C} = \{ \pa\ol\R\times\mface_{\qop,S}, \pa\ol\R\times\mface_{\qop,C} \},
\]
similarly $\pa\ol\R\times\bface_{\qop,F/S/C}$, etc.

\begin{definition}[Q-triple space]
\label{DefQPCTriple}
  The \emph{Q-triple space} of $X$ is the resolution
  \begin{align*}
    X^3_\Qop := &\bigl[ \ol\R\times X^3_\qop; \pa\ol\R\times\zface_{\qop,3}; \pa\ol\R\times\bface_{\qop,F/S/C}; \pa\ol\R\times\diag_{\qop,3}; \pa\ol\R\times\diag_{\qop,F/S/C}; \\
      &\quad \pa\ol\R\times(\diag_{\qop,F/S/C}\cap\,\mface_{\qop,F/S/C}); \pa\ol\R\times\mface_{\qop,F/S/C}; \\
      &\quad \pa\ol\R\times(\diag_{\qop,F/S/C}\cap\,\mface_{\qop,3}); \pa\ol\R\times\mface_{\qop,3} \bigr].
  \end{align*}
\end{definition}

\begin{lemma}[b-fibrations from the Q-triple space]
\label{LemmaQPCbFib}
  The projection map $\ol{\R_\sigma}\times[0,1]_\bhm\times X^3\to\ol\R\times[0,1]\times X^2$ to the first and second factor of $X$, i.e.\ $(\sigma,\bhm,x,x',x'')\mapsto(\sigma,\bhm,x,x')$, lifts to a b-fibration $\pi_F\colon X_\Qop^3\to X_\Qop^2$, similarly for the lifts $\pi_S$, $\pi_C\colon X^3_\Qop\to X^2_\Qop$ of the projections to the second and third, resp.\ first and third factor of $X^3$.
\end{lemma}
\begin{proof}
  Denote the lifted projection from Lemma~\ref{LemmaqbFib} by $\pi_{\qop,F}$. We make use of the description~\eqref{EqqPreimages} of the preimages of boundary hypersurfaces of $X^2_\qop$ under $\pi_{\qop,F}$. We start with the b-fibration $\Id\times\pi_{\qop,F}\colon\ol\R\times X^3_\qop\to\ol\R\times X^2_\qop$. By \cite[Proposition~5.12.1]{MelroseDiffOnMwc}, this map lifts to a b-fibration
  \[
    \bigl[ \ol\R\times X^3_\qop; \pa\ol\R\times\zface_{\qop,3}; \pa\ol\R\times\bface_{\qop,F} \bigr] \to \bigl[ \ol\R\times X^2_\qop; \pa\ol\R\times\zface_{\qop,2} \bigr].
  \]
  We next blow up $\pa\ol\R\times\lb_{\qop,2}$ and $\pa\ol\R\times\rb_{\qop,2}$ in the image; blowing up the preimages in the domain---see~\eqref{EqqPreimages}--- thus gives a b-fibration
  \begin{align*}
    &\bigl[ \ol\R\times X^3_\qop; \pa\ol\R\times\zface_{\qop,3}; \pa\ol\R\times\bface_{\qop,F/S/C}; \pa\ol\R\times\mface_{\qop,S/C} \bigr] \\
    &\hspace{6em} \to \bigl[ \ol\R\times X^2_\qop; \pa\ol\R\times\zface_{\qop,2}; \pa\ol\R\times(\lb_{\qop,2}\cup\rb_{\qop,2}) \bigr].
  \end{align*}
  We used here that $\pa\ol\R\times\mface_{\qop,S}$ and $\pa\ol\R\times\bface_{\qop,S}$ are disjoint to commute their blow-ups. Next, we blow up $\pa\ol\R\times\diag_\qop$ in the range and correspondingly $\pa\ol\R\times\diag_{\qop,F}$ in the domain; we may subsequently also blow up $\pa\ol\R\times\diag_{\qop,3}$ in the domain, as the lifted projection is b-transversal to this. This produces a b-fibration
  \begin{equation}
  \label{EqQPCbFib0}
  \begin{split}
    &\bigl[ \ol\R\times X^3_\qop; \pa\ol\R\times\zface_{\qop,3}; \pa\ol\R\times\bface_{\qop,F/S/C}; \pa\ol\R\times\diag_{\qop,3}; \pa\ol\R\times\diag_{\qop,F}; \pa\ol\R\times\mface_{\qop,S/C} \bigr] \\
    &\hspace{6em} \to X^2_{\Qop,\flat} := \bigl[ \ol\R\times X^2_\qop; \pa\ol\R\times\zface_{\qop,2}; \pa\ol\R\times\diag_\qop, \pa\ol\R\times(\lb_{\qop,2}\cup\rb_{\qop,2}) \bigr].
  \end{split}
  \end{equation}
  Here we use that $\pa\ol\R\times\diag_{\qop,3}\subset\pa\ol\R\times\diag_{\qop,F}$, which implies that we can switch the order of their blow-ups; and moreover $\mface_{\qop,S}$ and $\mface_{\qop,C}$ are disjoint from $\diag_{\qop,3}$ and $\diag_{\qop,F}$, hence their blow-ups can be commuted through to the end.

  In the domain, we next blow up $\pa\ol\R\times(\diag_{\qop,*}\cap\,\mface_{\qop,*})$ for $*=S,C$ (whose lifts get mapped diffeomorphically onto the lifts of $\pa\ol\R\times\lb_{\qop,2}$ and $\pa\ol\R\times\rb_{\qop,2}$); they can be commuted through the blow-ups of their supersets $\pa\ol\R\times\mface_{\qop,S/C}$. We thus obtain a b-fibration
  \begin{equation}
  \label{EqQPCbFib1}
  \begin{split}
    &\bigl[ \ol\R\times X^3_\qop; \pa\ol\R\times\zface_{\qop,3}; \pa\ol\R\times\bface_{\qop,F/S/C}; \pa\ol\R\times\diag_{\qop,3}; \pa\ol\R\times\diag_{\qop,F};  \\
    &\hspace{11em} \pa\ol\R\times(\diag_{\qop,S/C}\cap\,\mface_{\qop,S/C}); \pa\ol\R\times\mface_{\qop,S/C} \bigr] \to X^2_{\Qop,\flat}.
  \end{split}
  \end{equation}
  We can then blow up $\pa\ol\R\times\diag_{\qop,S}$ in the domain; this blow-up can be commuted through that of $\pa\ol\R\times\mface_{\qop,*}$ for $*=S$ (since the intersection $\pa\ol\R\times(\diag_{\qop,S}\cap\,\mface_{\qop,S})$ is blown up before) and also for $*=C$ (by disjointness), and then it can be commuted further through its subset $\pa\ol\R\times(\diag_{\qop,S}\cap\,\mface_{\qop,S})$. Arguing similarly for the blow-up of $\pa\ol\R\times\diag_{\qop,C}$, the map~\eqref{EqQPCbFib1} thus lifts to a b-fibration
  \begin{equation}
  \label{EqQPCbFib2}
  \begin{split}
    &\bigl[ \ol\R\times X^3_\qop; \pa\ol\R\times\zface_{\qop,3}; \pa\ol\R\times\bface_{\qop,F/S/C}; \pa\ol\R\times\diag_{\qop,3}; \pa\ol\R\times\diag_{\qop,F/S/C};  \\
    &\hspace{13em} \pa\ol\R\times(\diag_{\qop,S/C}\cap\,\mface_{\qop,S/C}); \pa\ol\R\times\mface_{\qop,S/C} \bigr] \to X^2_{\Qop,\flat}.
  \end{split}
  \end{equation}

  Next, blowing up $\pa\ol\R\times(\diag_\qop\cap\,\mface_{\qop,2})$ in the range, and using~\eqref{EqqPreimages} to deduce that we need to blow up $\pa\ol\R\times(\diag_{\qop,F}\cap\,\mface_{\qop,F})$ and $\pa\ol\R\times(\diag_{\qop,F}\cap\,\mface_{\qop,3})$ in the domain, we infer that the map~\eqref{EqQPCbFib2} lifts further to a b-fibration
  \begin{align*}
    &\bigl[ \ol\R\times X^3_\qop; \pa\ol\R\times\zface_{\qop,3}; \pa\ol\R\times\bface_{\qop,F/S/C}; \pa\ol\R\times\diag_{\qop,3}; \pa\ol\R\times\diag_{\qop,F/S/C}; \\
    &\qquad \pa\ol\R\times(\diag_{\qop,F/S/C}\cap\,\mface_{\qop,F/S/C}); \pa\ol\R\times\mface_{\qop,S/C}; \pa\ol\R\times(\diag_{\qop,F}\cap\,\mface_{\qop,3}) \bigr] \\
    &\ \to X_{\Qop,\sharp}^2 := \bigl[ \ol\R\times X^2_\qop; \pa\ol\R\times\zface_{\qop,2}; \pa\ol\R\times\diag_\qop, \pa\ol\R\times(\diag_\qop\cap\,\mface_{\qop,2}), \pa\ol\R\times(\lb_{\qop,2}\cup\rb_{\qop,2}) \bigr].
  \end{align*}
  For the commutation of blow-ups, we use here that $\diag_{\qop,F}$ is disjoint from $\mface_{\qop,S/C}$. To restore some symmetry, we then blow up $\pa\ol\R\times(\diag_{\qop,*}\cap\,\mface_{\qop,3})$ in the domain for $*=S,C$; these get mapped diffeomorphically onto the lift of $\pa\ol\R\times\mface_{\qop,2}$. Thus, we get a b-fibration
  \begin{align*}
    &\bigl[ \ol\R\times X^3_\qop; \pa\ol\R\times\zface_{\qop,3}; \pa\ol\R\times\bface_{\qop,F/S/C}; \pa\ol\R\times\diag_{\qop,3}; \pa\ol\R\times\diag_{\qop,F/S/C}; \\
    &\qquad \pa\ol\R\times(\diag_{\qop,F/S/C}\cap\,\mface_{\qop,F/S/C}); \pa\ol\R\times\mface_{\qop,S/C}; \pa\ol\R\times(\diag_{\qop,F/S/C}\cap\,\mface_{\qop,3}) \bigr] \to X_{\Qop,\sharp}^2.
  \end{align*}

  Finally, we again use \cite[Proposition~5.12.1]{MelroseDiffOnMwc} to lift this map to a b-fibration under the blow-up of $\pa\ol\R\times\mface_{\qop,2}$ in the range (producing $X^2_\Qop$) and of the lifts of its preimages $\pa\ol\R\times\mface_{\qop,F}$ and $\pa\ol\R\times\mface_{\qop,3}$ (in this order) in the domain; the resulting domain is naturally diffeomorphic to $X^3_\Qop$, since the blow-up of $\pa\ol\R\times\mface_{\qop,F}$ can be commuted through that of $\pa\ol\R\times(\diag_{\qop,*}\cap\,\mface_{\qop,3})$ for $*=F$ (since the set $\pa\ol\R\times(\diag_{\qop,F}\cap\,\mface_{\qop,F})$ containing their intersection is blown up earlier) and for $*=S,C$ (by disjointness). This finishes the proof.
\end{proof}

\begin{prop}[Composition of Q-ps.d.o.s]
\label{PropQPC}
  Let $A_j\in\PsiQ^{s_j,\alpha_j}(X)$, $j=1,2$. Then $A_1\circ A_2\in\PsiQ^{s_1+s_2,\alpha_1+\alpha_2}(X)$. The same holds true when working with $\Psi_{\Qop,\cl}$ instead.
\end{prop}
\begin{proof}
  The proof is similar to that of Proposition~\ref{PropqComp}. By Remark~\ref{RmkQPDefFn}, it suffices to consider the case $\alpha_1=\alpha_2=(0,0,0,0,0)$. Write the Schwartz kernel $\kappa$ of $A_1\circ A_2$ in terms of the Schwartz kernels $\kappa_1,\kappa_2$ of $A_1,A_2$ as
  \[
    \kappa = (\nu_1\nu_2)^{-1}(\pi_C)_* \bigl( \pi_F^*\kappa_1 \cdot \pi_S^*\kappa_2 \cdot \pi_C^*\nu_1\cdot\pi^*\nu_2 \bigr)
  \]
  where $0<\nu_1\in\CI(X_\Qop;\OmegaQ X)$ is an arbitrary q-density, and $\nu_2=\frac{\dd\bhm}{\bhm}\frac{\dd\sigma}{\la\sigma\ra}$ is a b-density on $\ol{\R_\sigma}\times[0,2)_\bhm$ with $\pi\colon X_\Qop^3\to\ol\R\times[0,1]$ denoting the lifted projection. One can then check that the term in parentheses is then a bounded conormal section of $\pi_F^*\OmegaQ X\otimes\pi_S^*\OmegaQ X\otimes\pi_C^*\OmegaQ X\otimes\pi^*\Omegab_{\ol\R\times[0,1]}(\ol\R\times[0,2))$ which vanishes to infinite order at the boundary hypersurfaces of $X^3_\Qop$ which map to $\iface'_2$, $\sface'_2$, $\lb_2$, $\rb_2$, $\tlb_2$, or $\trb_2$ under $\pi_C$; thus, it is a bounded conormal section of $\Omegab X^3_\Qop$ which vanishes at the aforementioned boundary hypersurfaces. The conclusion then follows using pullback and pushforward results for conormal distributions.
\end{proof}

\subsection{Q-Sobolev spaces}
\label{SsQH}

We now assume that $X$ is compact. We can define weighted Sobolev spaces (corresponding to the Lie algebra $\VQ(X)$) of integer differential order in the usual manner, analogously to Definition~\ref{DefqSob}; we leave it to the reader to spell this out. Here, we instead immediately record the definition for general orders, allowing in particular also for variable orders at $\iface$:

\begin{definition}[Weighted Q-Sobolev spaces]
\label{DefQH}
  Fix any positive weighted Q-density $\nu$ on $X_\Qop$, i.e.\ an element $\nu=(\prod\rho_H^{\nu_H})\nu_0$ where $0<\nu_0\in\CI(X_\Qop,\OmegaQ X)$ and $\nu_H\in\R$, and $H$ ranges over the boundary hypersurfaces $\zface,\mface,\nface,\iface,\sface$. Thus, the restriction $\nu_{\bhm,\sigma}$ is a smooth positive density on $X$ for any $\bhm\in(0,1]$, $\sigma\in\R$. Let $s\in\R$ and $l,\gamma,l',r,b\in\R$; put $w:=\rho_\zface^l\rho_\mface^\gamma\rho_\nface^{l'}\rho_\iface^r\rho_\sface^b$. Then for $s\geq 0$, and for $\bhm\in(0,1]$ and $\sigma\in\R$, we put
  \begin{subequations}
  \begin{equation}
  \label{EqQH}
    H_{\Qop,\bhm,\sigma}^{s,(l,\gamma,l',r,b)}(X,\nu) = H^s(X)
  \end{equation}
  with $(\bhm,\sigma)$-dependent norm
  \begin{equation}
  \label{EqQHNorm}
    \|u\|_{H_{\Qop,\bhm,\sigma}^{s,(l,\gamma,l',r,b)}(X,\nu)}^2 := \|w^{-1}u\|_{L^2(X,\nu_{\bhm,\sigma})}^2 + \|w^{-1}A_{\bhm,\sigma} u\|_{L^2(X,\nu_{\bhm,\sigma})},
  \end{equation}
  \end{subequations}
  where $A=(A_{\bhm,\sigma})\in\PsiQ^s(X)$ is any fixed Q-ps.d.o.\ with elliptic principal symbol. For $s<0$, we define the space~\eqref{EqQH} as a Hilbert space by letting it be the dual space (with respect to the inner product on $L^2(X,\nu_{\bhm,\sigma})$) of $H_{\Qop,\bhm,\sigma}^{-s,(-l,-\gamma,-l',-r,-b)}(X,\nu)$.\footnote{Equivalently, fixing an elliptic operator $A\in\PsiQ^{-s}(X)$, it is the space of all distributions of the form $u=u_0+A u_1$ where $u_0,u_1\in w L^2(X)$, equipped with the norm $\inf_{u=u_0+A u_1}\|w^{-1}u_0\|_{L^2(X,\nu_{\bhm,\sigma})}+\|w^{-1}u_1\|_{L^2(X,\nu_{\bhm,\sigma})}$; cf.\ \cite[Appendix~A]{MelroseVasyWunschEdge} for a general discussion of the underlying functional analysis.} Finally, for variable orders $\sfr\in\CI(\ol{\TQ^*_\iface}X)$, we define the norm on $H_{\Qop,\bhm,\sigma}^{s,(l,\gamma,l',\sfr,b)}(X,\nu)=H^s(X)$ to be
  \[
    \|u\|_{H_{\Qop,\bhm,\sigma}^{s,(l,\gamma,l',\sfr,b)}(X,\nu)}^2 := \|u\|_{H_{\Qop,\bhm,\sigma}^{s,(l,\gamma,l',r_0,b)}(X,\nu)}^2 + \|A u\|_{L^2(X,\nu_{\bhm,\sigma})}^2
  \]
  where $r_0=\min\sfr$, and where $A\in\PsiQ^{s,(l,\gamma,l',\sfr,b)}(X)$ is a fixed elliptic operator.
\end{definition}

We claim that any $A\in\PsiQ^0(X)$ is uniformly (for $\bhm\in(0,1]$ and $\sigma\in\R$) bounded on $L^2(X,\nu)$ when $0<\nu\in\CI(X_\Qop,\OmegaQ X)$ is a positive Q-density. As in~\S\ref{Ssq}, the proof can be reduced, using H\"or\-man\-der's square root trick, to the case that $A\in\PsiQ^{-\infty,(0,0,0,-\infty,-\infty)}(X)$; thus, the Schwartz kernel $\kappa$ of $A$ is a bounded conormal right Q-density on $X_\Qop^2$ which vanishes to infinite order at all boundary hypersurfaces except $\zface_2$, $\mface_2$, and $\nface_2$. The pushforward along the projection $X_\Qop^2\to X_\Qop$ (see Lemma~\ref{LemmaQPbFib}) is thus bounded (on $\bhm^{-1}([0,1])$) and conormal on $X_\Qop$ (and vanishes to infinite order at $\iface$ and $\sface$). The Schur test implies the claim. Directly from Definition~\ref{DefQH}, one can then show that for any orders $s,\tilde s\in\R$, $l,\tilde l,\gamma,\tilde\gamma,l',\tilde l',r,\tilde r,b,\tilde b\in\R$, any element $A=(A_{\bhm,\sigma})\in\PsiQ^{s,(l,\gamma,l',r,b)}(X)$ defines a uniformly bounded (as $\bhm,\sigma$ ranges over $(0,1]\times\R$) family of maps
\[
  A_{\bhm,\sigma} \colon H_{\Qop,\bhm,\sigma}^{\tilde s,(\tilde l,\tilde\gamma,\tilde l',\tilde r,\tilde b)}(X,\nu) \to H_{\Qop,\bhm,\sigma}^{\tilde s-s,(\tilde l-l,\tilde\gamma-\gamma,\tilde l'-l',\tilde r-r,\tilde b-b)}(X,\nu),
\]
similarly when the $\iface$-order is variable.

We next show how to relate Q-Sobolev spaces (and their norms) to b-, $\scbtop$-, $\chop$-, and semiclassical scattering Sobolev spaces near the various boundary hypersurfaces of $X_\Qop$, see Proposition~\ref{PropQStruct}, Definition~\ref{DefQPieces}, and equation~\eqref{EqQPieces}. We restrict attention to a certain class of $\sigma$-independent densities $\nu$, which are lifts of weighted q-densities on $X_\qop$ along the projection off the $\sigma$-coordinate.

\begin{prop}[Relationships between Sobolev spaces]
\label{PropQHRel}
  Fix a $\sigma$-independent density $\nu$ on $X_\Qop$ which is of the form $\nu=\rho_{\zface_q}^{n/2}\nu_0$, $0<\nu_0\in\CI(X,\Omegaq X)$, as in Proposition~\usref{PropqHRel}. Let $\sfr\in\CI(\ol{\TQ^*_\iface}X)$ be an order function which in $|x|<r_0$ (for some $r_0>0$) is invariant under the lift to $\TQ^*X$ of the dilation action $(\tilde\sigma,\bhm,x)\mapsto(\tilde\sigma,\lambda\bhm,\lambda x)$.
  \begin{enumerate}
  \item\label{ItQHRelzf} Put $\phi_\zface\colon\R\times(0,1]\times\hat X^\circ\ni(\sigma,\bhm,\hat x)\mapsto(\sigma,\bhm,\bhm\hat x)\in X_\Qop$, and let $\chi\in\CI(X_\Qop)$ be identically $1$ near $\zface$ and supported in a collar neighborhood thereof. Then for $\bhm\in(0,1]$ and $\sigma\in\R$, we have a uniform equivalence (in the same sense as in Proposition~\usref{PropqHRel})
    \begin{equation}
    \label{EqQHRelzf}
      \|\chi u\|_{H_{\Qop,\bhm,\sigma}^{s,(l,\gamma,l',\sfr,b)}(X)} \sim \la\sigma\ra^{l'-l}\bhm^{\frac{n}{2}-l}\|\phi_\zface^*(\chi u)|_{\sigma,\bhm}\|_{\Hb^{s,\gamma-l}(\hat X,|\dd\hat x|)}.
    \end{equation}
  \item\label{ItQHRelmfsemi} Put $\phi_{\mface,\pm,\semi}\colon(0,1]\times(0,1]\times\dot X^\circ\ni(h,\bhm,x)\mapsto(\pm h^{-1},\bhm,x)\in X_\Qop$, and let $\chi\in\CI(X_\Qop)$ be identically $1$ near $\mface$ and supported in a collar neighborhood thereof. Then, uniformly for $\bhm\in(0,1]$ and $h\in(0,1]$, we have
    \begin{equation}
    \label{EqQHRelmfsemi}
      \|\chi u\|_{H_{\Qop,\bhm,\pm h^{-1}}^{s,(l,\gamma,l',\sfr,b)}(X)} \sim \bhm^{-\gamma} \| \phi_{\mface,\pm,\semi}^*(\chi u)|_{h,\bhm} \|_{H_{\cop,h}^{s,l-\gamma,l'-\gamma,\sfr-\gamma}(\dot X,\nu_\cop)},
    \end{equation}
    where $\nu_\cop$ is the lift of a smooth positive density on $X$ to $\dot X$ as in Proposition~\usref{PropqHRel}\eqref{ItqHRelmf}.
  \item\label{ItQHRelnflow} Put $\phi_{\nface_{\pm,\low}}\colon\pm(0,1]\times(0,1]\times\hat X\ni(\tilde\sigma,\bhm,\hat x)\mapsto(\frac{\tilde\sigma}{\bhm},\bhm,\bhm\hat x)\in X_\Qop$, and let $\chi\in\CI(X_\Qop)$ be identically $1$ near $\nface_\pm$ and supported in a collar neighborhood thereof. Then, uniformly for $\tilde\sigma\in(0,1]$ and $\bhm\in(0,1]$,
    \begin{equation}
    \label{EqQHRelnflow}
      \|\chi u\|_{H_{\Qop,\bhm,\pm h^{-1}}^{s,(l,\gamma,l',\sfr,b)}(X)} \sim \bhm^{\frac{n}{2}-l'}\|\phi_{\nface_{\pm,\low}}^*(\chi u)|_{\tilde\sigma,\bhm}\|_{H_{\scbtop,\tilde\sigma}^{s,\sfr-l',\gamma-l',l-l'}(\hat X,|\dd\hat x|)}.
    \end{equation}
  \item\label{ItQHRelnfsemi} Put $\phi_{\nface_{\pm,\tilde\semi}}\colon(0,1]\times(0,1]\times\hat X\ni(\tilde h,\bhm,\hat x)\mapsto(\pm(\tilde h\bhm)^{-1},\bhm,\bhm\hat x)\in X_\Qop$, and let $\chi\in\CI(X_\Qop)$ be as in part~\eqref{ItQHRelnflow}. Then, uniformly for $\tilde h\in(0,1]$ and $\bhm\in(0,1]$,
  \begin{equation}
  \label{EqQHRelnfsemi}
      \|\chi u\|_{H_{\Qop,\bhm,\pm(\tilde h\bhm)^{-1}}^{s,(l,\gamma,l',\sfr,b)}(X)} \sim \bhm^{\frac{n}{2}-l'}\|\phi_{\nface_{\pm,\tilde\semi}}^*(\chi u)|_{\tilde h,\bhm}\|_{H_{\scop,\tilde h}^{s,\sfr-l',b}(\hat X,|\dd\hat x|)}.
  \end{equation}
  \end{enumerate}
\end{prop}

We remark that the invariance assumption on $\sfr$ is only used in parts~\eqref{ItQHRelnflow}--\eqref{ItQHRelnfsemi} and made there for simplicity; note that the assumption depends on the choice of local coordinates $x\in\R^n$ around $0\in X$. (Without this assumption, one gets slightly lossy two-sided estimates mirroring those in \cite[Corollary~3.7(2)]{HintzConicProp}; these would still be sufficient for our application.)

\begin{proof}[Proof of Proposition~\usref{PropQHRel}]
  It suffices to consider the case that all constant weights $l,\gamma,l',b$ are equal to $0$; furthermore, one can restrict to the case $s\geq 0$ since the case of $s<0$ then follows by duality. The $L^2$-case $s=0$ follows, for all four parts, as in Proposition~\ref{PropqHRel}.

  Part~\eqref{ItQHRelzf} for $s>0$ is then a parameter-dependent version of the estimate~\eqref{EqqHRelzf}, and the proof proceeds in the same manner: one extends the Schwartz kernel of an elliptic b-ps.d.o.\ $A_0\in\Psib^s(\hat X)$ via dilation-invariance (in $(\bhm,x,x')$) and translation-invariance (in $\sigma$), and cuts off the resulting kernel to a collar neighborhood of $\zface_2$ to obtain a Q-ps.d.o.\ $A$ which is elliptic near $\zface$ and can thus be used to compute $H_\Qop^s(X)$-norms in~\eqref{EqQHNorm}.

  For part~\eqref{ItQHRelmfsemi}, we fix an operator $A_0\in\Psich^{s,0,0,\sfr}(\dot X)$ with elliptic principal symbol. By Corollary~\ref{CorQPNormal}, this is the $\mface_{\pm,\semi}$-normal operator of some $A\in\PsiQ^{s,(0,0,0,\sfr,0)}(X)$, and in fact we can take the Schwartz kernel of $A$ to be given by the pushforward of the Schwartz kernel of $A_0$ (considered as a $\bhm$-independent distribution) along $\phi_{\mface,\pm,\semi}$, cut off in both factors to a collar neighborhood of $\mface_2$. The uniform equivalence~\eqref{EqQHRelmfsemi} then follows by arguments completely analogous to those in the proof of Proposition~\ref{PropqHRel}.

  The proof of parts~\eqref{ItQHRelnflow}--\eqref{ItQHRelnfsemi} is similar. The assumption on the order function $\sfr$ ensures that the cutoff (to a collar neighborhood of $\nface$ in both factors on the level of the Schwartz kernel) of the dilation-invariant extension off $\nface_2$ of an elliptic operator in $\Psiscbt^{s,\sfr,0,0}(\hat X)$ lies in $\PsiQ^{s,(0,0,0,\sfr,0)}(X)$.
\end{proof}

Finally, when $\Omega\subset X_\Qop$ is an open set, and writing $\alpha=(l,\gamma,l',\sfr,b)$, we denote by
\begin{equation}
\label{EqQHExt}
  \Hsupp_{\Qop,\bhm,\sigma}^{s,\alpha}(\ol\Omega) = \bigl\{ u\in H_{\Qop,\bhm,\sigma}^{s,\alpha}(X)\colon \supp u\subset\ol\Omega \bigr\},\quad
  \Hext_{\Qop,\bhm,\sigma}^{s,\alpha}(\Omega) = \bigl\{ u|_\Omega \colon u\in H_{\Qop,\bhm,\sigma}^{s,\alpha}(X) \bigr\}
\end{equation}
the spaces of supported, resp.\ extendible distribution (using H\"ormander's notation \cite[Appendix~B]{HormanderAnalysisPDE3}). The former space carries the subspace topology, and the latter space the quotient topology of $H_{\Qop,\bhm,\sigma}^{s,\alpha}(X)/\Hsupp_{\Qop,\bhm,\sigma}^{s,\alpha}(X_\Qop\setminus\Omega)$. In our application, we will take, in some fixed local coordinates $x\in\R^n$ around $0\in X$,
\[
  \Omega=X_\Qop\cap\{\bhm<r<2\}=X_\Qop\cap\{\hat r>1,\ r<2\},
\]
and the relationships recorded in Proposition~\ref{PropQHRel} remain valid upon using extendible Q-Sobolev spaces on $\Omega$ as well as extendible (b-, $\scbtop$-, and semiclassical scattering) Sobolev spaces on $\hat\Omega=\hat X\cap\{\hat r>1\}$ in~\eqref{EqQHRelzf}, \eqref{EqQHRelnflow}, \eqref{EqQHRelnfsemi}.

\section{Quasinormal modes of massless and massive scalar waves}
\label{SK}

In this section, we will prove Theorems~\ref{ThmI} and \ref{ThmIKG}. As discussed in~\S\ref{SsIS}, we may fix $\Lambda=3$. Moreover, we fix the ratio
\[
  \hat\bha := \frac{\bha}{\bhm} \in (-1,1).
\]
All estimates in this section will be uniform in the parameter $\hat\bha\in[-1+\eps,1-\eps]$ for any fixed $\eps>0$.

In~\S\ref{SsKL}, we fix some notation for the degenerating family of Kerr--de~Sitter spacetimes with parameters $(\Lambda,\bhm,\bha)=(3,\bhm,\hat\bha\bhm)$, with $\bhm\searrow 0$. In~\S\ref{SsKMain}, we recall the notions of generalized resonant states and the multiplicity of resonances; these feature in the detailed version of Theorem~\ref{ThmI}, see Theorem~\ref{ThmK}. As a preparation for the proof of Theorem~\ref{ThmK}, we show in~\S\ref{SsKS} how the spectral family of the wave operator on the degenerating Kerr--de~Sitter spacetimes fits into the framework of Q-analysis. The remaining sections, \S\S\ref{SsKSy}--\ref{SsKU}, contain the proof of Theorem~\ref{ThmK}; an outline is provided at the end of~\S\ref{SsKS}. In~\S\ref{SsKG}, we explain the minor modifications needed for the analysis of the Klein--Gordon equation, and thus prove Theorem~\ref{ThmIKG}.

\subsection{Limits of Kerr--de~Sitter metrics}
\label{SsKL}

Since the quantities involved in the definition~\eqref{EqIMetric} of the KdS metric depend only on $\bhm$ via $(\Lambda,\bhm,\bha)=(3,\bhm,\hat\bha\bhm)$, we denote them by
\begin{alignat*}{2}
  \mu_\bhm(r) &:= (r^2+\hat\bha^2\bhm^2)(1-r^2) - 2\bhm r, &\qquad
  b_\bhm &:= 1+\hat\bha^2\bhm^2, \\
  c_\bhm(\theta) &:= 1+\hat\bha^2\bhm^2\cos^2\theta, &\qquad
  \varrho_\bhm(r,\theta) &:= r^2+\hat\bha^2\bhm^2\cos^2\theta.
\end{alignat*}
As we shall prove momentarily, for $\bhm>0$ sufficiently small, the parameters $(3,\bhm,\hat\bha\bhm)$ are subextremal. We denote the roots of $\mu_\bhm$ by
\[
  r_\bhm^-<r_\bhm^C<r_\bhm^e<r_\bhm^c.
\]
\begin{lemma}[Roots of $\mu_\bhm$]
\label{LemmaKRoots}
  Define $\hat r^C:=1-\sqrt{1-\hat\bha^2}$, $\hat r^e:=1+\sqrt{1-\hat\bha^2}$. For sufficiently small $\bhm_0>0$ (depending on $\hat\bha$), and writing $\CI=\CI([0,\bhm_0])$, we have
  \begin{alignat*}{2}
    r_\bhm^-&\equiv-1\bmod \bhm\CI,&\quad
    r_\bhm^C&\equiv \bhm\hat r^C \bmod \bhm^2\CI, \\
    r_\bhm^e&\equiv \bhm\hat r^e \bmod \bhm^2\CI,&\quad
    r_\bhm^c&\equiv 1\bmod\bhm\CI.
  \end{alignat*}
\end{lemma}
\begin{proof}
  The simple roots of $\mu_0(r)=r^2(1-r^2)$ at $r=\pm 1$ extend to real analytic functions $r_\bhm^-=-1+\cO(\bhm)$ and $r_\bhm^c=1+\cO(\bhm)$ for small real $\bhm$. Note next that
  \[
    \bhm^{-2}\mu_\bhm(\bhm\hat r) = \hat r^2-2\hat r+\hat\bha^2 - \bhm^2(\hat r^2+\hat\bha^2)\hat r^2,
  \]
  for $\bhm=0$, has two simple roots at $\hat r=\hat r^C,\hat r^e$, which extend to real analytic functions $\hat r^C_\bhm,\hat r^e_\bhm$ for small $\bhm$, giving rise to the roots $r_\bhm^C=\bhm\hat r^C_\bhm$, $r_\bhm^e=\bhm\hat r^e_\bhm$, of $\mu_\bhm$.
\end{proof}

We use the coordinates $(t_*,r,\theta,\phi_*)$, see~\eqref{EqIKerrStar}, which we define using
\begin{equation}
\label{EqKFm}
  F_{3,\bhm,\hat\bha\bhm}(r) = F_\bhm(r) := -\chi^e\Bigl(\frac{r-r_\bhm^e}{\bhm}\Bigr) + \chi^c(r-r_\bhm^c),
\end{equation}
where $\chi^e\in\CIc(\R)$ and $\chi^c\in\CIc((-1,\infty))$ are both equal to $1$ at $0$. (Thus, the functions $t_*,\phi_*$ defined with respect to the choice of $F_\bhm$ here differ from those defined using the choice of $F_{3,\bhm,\hat\bha\bhm}$ in~\S\ref{SI} by the addition of smooth functions of $r$.) We fix $\chi^e,\chi^c$ in Lemma~\ref{LemmaKdSChi} below. The KdS metric $g_\bhm:=g_{3,\bhm,\hat\bha\bhm}$ takes the form
\begin{equation}
\label{EqKExt}
\begin{split}
  g_\bhm &= -\frac{\mu_\bhm(r)}{b_\bhm^2\varrho_\bhm^2(r,\theta)}\bigl(\dd t_*-\hat\bha\bhm\,\sin^2\theta\,\dd\phi_*\bigr)^2 - \frac{2 F_\bhm(r)}{b_\bhm}(\dd t_*-\hat\bha\bhm\,\sin^2\theta\,\dd\phi_*)\,\dd r \\
    &\quad + \varrho_\bhm^2(r,\theta)\frac{1{-}F_\bhm(r)^2}{\mu_\bhm(r)}\,\dd r^2 + \varrho_\bhm^2(r,\theta)\frac{\dd\theta^2}{c_\bhm(\theta)} + \frac{c_\bhm(\theta)\sin^2\theta}{b_\bhm^2\varrho_\bhm^2(r,\theta)}\bigl( (r^2+(\hat\bha\bhm)^2)\,\dd\phi_* - \hat\bha\bhm\,\dd t_*\bigr)^2.
\end{split}
\end{equation}
The dual metric is
\begin{equation}
\label{EqKExtDual}
\begin{split}
  g_\bhm^{-1} &= \varrho_\bhm(r,\theta)^{-2}\Bigl( -\frac{b_\bhm^2(1-F_\bhm(r)^2)}{\mu_\bhm(r)}\bigl((r^2+(\hat\bha\bhm)^2)\pa_{t_*}+\hat\bha\bhm\pa_{\phi_*}\bigr)^2 + \mu_\bhm(r)\pa_r^2 + c_\bhm(\theta)\pa_\theta^2 \\
    &\quad - 2 b_\bhm F_\bhm(r)\bigl((r^2+(\hat\bha\bhm)^2)\pa_{t_*}+\hat\bha\bhm\pa_{\phi_*}\bigr)\otimes_s\pa_r + \frac{b_\bhm^2}{c_\bhm(\theta)\sin^2\theta}\bigl(\pa_{\phi_*}+\hat\bha\bhm\,\sin^2\theta\,\pa_{t_*}\bigr)^2\Bigr).
\end{split}
\end{equation}

\begin{lemma}[Choice of time function]
\label{LemmaKdSChi}
  There exist smooth functions $\chi^e\in\CIc(\R)$ and $\chi^c\in\CIc((-1,\infty))$ with $\chi^e(0)=1$ and $\chi^c(0)=1$ so that $\dd t_*$ is (past) timelike with respect to $g_\bhm$ on $\R_{t_*}\times[\bhm,2]_r\times\Sph_{\theta,\phi_*}^2$ when $\bhm\in(0,\bhm_0]$ with $\bhm_0>0$ sufficiently small.
\end{lemma}

This can be proved directly by adapting the arguments of \cite[\S6.1]{VasyMicroKerrdS} to the present parameter-dependent setting; we postpone an alternative perturbative proof off the two limiting (de Sitter and Kerr) metrics until after the proof of Lemma~\ref{LemmaKMetric} below.

In $r>r_0$ for any $r_0>0$, the metric $g_\bhm$ converges, as $\bhm\searrow 0$, to the metric
\begin{equation}
\label{EqKdS}
\begin{split}
  g_\dS &:= -(1-r^2)\dd t_*^2 - 2\tilde\chi^c(r)\dd t_*\,\dd r + \frac{1-\tilde\chi^c(r)^2}{1-r^2}\dd r^2 + r^2\slg, \\
  g_\dS^{-1} &= -\frac{1-\tilde\chi^c(r)^2}{1-r^2}\pa_{t_*}^2 - 2\tilde\chi^c(r)\pa_{t_*}\otimes_s\pa_r + (1-r^2)\pa_r^2 + r^{-2}\slg^{-1},
\end{split}
\end{equation}
where $\tilde\chi^c(r):=\chi^c(r-1)$, and $\slg:=\dd\theta^2+\sin^2\theta\,\dd\phi_*^2$ is the standard metric on $\Sph^2$. Thus, $g_\dS$ is the de~Sitter metric\footnote{If we change coordinates via $t=t_*+T_\dS(r)$ where $T'_\dS(r)=\frac{\tilde\chi^c(r)}{1-r^2}$, then $g_\dS=-(1-r^2)\dd t_*^2+(1-r^2)^{-1}\dd r^2+r^2\slg$ is the de~Sitter metric in static coordinates.}---a nondegenerate Lorentzian metric on
\begin{equation}
\label{EqKdSSpace}
  \R_{t_*} \times X,\qquad X := B(0,3)=\{ x\in\R^3\colon |x|<3 \},
\end{equation}
with $(r,\theta,\phi_*)$ denoting polar coordinates on $X$. (We stress that $g_\dS$ is in fact smooth across $x=0$, though the geometry, resp.\ analysis of the limit $\bhm\searrow 0$ do see a remnant of the disappeared KdS black hole in the form of a conical singularity, resp.\ b-Sobolev spaces with weights at $r=0$.)

On the other hand, if we set $\hat t_*:=\bhm t_*$ and $\hat r:=\bhm r$ and express $g_\bhm$ in the coordinates $(\hat t_*,\hat r,\theta,\phi_*)$, then for $\hat r$ in any closed subinterval of $(\hat r^C,\infty)$, the rescaled metric $\bhm^{-2}g_\bhm$ converges, as $\bhm\searrow 0$, to the metric
\begin{subequations}
\begin{equation}
\label{EqKKerr}
\begin{split}
  &\hat g = -\frac{\hat\mu(\hat r)}{\hat\varrho^2(r,\theta)}\bigl(\dd\hat t_*-\hat\bha\,\sin^2\theta\,\dd\phi_*\bigr)^2 + 2\tilde\chi^e(\hat r)(\dd\hat t_*-\hat\bha\,\sin^2\theta\,\dd\phi_*)\,\dd\hat r \\
  &\qquad\qquad + \hat\varrho^2\frac{1-\tilde\chi^e(\hat r)^2}{\hat\mu(r)}\dd\hat r^2 + \hat\varrho^2(\hat r,\theta)\dd\theta^2 + \frac{\sin^2\theta}{\hat\varrho^2(\hat r,\theta)}\bigl((\hat r^2+\hat\bha^2)\dd\phi_*-\hat\bha\,\dd\hat t_*\bigr)^2, \\
  &\quad \hat\mu(\hat r):=\hat r^2-2\hat r+\hat\bha^2,\qquad
    \hat\varrho^2(\hat r,\theta):=\hat r^2+\hat\bha^2\cos^2\theta,\qquad
    \tilde\chi^e(\hat r)=\chi^e(\hat r-\hat r^e),
\end{split}
\end{equation}
of a Kerr black hole with mass $1$ and angular momentum $\hat\bha$.\footnote{A coordinate change in $\hat t_*$ and $\phi_*$ brings~\eqref{EqKKerr} into Boyer--Lindquist form.} The dual metric is
\begin{equation}
\label{EqKKerrDual}
\begin{split}
  \hat g^{-1} &= \hat\varrho(\hat r,\theta)^{-2}\Bigl(-\frac{1-\tilde\chi^e(\hat r)^2}{\hat\mu(\hat r)}\bigl((\hat r^2+\hat\bha^2)\pa_{\hat t_*}+\hat\bha\pa_{\phi_*}\bigr)^2 + \hat\mu(\hat r)\pa_{\hat r}^2 + \pa_\theta^2 \\
    &\qquad\qquad + 2\tilde\chi^e(\hat r)\bigl((\hat r^2+\hat\bha^2)\pa_{\hat t_*}+\hat\bha\pa_{\phi_*}\bigr)\otimes_s\pa_{\hat r} + \frac{1}{\sin^2\theta}\bigl(\pa_{\phi_*}+\hat\bha\,\sin^2\theta\,\pa_{\hat t_*}\bigr)^2\Bigr).
\end{split}
\end{equation}
This is a smooth nondegenerate Lorentzian metric on
\begin{equation}
\label{EqKKerrSpace}
  \R_{\hat t_*}\times(\hat r^C,\infty)_{\hat r}\times\Sph^2_{\theta,\phi_*}.
\end{equation}
\end{subequations}

The wave operators associated with the metrics $g_\bhm$, $g_\dS$, and $\hat g$ have as principal symbols the respective dual metric functions:

\begin{definition}[Dual metric functions]
\label{DefKDual}
  Let $\wt X_\bhm=\wt X_{3,\bhm,\hat\bha\bhm}$ (see~\eqref{EqIKdSExt}). The \emph{dual metric function} $G_\bhm\in\CI(T^*(\R_{t_*}\times\wt X_\bhm))$ of $g_\bhm$ is defined as
  \[
    G_\bhm(\zeta) = |\zeta|^2_{g_\bhm(z)^{-1}},\qquad
      z=(t_*,x)\in\R_{t_*}\times\wt X_\bhm,\quad \zeta\in T^*_z(\R_{t_*}\times\wt M_\bhm).
  \]
  The analogously defined dual metric functions of $g_\dS$ and $\hat g$ are denoted
  \[
    G_\dS\in\CI(T^*(\R_{t_*}\times X)),\quad\text{resp.}\quad\hat G\in\CI\bigl(T^*\bigl(\R_{\hat t_*}\times(\hat r^C,\infty)\times\Sph^2\bigr)\bigr).
  \]
\end{definition}

When $\bhm_0>0$ is sufficiently small, then Lemma~\ref{LemmaKRoots} implies that $r_\bhm^C < \bhm < r_\bhm^e < r_\bhm^c < 2$ for all $\bhm\in(0,\bhm_0]$ when $\bhm_0>0$ is sufficiently small. Put
\begin{equation}
\label{EqKSpatMfd}
  \Omega_\bhm := (\bhm,2)_r \times \Sph^2.
\end{equation}
Then, in the notation of~\eqref{EqIKdSDOC} and \eqref{EqIKdSExt}, the manifold $\R_{t_*}\times \ol{\Omega_\bhm}\subset\wt M_\bhm:=\wt M_{3,\bhm,\hat\bha\bhm}$ contains a neighborhood of the closure of $M_{3,\bhm,\hat\bha\bhm}^{\rm DOC}$.

\subsection{Resonances, multiplicity, and the main theorem}
\label{SsKMain}

We now prepare the precise statement of Theorem~\ref{ThmI}.

\begin{definition}[Spectral family]
\label{DefKFamily}
  For $\sigma\in\C$, we define\footnote{The more usual notation would be $\wh{\Box_{g_\bhm}}(\sigma)$. We do not use a hat here, however, to avoid overloading the notation.}
  \[
    \Box_{g_\bhm}(\sigma)\in\Diff^2(\ol{\Omega_\bhm})
  \]
  to be the unique operator with $\Box_{g_\bhm}(e^{-i\sigma t_*}u)=e^{-i\sigma t_*}\Box_{g_\bhm}(\sigma)u$ for $u\in\CIc(\Omega_\bhm)$. With $\Omega_\dS:=B(0,2)\subset X=B(0,3)$, we similarly define the spectral family of $\Box_{g_\dS}$, denoted
  \[
    \Box_{g_\dS}(\sigma)\in\Diff^2(\ol{\Omega_\dS}),\qquad \sigma\in\C.
  \]
  We finally denote by
  \[
    \Box_{\hat g}(\tilde\sigma)\in\Diff^2\bigl([1,\infty)_{\hat r}\times\Sph^2\bigr),\qquad \tilde\sigma\in\C,
  \]
  the spectral family of $\Box_{\hat g}$, so $\Box_{\hat g}(e^{-i\tilde\sigma\hat t_*}u)=e^{-i\tilde\sigma\hat t_*}\Box_{\hat g}(\tilde\sigma)u$ for $u\in\CIc((1,\infty)\times\Sph^2)$.
\end{definition}

Informally, $\Box_{g_\bhm}(\sigma)$, $\Box_\dS(\sigma)$, resp.\ $\Box_{\hat g}(\tilde\sigma)$ is obtained from $\Box_{g_\bhm}$, $\Box_{g_\dS}$, resp.\ $\Box_{\hat g}$ by replacing $\pa_{t_*}$, resp.\ $\pa_{\hat t_*}$ by $-i\sigma$, resp.\ $-i\tilde\sigma$. Thus, the spectral families are polynomials (hence holomorphic) in $\sigma$, resp.\ $\tilde\sigma$.

\begin{definition}[Space of resonant states]
\label{DefKRes}
  For $\sigma\in\R$, we define $\Res_\bhm(\sigma)\subset\CI(\R_{t_*}\times \ol{\Omega_\bhm})$ as the space of all \emph{generalized resonant states} $u=u(t_*,x)$ of $\Box_{g_\bhm}$ at frequency $\sigma$, i.e.\ solutions $u$ of $\Box_{g_\bhm}u=0$ which for some $n\in\N_0$ can be written as $u=\sum_{k=0}^n t_*^k e^{-i\sigma t_*}u_k(x)$ where $u_k\in\CI(\ol{\Omega_\bhm})$. The \emph{multiplicity} of $\sigma$ is
  \[
    m_\bhm(\sigma) := \dim\Res_\bhm(\sigma).
  \]
  We similarly define $\Res_\dS(\sigma)\subset\CI(\R_{t_*}\times\ol{\Omega_\dS})$ and $m_\dS(\sigma)$ with respect to $\Box_{g_\dS}$.
\end{definition}

Thus, $\sigma\in\QNM(\bhm):=\QNM(3,\bhm,\hat\bha\bhm)$ if and only if $\Res_\bhm(\sigma)\neq\{0\}$, i.e.\ $m_\bhm(\sigma)\neq 0$.\footnote{Note that the existence of a smooth resonant state is independent of the choice of the function $F_\bhm$ in~\eqref{EqKFm} or $F_{\Lambda,\bhm,\bha}$ in~\eqref{EqIKerrStar}, as long as these functions equal $-1$, resp.\ $+1$ at the event, resp.\ cosmological horizon.} For sufficiently small $\bhm$, the Fredholm theory of \cite[\S6]{VasyMicroKerrdS} can be shown to apply to $\Box_{g_\bhm}(\sigma)$ (see also~\eqref{EqKBdMap} below), and thus $\Box_{g_\bhm}(\sigma)^{-1}$ is a meromorphic family of operators on $\CI(\ol{\Omega_\bhm})$. As shown in \cite[\S5.1.1]{HintzVasyKdSStability}, an equivalent definition of $\Res_\bhm(\sigma)$ is then
\begin{equation}
\label{EqKResEquiv}
  \Res_\bhm(\sigma) = \biggl\{ \res_{\zeta=\sigma}\Bigl( e^{-i\zeta t_*}\Box_{g_\bhm}(\zeta)^{-1}p(\zeta)\Bigr) \colon p(\zeta)\ \text{is a polynomial with values in}\ \CI(\ol{\Omega_\bhm}) \biggr\},
\end{equation}
and the multiplicity can be computed via
\begin{equation}
\label{EqKResMult}
  m_\bhm(\sigma) = \frac{1}{2\pi i}\tr\oint_\sigma\Box_{g_\bhm}(\zeta)^{-1}\pa_\zeta\Box_{g_\bhm}(\zeta)\,\dd\zeta,
\end{equation}
where $\oint_\sigma$ is the contour integral over a circle enclosing $\sigma$ counterclockwise which contains no resonances other than $\sigma$. (The integral is a finite rank operator on $\CI(\ol{\Omega_\bhm})$, and hence its trace is well-defined.) There are analogous expressions for $\Res_\dS(\sigma)$ and $m_\dS(\sigma)$.

\begin{definition}[Quasinormal modes with multiplicity]
\label{DefKQNM}
  For $\bhm\in(0,\bhm_0]$, we put
  \begin{alignat*}{3}
    \QNM^*(\bhm) &:= \bigl\{ (\sigma,m_\bhm(\sigma))\in\C\times\N \colon &&m_\bhm(\sigma)\geq 1 \bigr\} &&\subset \C\times\N, \\
    \QNM_\dS^* &:= \bigl\{ (\sigma,m_\dS(\sigma))\in\C\times\N \colon &&m_\dS(\sigma)\geq 1 \bigr\} &&\subset \C\times\N.
  \end{alignat*}
  Furthermore, $\QNM(\bhm)=\QNM(3,\bhm,\hat\bha\bhm)$ is the projection of $\QNM^*(\bhm)$ to the first factor, and $\QNM_\dS$ is the projection of $\QNM^*_\dS$ to the first factor.
\end{definition}

\begin{lemma}[QNMs of de~Sitter space]
\label{LemmaKdSQNM}
  We have $\QNM_\dS=-i\N_0$, and
  \[
    \QNM_\dS^* = \{ (-i\ell, m) \colon \ell\in\N_0,\ m=m_\dS(-i\ell) \}
  \]
  where
  \begin{equation}
  \label{EqKdSQNM}
    m_\dS(-i\ell) = \begin{cases} 1, & \ell=0, \\ \ell^2+2, & \ell\geq 1. \end{cases}
  \end{equation}
\end{lemma}
\begin{proof}
  This follows from \cite[Proposition~2.1]{HintzXieSdS} upon setting $\nu=0$, thus $\lambda_-(\nu)=0$ and $\lambda_+(\nu)=3$ in the notation of the reference. Indeed, for $l\in\N_0$, the space of generalized resonant states with angular dependence given by a degree $l$ spherical harmonic is non-trivial exactly at all spectral parameters $-i\ell$ for $\ell\in(l+2\N_0)\cup(3+l+2\N_0)$, and at each such resonance has dimension $2 l+1$. This gives
  \[
    m_\dS(-i\ell) = \sum_{\genfrac{}{}{0pt}{}{l\in\N_0}{\ell-l\in(2\N_0)\cup(3+2\N_0)}} (2 l+1) = (2\ell+1) + \sum_{k=0}^{\ell-2} (2 k+1).
  \]
  For $\ell=0$, resp.\ $1$, this evaluates to $1$, resp.\ $3=1^2+2$. For $\ell\geq 2$ the second sum is $(\ell-1)^2=\ell^2-2\ell+1$. This gives~\eqref{EqKdSQNM}.
\end{proof}

\begin{thm}[Quasinormal modes of KdS black holes away from extremality: detailed version]
\label{ThmK}
  Let $C_1>0$ be such that $\Im\sigma\neq-C_1$ for all $\sigma\in\QNM_\dS$. Let $\eps>0$ be such that for each $\sigma_*\in\QNM_\dS$ with $\Im\sigma_*\geq-C_1$, the only $\sigma\in\QNM_\dS$ with $|\sigma-\sigma_*|\leq 2\eps$ is $\sigma_*$ itself.\footnote{Thus, one can take any $\eps<\half$. The present formulation generalizes without change to the case of the Klein--Gordon equation.} Then there exists $\bhm_1>0$ so that the following statements hold.
  \begin{enumerate}
  \item\label{ItK1} If $\bhm\in(0,\bhm_1]$ and $\sigma\in\QNM(\bhm)$, $\Im\sigma\geq-C_1$, then there exists $\sigma_*\in\QNM_\dS$ so that $|\sigma-\sigma_*|\leq\eps$.
  \item\label{ItK2} The total multiplicity of QNMs near $\sigma_*\in\QNM_\dS$ with $\Im\sigma_*\geq-C_1$ is independent of $\bhm$, that is,
    \[
      m_\dS(\sigma_*) = \sum_{\genfrac{}{}{0pt}{}{\sigma\in\QNM(\bhm)}{|\sigma-\sigma_*|\leq\eps}} m_\bhm(\sigma),\qquad \bhm\in(0,\bhm_1].
    \]
  \item\label{ItK0} The only resonance $\sigma\in\QNM(\bhm)$ with $|\sigma|\leq\eps$ is $\sigma=0$, with $m_\bhm(0)=1$, and $\Res_\bhm(0)$ consists of all constant functions on $\R_{t_*}\times\ol{\Omega_\bhm}$.
  \item\label{ItK3} Let $K=[r_0,2]\times\Sph^2$, and let $\sigma_*\in\QNM_\dS$ with $\Im\sigma_*\geq-C_1$. Then for all sufficiently small $r_0>0$, the space
    \begin{equation}
    \label{EqK3}
      \biggl\{ u|_{[0,1]_{t_*}\times K} \colon u \in \sum_{\genfrac{}{}{0pt}{}{\sigma\in\QNM(\bhm)}{|\sigma-\sigma_*|\leq\eps}} \Res_\bhm(\sigma) \biggr\}
    \end{equation}
    has dimension $m_\dS(\sigma_*)$ and converges to $\{u|_{[0,1]\times K}\colon u\in\Res_\dS(\sigma_*)\}$ in the topology of $\CI([0,1]\times K)$. (That is, there exists an $\bhm$-dependent basis $u_{\bhm,1},\ldots,u_{\bhm,m_\dS(\sigma_*)}$ of the space~\eqref{EqK3} which converges in $\CI([0,1]\times K)$ to a basis of $\Res_\dS(\sigma_*)|_{[0,1]\times K}$.)
  \end{enumerate}
\end{thm}

Parts~\eqref{ItK1} and \eqref{ItK2} together give precise meaning to the statement that the quasinormal modes of Kerr--de~Sitter space with parameters $(\Lambda,\bhm,\bha)=(3,\bhm,\hat\bha\bhm)$ converge \emph{with multiplicity} to those of de~Sitter space in any half space $\Im\sigma\geq -C_1$ as $\bhm\searrow 0$.

\subsection{The spectral family as a Q-differential operator}
\label{SsKS}

As the starting point for the proof of Theorem~\ref{ThmK}, we now place $\Box_{g_\bhm}(\sigma)$ into the context of q- and Q-analysis. We use the terminology of~\S\ref{SQ}, with two small modifications: \begin{enumerate*} \item the mass $\bhm$ will be restricted to a short interval $[0,\bhm_0]$ (rather than $[0,1]$) where $\bhm_0>0$ is chosen according to the requirement stated before~\eqref{EqKSpatMfd}; and \item we shall write $\sigma_0$ for the real parameter that was previously denoted $\sigma$ in~\S\S\ref{SsQP}--\ref{SsQH}.\end{enumerate*} We reserve the symbol $\sigma$ for the spectral parameter (which might be complex).

Let $X$ denote a 3-dimensional torus; we work in a local coordinate chart $B(0,3)$ near a point $0\in X$ as in~\eqref{EqKdSSpace}. (We make $X$ compact merely so that Sobolev spaces are well-defined.) At fixed (or more generally for bounded) frequencies $\sigma\in\C$, our analysis will take place in the domain
\begin{subequations}
\begin{equation}
\label{EqKSOmegaq}
  \Omega_\qop := \{ \hat r>1,\ r<2 \} \cap X_\qop.
\end{equation}
Thus, $\Omega_\qop$ resolves $\bigsqcup_{\bhm\in(0,\bhm_0]}\{\bhm\}\times\Omega_\bhm$ in the singular limit $\bhm\searrow 0$, and we have
\begin{equation}
\label{EqKSOmegahat}
  \hat\Omega := \Omega_\qop\cap\,\zface_\qop = \{ \hat r>1 \} \cap \zface_\qop,\qquad
  \dot\Omega := \Omega_\qop\cap\,\mface_\qop = [0,2)_r\times\Sph^2.
\end{equation}
(Here, $\hat\Omega$ is a subset of the spatial manifold in~\eqref{EqKKerrSpace}; the radius $1$ is chosen for notational convenience.) We denote by $\ol{\Omega_\qop}=\{\hat r\geq 1,\ r\leq 2\}\cap X_\qop$ and $\ol{\hat\Omega}=\{\hat r\geq 1\}\cap\,\zface_\qop$ the closures of $\Omega_\qop$ and $\hat\Omega$ inside $X_\qop$. See Figure~\ref{FigKSOmega}.

\begin{figure}[!ht]
\centering
\includegraphics{FigKSOmega}
\caption{The domains $\Omega_\qop$, $\hat\Omega$, and $\dot\Omega\subset X_\qop$ defined in~\eqref{EqKSOmegaq} and~\eqref{EqKSOmegahat}, without the factor $\Sph^2$.}
\label{FigKSOmega}
\end{figure}

On the Q-single space $X_\Qop$, we shall work on the lift of $\ol{\R_{\sigma_0}}\times\Omega_\qop$,
\begin{equation}
\label{EqKSOmegaQ}
  \Omega_\Qop := \{ \hat r>1,\ r<2 \} \cap X_\Qop.
\end{equation}
\end{subequations}

We need to analyze also non-real frequencies $\sigma$. For now, we work in strips $\{\Im\sigma\leq C_1\}$ for arbitrary fixed $C_1>0$, and the total space of our analysis is therefore
\[
  [-C_1,C_1]\times\Omega_\Qop \subset [-C_1,C_1]\times X_\Qop.
\]
(The modifications needed to treat all of $\{\Im\sigma\geq-C_1\}$ will be discussed in~\S\ref{SsKU}.) The total spectral family $(\bhm,\sigma)\mapsto\Box_{g_\bhm}(\sigma)$, where $\bhm\in(0,\bhm_0]$ and $\sigma=\sigma_0+i\sigma_1$ with $\sigma_0\in\R$, $\sigma_1\in[-C_1,C_1]$, defines an element
\[
  \Box(\cdot+i\sigma_1)\in\Diff^2(\ol{\Omega_\Qop}\cap\{\bhm>0\}),
\]
with smooth dependence on $\sigma_1$. The following key result puts the total spectral family into the Q-analytic framework developed in~\S\ref{SQ}, and is indeed the motivation for the development of this framework.

\begin{prop}[Properties of the total spectral family]
\label{PropKS}
  The total spectral family $\Box(\cdot+i\sigma_1)$ satisfies
  \begin{equation}
  \label{EqKSCat}
    \Box(\cdot+i\sigma_1) \in \DiffQ^{2,(2,0,2,2,2)}(\ol{\Omega_\Qop}) = \rho_\zface^{-2}\rho_\mface^0\rho_\nface^{-2}\rho_\sface^{-2}\rho_\iface^{-2}\DiffQ^2(\ol{\Omega_\Qop}),
  \end{equation}
  and depends smoothly on $\sigma_1\in[-C_1,C_1]$. Moreover, in the notation of Corollary~\usref{CorQPNormal}:
  \begin{enumerate}
  \item\label{ItKSSymb} The Q-principal symbol of $\Box(\cdot+i\sigma_1)$ is $G(\cdot+i\sigma_1,-;-,-)$, given by
    \begin{equation}
    \label{EqKSSymb}
      G\colon(\sigma,\bhm;x,\xi)\mapsto G_\bhm|_x(-\sigma\,\dd t_*+\xi),
    \end{equation}
    where $x\in\ol{\Omega_\bhm}$, $\xi\in T^*_x\ol{\Omega_\bhm}$, and $\sigma=\sigma_0+i\sigma_1$, in the sense that $\sigmaQ^{2,(2,0,2,2,2)}(\Box(\cdot+i\sigma_1))$ is given by the equivalence class of $G(\cdot+i\sigma_1)$ in $(S^{2,(2,0,2,2,2)}/S^{1,(2,0,2,1,1)})(\TQ_{\ol{\Omega_\Qop}}^*X)$.
  \item\label{ItKSzf} We have $N_\zface(\bhm^2\Box(\cdot+i\sigma_1))=\Box_{\hat g}(0)$ (regarded as a $\sigma_0$-independent operator on $\ol{\R_{\sigma_0}}\times\ol{\hat\Omega}\subset\zface$, cf.\ Proposition~\usref{PropQStruct}\eqref{ItQStructzf}).
  \item\label{ItKSnf} For $\tilde\sigma_0\in\R\setminus\{0\}$, we have $N_{\nface_{\tilde\sigma_0}}(\bhm^2\Box(\cdot+i\sigma_1))=\Box_{\hat g}(\tilde\sigma_0)$.
  \item\label{ItKSmf} For $\sigma_0\in\R$, we have $N_{\mface_{\sigma_0}}(\Box(\cdot+i\sigma_1))=\Box_{g_\dS}(\sigma)$, where $\sigma=\sigma_0+i\sigma_1$.
  \end{enumerate}
\end{prop}

One can prove this by direct calculation using the form~\eqref{EqKExt} of the KdS metric. We instead give a conceptual proof, which highlights the relevant structural properties of the family of metrics $g_\bhm$.\footnote{This route is longer, but it has the advantage of allowing for straightforward generalizations of Proposition~\ref{PropKS} to spectral families of other geometric operators---even if in the present paper we do not discuss such generalizations.} To begin with, we define
\[
  M:=\R_{t_*}\times X,\qquad
  \dot M:=\R_{t_*}\times\dot X,\qquad
  \hat M:=\R_{\hat t_*}\times\hat X,
\]
and identify $X$ with $\{0\}\times X\subset M$, likewise $\dot X\subset\dot M$ and $\hat X\subset\hat M$. Smooth stationary metrics on $M$ can be identified with smooth sections of $S^2 T^*_X M\to X$, likewise for $\dot M$, $\hat M$.

Denote now by
\[
  \pi_\qop\colon X_\qop\to X,\qquad
  \pi_\Qop\colon X_\Qop\to X
\]
the lifts of the projection maps $[0,\bhm_0]\times X\ni(\bhm,x)\mapsto x\in X$ and $\ol\R\times[0,\bhm_0]\times X\to X$, respectively. The pullback bundle $\pi_\qop^*T_X M\to\Omega_\qop$ will play two roles. Firstly, it is a bundle in (the tensor powers of) which geometric objects are valued (see Lemma~\ref{LemmaKMetric} below). Secondly, in $\bhm>0$ its sections are smooth families of horizontal vector fields; in this latter regard, we note:

\begin{lemma}[Bundle isomorphisms]
\label{LemmaKBundleIso}
  Let $\dot\beta\colon\mface_\qop=\dot X=[X;\{0\}]\to X$ denote the blow-down map. Then the identity map $(\pi_\qop^* T_X M)|_{(0,x)}=T_x M=T_x\dot M$ for $x\in X\setminus\{0\}$ extends to a bundle isomorphism
  \begin{equation}
  \label{EqKBundleIsomf}
    (\pi_\qop^*T_X M)|_{\mface_\qop} = \dot\beta^* T_{\dot X}\dot M.
  \end{equation}
  Moreover, the map $(\pi_\qop^* T X)|_{(\bhm,x)}=T_x X\ni V\mapsto\bhm V\in\Tq_{(\bhm,x)}X$ (for $\bhm\in(0,\bhm_0]$) extends by continuity to a smooth bundle map on $X_\qop$ and then restricts to $\zface_\qop=\hat X$ as an isomorphism
  \begin{equation}
  \label{EqKBundleIsozf}
    \iota \colon (\pi_\qop^*T X)|_{\zface_\qop} \cong \Tsc\hat X\qquad \text{(via `multiplication by $\bhm$')}.
  \end{equation}
\end{lemma}
\begin{proof}
  For~\eqref{EqKBundleIsomf}, simply note that both bundles have, as smooth frames, the vector fields $\pa_{t_*}$ and $\pa_{x^j}$ ($j=1,2,3$). For~\eqref{EqKBundleIsozf}, note that $\bhm\pa_{x^j}=\pa_{\hat x^j}$ ($j=1,2,3$) is a frame of $\Tsc\hat X$.
\end{proof}

We shall also write $\iota$ for tensor powers of the isomorphism~\eqref{EqKBundleIsozf} or its adjoint. Writing $\ul\R=\hat X\times\R$ for the trivial bundle, we furthermore define the map
\begin{equation}
\label{EqKBundleSpacetime}
  \wt\iota \colon (\pi_\qop^* T_X M)|_{\zface_\qop} \xra{\cong} T_0\R_{\hat t_*}\oplus\Tsc\hat X,\qquad \pa_{t_*}\mapsto\pa_{\hat t_*},\quad V\mapsto\iota(V).
\end{equation}
(This is `multiplication by $\bhm$' for tangent vectors on the spacetime $M$.) Tensor powers of $\wt\iota$ or its adjoint are denoted by the same symbol.

\begin{lemma}[The family $g_\bhm$ on the q-single space]
\label{LemmaKMetric}
  For $\bhm\in(0,\bhm_0]$ and $x\in\ol{\Omega_\bhm}$, define the symmetric 2-tensor $g(\bhm,x)\in(\pi_\qop^* S^2 T^*_X M)|_{(\bhm,x)}=S^2 T^*_x M$ to be equal to $g_\bhm|_x$. Then
  \[
    g \in \CI\bigl(\ol{\Omega_\qop}; (\pi_\qop^*S^2 T^*_X M)|_{\ol{\Omega_\qop}}\bigr),\qquad
    g^{-1} \in \CI\bigl(\ol{\Omega_\qop}; (\pi_\qop^*S^2 T_X M)|_{\ol{\Omega_\qop}}\bigr).
  \]
  Moreover, $g|_{\mface_\qop}=g_\dS$ (under the identification~\eqref{EqKBundleIsomf}), and $\wt\iota^{-1}(g|_{\zface_\qop})=\hat g$. Furthermore, $|\dd g|=|\dd t_*||\dd g_X|$ where $|\dd g_X|\in\CI(\ol{\Omega_\qop};(\pi_\qop^*\Omega_X M)|_{\ol{\Omega_\qop}})$. (Explicitly, we have $|\dd g_X|=b^2\varrho^2\sin\theta\,|\dd r\,\dd\theta\,\dd\phi_*|$ where $b(\bhm)=b_\bhm$ and $\varrho(\bhm,r,\theta)=\varrho_\bhm(r,\theta)$.)
\end{lemma}
\begin{proof}
  On $\ol{\Omega_\qop}$, the 1-forms $\dd t_*$, $\dd r$, $r\,\dd\theta$, and $r\,\dd\phi_*$ are smooth and nonzero sections of $\pi_\qop^*T^*_X M$. It thus suffices to show that the coefficients of $g_\bhm$ in~\eqref{EqKExt} (expressed in terms of symmetric tensor products of these 1-forms) are elements of $\CI(\ol{\Omega_\qop})$. On $\ol{\Omega_\qop}$, smooth coordinates are given by $\hat\rho=\frac{\bhm}{r}\in[0,1]$, $r\geq 0$, and $\theta,\phi_*$, and we then note that
  \[
    \frac{\mu_\bhm(r)}{b_\bhm^2\varrho_\bhm^2(r,\theta)} = \frac{(1+\hat\bha^2\hat\rho^2)(1-r^2)-2\hat\rho}{(1+\hat\bha^2\hat\rho^2 r^2)(1+\hat\bha^2\hat\rho^2\cos^2\theta)}
  \]
  is indeed smooth in these coordinates, similarly for the other coefficients of $g_\bhm$; note in particular that $F_\bhm=-\chi^e(\hat\rho^{-1})+\chi^c(r)$ is smooth. The membership of $g^{-1}$ follows similarly by inspection of the coefficients of $g_\bhm^{-1}$ in~\eqref{EqKExtDual} in the basis $\pa_{t_*}$, $\pa_r$, $r^{-1}\pa_\theta$, $r^{-1}\pa_{\phi_*}$.

  The computation of $g|_{\mface_\qop}$ was already performed in~\eqref{EqKdS}. The computation of $\iota^{-1}(g|_{\zface_\qop})$ amounts to taking the limit of $\bhm^{-2}g_\bhm$ as $\bhm\searrow 0$ for bounded $\hat r=|\hat x|$, which was done in~\eqref{EqKKerr}.
\end{proof}

We can now give a simple proof of Lemma~\ref{LemmaKdSChi}:

\begin{proof}[Proof of Lemma~\usref{LemmaKdSChi}]
  Using~\eqref{EqKdS} and writing $\tilde\chi^c(r)=1+(1-r^2)f(r)$, we compute
  \[
    |\dd t_*|_{g_\dS^{-1}}^2 = -\frac{1-(1+(1-r^2)f(r))^2}{1-r^2} = 2 f(r) + (1-r^2)f(r)^2.
  \]
  Note that in any region $r\leq r_1<1$, this is negative for $f(r)=-\frac{1}{1-r^2}$ (in which case $\tilde\chi^c=0$). More generally, in $r<1$, resp.\ at $r=1$, we have $|\dd t_*|_{g_\dS^{-1}}^2<0$ provided $-\frac{2}{1-r^2}<f(r)<0$, resp.\ $f(1)<0$. For $1<r\leq 3$, it is enough to ensure $f(r)<0$. We can thus use a partition of unity to construct a smooth $f$ so that $\dd t_*$ is past timelike for $g_\dS$, and so that $\tilde\chi^c(r)=0$ for $r\leq\half$.

  Next, using~\eqref{EqKKerrDual} and writing $\tilde\chi^e(\hat r)=1+\hat\mu(\hat r)\hat f(\hat r)$, we compute
  \begin{align*}
    \hat\varrho(\hat r,\theta)^2|\dd\hat t_*|_{\hat g^{-1}}^2 &= -\frac{1-(1+\hat\mu(\hat r)\hat f(\hat r))^2}{\hat\mu(\hat r)}(\hat r^2+\hat\bha^2)^2 + \hat\bha^2\sin^2\theta \\
      &\leq \hat\mu(\hat r)(\hat r^2+\hat\bha^2)^2\hat f(\hat r)^2 + 2(\hat r^2+\hat\bha^2)^2\hat f(\hat r) + \hat\bha^2.
  \end{align*}
  When $\hat f(\hat r)=-\hat\mu(\hat r)^{-1}$ (so $\tilde\chi^e=0$), the right hand side evaluates to $-\frac{(\hat r^2+\hat\bha^2)^2}{\hat\mu(\hat r)}+\hat\bha^2$ which is negative when $\hat\mu(\hat r)>0$ (since upon multiplication by $\hat\mu(\hat r)$, this is $-(\hat r^2+\hat\bha^2)^2+\hat\bha^2(\hat r^2+\hat\bha^2-2\hat r)=-\hat r^4-\hat\bha^2\hat r^2-2\hat\bha^2\hat r$). At $\hat r=\hat r^e$, we require $\hat f(\hat r)<\frac{\hat\bha^2}{2(\hat r^2+\hat\bha^2)^2}$. Where $\hat\mu(\hat r)\neq 0$, the set of allowed values of $\hat f(\hat r)$ is a nonempty open interval (depending continuously on $\hat r$). We can thus find an appropriate $\hat f(\hat r)$ so that, moreover, $\tilde\chi^e(\hat r)=0$ for $\hat r\geq 3$, say.

  Having fixed $\tilde\chi^c,\tilde\chi^e$ and thus $\chi^c,\chi^e$ in~\eqref{EqKFm} in this manner, the past timelike nature of $\dd t_*$ with respect to $g_\bhm$ now follows by continuity for all sufficiently small $\bhm>0$ in view of Lemma~\ref{LemmaKMetric}.
\end{proof}

\begin{prop}[Spectral family of the connection of $g$]
\label{PropKSpectral}
  For $\bhm\in(0,\bhm_0]$ and $\sigma\in\C$, denote by $\nabla^{g_\bhm}(\sigma)\in\Diff^1\bigl(\ol{\Omega_\bhm};T_{\ol{\Omega_\bhm}}M,T_{\ol{\Omega_\bhm}}^*M\otimes T_{\ol{\Omega_\bhm}}M\bigr)$, $\sigma\in\C$, the spectral family of the Levi-Civita connection of $g_\bhm$, defined analogously to Definition~\usref{DefKFamily}. Denote by $\nabla^g(\cdot+i\sigma_1)\colon(0,\bhm_0]\times\R\ni(\bhm,\sigma_0)\mapsto\nabla^{g_\bhm}(\sigma_0+i\sigma_1)$ the total spectral family. Then
  \begin{equation}
  \label{EqKSpectralSymb}
    \nabla^g(\cdot+i\sigma_1) \in \DiffQ^{1,(1,0,1,1,1)}\bigl(\ol{\Omega_\Qop};(\pi_\Qop^*T_X M)|_{\ol{\Omega_\Qop}},(\pi_\Qop^*(T_X^*M\otimes T_X M))|_{\ol{\Omega_\Qop}}\bigr).
  \end{equation}
  Its principal symbol is $\sigmaQ^{1,(1,0,1,1,1)}(\nabla^g(\cdot+i\sigma_1))(\sigma_0,\bhm,x,\xi)=(-\sigma\,\dd t_*+\xi)\otimes(-)$,\footnote{It is irrelevant here which (rescaled) cotangent bundle $\xi$ lies in. For example, if we take $\xi\in\TQ^*X$ of unit size (with respect to any fixed positive definite fiber metric), then $\xi$ has size $(\rho_\zface\rho_\nface\rho_\iface\rho_\sface)^{-1}$ as an element of $\pi_\Qop^*(T^*_X M)$.} where $\sigma=\sigma_0+i\sigma_1$.\footnote{The term $-\sigma\,\dd t_*$ only contributes to the principal symbol in the high frequency regime $|\Re\sigma|=|\sigma_0|\gg 1$, where in view of the boundedness of $\sigma_1$ the contribution of $-i\sigma_1\,\dd t_*$ is subprincipal and therefore does, in fact, not contribute to the principal symbol. When relaxing the assumption that $\Im\sigma$ be bounded, the imaginary part of $\sigma$ does matter, however; see~\S\ref{SsKU}.} Moreover,
  \begin{subequations}
  \begin{align}
  \label{EqKSpectralzf}
    N_\zface\bigl(\bhm\nabla^g(\cdot+i\sigma_1)\bigr) &= \nabla^{\hat g}(0), \\
  \label{EqKSpectralnf}
    N_{\nface_{\tilde\sigma}}\bigl(\bhm\nabla^g(\cdot+i\sigma_1)\bigr) &= \nabla^{\hat g}(\tilde\sigma), \\
  \label{EqKSpectralmf}
    N_{\mface_{\sigma_0}}\bigl(\nabla^g(\cdot+i\sigma_1)\bigr) &= \nabla^{g_\dS}(\sigma),\quad \sigma=\sigma_0+i\sigma_1,
  \end{align}
  \end{subequations}
  where we use the isomorphism $\wt\iota$ from~\eqref{EqKBundleSpacetime} in the first two lines (to identify the bundles $(\pi_\Qop^*T_X M)|_{\zface_{\sigma_0}}=(\pi_\Qop^*T_X M)|_{\nface_{\tilde\sigma_0}}=(\pi_\qop^*T_X M)|_{\zface_\qop}$ with $T_0\R_{\hat t_*}\oplus\Tsc\hat X$, likewise for their duals), and the identification~\eqref{EqKBundleIsomf} in the third line.

  Analogous statements hold for the spectral family of the exterior derivative $\dd$, resp.\ the gradient $\nabla^g$ (with $\dd(\cdot+i\sigma_1)$ a map from complex-valued functions to sections of $\pi_\Qop^* T_X^* M$, resp.\ $\pi_\Qop^*T_X M$, over $\ol{\Omega_\Qop}$); see~\eqref{EqKSpectralD0}--\eqref{EqKSpectralD2} below for the case of $\dd$.
\end{prop}
\begin{proof}
  Consider first the exterior derivative $\dd u=(\pa_{t_*}u)\dd t_*+\dd_X u$, where, with $\ul\R=X\times\R$ denoting the trivial bundle, $\dd_X\in\Diff^1(X;\ul\R,T^*X)$ is the spatial exterior derivative. From~\eqref{EqQDiffEx}--\eqref{EqQDiffEx2} we then deduce that
  \begin{align*}
    \dd(\sigma)=-i\sigma\,\dd t_*+\dd_X &\in \DiffQ^{0,(0,0,1,1,1)}(X;\ul\R,\pi_\Qop^*T^*_X M) + \DiffQ^{1,(1,0,1,1,1)}(X;\ul\R,\pi_\Qop^*T^*_X M) \\
      &= \DiffQ^{1,(1,0,1,1,1)}(X;\ul\R,\pi_\Qop^*T^*_X M),
  \end{align*}
  now with $\ul\R=X_\Qop\times\R$. This explicit expression implies
  \begin{subequations}
  \begin{equation}
  \label{EqKSpectralD0}
    N_{\mface_{\sigma_0}}(\dd(\cdot+i\sigma_1))=\dot\dd(\sigma)
  \end{equation}
  where $\dot\dd$ is the exterior derivative operator on $\dot M=\R_{t_*}\times\dot X$. The principal symbol at $(\sigma_0,\bhm,x,\xi)$ is $\xi$. Considering the rescaling $\bhm\dd(\sigma)=-i\bhm\sigma\,\dd t_*+\bhm\dd_X$, note that $\iota(\dd_X u)=\sum_{j=1}^3(\pa_{x^j}u)\dd\hat x^j$ and $\bhm\pa_{x^j}=\pa_{\hat x^j}$, and therefore
  \begin{equation}
  \label{EqKSpectralD1}
    N_\zface(\bhm\dd(\cdot+i\sigma_1))=\hat\dd(0),
  \end{equation}
  with $\hat\dd$ the exterior derivative on $\hat M=\R_{\hat t_*}\times\hat X$. When $\sigma=\tilde\sigma/\bhm$, then $\bhm\dd(\sigma)=-i\tilde\sigma\,\dd t_*+\bhm\dd_X$; thus,
  \begin{equation}
  \label{EqKSpectralD2}
    N_{\nface_{\tilde\sigma_0}}(\bhm\dd(\cdot+i\sigma_1)) = \hat\dd(\tilde\sigma).
  \end{equation}
  \end{subequations}

  The analogous statements about the gradient $\nabla^g$ on functions follow from~\eqref{EqKSpectralD0}--\eqref{EqKSpectralD2} and the description of $g^{-1}$ in Lemma~\ref{LemmaKMetric}.

  The analysis of the Levi-Civita connection $\nabla^{g_\bhm}$ is similar. In terms of local coordinates $x=(x^1,x^2,x^3)$ on $X$ and the corresponding coordinates $(t_*,x^1,x^2,x^3)$ on $M$, the Christoffel symbols $\Gamma_{\mu\nu}^\lambda(g_\bhm)$ satisfy
  \[
    \Gamma_{\mu\nu}^\lambda(g_\bhm) = \half(g_\bhm^{-1})^{\lambda\kappa}\bigl(\pa_\mu(g_\bhm)_{\nu\kappa} + \pa_\nu(g_\bhm)_{\mu\kappa} - \pa_\kappa(g_\bhm)_{\mu\nu}\bigr) \in \rho_{\zface_\qop}^{-1}\CI(\ol{\Omega_\qop}) \subset \rho_\zface^{-1}\rho_\nface^{-1}\CI(\ol{\Omega_\Qop})
  \]
  in view of~\eqref{EqqVFX} and Lemma~\ref{LemmaKMetric}. Now, $\nabla^{g_\bhm}(u^\mu\pa_\mu)=(\pa_\nu u^\mu)\dd x^\nu\otimes\pa_\mu+u^\mu\Gamma_{\mu\nu}^\lambda\dd x^\nu\otimes\pa_\lambda$. Passing to the spectral family amounts to replacing $\pa_0$ by $-i\sigma$, and we therefore obtain~\eqref{EqKSpectralSymb} as in the discussion of $\dd$ above; also~\eqref{EqKSpectralmf} follows directly by taking the limit as $\bhm\searrow 0$ in $r>r_0>0$.
  
  When analyzing the normal operators of $\bhm\nabla^{g_\bhm}(\sigma)$ at $\zface$ and $\nface$, one works with the coordinates $(\hat t_*,\hat x)=(t_*,x)/\bhm$ and identifies the vector $u^\mu\pa_\mu$ with $u^\mu\pa_{\hat\mu}$ (where $\pa_{\hat 0}=\pa_{\hat t_*}$ and $\pa_{\hat j}=\pa_{\hat x^j}$, $j=1,2,3$); note also that the differential operator $\bhm\pa_\nu$ for $\nu=1,2,3$ is equal to $\pa_{\hat\nu}$, while the spectral family of $\bhm\pa_0=\bhm\pa_{t_*}=\pa_{\hat t_*}$ is $-i\bhm\sigma=-i\tilde\sigma$. To obtain~\eqref{EqKSpectralzf}--\eqref{EqKSpectralmf}, it then remains to note that for bounded $\hat x$, Lemma~\ref{LemmaKMetric} implies
  \[
    \lim_{\bhm\searrow 0}\bigl(\bhm\Gamma_{\mu\nu}^\lambda(g_\bhm)\bigr) = \Gamma_{\hat\mu\hat\nu}^{\hat\lambda}(\hat g).\qedhere
  \]
\end{proof}

\begin{proof}[Proof of Proposition~\usref{PropKS}]
  Since $\Box(\cdot+i\sigma_1)=\tr_g(\nabla^g(\cdot+i\sigma_1)\circ\nabla^g(\cdot+i\sigma_1))$ in the notation of Lemma~\ref{LemmaKMetric}, we only need to appeal to Proposition~\ref{PropKSpectral} and use the multiplicativity of the principal symbol and normal operator maps.
\end{proof}

The plan of the remainder of this section is as follows:
\begin{itemize}
\item In~\S\ref{SsKSy}, we work exclusively with the principal (and subprincipal) symbol of $\Box(\cdot+i\sigma_1)$; this is enough to deduce the absence of extremely high energy resonances ($|\sigma|\gg\bhm^{-1}$), see Remark~\ref{RmkKSyAbsence}. The same methods also prove the invertibility of the $\nface_{\pm,\tilde\semi}$-normal operator of $\Box(\cdot+i\sigma_1)$ at high energies, see Proposition~\ref{PropKSyKerr}.
\item In~\S\ref{SsKnf}, we study the inverse of the spectral family of the wave operator on a Kerr spacetime at small and bounded (real) energies, cf.\ Proposition~\ref{PropKS}\eqref{ItKSnf}. We first prove uniform bounds on its inverse away from zero energy (Proposition~\ref{PropKnfNz})---which suffices to obtain the absence of very high energy resonances ($|\sigma|\sim\bhm^{-1}$)---before proving uniform bounds down to zero energy (Lemma~\ref{LemmaKz} and Proposition~\ref{PropKnfZ}).
\item Having inverted all normal operators that are related to the singular Kerr limit, we then turn in~\S\ref{SsKmf} to the inversion of the spectral family on de~Sitter space at high energies. This then implies the absence of high energy resonances ($|\sigma|\geq h_0^{-1}\gg 1$) for all sufficiently small $\bhm$, see Corollary~\ref{CorKAHigh}.
\item In~\S\ref{SsKBd}, we finally control the resonances in the compact subset of $\C_\sigma$ to which they have been constrained at this point.
\item In~\S\ref{SsKU}, we explain the modifications necessary to treat the singular limit $\bhm\searrow 0$ not just in a strip of frequencies $\sigma$, but in a half space $\Im\sigma\geq -C_1$. This will complete the proof of Theorem~\ref{ThmK} (and thus of Theorem~\ref{ThmI}).
\item The minimal modifications necessary to treat the Klein--Gordon equation are discussed in~\S\ref{SsKG}.
\end{itemize}

Throughout, we will use the ($\bhm$-dependent) spatial volume density $|\dd g_X|$ on $X_\qop$, its restriction $|\dd(g_\dS)|_X|$ to $\mface_\qop$ (which is the spatial volume density for the de~Sitter metric, i.e.\ $|\dd g_\dS|=|\dd t_*||\dd(g_\dS)|_X|$), and the spatial volume density
\begin{equation}
\label{EqKSDensity}
  0<|\dd\hat g_{\hat X}| = \hat\varrho^2\sin\theta\,|\dd\hat r\,\dd\theta\,\dd\phi_*|\in\CI\bigl(\ol{\hat\Omega};\Omegasc_{\ol{\hat\Omega}}\hat X\bigr)
\end{equation}
of the Kerr metric on $\hat M=\R_{\hat t_*}\times\hat X$, defined via $|\dd\hat g|=|\dd\hat t_*| |\dd\hat g_{\hat X}|$.

\subsection{Symbolic analysis}
\label{SsKSy}

In this section, we exploit the information given by Proposition~\ref{PropKS}\eqref{ItKSSymb}. The symbolic estimates for $\Box(\cdot+i\sigma_1)$ on the Q-characteristic set are microlocal propagation estimates which are well-established in the literature.\footnote{The positive commutator arguments used for their proofs only make use of the principal symbol, and hence work in the Q-calculus as well.} Concretely, we shall use radial point estimates over the event and cosmological horizons following \cite[\S2.4]{VasyMicroKerrdS} as well as at spatial infinity for the Kerr model operators following \cite{MelroseEuclideanSpectralTheory,VasyZworskiScl}, and estimates at normally hyperbolic trapping \cite{DyatlovSpectralGaps}. The Q-algebra is furthermore related to the semiclassical cone algebra developed in \cite{HintzConicPowers,HintzConicProp}, and we use the radial point estimates established in \cite[\S4.4]{HintzConicProp} at the cone point $\pa\dot X$ in the high frequency regime (in the terminology of~\S\ref{SsIA}). There are further radial sets lying over $\iface\cap\nface$ (thus in the very high frequency regime) where we will prove Q-microlocal estimates by means of standard positive commutator arguments.

We denote by
\[
  \Sigma \subset \ol{\TQ^*_{\ol{\Omega_\Qop}}}X
\]
the characteristic set of $\Box(\cdot+i\sigma_1)$ (which is independent of $\sigma_1$), i.e.\ the closure of the zero set of $(\rho_\zface\rho_\nface\rho_\sface\rho_\iface)^2 G$ in the notation of Proposition~\ref{PropKS}; more precisely, $\Sigma$ is the union of the characteristic sets of $\Box(\cdot+i\sigma_1)$ lying in
\begin{equation}
\label{EqKSySymbolicFaces}
  \ol{\TQ^*_\sface}X,\quad\ol{\TQ^*_\iface}X,\quad\SQ^*X,
\end{equation}
where $\SQ^*X\subset\ol{\TQ^*}X$ denotes the boundary at fiber infinity. In this section, we show:

\begin{prop}[Symbolic estimates]
\label{PropKSy}
  Let $s,\gamma,l',b\in\R$, and let $\sfr\in\CI(\ol{\TQ^*_\iface}X)$. Suppose that $s>\half+C_1$, and that $\sfr-l'>-\half$, resp.\ $\sfr-l'<-\half$ at the incoming, resp.\ outgoing radial set over $\iface\cap\nface$ (see~\eqref{EqKSyRif}, \eqref{EqKSyRifm}, \eqref{EqKSyRifp} below). Suppose moreover that $\sfr$ is non-increasing along the flow of the Hamiltonian vector field $H_G$ of the principal symbol $G$ of $\Box(\cdot+i\sigma_0)$. Then for any $s_0\in\R$, $\sfr_0\in\CI(\ol{\TQ^*_\iface}X)$, and $b_0\in\R$, there exists $C>0$ so that for all $\sigma_0\in\R$, $\bhm\in(0,\bhm_0]$, and $\sigma_1\in[-C_1,C_1]$, we have
  \begin{equation}
  \label{EqKSy}
    \| u \|_{\Hext_{\Qop,\sigma_0,\bhm}^{s,(l,\gamma,l',\sfr,b)}(\Omega_\Qop)} \leq C\Bigl( \| \Box_{g_\bhm}(\sigma_0+i\sigma_1)u \|_{\Hext_{\Qop,\sigma_0,\bhm}^{s-1,(l-2,\gamma,l'-2,\sfr-1,b)}(\Omega_\Qop)} + \|u\|_{\Hext_{\Qop,\sigma_0,\bhm}^{s_0,(l,\gamma,l',\sfr_0,b_0)}(\Omega_\Qop)} \Bigr).
  \end{equation}
\end{prop}

The loss of one order in the Q-differential and $\iface$-decay sense ($s$ and $\sfr$) arises from real principal type or radial point propagation results, while the loss of two $\sface$-orders ($b$) arises at the trapped set.\footnote{Any loss in the $\sface$-order bigger than $1$ would be sufficient for this estimate, but we do not need a sharp estimate in the sequel.}

\begin{rmk}[Absence of very high energy resonances]
\label{RmkKSyAbsence}
  For sufficiently small $\tilde h=|\tilde\sigma|^{-1}>0$, the second, error, term on the right in~\eqref{EqKSy} is smaller than $\half$ times the left hand side. We conclude that any $u\in\Hext^s(\Omega_\bhm)$ with $\Box_{g_\bhm}(\sigma_0+i\sigma_1)u=0$ must vanish, provided $\tilde\sigma:=\bhm\sigma_0$ is sufficiently large in absolute value, and $\bhm>0$ is sufficiently small.
\end{rmk}

In the proof of Proposition~\ref{PropKSy}, we work our way systematically through the boundary faces~\eqref{EqKSySymbolicFaces} (over which the principal symbol is a well-defined function): first we work in $\ol{\TQ^*_\sface}X$ and $\ol{\TQ^*_\iface}X$, and then at fiber infinity $\SQ^*X\subset\ol{\TQ^*}X$. Since we work over the domain $\ol{\Omega_\Qop}$ where $r\geq\bhm$ (which we will henceforth not state explicitly anymore), the function $r$ is a defining function of $\zface\cup\nface$. Since in $\sigma>1$, the function $h=|\sigma|^{-1}$ is a defining function of $\nface\cup\sface\cup\iface$, we conclude that
\[
  \frac{h}{h+r}\ \text{is a defining function of}\ \sface\cup\iface.
\]
Furthermore, by the second part of~\eqref{EqQSpan}, smooth fiber-linear coordinates on $\TQ^*_{\ol{\Omega_\Qop}}X$ are given in polar coordinates $(r,\omega)$ on $X$ by writing the canonical 1-form as
\begin{equation}
\label{EqKSyTQCoord}
  \xi_\Qop \frac{h+r}{h}\frac{\dd r}{r} + \frac{h+r}{h}\eta_\Qop,\qquad \xi_\Qop\in\R,\ \eta_\Qop\in T^*\Sph^2.
\end{equation}

At radial and trapped sets, the subprincipal symbol of $\Box(\cdot+i\sigma_1)$ enters.

\begin{lemma}[Imaginary part]
\label{LemmaKSyIm}
  The operator
  \begin{equation}
  \label{EqKSyIm}
    \Im\Box(\cdot+i\sigma_1) := \frac{1}{2 i}\bigl(\Box(\cdot+i\sigma_1)-\Box(\cdot+i\sigma_1)^*\bigr) \in \DiffQ^{1,(2,0,1,1,1)}(X)
  \end{equation}
  has principal symbol $(\sigma_0,\bhm;x,\xi)\mapsto 2(\Im\sigma) g_\bhm^{-1}|_x(-\dd t_*,-(\Re\sigma)\,\dd t_*+\xi)$ where $\sigma=\sigma_0+i\sigma_1$.
\end{lemma}
\begin{proof}
  Since $\Box_{g_\bhm}$ is a symmetric operator on $\R_{t_*}\times\Omega_\bhm$ with respect to the volume form $|\dd g_\bhm|$, we have $\Box_{g_\bhm}(\sigma)^*=\Box_{g_\bhm}(\bar\sigma)$.

  For fixed $\sigma_1$, we have $\Im\Box(\cdot+i\sigma_1)\in\DiffQ^{1,(2,0,2,1,1)}(X)$ since the principal symbol of $\Box(\cdot+i\sigma_1)$ is real-valued; but since the $\nface_{\tilde\sigma_0}$-normal operators are symmetric (as they involve only \emph{real} frequencies), we obtain an order of improvement at $\nface$, leading to~\eqref{EqKSyIm}.

  Write $\Im\Box_{g_\bhm}(\sigma_0+i\sigma_1)=\Re\int_0^{\sigma_1} \pa_\sigma\Box_{g_\bhm}(\sigma_0+i\tau)\,\dd\tau$. Now,
  \[
    \pa_\sigma\Box_{g_\bhm}(\sigma)=\pa_\sigma(e^{i\sigma t_*}\Box_{g_\bhm} e^{-i\sigma t_*})=e^{i\sigma t_*}(i t_*\Box_{g_\bhm}-\Box_{g_\bhm} i t_*)e^{-i\sigma t_*}
  \]
  is the spectral family $-(i[\Box_{g_\bhm},t_*])(\sigma)$ of $-i[\Box_{g_\bhm},t_*]$, the principal symbol of which at $z=(0,x)\in\Omega_\bhm\subset M$ and $\zeta=-\sigma\,\dd t_*+\xi\in T^*_z M$ (where $\xi\in T^*_x X$) is $-H_{G_\bhm}t_*=\pa_\sigma G_\bhm$. Note then that $\pa_\sigma G_\bhm(z,\zeta)=2 g_\bhm^{-1}|_z(-\dd t_*,\zeta)$ (where $z=(t_*,x)$ and $\zeta\in T^*_X M$), and therefore the principal symbol of $-(i[\Box_{g_\bhm},t_*])(\sigma)$ is $2 g_\bhm^{-1}(-\dd t_*,-\sigma\,\dd t_*+\cdot)$. This implies the Lemma.
\end{proof}

\begin{notation}[Arbitrary orders]
\label{NotKOrder}
  In the arguments below, some orders of symbols on $\TQ^*X$ will be arbitrary by virtue of the symbols being supported away from some boundary hypersurfaces; in this case, we write `$*$' instead of specifying (arbitrary) orders at those boundary hypersurfaces. As an example, the lift of a compactly supported smooth function in $\hat r$ to $X_\Qop$ is an element of $S^{0,(0,*,0,*,0)}(\TQ^*X)$ (i.e.\ with the orders at $\mface$ and $\iface$ arbitrary). We use the same notation for Q-ps.d.o.s and Sobolev spaces.
\end{notation}

\subsubsection{Estimates near \texorpdfstring{$\iface$}{if}}
\label{SssKSyif}

We work at (large) positive frequencies $\sigma_0>1$ and indeed near $\sface_+\cup\iface_+$; the analysis in $\sigma_0<-1$ is completely analogous. Consider the semiclassical rescaling $h=|\sigma|^{-1}$, $z=h\sigma$,
\[
  G_{\semi,z}(h,\bhm,x,\xi) := |\sigma|^{-2}G(\sigma,\bhm,x,\xi) = G_\bhm|_x(-z\,\dd t_*+h\xi),\qquad \xi\in T^*_x\ol{\Omega_\bhm}.
\]
By Proposition~\ref{PropKS}\eqref{ItKSSymb} and the membership~\eqref{EqQDiffEx}, the symbol $G\in S^{2,(2,0,0,0,0)}$ is a quadratic form on the fibers of $\TQ^*X$ which is smooth down to $\sface_+\cup\iface_+$. Since $|\Im\sigma|\leq C_1$, we have $|z-1|\leq C h$, and therefore we can replace $G_{\semi,z}(h,\bhm,x,\xi)$ by
\begin{equation}
\label{EqKSyifGsemi}
  G_\semi(h,\bhm,x,\xi) := G_{\semi,1}(h,\bhm,x,\xi) = G_\bhm|_x(-\dd t_*+h\xi)
\end{equation}
without changing its principal part, i.e.\ its equivalence class modulo $S^{1,(2,0,0,-1,-1)}(\TQ^*X)$.

Let us consider a neighborhood of $\iface_+$. There, we have $h\lesssim r$, and thus $h/r$ is a joint defining function of $\iface_+\cup\sface_+$; replacing $\frac{h}{h+r}$ by $\frac{h}{r}$ in~\eqref{EqKSyTQCoord}, we write Q-covectors as
\begin{equation}
\label{EqKSyCoordif}
  h^{-1}(\xi\,\dd r + r\eta),\quad \xi\in\R,\ \eta\in T^*\Sph^2,
\end{equation}
with $\xi,\eta$ giving smooth fiber-linear coordinates. In terms of these, we have
\[
  G_\semi = G_\bhm(-\dd t_*+\xi\,\dd r+r\eta).
\]
At $\bhm=0$, this is the dual metric function $G_\dS$ of the de~Sitter metric, so from~\eqref{EqKdS} we find
\begin{equation}
\label{EqKSydS}
  G_\semi|_{\iface_+} = G_\dS = (1-r^2)\xi^2+|\eta|_{\slg^{-1}}^2 + 2\tilde\chi^c(r)\xi - \frac{1-\tilde\chi^c(r)^2}{1-r^2}.
\end{equation}

The structure of the characteristic set of~\eqref{EqKSydS} (in slightly different coordinates), as well as the dynamics of the null-bicharacteristic flow, was studied in detail in~\cite[\S4]{VasyMicroKerrdS}, with the caveat that now $r=0$ is resolved, i.e.\ blown up. (Recall here that $\iface_+=[0,\infty]_{\tilde\sigma}\times\dot X$ from Proposition~\ref{PropQStruct}\eqref{ItQStructif}.) We begin by noting that the Hamiltonian vector field in the coordinates~\eqref{EqKSyCoordif} takes the form
\[
  r h^{-1}H_p = (\pa_\xi p)(r\pa_r-\eta\pa_\eta) - \bigl((r\pa_r-\eta\pa_\eta)p\bigr)\pa_\xi + (\pa_\eta p)\pa_\omega-(\pa_\omega p)\pa_\eta,
\]
as can be seen by changing variables from the standard variables $(\xi_0,\eta_0)$ (with covectors written as $\xi_0\,\dd r+\eta_0$, thus $H_p=(\pa_{\xi_0}p)\pa_r-(\pa_r p)\pa_{\xi_0}+(\pa_{\eta_0}p)\pa_\omega-(\pa_\omega p)\pa_{\eta_0}$) to $(\xi,\eta)=(h\xi_0,h r^{-1}\eta_0)$. Thus, if $\half h^{-1}H_{G_\dS}r=(1-r^2)\xi+\tilde\chi^c(r)=0$ on $\Sigma$, then $0=G_\dS=|\eta|_{\slg^{-1}}^2-\frac{1}{1-r^2}$ forces $r<1$ when $(\xi,\eta)$ is finite, and then $\half(h^{-1}H_{G_\dS})^2 r=(1-r^2)h^{-1}H_{G_\dS}\xi=2 r^{-1}(1-r^2)(r^2\xi^2 + |\eta|^2_{\slg^{-1}}+\frac{r^2}{(1-r^2)^2})>0$. Therefore, the level sets of $r$ in $(0,1)$ are null-bicharacteristically convex.

At $r=0$ on the other hand, where $\tilde\chi^c(r)=0$, we have $G_\dS=\xi^2+|\eta|^2_{\slg^{-1}}-1$. The restriction of $\half r h^{-1}H_{G_\dS}$, as a b-vector field on $\TQ^*X$, to the characteristic set over $r=0$ is given by $\xi(r\pa_r-\eta\pa_\eta)+|\eta|^2_{\slg^{-1}}\pa_\xi+\eta\cdot\pa_\omega$ (at the center of $\slg$-normal coordinates $\omega$), which on the characteristic set is radial (i.e.\ vanishes as a vector field) only at
\begin{equation}
\label{EqKSyRif}
  \cR_{\iface_+,\pm} = \{ r=0,\ \xi=\pm 1,\ \eta=0 \} \cap \TQ^*_{\iface_+}X.
\end{equation}
The linearizations of $\half r h^{-1}H_{G_\dS}$ at these radial sets are
\begin{equation}
\label{EqKSyRadif}
  {\pm}(r\pa_r-\eta\pa_\eta),
\end{equation}
and inside the characteristic set over $r=0$, the $r h^{-1}H_{G_\dS}$-flow flows from the source at $\xi=-1$ to the sink at $\xi=+1$. This can be translated into an estimate by means of a standard symbolic positive commutator argument at radial sets; we sketch this near $\xi=-1$. Thus, using the local defining functions
\[
  \rho_\mface = \frac{\bhm}{h+\bhm},\qquad
  \rho_\nface = r,\qquad
  \rho_\iface = \frac{h+\bhm}{r},\qquad
  \rho_\sface = \frac{h}{h+\bhm},
\]
we consider a commutant (with constant orders, and recalling Notation~\ref{NotKOrder})
\begin{align*}
  a &= \rho_\mface^{-2\gamma} \rho_\nface^{-2 l'+2} \rho_\iface^{-2\sfr+1} \rho_\sface^{-2 b+1} \chi(\xi+1)\chi(|\eta|_{\slg^{-1}}^2)\chi(r)\chi(h/r)\chi(\bhm/r) \\
    &\qquad \in S^{*,(*,2\gamma,2 l'-2,2\sfr-1,2 b-1)}(\TQ^*X)
\end{align*}
where $\chi\in\CIc([0,2\delta))$ is identically $1$ on $[0,\delta]$ for some fixed small $\delta>0$, and satisfies $\chi'\leq 0$. The cutoffs localize to a neighborhood (in $\TQ^*X$) of $r=0$, $\xi=-1$, $\eta=0$ over $\iface_+$. Denote by $A=A^*\in\PsiQ^{*,(*,2\gamma,2 l'-2,2\sfr-1,2 b-1)}(X)$ a Q-quantization of $a$ (with Schwartz kernel supported in both factors in $\hat r>1$, $r<2$), and consider the $L^2$-pairing
\begin{equation}
\label{EqKSyCommCalc}
  2\Im\la\Box(\cdot+i\sigma_1)u,A u\ra = \la\cC u,u\ra,\qquad \cC := i[\Box(\cdot+i\sigma_1),A]+2(\Im\Box(\cdot+i\sigma_1))A.
\end{equation}
Thus, $c=\sigmaQ(\cC)\in S^{*,(*,2\gamma,2 l',2\sfr,2 b)}$, with the second summand of $\cC$ contributing an element of $S^{*,(*,2\gamma,2 l'-1,2\sfr,2 b)}$ by Lemma~\ref{LemmaKSyIm}, which is thus of lower order at $\nface$. By~\eqref{EqKSyRadif}, the rescaled symbol $\rho_\mface^{-2\gamma}\rho_\nface^{-2 l'}\rho_\iface^{-2\sfr}\rho_\sface^{-2 b}c$ is a positive multiple of $-2\sfr+2 l'-1$ at the radial set $\cR_{\iface_+,-}$; if $\sfr,l'$ are such that this is negative, then differentiation of $\chi(|\eta|^2)$ along $\eta\pa_\eta$ gives a contribution of the same sign (i.e.\ non-positive, and strictly negative where $\chi'<0$), and so does differentiation of $\chi(\bhm/r)$ along $-r\pa_r$ when $\delta>0$ is sufficiently small, whereas differentiation of $\chi(r)$ produces a nonnegative contribution which necessitates an a priori control assumption on $u$ on $\supp a\cap\supp\,\chi'(r)$. Therefore, in order to propagate Q-regularity from $r>0$ into the radial set, we need\footnote{This threshold condition is completely analogous to \cite[Theorem~4.10]{HintzConicProp}, where the notation $\sfb,\alpha$ is used instead of $\sfr,l'$.}
\begin{subequations}
\begin{equation}
\label{EqKSyRifm}
  \sfr > -\half + l'\quad \text{at}\ \cR_{\iface_+,-}.
\end{equation}
Under this assumption, we thus obtain a uniform (for $\sigma_0\in\R$, $\bhm\in(0,\bhm_0]$, and $\sigma_1\in[-C_1,C_1]$) estimate
\begin{equation}
\label{EqKSyEst}
  \|B u\|_{H_{\Qop,\sigma_0,\bhm}^{*,(*,\gamma,l',\sfr,b)}} \leq C\Bigl( \|\Box_{g_\bhm}(\sigma_0+i\sigma_1)u\|_{H_{\Qop,\sigma_0,\bhm}^{*,(*,\gamma,l'-2,\sfr-1,b-1)}} + \|E u\|_{H_{\Qop,\sigma_0,\bhm}^{*,(*,\gamma,l',\sfr,b)}} + \|u\|_{H_{\Qop,\sigma_0,\bhm}^{*,(*,\gamma,l',\sfr_0,b_0)}}\Bigr)
\end{equation}
for arbitrary $\sfr_0<\sfr$, $b_0<b$, for appropriate operators $B,E\in\PsiQ^0$ microlocalized in a neighborhood of $\iface_+$, where $B$ (quantizing a symbol arising from the elliptic leading order term of $c$ at $\cR_{\iface_+,-}$) is elliptic at $\cR_{\iface_+,-}$ and $E$ (quantizing a symbol arising from the term from differentiation of $\chi(r)$ above) can be taken to have operator wave front set contained in $r>0$.

Similarly, one can propagate regularity near $\cR_{\iface_+,+}$ from the a priori control regions $\bhm/r>0$ and a punctured neighborhood of $\cR_{\iface_+,+}$ inside of $r=0$ into $\cR_{\iface_+,+}$ itself, together with a uniform estimate that takes the same form, except now $E$ controls $u$ on these changed a priori control regions; the requirement on the orders is
\begin{equation}
\label{EqKSyRifp}
  \sfr < -\half+l'\quad\text{at}\ \cR_{\iface_+,+}.
\end{equation}
\end{subequations}
(Thus, an $\iface$-order $\sfr$ satisfying both~\eqref{EqKSyRifm} and \eqref{EqKSyRifp} must be variable. For real principal type propagation in between the two radial sets, one moreover needs $\sfr$ to be non-increasing along the direction of propagation; see e.g.\ \cite[\S4.1]{VasyMinicourse}.) We remark that if we restrict to bounded subsets of $\tilde\sigma\in[0,\infty)$, then the $\sface$-order $b$ becomes irrelevant, and thus the a priori control term in $\bhm/r>0$ (where also the $\iface$-order is irrelevant) is bounded by the overall error term (the last term in~\eqref{EqKSyEst}). This corresponds to the fact that the Q-calculus is not symbolic for finite Q-momenta away from $\iface\cup\sface$; instead, control at $\nface_{\tilde\sigma}$ requires the inversion of a model \emph{operator}, see~\S\ref{SsKnf} below.

The (microlocal) propagation estimates near $\iface_+$ but over $r>0$ are the same as those proved in \cite[\S4]{VasyMicroKerrdS}, except now the $\iface$-order $\sfr$ is variable---which, under the aforementioned monotonicity assumption on $\sfr$, does not necessitate any changes in the proofs of the propagation results. We sketch the computation of the null-bicharacteristic dynamics and of the positive commutator estimates in order to determine the relevant threshold conditions. To wit, we shall work near fiber infinity of the conormal bundle of the cosmological horizon $r=1$ of de~Sitter space; we work in $\xi<0$ and with the coordinates $\rho_\infty=|\xi|^{-1}$, $\hat\eta=\xi/\eta$ near fiber infinity. We may replace $G_\dS$ by the simpler expression $G_0=(1-r^2)\xi^2+|\eta|_{\slg^{-1}}^2$, for which one finds
\[
  \rho_\infty r h^{-1}H_{G_0} = -2(1-r^2)(r\pa_r-\hat\eta\pa_{\hat\eta}) + 2(r+r^{-1}|\hat\eta|^2)(\rho_\infty\pa_{\rho_\infty}+\hat\eta\pa_{\hat\eta}).
\]
At the radial set $w:=r-1=0$, $\rho_\infty=0$, $\hat\eta=0$, the linearization of this vector field is
\[
  4 w\pa_w + 2\rho_\infty\pa_{\rho_\infty} + 2\hat\eta\pa_{\hat\eta},
\]
and therefore this radial set is a source for the rescaled Hamiltonian flow. Thus, $H_{G_0}$ is to leading order at the radial set given by $\rho_\infty^{-1}h(4 w\pa_w+2\rho_\infty\pa_{\rho_\infty}+2\hat\eta\pa_{\hat\eta})$. Since we are working near $r=1$, we can take as local defining functions
\[
  \rho_\mface=\frac{\bhm}{h}, \qquad
  \rho_\iface=h+\bhm,\qquad
  \rho_\sface=\frac{h}{h+\bhm}.
\]
Consider the commutant
\[
  a = \rho_\infty^{-2 s+1}\rho_\mface^{-2\gamma}\rho_\iface^{-2\sfr+1}\rho_\sface^{-2 b+1} \chi(\rho_\infty)\chi(w^2)\chi(|\hat\eta|_{\slg^{-1}}^2)\chi(\bhm)\chi(h),
\]
with $\chi$ as before, and let $A=A^*$ denote a Q-quantization of $a$. The rescaled symbol $\rho_\infty^{2 s}\rho_\mface^{2\gamma}\rho_\iface^{2\sfr}\rho_\sface^{2 b}\,\sigmaQ(\cC)$ of the operator $\cC$ in~\eqref{EqKSyCommCalc} is now a sum of three types of terms: the first type arises from differentiating the weights of $a$, giving $-2(2 s-1)$ at the radial set; the second arises from differentiating the cutoffs in $\rho_\infty$, $w^2$, $|\hat\eta|_{\slg^{-1}}^2$, which give non-positive terms; and the third arises from the skew-adjoint part and at the radial set contributes (using Lemma~\ref{LemmaKSyIm}) twice $2\sigma_1 g_\dS^{-1}(-\dd t_*,-\dd r)$, so $-4\sigma_1$. In order to propagate out of the radial set, we thus need $-2(2 s-1)-4\sigma_1 < 0$, or equivalently
\[
  s > \half - \sigma_1,
\]
which in view of $\sigma_1\geq-C_1$ holds provided $s>\half+C_1$. Propagation out of the opposite radial set (at $r=1$ and $\xi>0$, $\xi^{-1}=0$, $\hat\eta=0$) requires the same threshold condition.

Finally, near $r=2$, say in $r\in[\frac32,2]$ for definiteness, we need to use energy estimates in order to deal with the presence of a Cauchy hypersurface at $r=2$; note that $\dd r$ is past timelike in this region. We can thus apply the semiclassical energy estimates of \cite[\S3.3]{VasyMicroKerrdS}, extended to general orders $s$, $\sfr$ using microlocal propagation results in a manner completely analogous to \cite[\S2.1.3]{HintzVasySemilinear}, in order to estimate $u$ in $\bar H_{\Qop,\sigma_0,\bhm}^{s,(*,\gamma,*,\sfr,b)}$ near $r=2$ in terms of its norm near $r=\frac32$.

\bigskip

To summarize, we can propagate Q-Sobolev regularity from the radial sets over the cosmological horizon towards the conic point $r=0$ and into $\cR_{\iface_+,-}$. For finite $\tilde\sigma$, this can be propagated further into $\cR_{\iface_+,+}$ and then outwards into $r>0$, at which point we have microlocal control on the whole Q-phase space over $\iface_+\cap\{r<2\}$; energy estimates near $r=2$ then give uniform control down to $r=2$. In order to complete the proof of the estimate~\eqref{EqKSy} for finite $\tilde\sigma$, it thus remains to control Q-regularity for bounded $\hat r$ and $r\simeq 1$, which is done in~\S\ref{SssKSyHor}.

In the semiclassical regime $\tilde\sigma\to\infty$, we cannot yet propagate into the outgoing radial set $\cR_{\iface_+,+}$ since this requires control on its unstable manifold also over $\nface^\circ\cap\sface$---which requires the analysis of the $\nface$-normal operator, i.e.\ the spectral family of the Kerr wave operator, at high energies. This is the subject of~\S\ref{SssKSyt} below. We remark that the radial point estimates at $\cR_{\iface_+,\pm}$ are, from the perspective of $\nface$, semiclassical scattering estimates in asymptotically Euclidean scattering; such estimates were first proved by Vasy--Zworski \cite{VasyZworskiScl} for high energy potential scattering on asymptotically Euclidean Riemannian manifolds.

\subsubsection{Estimates near \texorpdfstring{$\sface\cap\nface$}{the intersection of sf and nf}}
\label{SssKSyt}

Since the analysis in~\S\ref{SssKSyif} covers (an open neighborhood of) the corner $\sface\cap\nface\cap\iface$, we may work in a region $\hat r<\hat r_0$ for an arbitrary large $\hat r_0$; moreover, we work at large $|\tilde\sigma|=\tilde h^{-1}$, so local boundary defining functions are
\[
  \rho_\nface = \bhm,\quad
  \rho_\sface = \tilde h = \frac{h}{\bhm}.
\]
Our local coordinate system $\tilde h,\bhm,\hat r,\omega$ is disjoint from the other boundary hypersurfaces of $X_\Qop$. We introduce smooth fiber-linear coordinates on $\TQ^*X$ by writing the canonical 1-form as
\begin{equation}
\label{EqKSytCoord}
  \tilde h^{-1}(\xi\,\dd\hat r + \hat r\eta),\qquad \xi\in\R,\ \eta\in T^*\Sph^2.
\end{equation}
In these coordinates, the semiclassically rescaled principal symbol $G_\semi$ (see~\eqref{EqKSyifGsemi}) is, using the notation of Definition~\ref{DefKDual}, at $\bhm=0$ given by
\begin{equation}
\label{EqKSySym}
  G_\semi = \hat G|_{\hat x}(-\dd\hat t_*+\xi\,\dd\hat r + \hat r\eta).
\end{equation}
Indeed, this is the limit as $\bhm\searrow 0$ (for bounded $\hat x$) of $G_\bhm|_{\bhm\hat x}(-\dd(\bhm\hat t_*)+h\tilde h^{-1}(\xi\,\dd\hat r+\hat r\eta))$. But~\eqref{EqKSySym} is the semiclassical principal symbol of the spectral family $\tilde h^2\Box_{\hat g}(\tilde h^{-1})$; a full description of its characteristic set and null-bicharacteristic flow in the black hole exterior $\hat r>\hat r^e$ can be found in~\cite[\S\S3.1--3.2]{DyatlovWaveAsymptotics}. We in particular note that the trapped set of $G_\bhm$ lies over a fixed compact subset of radii $\hat r$ as $\bhm\searrow 0$; this follows from \cite{DyatlovWaveAsymptotics}. For us, it is convenient to use the fact that the trapped set depends \emph{smoothly} on $\bhm$ down to $\bhm=0$.\footnote{Using this fact runs counter to our insistence that only the Kerr model needs to be analyzed explicitly, whereas the Kerr--de Sitter wave operators are exclusively treated perturbatively. One may instead use that the trapping on subextremal Kerr spacetimes is $k$-normally hyperbolic (for any fixed $k$), as proved in \cite{WunschZworskiNormHypResolvent} and \cite[\S3.2]{DyatlovWaveAsymptotics}, together with the structural stability of such trapping \cite{HirschPughShubInvariantManifolds}, and note that the microlocal estimates \cite{DyatlovSpectralGaps} at the trapped set only require some large but finite amount of regularity of the defining functions for the stable and unstable manifolds; see \cite[Remark after Theorem~2]{DyatlovSpectralGaps}. Using the structural stability, the symbols of the semiclassical ps.d.o.s involved in the proofs of these estimates typically only depend continuously on the parameter $\bhm$, which is inconsequential for the standard semiclassical calculus (with continuous dependence on $\bhm\in[0,\bhm_0]$). The resulting uniform semiclassical estimates are then equivalent to estimates on Q-Sobolev spaces in the extremely high frequency regime under consideration here.} This is a consequence of the explicit description in \cite[Theorem~3.2]{PetersenVasySubextremal}, and by using this fact, one can apply the proof of \cite[Theorem~1]{DyatlovSpectralGaps} at once for the smooth family of trapped sets of $g_\bhm$. (For a direct positive commutator proof of these trapping estimates, albeit not in a semiclassical setting, see \cite[\S3]{HintzPolyTrap}.) Near the event horizon $\hat r=\hat r^e$ on the other hand, we can follow \cite[\S\S4.6 and 6.4]{VasyMicroKerrdS}, which applies in the present subextremal Kerr context (see also \cite[Lemma~4.3]{HintzPrice} and \cite[Theorem~4.3]{HaefnerHintzVasyKerr}).

Since the spectral parameter $\tilde\sigma$ is real---thus $\Box_{\hat g}(\tilde\sigma)$ is formally symmetric---the threshold regularity at the radial set at fiber infinity of the conormal bundle of the event horizon $\hat r=\hat r^e$ is equal to $\half$. For the same reason, the skew-adjoint part of $\Box_{\hat g}(\tilde\sigma)$ at the trapped set has vanishing principal symbol, and hence the estimates of \cite{DyatlovSpectralGaps} apply. (See \cite[Theorem~4.7]{HintzVasyQuasilinearKdS} for an explicit statement.)

Combining the trapping, radial point, microlocal propagation, elliptic regularity, and wave propagation (in $\hat r<\hat r^e$) results proves Proposition~\ref{PropKSy} at extremely high frequencies. We also record the following consequence of these estimates together with the radial point estimates proved in the previous section (cf.\ Proposition~\ref{PropQHRel}\eqref{ItQHRelnfsemi}):

\begin{prop}[Estimates for the Kerr spectral family at high energies]
\label{PropKSyKerr}
  There exists $\tilde h_0>0$ so that the following holds. Let $\sfr\in\CI(\ol{\Tsc^*_{\pa\hat X}}\hat X)$ is a variable order function so that $\sfr>-\half$, resp.\ $\sfr<-\half$ at the semiclassical incoming, resp.\ outgoing radial set over $\pa\hat X$, and so that $\sfr$ is monotone along the Hamilton flow inside the characteristic set.\footnote{At positive frequencies, these radial sets are given by $\cR_{\iface_+,-}$ and $\cR_{\iface_+,+}$ under the isomorphism of Corollary~\ref{CorQBundle}. In general, in the coordinates~\eqref{EqKSytCoord}, the incoming, resp.\ outgoing radial set is located at $(\xi,\eta)=(-1,0)$, resp.\ $(\xi,\eta)=(1,0)$, over $\pa\hat X$, see e.g.\ \cite[\S4.8]{VasyMinicourse} or \cite{MelroseEuclideanSpectralTheory} for the non-semiclassical setting, further \cite{VasyZworskiScl} for a global semiclassical commutator estimate, and \cite[\S5]{VasyLAPLag} for a refined semiclassical estimate.} Suppose $s>\half$. Then there exists $C>0$ so that
  \[
    \|u\|_{\Hext_{\scop,\tilde h}^{s,\sfr}(\hat\Omega)} \leq C\tilde h^{-2}\|\tilde h^2\Box_{\hat g}(\pm\tilde h^{-1})u\|_{\Hext_{\scop,\tilde h}^{s-1,\sfr+1}(\hat\Omega)},\quad 0<\tilde h\leq\tilde h_0.
  \]
\end{prop}

In our application to the uniform analysis of $\Box(\cdot+i\sigma_1)$, we shall apply Proposition~\ref{PropKSyKerr} with $\sfr-l'$ in place of $\sfr$ (in particular, the threshold conditions here match those of Proposition~\ref{PropKSy}).

\subsubsection{Non-semiclassical estimates near the horizons}
\label{SssKSyHor}

Note that the only parts of the characteristic set $\Sigma$ not covered by the previous arguments are the conormal bundles over the cosmological horizon near $\mface$ and the event horizon near $\zface$, as well as their flowouts. The radial point estimates at the conormal bundles were however already discussed in the (more delicate) semiclassical setting in the previous two sections, as were the propagation estimates (including energy estimates to deal with the Cauchy hypersurfaces at $\hat r=1$ and $r=2$). This completes the proof of Proposition~\ref{PropKSy}.

\subsection{Estimates for the \texorpdfstring{$\nface_\pm$-}{low/bounded energy nf-}normal operator}
\label{SsKnf}

We now turn to estimates for the various normal operators of $\Box(\cdot+i\sigma_1)$ which were computed in Proposition~\ref{PropKS}\eqref{ItKSzf}--\eqref{ItKSmf}. The symbolic estimates proved in~\S\ref{SsKSy} restrict to symbolic estimates for all model operators, in the sense that e.g.\ for positive commutator arguments the same commutants can be used (with fewer localizers, corresponding to working on a boundary hypersurface of $X_\Qop$); on the level of function spaces, this relies on Proposition~\ref{PropQHRel}.

\begin{prop}[Uniform bounds on Kerr at bounded nonzero energies]
\label{PropKnfNz}
  Let $c\in(0,1)$, $s>\half$, and let $\sfr$ be as in Proposition~\usref{PropKSyKerr}.\footnote{For bounded nonzero $\tilde\sigma$, one can drop the rescaling of $\xi$ and $\eta$ in~\eqref{EqKSytCoord}, thus writing covectors simply as $\xi\,\dd\hat r+\hat r\eta$; the outgoing radial set is then given by $(\xi,\eta)=(\tilde\sigma,0)$ over $\hat r=\infty$, and the incoming radial set by $(\xi,\eta)=(-\tilde\sigma,0)$.} Then there exists $C>0$ so that for all $\tilde\sigma\in\R$ with $|\tilde\sigma|\in[c,c^{-1}]$,
  \begin{equation}
  \label{EqKnfNz}
    \|u\|_{\Hext_\scop^{s,\sfr}(\hat\Omega)} \leq C\|\Box_{\hat g}(\tilde\sigma)u\|_{\Hext_\scop^{s-1,\sfr}(\hat\Omega)}.
  \end{equation}
\end{prop}
\begin{proof}
  The same symbolic arguments as in the previous section give the estimate
  \[
    \|u\|_{\Hext_\scop^{s,\sfr}(\hat\Omega)} \leq C\bigl(\|\Box_{\hat g}(\tilde\sigma)u\|_{\Hext_\scop^{s-1,\sfr}(\hat\Omega)} + \|u\|_{\Hext_\scop^{-N,-N}(\hat\Omega)}\bigr).
  \]
  for any $N$, which we take to satisfy $-N<\min(s,\sfr)$; thus, the embedding $\Hext_\scop^{s,\sfr}(\hat\Omega)\hra\Hext_\scop^{-N,-N}(\hat\Omega)$ is compact. The estimate~\eqref{EqKnfNz} (for a different constant $C$) then follows provided we show that any $u\in\Hext_\scop^{s,\sfr}(\hat\Omega)$ with $\Box_{\hat g}(\tilde\sigma)u=0$ necessarily vanishes. We reduce this to the mode stability result of Whiting and Shlapentokh-Rothman \cite{WhitingKerrModeStability,ShlapentokhRothmanModeStability} which we recalled in Theorem~\ref{ThmIKerr}.
  
  Radial point estimates at the conormal bundle of the event horizon, followed by propagation of regularity from there, imply that $u$ is smooth; at spatial infinity, $u$ has infinite scattering regularity since $\Box_{\hat g}(\tilde\sigma)$ is elliptic at high scattering frequencies. At the incoming radial set, $u$ has arbitrary scattering decay, and by propagating this to a punctured neighborhood of the outgoing radial set, we conclude that $u\ni\Hext_\scop^{\infty,\sfr'}(\hat\Omega)$ where $\sfr'$ is arbitrary except $\sfr'<-\half$ at the outgoing radial set. This can be further improved by means of module regularity at the outgoing radial set, i.e.\ stable regularity under application of $\hat r(\pa_{\hat r}-i\tilde\sigma)$ and spherical vector fields; this goes back to \cite[\S12]{MelroseEuclideanSpectralTheory} and \cite{HassellMelroseVasySymbolicOrderZero}, and is discussed in detail in the present setting in \cite[\S2.4]{GellRedmanHassellShapiroZhangHelmholtz} (see also \cite[Proposition~4.4]{BaskinVasyWunschRadMink} and \cite{HaberVasyPropagation}). We thus conclude that $e^{-i\tilde\sigma\hat r}u\in\Hbext^{\infty,l_0}(\hat\Omega)$ is conormal at $\hat r=\infty$ where $l_0<-\half$. Taking into account the modified asymptotics of outgoing spherical waves caused by the black hole mass (here $1$), we consider
  \[
    u_0(\hat r,\theta,\phi_*) := e^{-i\tilde\sigma\hat r}\hat r^{-2 i\tilde\sigma}u(\hat r,\theta,\phi_*).
  \]

  Thus, $u_0$ is conormal at $\hat\rho=\hat r^{-1}=0$, but we need more precise information. To this end, we observe that the equation satisfied by $u_0$ in the coordinates $(\hat\rho,\omega)\in[0,1)\times\Sph^2$ takes the form
  \[
    \bigl(2 i\tilde\sigma\hat\rho(\hat\rho\pa_{\hat\rho}-1) + \hat\rho^2 L\bigr)u_0 = 0,
  \]
  where $L\in\Diffb^2([0,1)_{\hat\rho}\times\Sph^2)$, see \cite[Definition~2.1, Lemma~2.7, and \S4]{HintzPrice}. Rewriting this as $(\hat\rho\pa_{\hat\rho}-1)u_0=\hat\rho L'u_0$ for a new operator $L'\in\Diffb^2$, the conormality of $u_0$ at $\hat\rho=0$ can be upgraded by an iterative procedure, based on the inversion of $\hat\rho\pa_{\hat\rho}-1$, to the fact that $u_0\in\hat\rho\CI([0,1)_{\hat\rho}\times\Sph^2)$. We can now apply Theorem~\ref{ThmIKerr} to conclude that $u_0=0$ (and thus $u=0$) in $\hat r\geq \hat r^e$.

  This then implies the vanishing of $u$ in $\hat r<\hat r^e$ as well: this can be shown by considering the projections of $u$ to its separated parts $e^{i m\phi_*}S(\theta)R(\hat r)$ and noting (by inspection of the dual metric~\eqref{EqKKerrDual}) that $R$ then satisfies an ODE which upon multiplication by $\hat\mu(\hat r)$ has a regular-singular point at $\hat r=\hat r^e$; hence the infinite order vanishing of $R$ at $\hat r^e$ implies $R\equiv 0$ also in $\hat r<\hat r^e$. The proof is complete.
\end{proof}

Uniform estimates near zero energy require, first of all, an estimate for the zero energy operator:

\begin{lemma}[Zero energy operator on Kerr]
\label{LemmaKz}
  Let $s>\half$ and $\gamma\in(-\frac32,-\half)$. Then
  \begin{equation}
  \label{EqKz}
    \|u\|_{\Hbext^{s,\gamma}(\hat\Omega)} \leq C\|\Box_{\hat g}(0)u\|_{\Hbext^{s-1,\gamma+2}(\hat\Omega)}.
  \end{equation}
\end{lemma}

Recall from Proposition~\ref{PropKS}\eqref{ItKSzf} that the $\zface$-normal operator of $\Box(\cdot+i\sigma_1)$ is independent of $\sigma_1\in[-C_1,C_1]$ and $\sigma_0\in\ol\R$, and equal to the Kerr zero energy operator $\Box_{\hat g}(0)$; thus, Lemma~\ref{LemmaKz} proves the invertibility of $N_\zface(\Box(\cdot+i\sigma_1))$.

\begin{proof}[Proof of Lemma~\usref{LemmaKz}]
  Combining the symbolic estimates proved in~\S\ref{SsKSy}---or rather their restrictions to $\zface\cap\nface$, cf.\ Proposition~\ref{PropQHRel}---with elliptic b-theory near $\hat\rho=\hat r^{-1}=0$, we obtain the estimate
  \[
    \|u\|_{\Hbext^{s,\gamma}(\hat\Omega)} \leq C\bigl(\|\Box_{\hat g}(0)u\|_{\Hbext^{s-1,\gamma+2}(\hat\Omega)} + \|u\|_{\Hbext^{-N,-N}(\hat\Omega)}\bigr).
  \]
  (The b-analysis at $\hat\rho=0$ uses that $\Box_{\hat g}(0)$ is, to leading order as a b-operator, the Euclidean Laplacian $\hat\rho^2((\hat\rho D_{\hat\rho})^2+i\hat\rho D_{\hat\rho}+\Delta_\slg)$. Upon separation into spherical harmonics, this is a rescaling of the regular-singular ODE $(\hat\rho\pa_{\hat\rho})^2-\hat\rho\pa_{\hat\rho}-\ell(\ell+1)$, with $\ell\in\N_0$ labeling the degree of the spherical harmonic; the indicial solutions are $\hat\rho^{\ell+1}$ and $\hat\rho^{-\ell}$, and the choice of weight $\gamma$ ensures that the weighted $L^2$-space $\Hbext^{0,\gamma}(\hat\Omega)$ contains, for all $\ell$, the solution $\hat\rho^{\ell+1}$ but not $\hat\rho^{-\ell}$. See also \cite[Theorem~2.1]{GuillarmouHassellResI}.) Since the inclusion $\Hbext^{s,\gamma}(\hat\Omega)\hra\Hbext^{-N,-N}(\hat\Omega)$ is compact, it remains to prove the triviality of $\ker\Box_{\hat g}(0)$. This can be checked using explicit computations with special functions (as remarked in \cite{PressTeukolskyKerrII,TeukolskySeparation}), but we give a softer proof here, following \cite{HaefnerVasyKerrUnfinished}. 

  In view of~\eqref{EqKKerrDual} and~\eqref{EqKSDensity}, the operator $\Box_{\hat g}(0)$ is explicitly given by
  \begin{align*}
    \hat\varrho^2\Box_{\hat g}(0)&=D_{\hat r}\hat\mu(r) D_{\hat r} + \Delta_\slg - \frac{1-\chi^e(\hat r)^2}{\hat\mu(r)}(\bha D_{\phi_*})^2 + \bigl(\chi^e(\hat r)D_{\hat r}+D_{\hat r}\chi^e(\hat r)\bigr)\bha D_{\phi_*} \\
      &= D_{\hat r}\hat\mu(\hat r)D_{\hat r} + \Delta_\slg - \frac{\bha^2}{\hat\mu(\hat r)}D_\phi^2,
  \end{align*}
  where in the second line we passed to $\phi=\phi_*+\Phi(\hat r)$ with $\Phi'(\hat r)=-\frac{\bha\chi^e(\hat r)}{\hat\mu(\hat r)}$; note that
  \begin{equation}
  \label{EqKzPhi}
    \Phi(\hat r)=-\frac{\bha}{\beta}\log(\hat r-\hat r^e)+\tilde\Phi(\hat r),\qquad
    \beta:=\hat\mu'(\hat r^e)=\hat r^e-\hat r^c=2\sqrt{1-\hat\bha^2},
  \end{equation}
  with $\tilde\Phi$ smooth down to $\hat r=\hat r^e$. We may also arrange that $\Phi(\hat r)=0$ for large $\hat r$.

  Let now $u\in\ker\Box_{\hat g}(0)$. First of all, we have $u\in\Hbext^{\infty,\gamma}(\hat\Omega)$: conormality at, and smoothness near spatial infinity follows from the ellipticity (for large $\hat r$) of $\Box_{\hat g}(0)$ as a weighted b-differential operator, whereas smoothness near the ergoregion and in the black hole interior follows by combining radial point estimates at the event horizon and propagation estimates in the ergoregion and in the black hole interior $\hat r<\hat r^e$. Sobolev embedding for $u\in\Hbext^{\infty,\gamma+\frac32}(\hat\Omega,|\frac{\dd\hat r}{\hat r}\dd\slg|)$ implies that $|D_{\hat r}^j u|=\cO(\hat r^{-\gamma-\frac32-j})=o(\hat r^{-j})$ for any $j\in\N_0$ as $\hat r\to\infty$.

  Projecting $u(r,\theta,\phi_*)$ in the angular variables to a fixed spherical harmonic $Y_{\ell m}(\theta,\phi_*)=e^{i m\phi_*}S_{\ell m}(\theta)$, where $\ell\in\N_0$ and $m\in\Z\cap[-\ell,\ell]$, produces a separated solution
  \begin{equation}
  \label{EqKzvvstar}
    v_*(\hat r)Y_{\ell m}(\theta,\phi_*)=v(\hat r)Y_{\ell m}(\theta,\phi),\qquad
    v(\hat r)=e^{-i m\Phi(\hat r)}v_*(\hat r),
  \end{equation}
  where $v_*\in\CI([1,\infty)_{\hat r})$ satisfies $|v_*|=o(1)$ as $\hat r\to\infty$, and $v$ (which equals $v_*$ for large $\hat r$) satisfies
  \begin{equation}
  \label{EqKzODE}
    \Bigl(D_{\hat r}\hat\mu D_{\hat r} - \frac{\bha^2 m^2}{\hat\mu} + \ell(\ell+1)\Bigr)v=0.
  \end{equation}
  This is a regular-singular ODE at $\hat r=\infty$, with indicial solutions $\hat r^\ell$ (which does not decay as $\hat r\to\infty$) and $\hat r^{-\ell-1}$, and therefore we have $|v|=\cO(\hat r^{-\ell-1})$ and thus $|D_{\hat r}^j v|=\cO(\hat r^{-\ell-1-j})$ for all $j\in\N_0$.

  We first study the case $m\bha=0$, i.e.\ $\bha=0$ or $m=0$. Then $v$ is smooth on $[\hat r^e,\infty)$; upon multiplying~\eqref{EqKzODE} by $\bar v$ and integrating over $\hat r\in(\hat r^e,\infty)$, we may integrate by parts in view of $|v|=\cO(\hat r^{-1})$ and $|v'|=\cO(\hat r^{-2})$ as $\hat r\to\infty$. For $\ell=0$, we obtain $v'=0$, hence $v$ is constant and therefore must vanish since $v$ is required to decay at infinity; for $\ell\geq 1$, we obtain $v=0$ directly.

  When $m,\bha\neq 0$, the rescaling of~\eqref{EqKzODE} by $\hat\mu$ is of regular-singular type at $\hat\mu=0$, and by~\eqref{EqKzPhi} and \eqref{EqKzvvstar}, we have $v(\hat r)=(\hat r-\hat r^e)^{i m\bha/\beta}w(\hat r)$ where $w(\hat r)$ is smooth down to $\hat r=\hat r^e$. The Wronskian
  \[
   W:=\Im\bigl(v(\hat r)\mu D_{\hat r}\bar v(\hat r)\bigr)
 \]
 is constant, but decays to zero as $\hat r\to\infty$, and hence $W=0$. On the other hand, by evaluating its limit as $\hat r\searrow\hat r^e$, one finds $W=m\hat\bha|w(\hat r^e)|^2$; thus, $w(\hat r^e)=0$, and since the other indicial root of~\eqref{EqKzODE} is $-i m\bha/\beta\notin i m\bha/\beta-\N_0$, we conclude that $w$ vanishes identically, and therefore so does $v_*$ in $\hat r\geq\hat r^e$.

  Having shown that $v_*=0$ on $[\hat r^e,\infty)$, we obtain $v_*=0$ also on $[1,\hat r^e]$ since $v_*(\hat r)$ vanishes to infinite order at $\hat r=\hat r^e$ and satisfies $0=\hat\mu\hat\varrho^2\Box_{\hat g}(0)(v_* Y_{\ell m})$, which is a regular-singular ODE at $\hat r=\hat r^e$.
\end{proof}

Next, the transition between zero and nonzero frequencies is governed by a model operator on an exact cone; for purely imaginary spectral parameters, this was introduced in \cite{GuillarmouHassellResI}, while in the present context of real spectral parameters, this model operator was introduced in \cite[\S5]{GuillarmouHassellSikoraResIII}; see also \cite[Definition~2.4, \S5]{VasyLowEnergyLag}. In the following result, we work on the transition face $\tface\subset\hat X_\scbtop$, which (recalling the coordinates~\eqref{EqqCoord2} and \eqref{EqQCoord}) is
\[
  \tface = [0,\infty]_{\tilde r}\times\Sph^2,\qquad \tilde r=|\tilde\sigma|\hat r
\]
by Proposition~\ref{PropQStruct}\eqref{ItQStructnf}. Concretely, the $\tface$-normal operator of $\tilde\sigma^{-2}\Box_{\hat g}(\tilde\sigma)$ is
\begin{equation}
\label{EqKtfOp}
  \Box_\tface(1) := N_\tface(\Box_{\hat g}(\cdot)) = \tilde\Delta + 1,\qquad \tilde\Delta=D_{\tilde r}^2-\frac{2 i}{\tilde r}D_{\tilde r} + \tilde r^{-2}\Delta_\slg,
\end{equation}
see \cite[\S\S4.1 and 6]{VasyLowEnergy}.\footnote{This is the conjugation of the model operator in \cite[Definition~2.20]{HintzPrice} by $e^{i\tilde r}$. We remark that in \cite{HintzPrice}, which is based on \cite{VasyLowEnergyLag}, the analytic setup focuses on precise second microlocal/module regularity at the outgoing radial set, whereas in the present paper variable order estimates are sufficient.} On $\tface$, we work with the volume density $\tilde r^2|\dd\tilde r\,\dd\slg|$, and with Sobolev spaces
\[
  H_{\scop,\bop}^{s,\sfr,l}(\tface)
\]
which are scattering Sobolev spaces near $\tilde\rho=0$ (with variable decay order $\sfr$) and b-Sobolev spaces near $\tilde r=0$ (with decay order $l$ there). Note that
\[
  \Box_\tface(1) \in \Diff_{\scop,\bop}^{2,0,2}(\tface) = \Bigl(\frac{\tilde r}{\tilde r+1}\Bigr)^{-2}\Diff_{\scop,\bop}^2(\tface)
\]
is an unweighted scattering operator near $\tilde\rho=0$, and a weighted b-operator near $\tilde r=0$. The b-normal operator of $\tilde r^2\Box_\tface(1)$ at $\tilde r=0$ is $(\tilde r D_{\tilde r})^2-i\tilde r D_{\tilde r}+\Delta_\slg$, with indicial solutions $\hat r^{-\ell-1}Y_{\ell m}$ and $\hat r^\ell Y_{\ell m}$; the range $(\frac12,\frac32)$ of weights in Lemma~\ref{LemmaKtf} disallows the former, more singular, solution. The outgoing and incoming radial sets are as usual the graphs at $\tilde\rho=\tilde r^{-1}=0$ of $\dd\tilde r$ and $-\dd\tilde r$, respectively.

\begin{lemma}[Estimates for the $\tface$-normal operator]
\label{LemmaKtf}
  Let $s\in\R$, $l\in(\frac12,\frac32)$, and suppose $\sfr\in\CI(\ol{{}^{\scop,\bop}T^*_{\hat\rho^{-1}(0)}}\tface)$ is a variable order function which is monotone along the flow of the Hamiltonian vector field of the principal symbol of $\Box_\tface(1)$, and which satisfies $\sfr>-\half$, resp.\ $\sfr<-\half$ at the incoming, resp.\ outgoing radial set. Then there exists a constant $C>0$ so that
  \begin{equation}
  \label{EqKtfEst}
    \|u\|_{H_{\scop,\bop}^{s,\sfr,l}(\tface)} \leq C \| \Box_\tface(1)u \|_{H_{\scop,\bop}^{s-2,\sfr+1,l-2}(\tface)}.
  \end{equation}
\end{lemma}
\begin{proof}
  Radial point estimates at the scattering end $\tilde\rho=0$, and elliptic b-estimates at the small end $\tilde r=0$ of the cone $\tface$ give the estimate~\eqref{EqKtfEst} except for the presence of an additional, relatively compact, error term $C\|u\|_{H_{\scop,\bop}^{-N,-N,-N}(\tface)}$ on the right. The estimate~\eqref{EqKtfEst} thus follows from the nonexistence of outgoing elements in the kernel of $\Box_\tface(1)$, which is standard; it can be proved upon separation into spherical harmonics using Wronskian arguments, or by inspection of the asymptotic behavior of the explicit (Bessel function) solutions as done in \cite[\S\S3.4--3.5]{GuillarmouHassellResI} or \cite[Lemma~5.10]{HintzConicProp}.
\end{proof}

Lemmas~\ref{LemmaKz} and \ref{LemmaKtf} provide the normal operator estimates for the uniform low energy analysis of $\Box_{\hat g}(\cdot)\in\Diffscbt^{2,0,2,0}(\ol{\hat\Omega})\subset\Psiscbt^{2,0,2,0}(\ol{\hat\Omega})$ on the $\scbtop$-transition-Sobolev spaces $\Hext_{\scbtop,\tilde\sigma}^{s,\sfr,\gamma,l'}(\hat\Omega)$ introduced in~\S\ref{SsPH}, with $\gamma$ and $l'$ the weights at $\tface$ and $\zface$, respectively, and $\sfr\in\CI(\ol{\Tscbt^*_\scface}\hat X)$ denoting a variable scattering decay order function. Near $\scface\subset\hat X_\scbtop$, a defining function of $\scface$ is $\tilde\rho=\tilde r^{-1}$, and thus we can write $\scbtop$-covectors (cf.\ \eqref{EqPscbtFrame}) as
\[
  -\xi\frac{\dd\tilde\rho}{\tilde\rho^2}+\frac{\eta}{\tilde\rho} = \xi\,\dd\tilde r + \tilde r\eta = |\tilde\sigma| \bigl( \xi\,\dd\hat r+\hat r\eta \bigr)
\]
where $\eta\in T^*\Sph^2$. For $\tilde\sigma>0$, the outgoing (incoming) radial set is then given by $\xi=1$ ($\xi=-1$), $\eta=0$, $\hat\rho=0$, and the signs are reversed when $\tilde\sigma<0$.

\begin{prop}[Uniform bounds on Kerr near zero energy]
\label{PropKnfZ}
  Let $s>\half$, $l,\gamma\in\R$, and suppose $\gamma-l\in(-\frac32,-\half)$. Suppose $\sfr$ is a variable order function that is monotone along the Hamiltonian flow of the principal symbol of $\Box_\tface(1)$, and which satisfies $\sfr>\half$, resp.\ $\sfr<-\half$ at the incoming, resp.\ outgoing radial set. Then there exists $C>0$ so that, for $\tilde\sigma\in\pm[0,1]$, we have
  \begin{equation}
  \label{EqKnfZ}
    \|u\|_{\bar H_{\scbtop,\tilde\sigma}^{s,\sfr,\gamma,l}(\hat\Omega)} \leq C\|\Box_{\hat g}(\tilde\sigma)u\|_{\bar H_{\scbtop,\tilde\sigma}^{s-1,\sfr+1,\gamma+2,l}(\hat\Omega)}.
  \end{equation}
\end{prop}

These estimates are closely related to those proved by Vasy in~\cite{VasyLowEnergyLag}; but whereas Vasy uses a second microlocal algebra which allows for precise module regularity control at the outgoing radial set (roughly speaking allowing the order $\sfr$ to be constant---and thus high---except for a jump right at the outgoing radial set), we prove a less precise estimate on variable order spaces here. Thus, using the simpler $\scbtop$-ps.d.o.\ algebra already introduced by Guillarmou--Hassell \cite{GuillarmouHassellResI}, we are still able to prove uniform low energy resolvent estimates.

\begin{proof}[Proof of Proposition~\usref{PropKnfZ}]
  Via multiplication by $|\tilde\sigma|^l$, one may reduce to the case that $l=0$. When $|\tilde\sigma|$ is bounded away from $0$, the estimate~\eqref{EqKnfZ} is the content of Proposition~\ref{PropKnfNz}. Symbolic estimates (which at the incoming radial set only require $\sfr>-\half$) give
  \[
    \|u\|_{\bar H_{\scbtop,\tilde\sigma}^{s,\sfr,\gamma,0}(\hat\Omega)} \leq C\bigl( \|\Box_{\hat g}(\tilde\sigma)u\|_{\bar H_{\scbtop,\tilde\sigma}^{s-1,\sfr+1,\gamma+2,0}(\hat\Omega)} + \|u\|_{\bar H_{\scbtop,\tilde\sigma}^{s_0,\sfr_0,\gamma,0}(\hat\Omega)}\bigr)
  \]
  for any $s_0<s$ and $\sfr_0<\sfr$; we shall take $s_0\in(\half,s)$, and choose $\sfr_0<\sfr-1$ with $\sfr_0>-\half$ at the incoming radial set and monotone along the Hamiltonian flow. Let $\chi=\chi(\tilde\sigma/\hat\rho)=\chi(\tilde r)\in\CIc([0,1))$ denote a cutoff, identically $1$ near $0$, to a neighborhood of $\zface$. Then by writing $u=\chi(\tilde\rho)u+(1-\chi(\tilde\rho))u$, with the second summand supported away from $\zface$, we have
  \[
    \|u\|_{\bar H_{\scbtop,\tilde\sigma}^{s_0,\sfr_0,\gamma,0}(\hat\Omega)} \leq \|\chi(\tilde\rho)u\|_{\bar H_{\scbtop,\tilde\sigma}^{s_0,\sfr_0,\gamma,0}(\hat\Omega)} + C\|u\|_{\bar H_{\scbtop,\tilde\sigma}^{s_0,\sfr_0,\gamma,-N}(\hat\Omega)}
  \]
  for any fixed $N$; we take $N=1$. Moreover, uniformly for $\tilde\sigma\in[0,1]$,
  \[
    \|\chi(\tilde\rho)u\|_{\bar H_{\scbtop,\tilde\sigma}^{s_0,\sfr_0,\gamma,0}(\hat\Omega)} \leq C\|\chi(\tilde\rho)u\|_{\Hbext^{s_0,\gamma}(\hat\Omega)}.
  \]
  (In fact, the norms on both sides, in the presence of the cutoff $\chi(\tilde\rho)$, are uniformly equivalent, see~\eqref{EqPHEquivzf}.) Using Lemma~\ref{LemmaKz},
  \begin{align*}
    \|\chi(\tilde\rho)u\|_{\Hbext^{s_0,\gamma}(\hat\Omega)} &\leq C\bigl( \|\chi(\tilde\rho)\Box_{\hat g}(0)u\|_{\Hbext^{s_0-1,\gamma+2}(\hat\Omega)} + \|[\Box_{\hat g}(0),\chi(\tilde\rho)]u\|_{\Hbext^{s_0-1,\gamma+2}(\hat\Omega)}\bigr) \\
      &\leq C\bigl( \|\chi(\tilde\rho)\Box_{\hat g}(0)u\|_{\Hext_{\scbtop,\tilde\sigma}^{s_0-1,*,\gamma+2,0}(\hat\Omega)} + \|[\Box_{\hat g}(0),\chi(\tilde\rho)]u\|_{\Hext_{\scbtop,\tilde\sigma}^{s_0-1,*,\gamma+2,0}(\hat\Omega)}\bigr) \\
      &\leq C\bigl( \|\Box_{\hat g}(\tilde\sigma)u\|_{\Hext_{\scbtop,\tilde\sigma}^{s_0-1,*,\gamma+2,0}(\hat\Omega)} + \|u\|_{\Hext_{\scbtop,\tilde\sigma}^{s_0,*,\gamma,-1}(\hat\Omega)}\bigr),
  \end{align*}
  where the `$*$' indicates that the order is arbitrary; we use here that $\chi(\tilde\rho)(\Box_{\hat g}(\tilde\sigma)-\Box_{\hat g}(0))\in\Diffscbt^{1,0,2,-1}(\hat\Omega)$ and $[\Box_{\hat g}(0),\chi(\tilde\rho)]\in\Diffscbt^{1,-\infty,2,-\infty}(\hat\Omega)$.

  We have now obtained the improved estimate
  \[
    \|u\|_{\bar H_{\scbtop,\tilde\sigma}^{s,\sfr,\gamma,0}(\hat\Omega)} \leq C\bigl( \|\Box_{\hat g}(\tilde\sigma)u\|_{\bar H_{\scbtop,\tilde\sigma}^{s-1,\sfr+1,\gamma+2,0}(\hat\Omega)} + \|u\|_{\bar H_{\scbtop,\tilde\sigma}^{s_0,\sfr_0,\gamma,-1}(\hat\Omega)}\bigr).
  \]
  The next step is to strengthen this further by weakening the weight of the error term at $\tface$. To this end, we fix a cutoff $\psi\in\CI(\hat X_\scbtop)$ which is supported in a small collar neighborhood of $\tface\subset\hat X_\scbtop$ and identically $1$ near $\tface$; then for any $\delta\in(0,1]$ and $N\in\R$, we have
  \[
   \|u\|_{\bar H_{\scbtop,\tilde\sigma}^{s_0,\sfr_0,\gamma,\delta}(\hat\Omega)} \leq \|\psi u\|_{\bar H_{\scbtop,\tilde\sigma}^{s_0,\sfr_0,\gamma,\delta}(\hat\Omega)} + C\|u\|_{\bar H_{\scbtop,\tilde\sigma}^{s_0,\sfr_0,-N,\delta}(\hat\Omega)}.
  \]
  We can estimate the first term, using~\eqref{EqPHEquivtf} and Lemma~\ref{LemmaKtf}, via pullback along the coordinate change $\phi\colon(\tilde\sigma,\tilde\rho,\omega)\mapsto(\tilde\sigma,\tilde\sigma\tilde\rho,\omega)\in[0,1]\times[0,1)_{\hat\rho}\times\Sph^2_\omega$, similarly to above by
  \begin{align}
    &|\tilde\sigma|^{\delta+\frac32}\|\phi^*(\psi u)\|_{H_{\scop,\bop}^{s_0,\sfr_0,-\gamma+\delta}(\tface)} \nonumber\\
    &\quad\leq C|\tilde\sigma|^{\delta+\frac32}\Bigl( \|\phi^*(\psi)\Box_\tface(1)(\phi^*u)\|_{H_{\scop,\bop}^{s_0-2,\sfr_0+1,-\gamma+\delta-2}(\tface)} + \|[\Box_\tface(1),\phi^*(\psi)]\phi^*u\|_{H_{\scop,\bop}^{s_0,\sfr_0,-\gamma+\delta}(\tface)}\Bigr) \nonumber\\
  \label{EqKnfZalmost}
    &\quad\leq C\Bigl( \|\psi\Box_{\hat g}(\tilde\sigma)u\|_{\Hext_{\scbtop,\tilde\sigma}^{s_0-2,\sfr_0+1,\gamma+2,\delta}(\hat\Omega)} + \|u\|_{\Hext_{\scbtop,\tilde\sigma}^{s_0,\sfr_0+1,\gamma-1,\delta}(\hat\Omega)}\Bigr),
  \end{align}
  where we fix $\delta>0$ so small that $-\gamma+\delta\in(\frac12,\frac32)$. Here, we used that $\psi(\Box_{\hat g}(\tilde\sigma)-\phi_*(\tilde\sigma^2\Box_\tface(1)))\in\Diffscbt^{2,0,-3,0}(\hat\Omega)$ (which is the statement that $\tilde\sigma^2\Box_\tface(1)$ is the $\tface$-normal operator of $\Box_{\hat g}(\tilde\sigma)$).

  Altogether, increasing the $\tface$-order of the final term in~\eqref{EqKnfZalmost} to $\gamma-\delta$ (thus making this term larger) for convenience, we have shown
  \begin{align*}
    \|u\|_{\bar H_{\scbtop,\tilde\sigma}^{s,\sfr,\gamma,0}(\hat\Omega)} &\leq C\bigl( \|\Box_{\hat g}(\tilde\sigma)u\|_{\bar H_{\scbtop,\tilde\sigma}^{s-1,\sfr+1,\gamma+2,0}(\hat\Omega)} + \|u\|_{\bar H_{\scbtop,\tilde\sigma}^{s_0,\sfr_0+1,\gamma-\delta,\delta}(\hat\Omega)}\bigr) \\
      &\leq C\|\Box_{\hat g}(\tilde\sigma)u\|_{\bar H_{\scbtop,\tilde\sigma}^{s-1,\sfr+1,\gamma+2,0}(\hat\Omega)} + C|\tilde\sigma|^\delta \|u\|_{\bar H_{\scbtop,\tilde\sigma}^{s_0,\sfr_0+1,\gamma,0}(\hat\Omega)}.
  \end{align*}
  Since $s_0<s$ and  $\sfr_0+1<\sfr$, the second term, for sufficiently small $|\tilde\sigma|$, can be absorbed into the left hand side. The proof is complete.
\end{proof}

\subsection{Estimates for the \texorpdfstring{$\mface_{\pm,\semi}$-}{high energy mf-}normal operator}
\label{SsKmf}

Having proved estimates for all normal operators related to the Kerr model, we now turn to the de~Sitter model at $\mface$ and prove high energy estimates. Since the de~Sitter model involves, analytically and geometrically, a cone point due to the blow-up of the spatial manifold $X$ at $0\in X$, these estimates do not follow from \cite[\S4]{VasyMicroKerrdS}. Rather, they involve propagation estimates on semiclassical cone spaces; indeed one can quote \cite[Theorem~4.10]{HintzConicProp}. The details are as follows. By Proposition~\ref{PropKS}\eqref{ItKSmf}, the $\mface$-normal operator of $\Box(\cdot+i\sigma_1)$ is the operator family $\sigma_0\mapsto\Box_{g_\dS}(\sigma_0+i\sigma_1)$. In the high energy regime $h=|\sigma_0|^{-1}\leq 1$, $\pm\sigma_0>0$, we rescale this to
\begin{equation}
\label{EqKmOp}
  h\mapsto h^2\Box_{g_\dS}(\pm h^{-1}+i\sigma_1).
\end{equation}
Near the lift $\sface$ of $h=0$ to $\dot X_\chop\subset\mface$, we have coordinates $\tilde h=h/r$, $r$, and $\omega\in\Sph^2$, and Q-covectors can be written as $h^{-1}(\xi\,\dd r+r\eta)$, $\xi\in\R$, $\eta\in T^*\Sph^2$, as in~\eqref{EqKSyCoordif}. In view of~\eqref{EqKdS}, the semiclassical cone principal symbol of~\eqref{EqKmOp} is then
\[
  (1-r^2)\xi^2 + |\eta|_{\slg^{-1}}^2 - \frac{1}{1-r^2}.
\]
The outgoing and incoming radial sets were computed already in~\S\ref{SssKSyif}, see~\eqref{EqKSyRif}. (Indeed, in view of Corollary~\ref{CorQBundle}, we have $\TQ_{\mface_{+,\semi}\cap\,\iface_+}X\cong\Tch^*_\sface\dot X$.) Furthermore, the $\tface$-model operator of~\eqref{EqKmOp} only depends on the metric $g_\dS$ at the point $0$ where it is the Minkowski metric on $\R_{t_*}\times X$, and therefore the model operator is
\[
  \Box_\tface(1) = D_{\tilde r}^2 - \frac{2 i}{\tilde r}D_{\tilde r} + \tilde r^2\Delta_\slg + 1, \qquad \tilde r=\frac{r}{h}.
\]
This is of course the same operator as in~\eqref{EqKtfOp}, since it is the restriction of $h^2\Box(\cdot+i\sigma_1)$ to the boundary face $\mface\cap\nface$ (see also Figure~\ref{FigQSingle}). Notice how what here is a model problem at high energy right at the conic singularity of the spatial de~Sitter manifold blown up at $0$ is the same as a model problem at low energy at spatial infinity of the asymptotically flat spatial Kerr manifold.

\begin{prop}[High energy estimates on de~Sitter space]
\label{PropKmf}
  There exists $h_0>0$ so that the following holds. Let $s>\half+C_1$, $l\in(\frac12,\frac32)$, $l'\in\R$ and $\sfr\in\CI(\ol{\Tch^*_\sface}\dot X)$, and assume that $\sfr$ is monotone along the Hamiltonian flow of the semiclassical cone principal symbol of $\Box_{g_\dS}(\sigma_0+i\sigma_1)$ (with $h=\pm\sigma_0^{-1}\geq 0$ the semiclassical parameter), and so that $\sfr-l'>\half$, resp.\ $\sfr-l'<-\half$ at the incoming, resp.\ outgoing radial set. Then there exists $C>0$ so that
  \[
    \|u\|_{\Hext_{\cop,h}^{s,l,l',\sfr}(\dot\Omega)} \leq C\|h^2\Box_{g_\dS}(\pm h^{-1}+i\sigma_1)u\|_{\Hext_{\cop,h}^{s-1,l-2,l',\sfr+1}(\dot\Omega)},\qquad 0<h\leq h_0.
  \]
 (Recall here the notation $\dot\Omega$ from~\eqref{EqKSOmegahat}.)
\end{prop}
\begin{proof}
  Via multiplication by $h^{l'}$, we can reduce to the case $l'=0$. Using the assumptions on $s$ and $\sfr$, symbolic estimates (which control elements of semiclassical cone Sobolev spaces in the sense of regularity $s$ and semiclassical order $\sfr$) give
  \begin{equation}
  \label{EqKmAlmost}
    \|u\|_{\Hext_{\cop,h}^{s,l,0,\sfr}(\dot\Omega)} \leq C\bigl(\|h^2\Box_{g_\dS}(\pm h^{-1}+i\sigma_1)u\|_{\Hext_{\cop,h}^{s-1,l-2,0,\sfr+1}(\dot\Omega)} + \|u\|_{\Hext_{\cop,h}^{-N,l,0,\sfr_0}(\dot\Omega)}\bigr)
  \end{equation}
  for any fixed $N$ and $\sfr_0<\sfr$, which we fix subject to $\sfr_0<\sfr-1$, and $\sfr_0>-\half$ at the incoming radial set. The error term can then be estimated in terms of the $\tface$-normal operator $\Box_\tface(1)$ by using Lemma~\ref{LemmaKtf} in a manner completely analogous to the proof of Proposition~\ref{PropKnfZ}; this is where the assumption $l\in(\frac12,\frac32)$ is used. Thus, the last, error, term in~\eqref{EqKmAlmost} can be replaced by $\|u\|_{\Hext_{\cop,h}^{-N,l,-1,\sfr_0+1}(\dot\Omega)}\leq C h^\delta\|u\|_{\Hext_{\cop,h}^{-N,l,0,\sfr}(\dot\Omega)}$ if we choose $\delta>0$ small enough so that $\sfr_0+1+\delta<\sfr$ still.\footnote{The assumption $\sfr-l'>\half$ can be weakened to $\sfr-l'>-\half$, see the end of the proof of \cite[Theorem~4.10]{HintzConicProp}, but we do not need this precision here.} For small $h>0$, this error term can then be absorbed into the left hand side of~\eqref{EqKmAlmost}, finishing the proof.
\end{proof}

\subsection{Absence of high energy resonances}
\label{SsKA}

By combining the estimates proved in~\S\S\ref{SsKSy}--\ref{SsKmf}, we can now show:

\begin{prop}[Uniform estimates at high energies]
\label{PropKA}
  Let $s,\gamma,l',b\in\R$, and let $\sfr\in\CI(\ol{\TQ^*_\iface}X)$ be a variable order. Suppose that $s>\frac32+C_1$, $\gamma-l\in(-\frac32,-\half)$, and that $\sfr-l'>\half$, resp.\ $\sfr-l'<-\half$ at the incoming, resp.\ outgoing radial set over $\iface\cap\nface$. Suppose moreover that $\sfr$ is non-increasing along the Hamiltonian flow of the principal symbol of $\Box(\cdot+i\sigma_1)$. Let $h_0>0$ be as in Proposition~\usref{PropKmf} (i.e.\ sufficiently small). Then for any fixed $s_0<s$, $l_0<l$, $\gamma_0<\gamma$, $l'_0<l'$, $\sfr_0<\sfr$, $b_0<b$, there exists a constant $C>0$ so that for $|\sigma_0|\geq h_0^{-1}$, we have the uniform (for $|\sigma_0|\geq h_0^{-1}$, $\sigma_1\in[-C_1,C_1]$, $\bhm\in(0,\bhm_0]$) estimate
  \begin{equation}
  \label{EqKAHi}
    \|u\|_{\Hext_{\Qop,\sigma_0,\bhm}^{s,(l,\gamma,l',\sfr,b)}(\Omega_\Qop)} \leq C\bigl( \|\Box_{g_\bhm}(\sigma_0+i\sigma_1)u\|_{\Hext_{\Qop,\sigma_0,\bhm}^{s-1,(l-2,\gamma,l'-2,\sfr-1,b)}(\Omega_\Qop)} + \|u\|_{\Hext_{\Qop,\sigma_0,\bhm}^{s_0,(l_0,\gamma_0,l'_0,\sfr_0,b_0)}(\Omega_\Qop)}\bigr)
  \end{equation}
  In the remaining bounded frequency range $|\sigma_0|\leq h_0^{-1}$, we have a uniform (for $\sigma_0\in[-h_0^{-1},h_0^{-1}]$, $\sigma_1\in[-C_1,C_1]$, $\bhm\in(0,\bhm_0]$) estimate
  \begin{equation}
  \label{EqKABdd}
    \|u\|_{\Hext_{\qop,\bhm}^{s,(l,\gamma)}(\Omega_\qop)} \leq C\bigl( \|\Box_{g_\bhm}(\sigma_0+i\sigma_1)u\|_{\Hext_{\qop,\bhm}^{s-1,(l-2,\gamma)}(\Omega_\qop)} + \|u\|_{\Hext_{\qop,\bhm}^{s_0,(l_0,\gamma)}(\Omega_\qop)}\bigr).
  \end{equation}
\end{prop}
\begin{proof}
  This follows analogously to the proof of Proposition~\ref{PropKnfZ} from successive improvements of the error term of the symbolic estimate of Proposition~\ref{PropKSy} by means of the normal operator estimates proved in~\S\S\ref{SsKnf}--\ref{SsKmf}; the function spaces are related via Proposition~\ref{PropQHRel}.
  
  Thus, we fix $s_0<s-1$, $\sfr_0<\sfr-1$, and $b_0<b-2$ subject to the conditions that $s_0>\half+C_1$, and $\sfr_0-l'>-\half$ at the incoming radial set, and start with the estimate~\eqref{EqKSy}. We weaken the error term in~\eqref{EqKSy} at $\zface$: let $\chi$ be a cutoff to a neighborhood of $\zface$ as in Proposition~\ref{PropQHRel}\eqref{ItQHRelzf}, then Lemma~\ref{LemmaKz} implies (omitting the coordinate change $\phi_\zface$ from the notation)
  \begin{align*}
    &\|\chi u\|_{\Hext_{\Qop,\sigma_0,\bhm}^{s_0,(l,\gamma,l',*,*)}(\Omega_\Qop)}  \\
    &\quad \leq C\la\sigma\ra^{l'-l}\bhm^{\frac32-l}\|\chi u\|_{\Hbext^{s_0,\gamma-l}(\hat\Omega)} \\
    &\quad \leq C\la\sigma\ra^{(l'-2)-(l-2)}\bhm^{\frac32-(l-2)} \\
    &\quad\qquad \times \bigl(\|\chi\bhm^{-2}\Box_{\hat g}(0)u\|_{\Hbext^{s_0-1,\gamma-l+2}(\hat\Omega)} + \|\bhm^{-2}[\Box_{\hat g}(0),\chi]u\|_{\Hbext^{s_0-1,\gamma-l+2}(\hat\Omega)}\bigr) \\
    &\quad \leq C\bigl( \|\chi\Box_{g_\bhm}(\sigma_0+i\sigma_1)u\|_{\Hext_{\Qop,\sigma_0,\bhm}^{s_0-1,(l-2,\gamma,l'-2,*,*)}(\Omega_\Qop)} + \|u\|_{\Hext_{\Qop,\sigma_0,\bhm}^{s_0+1,(l-1,\gamma,l',*,*)}}\bigr),
  \end{align*}
  where we used $\gamma-l\in(-\frac32,-\half)$ in the application of Lemma~\ref{LemmaKz}, and the fact that $\chi(\Box_{g_\bhm}(\sigma_0+i\sigma_1)-\bhm^{-2}\Box_{\hat g}(0))\in\DiffQ^{2,(1,0,2,*,*)}$ (from Proposition~\ref{PropKS}\eqref{ItKSzf}); also $\bhm^{-2}[\Box_{\hat g}(0),\chi]\in\DiffQ^{1,(-\infty,0,2,*,*)}$ is a fortiori of this class. Since on the other hand for Q-Sobolev norms of $(1-\chi)u$ (which is supported away from $\zface$) the weight at $\zface$ is arbitrary, we can now improve the symbolic estimate~\eqref{EqKSy} (as far as the $\zface$-weight is concerned) to
  \[
    \| u \|_{\Hext_{\Qop,\sigma_0,\bhm}^{s,(l,\gamma,l',\sfr,b)}(\Omega_\Qop)} \leq C\Bigl( \| \Box_{g_\bhm}(\sigma_0+i\sigma_1)u \|_{\Hext_{\Qop,\sigma_0,\bhm}^{s-1,(l-2,\gamma,l'-2,\sfr-1,b)}(\Omega_\Qop)} + \|u\|_{\Hext_{\Qop,\sigma_0,\bhm}^{s_0+1,(l-\delta,\gamma,l',\sfr_0,b_0)}(\Omega_\Qop)} \Bigr)
  \]
  for any $\delta\in(0,1]$. For any fixed compact interval of $\sigma_0$, this implies the uniform estimate~\eqref{EqKABdd}. (The weight $l-\delta$ can be reduced to any $l_0$ using an interpolation argument.) Note also that we can apply Proposition~\ref{PropKSy} to the error term here and thereby reduce its differential order back to $s_0$.

  We work on the resulting error term $\|u\|_{\Hext_{\Qop,\sigma_0,\bhm}^{s_0,(l-\delta,\gamma,l',\sfr_0,b_0)}(\Omega_\Qop)}$ further by inverting the $\nface$-normal operator, which is $\bhm^{-2}\Box_{\hat g}(\bhm\cdot)$ by Proposition~\ref{PropKS}\eqref{ItKSnf}. Thus, reusing the symbol $\chi$ to now denote a cutoff to a collar neighborhood of $\nface$ which is identically $1$ near $\nface$, we use Proposition~\ref{PropQHRel}\eqref{ItQHRelnflow} and Proposition~\ref{PropKnfZ} to estimate, for $\tilde\sigma_0=\bhm\sigma_0$ with $\tilde\sigma_0\in\pm[0,1]$,
  \begin{align*}
    &\|\chi u\|_{\Hext_{\Qop,\sigma_0,\bhm}^{s_0,(l-\delta,\gamma,l',\sfr_0,*)}(\Omega_\Qop)} \\
    &\quad \leq \bhm^{\frac32-l'}\|\chi u\|_{\Hext_{\scbtop,\tilde\sigma_0}^{s_0,\sfr_0-l',\gamma-l',l-\delta-l'}(\hat\Omega)} \\
    &\quad\leq C \bhm^{\frac32-(l'-2)}\bigl( \|\chi \bhm^{-2}\Box_{\hat g}(\tilde\sigma_0)u\|_{\Hext_{\scbtop,\tilde\sigma_0}^{s_0-1,\sfr_0-l'+1,\gamma-l'+2,l-\delta-l'}(\hat\Omega)} \\
    &\quad\hspace{6em} + \|\bhm^{-2}[\Box_{\hat g}(\tilde\sigma_0),\chi]u\|_{\Hext_{\scbtop,\tilde\sigma_0}^{s_0-1,\sfr_0-l'+1,\gamma-l'+2,l-\delta-l'}(\hat\Omega)}\bigr) \\
    &\quad\leq C\bigl( \bigl\| \chi\Box(\cdot+i\sigma_1)u \|_{\Hext_{\Qop,\sigma_0,\bhm}^{s_0-1,(l-2-\delta,\gamma,l'-2,\sfr_0-1,*)}(\Omega_\Qop)} + \|u\|_{\Hext_{\Qop,\sigma_0,\bhm}^{s_0+1,(l-\delta,\gamma,l'-\delta',\sfr_0,*)}(\Omega_\Qop)} \bigr).
  \end{align*}
  Here, we fix $\delta>0$ sufficiently small so that $\gamma-(l-\delta)\in(-\frac32,-\half)$; and $\delta'\in(0,1]$ can be chosen arbitrarily. A completely analogous argument, now using Proposition~\ref{PropQHRel}\eqref{ItQHRelnfsemi} and Propositions~\ref{PropKSyKerr} and~\ref{PropKnfNz}, gives the high energy estimate (for $\tilde\sigma_0\in\pm[1,\infty]$)
  \begin{align*}
    &\|\chi u\|_{\Hext_{\Qop,\sigma_0,\bhm}^{s_0,(l-\delta,\gamma,l',\sfr_0,b_0)}(\Omega_\Qop)} \\
    &\quad \leq C\bigl( \bigl\| \chi\Box_{g_\bhm}(\sigma_0+i\sigma_1)u \|_{\Hext_{\Qop,\sigma_0,\bhm}^{s_0-1,(l-2-\delta,\gamma,l'-2,\sfr_0-1,b_0)}(\Omega_\Qop)} + \|u\|_{\Hext_{\Qop,\sigma_0,\bhm}^{s_0+1,(l-\delta,\gamma,l'-\delta',\sfr_0+1,b_0+2)}(\Omega_\Qop)} \bigr).
  \end{align*}
  (The semiclassical order $b_0+2$ of the error term arises from the fact that $\Box(\cdot+i\sigma_1)$ differs, near $\sface$, from its $\nface$-normal operator by an operator of class $\DiffQ^{2,(*,*,1,2,2)}$.) On the other hand, $(1-\chi)u$ is supported away from $\nface$, and hence for its Q-Sobolev norms the order at $\nface$ is arbitrary. We have thus established the uniform estimate
  \begin{equation}
  \label{EqKAAlmost}
    \| u \|_{\Hext_{\Qop,\sigma_0,\bhm}^{s,(l,\gamma,l',\sfr,b)}(\Omega_\Qop)} \leq C\Bigl( \| \Box_{g_\bhm}(\sigma_0+i\sigma_1)u \|_{\Hext_{\Qop,\sigma_0,\bhm}^{s-1,(l-2,\gamma,l'-2,\sfr-1,b)}(\Omega_\Qop)} + \|u\|_{\Hext_{\Qop,\sigma_0,\bhm}^{s_0,(l-\delta,\gamma,l'-\delta',\sfr_0,b_0)}(\Omega_\Qop)} \Bigr),
  \end{equation}
  where we used the symbolic estimate~\eqref{EqKSy} again to reduce the differential and semiclassical order to $s_0$ and $b_0$ (using $b_0+2<b$).

  Finally, for $|\sigma_0|^{-1}\leq h_0$, we can apply Proposition~\ref{PropKmf} to the localization of the error term in~\eqref{EqKAAlmost} to a collar neighborhood of $\mface$ and to these high frequencies; using Proposition~\ref{PropQHRel}\eqref{ItQHRelmfsemi} to pass between Q- and semiclassical cone Sobolev spaces, and using that $l-\delta-\gamma\in(\frac12,\frac32)$ and $(\sfr_0-\gamma)-(l'-\delta')>-\half$, resp.\ $<-\half$ at the incoming, resp.\ outgoing radial set when $\delta,\delta'>0$ are sufficiently small, an application of Proposition~\ref{PropKmf} improves~\eqref{EqKAAlmost} to the desired estimate~\eqref{EqKAHi}.
\end{proof}

\begin{cor}[Absence of high energy resonances]
\label{CorKAHigh}
  There exists $\bhm_1>0$ so that for all $\bhm\in(0,\bhm_1]$, $\sigma_0\in\R$ with $|\sigma_0|\geq h_0^{-1}$, and $\sigma_1\in[-C_1,C_1]$, we have $\sigma_0+i\sigma_1\notin\QNM(\bhm)$.
\end{cor}
\begin{proof}
  For $|\sigma_0|\geq h_0^{-1}$, the final, error, term in the estimate~\eqref{EqKAHi} is \emph{small} compared to the left hand side, since
  \[
    \|u\|_{\Hext_{\Qop,\sigma_0,\bhm}^{s_0,(l_0,\gamma_0,l'_0,\sfr_0,b_0)}(\Omega_\Qop)} = \bhm^\delta \|u\|_{\Hext_{\Qop,\sigma_0,\bhm}^{s_0,(l_0+\delta,\gamma_0+\delta,l'_0+\delta,\sfr_0+\delta,b_0)}(\Omega_\Qop)},
  \]
  and $l_0+\delta<l$, $\gamma_0+\delta<\gamma$, $l'_0+\delta<l'$, $\sfr_0+\delta<\sfr$, and $b_0<b$ for sufficiently small $\delta>0$. Thus, when $\bhm_1>0$ is sufficiently small, then for $\bhm\in(0,\bhm_1]$, the estimate~\eqref{EqKAHi} implies
  \[
    \|u\|_{\Hext_{\Qop,\sigma_0,\bhm}^{s,(l,\gamma,l',\sfr,b)}(\Omega_\Qop)} \leq C\|\Box_{g_\bhm}(\sigma_0+i\sigma_1)u\|_{\Hext_{\Qop,\sigma_0,\bhm}^{s-1,(l-2,\gamma,l'-2,\sfr-1,b)}(\Omega_\Qop)}.
  \]
  This implies the claim.
\end{proof}

\subsection{Uniform control of bounded frequencies}
\label{SsKBd}

Having proved that all resonances $\sigma\in\QNM(\bhm)$ with $\sigma_1=\Im\sigma\in[-C_1,C_1]$ satisfy $|\Re\sigma_0|<C_0$ for $C_0=h_0^{-1}$, we may now work with the holomorphic family
\[
  B := [-C_0,C_0]+i[-C_1,C_1] \ni \sigma \mapsto \bigl( (0,\bhm_0]\ni\bhm \mapsto \Box_{g_\bhm}(\sigma) \bigr) \in \Diffq^{2,2,0}(\ol{\Omega_\qop})
\]
of q-differential operators. For this family, we have the uniform estimate~\eqref{EqKABdd}; for its $\mface_\sigma$-normal operator $\Box_{g_\dS}(\sigma)$, we moreover have uniform estimates
\begin{equation}
\label{EqKBddS}
  \|u\|_{\Hbext^{s,l}(\dot\Omega)} \leq C\bigl( \|\Box_{g_\dS}(\sigma)u\|_{\Hbext^{s-1,l-2}(\dot\Omega)} + \|u\|_{\Hbext^{s_0,l_0}(\dot\Omega)} \bigr)
\end{equation}
for any fixed $s_0<s$, $l_0<l$ when $s>\half+C_1$, $l\in(\frac12,\frac32)$. (This follows under the stronger requirement $s>\frac32+C_1$ from the relationship between q- and b-Sobolev spaces, see Proposition~\ref{PropqHRel}\eqref{ItqHRelmf}; or it follows directly by combining elliptic b-theory in $r<1$, and radial point and propagation estimates in $r\geq 1$.)

For the following result, we recall that $\dot\beta\colon\dot X\to X$ is the blow-down map (used before in Lemma~\ref{LemmaKBundleIso}), and we recall $\Omega_\dS:=B(0,2)\subset X$ from Definition~\ref{DefKFamily}.

\begin{lemma}[Properties of the spectral family on de~Sitter space]
\label{LemmaKBddS}
  Let $s>\half+C_1$ and $l\in(\frac12,\frac32)$. Then for all $\sigma\in B$, the operator\footnote{Note that $\Box_{g_\dS}(\sigma)-\Box_{g_\dS}(0)\in r^{-2}\Diffb^1(\ol{\dot\Omega})$, and therefore the domain in~\eqref{EqKBddSOp} can be defined equivalently using $\Box_{g_\dS}(\sigma)$ in place of $\Box_{g_\dS}(0)$.}
  \begin{equation}
  \label{EqKBddSOp}
    \Box_{g_\dS}(\sigma) \colon \bigl\{ u\in\Hbext^{s,l}(\dot\Omega) \colon \Box_{g_\dS}(0)u\in\Hbext^{s-1,l-2}(\dot\Omega) \bigr\} \to \Hbext^{s-1,l-2}(\dot\Omega)
  \end{equation}
  is Fredholm and has index $0$. Moreover, if $u$ lies in its kernel, then $u=\dot\beta^*v$ where $v\in\CI(\ol{\Omega_\dS})$.
\end{lemma}
\begin{proof}
  We complement~\eqref{EqKBddS} by an analogous estimate for the adjoint operator on the dual function spaces,
  \[
    \|v\|_{\Hbsupp^{-s+1,-l+2}\bigl(\ol{\dot\Omega}\bigr)} \leq C\bigl( \|\Box_{g_\dS}(\sigma)^*v\|_{\Hbsupp^{-s,-l}\bigl(\ol{\dot\Omega}\bigr)} + \|v\|_{\Hbsupp^{s_1,l_1}\bigl(\ol{\dot\Omega}\bigr)}\bigr)
  \]
  for any $s_1<-s+1$, $l_1<-l+2$. This is proved as in \cite[\S4]{VasyMicroKerrdS} (see also \cite{ZworskiRevisitVasy}) using radial point and propagation estimates which propagate in the opposite direction compared to the proof of~\eqref{EqKBddS}, with the caveat that at the conic singularity $r=0$, one uses elliptic b-theory and $-l+2\in(\half,\frac32)$. Together with~\eqref{EqKBddS}, this implies that~\eqref{EqKBddSOp} is Fredholm.

  The high energy estimates of Proposition~\ref{PropKmf} imply that~\eqref{EqKBddSOp} is injective when $|\Re\sigma|$ is sufficiently large. One can similarly prove adjoint versions of the high energy estimates of Proposition~\ref{PropKmf}, which imply the triviality of the kernel of the adjoint $\Box_{g_\dS}(\sigma)^*$ on $\Hbsupp^{-s+1,-l+2}(\ol{\dot\Omega})$. Thus, the operator~\eqref{EqKBddSOp} is invertible for large $|\Re\sigma|$, and therefore Fredholm of index $0$ for all $\sigma\in B$ since the Fredholm index is constant.

  If $u\in\ker\Box_{g_\dS}(\sigma)$, then $u\in\Hbext^{\infty,l}(\dot\Omega)$ by elliptic regularity and radial point and propagation estimates. But interpolating between the maps $\bar H^0_\bop(\dot\Omega)\hra L^2(\Omega_\dS)$ and $\bar H^{1,1}_\bop(\dot\Omega)\hra\Hext^1(\Omega_\dS)$ implies that $u=\beta^*v$ where $v\in\Hext^l(\Omega_\dS)\cap\CI(\ol{\Omega_\dS}\setminus\{0\})$. Therefore $\Box_{g_\dS}(\sigma)v$, as an extendible distribution on $\Omega_\dS$, has support in $\{0\}$ but Sobolev regularity $\geq l-2$ (since $\Box_{g_\dS}(\sigma)\in\Diff^2(\Omega_\dS)$). Since $l-2>-\frac32$, we must have $\Box_{g_\dS}(\sigma)v=0$, and therefore $v$ is smooth near $0$ by elliptic regularity. (One can also prove this directly by expanding $u$ near $r=0$ into spherical harmonics and solving the resulting family of regular-singular ODEs at $r=0$.)
\end{proof}

Similarly, we can complement~\eqref{EqKABdd} by a uniform adjoint estimate
\[
  \|u\|_{\Hsupp_{\qop,\bhm}^{-s+1,(-l+2,-\gamma)}(\ol{\Omega_\qop})} \leq C\bigl( \|\Box_{g_\bhm}(\sigma_0+i\sigma_1)^*u\|_{\Hsupp_{\qop,\bhm}^{-s,(-l,-\gamma)}(\ol{\Omega_\qop})} + \|u\|_{\Hsupp_{\qop,\bhm}^{s_0,(l_0,-\gamma)}(\ol{\Omega_\qop})}
\]
for $s_0<-s+1$, $l_0<-l+2$. For any $\bhm>0$, the two estimates together imply that
\begin{equation}
\label{EqKBdMap}
  \Box_{g_\bhm}(\sigma) \colon \cH_\bhm^s := \bigl\{ u\in\Hext^s(\ol{\Omega_\bhm}) \colon \Box_{g_\bhm}(0)u\in\Hext^{s-1}(\ol{\Omega_\bhm}) \bigr\} \to \Hext^{s-1}(\ol{\Omega_\bhm})
\end{equation}
is Fredholm; and it is invertible for $\sigma=\sigma_0+i\sigma_1$, $\sigma_1\in[-C_1,C_1]$, provided $|\Re\sigma_0|$ is sufficiently large, as follows from the absence of a kernel in this semiclassical regime (proved in Corollary~\ref{CorKAHigh}) and of a cokernel (proved by means of an adjoint version of the estimate~\eqref{EqKAHi}). Thus, the map~\eqref{EqKBdMap} has Fredholm index $0$ and a meromorphic inverse.

The following two complementary results describe KdS QNMs for small masses as perturbations of dS QNMs.

\begin{prop}[Absence of QNMs away from de~Sitter QNMs]
\label{PropKBdNo}
  Suppose $\sigma_*\in B$ is such that $\CI(\ol{\Omega_\dS})\cap\ker\Box_{g_\dS}(\sigma_*)$ is trivial. Then there exists $\eps>0$ so that for all $\bhm\in(0,\eps]$ and $\sigma\in R$ with $|\sigma-\sigma_*|\leq\eps$, we have $\sigma\notin\QNM(\bhm)$.
\end{prop}

By Lemma~\ref{LemmaKdSQNM}, the assumption on $\sigma_*$ is equivalent to $\sigma_*\notin -i\N_0$.

\begin{proof}[Proof of Proposition~\usref{PropKBdNo}]
  In view of the uniform Fredholm estimates for the spectral family of $\Box_{g_\dS}$, the assumption is satisfied for an open set of $\sigma_*$ (see \cite[\S2.7]{VasyMicroKerrdS} for the relevant functional analysis). Thus, when $\eps>0$ is sufficiently small, then for $\sigma\in B$ with $|\sigma-\sigma_*|\leq\eps$, we have
  \[
    \|u\|_{\Hbext^{s,l}(\dot\Omega)} \leq C\|\Box_{g_\dS}(\sigma)u\|_{\Hbext^{s-1,l-2}(\dot\Omega)}
  \]
  for any fixed $s>\half+C_1$ and $l\in(\frac12,\frac32)$. Using this estimate, with $\half+C_1<s_0<s-1$ and $l_0-\gamma$ in place of $s,l$, we can improve the error term in~\eqref{EqKABdd} (applied with $s>\frac32+C_1$) to $\|u\|_{\Hext_{\qop,\bhm}^{s_0+1,(l_0,\gamma-1)}(\Omega_\qop)}$ (provided $l_0<l$ is sufficiently close to $l$ so that $l_0-\gamma\in(\frac12,\frac32)$ still) by exploiting the relationship between q- and b-Sobolev spaces near $\mface_\qop$, see Proposition~\ref{PropqHRel}\eqref{ItqHRelmf}. But this new error term is now small when $\bhm>0$ is sufficiently small, and can thus be absorbed into the left hand side of~\eqref{EqKABdd}. The resulting estimate, for $\bhm\leq\eps$, is
  \[
    \|u\|_{\Hext_{\qop,\bhm}^{s,(l,\gamma)}(\Omega_\qop)} \leq C \|\Box_{g_\bhm}(\sigma)u\|_{\Hext_{\qop,\bhm}^{s-1,(l-2,\gamma)}(\Omega_\qop)},\qquad |\sigma-\sigma_*|\leq\eps.
  \]
  This implies the triviality of $\ker\Box_{g_\bhm}(\sigma)$ and finishes the proof.
\end{proof}

\begin{prop}[Kerr--de~Sitter QNMs near de~Sitter QNMs]
\label{PropKBdYes}
  Let $\sigma_*\in B^\circ\cap\QNM_\dS$, and let $\eps_0>0$ be so small that $\QNM_\dS\cap\{|\sigma-\sigma_*|\leq 2\eps_0\}=\{\sigma_*\}$. Then for sufficiently small $\eps\in(0,\eps_0]$, there exists $\bhm_1>0$ so that
  \begin{equation}
  \label{EqKBdYes}
    m_\dS(\sigma_*) = \sum_{|\sigma-\sigma_*|\leq\eps} m_\bhm(\sigma),\qquad \bhm\in(0,\bhm_1].
  \end{equation}
  Moreover, for all sufficiently small $r_0>0$ and $K:=[r_0,2]_r\times\Sph^2$, the restriction of $\sum_{|\sigma-\sigma_*|\leq\eps} \Res_\bhm(\sigma)$ to $[0,1]_{t_*}\times K$ converges to $\Res_\dS(\sigma_*)$ in the topology of $\CI([0,1]\times K)$.
\end{prop}
\begin{proof}
  The proof is an elaboration on \cite[Theorem~1.1]{HintzXieSdS}. Thus, for $s>\frac32+C_1$ and $l\in(\half,\frac32)$, let
  \[
    K_0 = \ker_{\Hbext^{s,l}(\dot\Omega)}\Box_{g_\dS}(\sigma_*),\qquad
    K_0^* = \ker_{\Hbsupp^{-s+1,-l-2}\bigl(\ol{\dot\Omega}\bigr)} \Box_{g_\dS}(\sigma_*)^*.
  \]
  (Note that $e^{-i\sigma_* t_*}K_0\subseteq\Res_\dS(\sigma_*)$, but equality need not hold.) By Lemma~\ref{LemmaKBddS}, the spaces $K_0$ and $K_0^*$ have equal dimension $d\geq 1$. Choose $r^\flat>0$ and functions $u_j^\flat,u_j^\sharp\in\CIc(\dot\Omega\cap\{r>r^\flat\})$, $j=1,\ldots,d$, so that the maps $K_0\ni u\mapsto(\la u,u_j^\flat\ra)_{j=1,\ldots,d}\in\C^d$ and $K_0^*\ni u^*\mapsto(\la u^*,u_j^\sharp\ra)_{j=1,\ldots,d}\in\C^d$ are isomorphisms. Define the operators
  \begin{alignat*}{3}
    R_+ &\colon &\Hbext^{s,l}(\dot\Omega) &\ni u && \mapsto ( \la u,u_j^\flat\ra )_{j=1,\ldots, d} \in \C^d, \\
    R_- &\colon &\C^d &\ni (w_j)_{j=1,\ldots,d} && \mapsto \sum_{j=1}^d w_j u_j^\sharp \in \CIc(\dot\Omega\setminus\pa\dot X).
  \end{alignat*}
  Recalling the definition of $\cH_\bhm^s$ from~\eqref{EqKBdMap}, the operator
  \[
    P_\bhm(\sigma) := \begin{pmatrix} \Box_{g_\bhm}(\sigma) & R_- \\ R_+ & 0 \end{pmatrix} \colon \cH_\bhm^s \oplus \C^d \to \Hbext^{s-1,l-2}(\Omega_\bhm) \oplus \C^d
  \]
  is Fredholm of index $0$.

  The uniform estimate~\eqref{EqKABdd} (with $\gamma=0$) for $\Box_{g_\bhm}(\sigma)$ for $\sigma\in B$ implies
  \begin{equation}
  \label{EqKBdYesAlmost}
    \|(u,w)\|_{\Hext_{\qop,\bhm}^{s,(l,0)}(\Omega_\qop) \oplus \C^d} \leq C\bigl( \| P_\bhm(\sigma)(u,w) \|_{\Hext_{\qop,\bhm}^{s-1,(l-2,0)}(\Omega_\qop)\oplus\C^d} + \|(u,w)\|_{\Hext_{\qop,\bhm}^{s_0,(l_0,0)}(\Omega_\qop)\oplus\C^d} \bigr).
  \end{equation}
  But now the $\mface_\qop$-normal operator
  \[
    P_\dS(\sigma) := \begin{pmatrix} \Box_{g_\dS}(\sigma) & R_- \\ R_+ & 0 \end{pmatrix}
  \]
  has trivial nullspace for $\sigma=\sigma_*$ by construction, and thus for $|\sigma-\sigma_*|\leq 2\eps$ if we shrink $\eps>0$; we may assume that $2\eps$ is smaller than the distance from $\sigma_*$ to $\pa B$. Therefore, $P_\dS$ obeys an estimate
  \[
    \|(u,w)\|_{\Hbext^{s,l}(\dot\Omega)\oplus\C^d} \leq C\| P_\dS(\sigma)u \|_{\Hbext^{s-1,l-2}(\dot\Omega)\oplus\C^d},\qquad |\sigma-\sigma_*|\leq 2\eps.
  \]
  As in the proof of Proposition~\ref{PropKBdNo}, this can then be used to weaken the norm on the error term in~\eqref{EqKBdYesAlmost} to $\|(u,0)\|_{\Hext_{\qop,\bhm}^{s_0+2,(l_0,-1)}(\Omega_\qop)\oplus\C^d}$; this weakened error term can be absorbed into the left hand side of~\eqref{EqKBdYesAlmost} when $\bhm\in(0,\bhm_1]$ for a sufficiently small $\bhm_1>0$, and for all $\sigma\in\C$ with $|\sigma-\sigma_*|\leq 2\eps$. (Here, as in the proof of Proposition~\ref{PropKBdNo}, we need to use $s>\frac32+C_1$ and take $\half+C_1<s_0<s-1$.) Therefore, the operator $P_\bhm(\sigma)$ is injective and thus invertible for such $\bhm,\sigma$; we write its inverse as
  \[
    P_\bhm(\sigma)^{-1} = \begin{pmatrix} A_\bhm(\sigma) & B_\bhm(\sigma) \\ C_\bhm(\sigma) & D_\bhm(\sigma) \end{pmatrix},\qquad \bhm\in(0,\bhm_1],\quad |\sigma-\sigma_*|\leq 2\eps.
  \]

  By the Schur complement formula, $\Box_{g_\bhm}(\sigma)$ is invertible on $\CI(\ol{\Omega_\bhm})$ (or, equivalently, as a map~\eqref{EqKBdMap}) if and only if the $d\times d$ matrix $D_\bhm(\sigma)$ is invertible; concretely, we have
  \begin{align}
  \label{EqKBdYesBoxInv}
    \Box_{g_\bhm}(\sigma)^{-1} &= A_\bhm(\sigma)-B_\bhm(\sigma)D_\bhm(\sigma)^{-1}C_\bhm(\sigma), \\
    D_\bhm(\sigma)^{-1} &= -R_+ \Box_{g_\bhm}(\sigma)^{-1} R_-. \nonumber
  \end{align}
  Upon setting
  \[
    m'_\bhm(\sigma') := \frac{1}{2\pi i}\tr\oint_{\sigma'} D_\bhm(\sigma)^{-1}\pa_\sigma D_\bhm(\sigma)\,\dd\sigma
  \]
  these formulas imply $m_\bhm(\sigma')\leq m'_\bhm(\sigma')$ and $m'_\bhm(\sigma')\leq m_\bhm(\sigma)$, and therefore $m_\bhm(\sigma)=m'_\bhm(\sigma')$. (Since $D_\bhm(\sigma)$ is an analytic family in $\sigma$ of $d\times d$ matrices, $m_\bhm(\sigma')$ is the order of vanishing of $\det D_\bhm(\sigma)$ at $\sigma=\sigma'$.) We similarly have
  \[
    P_\dS(\sigma)^{-1} = \begin{pmatrix} A_\dS(\sigma) & B_\dS(\sigma) \\ C_\dS(\sigma) & D_\dS(\sigma) \end{pmatrix},\qquad
    m_\dS(\sigma')=\frac{1}{2\pi i}\oint_{\sigma'}D_\dS(\sigma)^{-1}\pa_\sigma D_\dS(\sigma)\,\dd\sigma.
  \]

  Set $D_0(\sigma):=D_\dS(\sigma)$. We then claim that $D_\bhm(\sigma)$ is continuous in $\bhm\in[0,\bhm_1]$ with values in holomorphic families (in $|\sigma-\sigma_*|\leq\frac32\eps$) of $d\times d$ matrices; to this end, since $D_0(\sigma)$ is holomorphic, it suffices to prove the continuity of $D_\bhm(\sigma)$ in $\bhm$ for any fixed $\sigma$ with $|\sigma-\sigma_*|\leq\frac32\eps$. Thus, let $w\in\C^d$ and consider
  \[
    (u_\bhm,w_\bhm) = P_\bhm(\sigma)^{-1} (0,w);
  \]
  we need to show that $w_\bhm=D_\bhm(\sigma)w\to D_\dS(\sigma)w$ as $\bhm\searrow 0$. But $u_\bhm\in\Hext_{\qop,\bhm}^{s,(l,0)}(\Omega_\qop)$ and $w_\bhm\in\C^d$ are uniformly bounded. Fixing $\chi\in\CIc([0,1))$ with $\chi=1$ on $[0,\half]$, this implies in view of Proposition~\ref{PropqHRel}\eqref{ItqHRelmf} that $u'_\bhm:=\chi(\bhm/r)u_\bhm\in\Hbext^{s,l}(\dot\Omega)$ is uniformly bounded. Upon passing to a subsequence of black hole masses $\bhm_j$ with $\bhm_j\searrow 0$ as $j\to\infty$, we may assume that $u'_{\bhm_j}\weakto u_0\in\Hbext^{s,l}(\dot\Omega)$ and $w_{\bhm_j}\to w_0$. When $\bhm_j$ is so small that $\chi(\bhm_j/r)=1$ for $r>r^\flat$, then $u'_{\bhm_j}|_{r>r^\flat}=u_{\bhm_j}|_{r>r^\flat}$ and therefore $R_+ u'_{\bhm_j}=R_+ u_{\bhm_j}=w$; thus, by taking the weak limit of
  \[
    P_{\bhm_j}(\sigma)(u'_{\bhm_j},w_{\bhm_j})=\bigl([\Box_{g_{\bhm_j}},\chi(\bhm_j/r)]u_{\bhm_j},w\bigr),
  \]
  as $j\to\infty$, we obtain
  \[
    \Box_{g_\dS}(\sigma)u_0 + R_- w_0 = 0,\qquad
    R_+ u_0 = w.
  \]
  But since $P_\dS(\sigma)$ is invertible, we must have $(u_0,w_0)=P_\dS(\sigma)^{-1}(0,w)$, so $w_0=D_\dS(\sigma)w$. The weak subsequential limit $(u_0,w_0)=(B_\dS(\sigma)w,D_\dS(\sigma)w)$ is therefore unique, and in particular $w_\bhm\to D_\dS(\sigma) w$, as claimed. For later use, we note that for any fixed $r_0>0$, this also shows that $(B_\bhm(\sigma)w)|_{r>r_0}=u_\bhm|_{r>r_0}=u'_\bhm|_{r>r_0}\weakto u_0|_{r>r_0}=(B_\dS(\sigma)w)|_{r>r_0}$ in $\Hext^s([r_0,2]\times\Sph^2)$ as $\bhm\searrow 0$ (where the second equality holds when $\bhm$ is so small that $\chi(\bhm/r)=1$ for $r>r_0$), and since $s$ here is arbitrary, we indeed have strong convergence
  \begin{equation}
  \label{EqKBdYesB}
    (B_\bhm(\sigma)w)|_{r>r_0} \to (B_\dS(\sigma)w)|_{r>r_0}\quad \text{in}\ \CI([r_0,2]\times\Sph^2),
  \end{equation}
  uniformly in $\sigma$ when $|\sigma-\sigma_*|\leq\frac32\eps$.
  
  As a consequence, if $\gamma=\{|\sigma-\sigma_*|=\eps\}\subset B$, oriented counterclockwise, then, for $\bhm_1>0$ so small that $\gamma\cap\QNM(\bhm)=\emptyset$ for all $\bhm\in(0,\bhm_1]$ (such $\bhm_1$ exists by Proposition~\ref{PropKBdNo}), we have
  \begin{align*}
    m_\dS(\sigma_*) = \sum_{|\sigma-\sigma_*|\leq\eps} m_\dS(\sigma) &= \frac{1}{2\pi i}\oint_\gamma D_\dS(\sigma)^{-1}\pa_\sigma D_\dS(\sigma)\,\dd\sigma \\
      &= \frac{1}{2\pi i}\oint_\gamma D_\bhm(\sigma)^{-1}\pa_\sigma D_\bhm(\sigma)\,\dd\sigma = \sum_{|\sigma-\sigma_*|\leq\eps} m_\bhm(\sigma),
  \end{align*}
  as asserted in~\eqref{EqKBdYes}.

  Finally, we can choose a number $r_\flat>0$ and polynomials $p_j=p_j(\zeta)$ with values in $\CIc(\dot\Omega\cap\{r>r_\flat\})$ for $j=1,\ldots,m_\dS(\sigma_*)$ so that $\Res_\dS(\sigma_*)$ has as a basis
  \[
    u_{\dS,j}(t_*,x)=\res_{\zeta=\sigma_*}\bigl(e^{-i t_*\zeta}\Box_{g_\dS}(\zeta)^{-1}p_j(\zeta)\bigr),\qquad j=1,\ldots,m_\dS(\sigma_*).
  \]
  The restrictions of $u_{\dS,j}$ to $[0,1]_{t_*}\times K$ remain linearly independent for any $K=[r_0,2]\times\Sph^2$ when $r_0\in(0,2)$ is sufficiently small.\footnote{Since the $u_{\dS,j}$ are analytic in an appropriate coordinate system---see \cite{HintzXiedS} for explicit formulas and \cite{GalkowskiZworskiHypo} for a general argument---the smallness requirement on $r_0$ is in fact unnecessary.} With $\gamma$ as above, we can then set
  \begin{align*}
    u_{\bhm,j}(t_*,x) &= \frac{1}{2\pi i}\oint_\gamma e^{-i t_*\zeta} \Box_\bhm(\zeta)^{-1}p_j(\zeta)\,\dd\zeta = -\frac{1}{2\pi i}\oint_\gamma e^{-i t_*\zeta} B_\bhm(\zeta)D_\bhm(\zeta)^{-1}C_\bhm(\zeta)p_j(\zeta)\,\dd\zeta \\
      &\in \sum_{|\sigma-\sigma_*|\leq\eps} \Res_\bhm(\sigma),
  \end{align*}
  where we used~\eqref{EqKBdYesBoxInv} and the holomorphicity of $A_\bhm(\zeta)$ in $\zeta$. Since the span of $p_j(\zeta)$, where $j$ and $\zeta$ range over $1,\ldots,m_\dS(\sigma_*)$ and $\C$ respectively, is a fixed finite-dimensional subspace of $\CIc(\dot\Omega\cap\{r>r_\flat\})$, one can prove the uniform convergence $C_\bhm(\zeta)p_j(\zeta)\to C_\dS(\zeta)p_j(\zeta)$ in $\C^d$ for $|\zeta-\sigma_*|\leq\frac32\eps$ using arguments analogous to those leading to~\eqref{EqKBdYesB}. Using the already established convergence of $D_\bhm(\zeta)$ and $B_\bhm(\zeta)$, we thus conclude that $u_{\bhm,j}|_{[0,1]\times K}\to u_{\dS,j}|_{[0,1]\times K}$ in $\CI([0,1]\times K)$. In particular, for all sufficiently small $\bhm>0$, the span of $u_{\bhm,1},\ldots,u_{\bhm,m_\dS(\sigma_*)}$ is $m_\dS(\sigma_*)$-dimensional. But since we already proved $\dim\sum_{|\sigma-\sigma_*|\leq\eps}\Res_\bhm(\sigma)=m_\dS(\sigma_*)$, the $u_{\bhm,j}$, $j=1,\ldots,m_\dS(\sigma_*)$, span the full space $\sum_{|\sigma-\sigma_*|\leq\eps}\Res_\bhm(\sigma)$. The proof is complete.
\end{proof}

In particular, for $\sigma_*=0$, the equation~\eqref{EqKBdYes} gives $1=m_\dS(0)=\sum_{|\sigma-\sigma_*|\leq\eps}m_\bhm(\sigma)$, and therefore there exists a single resonance $\sigma(\bhm)\in\QNM(\bhm)$ with $|\sigma(\bhm)|\leq\eps$. But since constant functions on $\R_{t_*}\times\ol{\Omega_\bhm}$ lie in the nullspace of $\Box_{g_\bhm}$, we have $0\in\QNM(\bhm)$; therefore, necessarily, $\sigma(\bhm)=0$, with $\Res_\bhm(0)$ equal to the space of constant functions. This proves part~\eqref{ItK0} of Theorem~\ref{ThmK}.

In order to finish the proof of Theorem~\ref{ThmK}, it now remains to show that there exists $h_1>0$ so that for $\sigma\in\QNM(\bhm)$ we have $\Im\sigma\leq h_1^{-1}$ for all sufficiently small $\bhm$; that is, we need to prove uniform estimates not just in strips (as done so far) but also in the full upper half plane. We turn to this next.

\subsection{Uniform analysis in a half space}
\label{SsKU}

We now work in the complement $\{\sigma\in\C\colon \Im\sigma\geq 0, |\sigma|\geq 1\}$ of the unit ball in the closed upper half plane; we parameterize this set via
\[
  [0,\pi]_\vartheta \times [1,\infty)_{\sigma_0} \mapsto \sigma=e^{i\vartheta}\sigma_0.
\]
We can then regard the spectral family $\Box_\bhm(\sigma)$ as a smooth family
\[
  [0,\pi] \ni \vartheta \mapsto \bigl( [1,\infty) \times (0,\bhm_0] \ni (\sigma_0,\bhm)\mapsto \Box_{g_\bhm}(e^{i\vartheta}\sigma_0) \bigr).
\]
In the Q-single space $X_\Qop$, we work only in $\sigma_0\geq 1$. The analogues of Proposition~\ref{PropKS} and Lemma~\ref{LemmaKSyIm} in this setting are then:

\begin{prop}[Properties of the spectral family]
\label{PropKUS}
  We have
  \[
    \Box(e^{i\vartheta}\cdot) \in \DiffQ^{2,(2,0,2,2,2)}(\ol{\Omega_\Qop}),
  \]
  with smooth dependence on $\vartheta\in[0,\pi]$. Moreover:
  \begin{enumerate}
  \item the Q-principal symbol of $\Box(e^{i\vartheta}\cdot)$ is given by~\eqref{EqKSSymb} with $\sigma=e^{i\vartheta}\sigma_0$;
  \item we have $N_\zface(\bhm^2\Box(e^{i\vartheta}\cdot))=\Box_{\hat g}(0)$;
  \item for $\tilde\sigma_0>0$, we have $N_{\nface_{\tilde\sigma_0}}(\Box(e^{i\vartheta}\cdot))=\Box_{\hat g}(e^{i\vartheta}\tilde\sigma_0)$;
  \item for $\sigma_0\geq 1$, we have $N_{\mface_{\sigma_0}}(\Box(e^{i\vartheta}\cdot))=\Box_{g_\dS}(e^{i\vartheta}\sigma_0)$;
  \item the principal symbol of $\Im\Box(e^{i\vartheta}\cdot)$ is
    \[
      (\sigma_0,\bhm;x,\xi)\mapsto 2(\Im\sigma) g_\bhm^{-1}|_x(-\dd t_*,-(\Re\sigma)\dd t_*+\xi),\qquad \sigma=e^{i\vartheta}\sigma_0.
    \]
  \end{enumerate}
\end{prop}

Since we arranged for $\dd t_*$ to be past timelike (see Lemma~\ref{LemmaKdSChi}), the symbolic estimates of~\S\ref{SsKSy} apply uniformly for $\vartheta\in[0,\pi]$ (thus $\Im e^{i\vartheta}=\sin\vartheta\geq 0$), cf.\ \cite[\S7]{VasyMicroKerrdS}; for $\vartheta\in(0,\pi)$, these are propagation estimates with complex absorption which permit propagation in the causal future direction along the Hamiltonian flow. In particular, at the radial points at spatial infinity from the perspective of the Kerr model problems at $\nface$, the need to obtain uniform estimates in $\Im\sigma\geq 0$ down to $\Im\sigma=0$ is what forces the choice of $\cR_{\iface_+,-}$ as the incoming and $\cR_{\iface_+,+}$ as the outgoing radial set (rather than the other way around). (This is the essence of the scattering microlocal proof of the limiting absorption principle, see \cite[\S\S 9 and 14]{MelroseEuclideanSpectralTheory}, or \cite[Proposition~4.13]{VasyMinicourse}.)

Next, the $\zface$-model problem is unchanged, with estimates for it provided by Lemma~\ref{LemmaKz}. For the $\nface$-model problem at frequencies $e^{i\vartheta}\tilde\sigma_0$ with $\tilde\sigma_0$ bounded away from $0$ and $\infty$, one similarly has uniform (in $\vartheta\in[0,\pi]$) symbolic estimates on the same function spaces as in Proposition~\ref{PropKnfNz}, and for the triviality of $\ker\Box_{\hat g}(e^{i\vartheta}\tilde\sigma_0)$ one can use Whiting's original result \cite{WhitingKerrModeStability} for $\vartheta\in(0,\pi)$ (or \cite{ShlapentokhRothmanModeStability,CasalsTeixeiradCModes} for all $\vartheta\in[0,\pi]$). For the uniform low energy estimate~\eqref{EqKnfZ} for $\tilde\sigma=e^{i\vartheta}\tilde\sigma_0$, $\tilde\sigma_0\in[0,1]$, the only additional ingredient is a uniform estimate
\[
  \|u\|_{H_{\scop,\bop}^{s,\sfr,l}(\tface)} \leq C\|\Box_\tface(e^{i\vartheta})u\|_{H_{\scop,\bop}^{s-2,\sfr+1,l-2}(\tface)},\qquad \vartheta\in[0,\pi],
\]
for the $\tface$-model operator $\Box_\tface(e^{i\vartheta})=\tilde\Delta+e^{i\vartheta}$ (see~\eqref{EqKtfOp}); this is again a consequence of uniform symbolic estimates together with the triviality of $\ker\Box_\tface(e^{i\vartheta})$, which for $\vartheta=0,\pi$ was proved in Lemma~\ref{LemmaKtf} and which for $\vartheta\in(0,\pi)$ follows via a direct integration by parts (since tempered elements of $\ker\Box_\tface(e^{i\vartheta})$ are then automatically rapidly decaying as $\tilde r\to\infty$). The uniform high energy estimates of Proposition~\ref{PropKmf} continue to hold for $h^2\Box_{g_\dS}(h^{-1}e^{i\vartheta})$ when $h>0$ is sufficiently small.

These symbolic and normal operator estimates can then be combined as in the proof of Proposition~\ref{PropKA} and yield, as in Corollary~\ref{CorKAHigh}, the existence of $\bhm_1>0$ and $h_1>0$ so that for all $\bhm\in(0,\bhm_1]$ and $\sigma_0\geq h_1^{-1}$ we have $e^{i\vartheta}\sigma_0\notin\QNM(\bhm)$ for all $\vartheta\in[0,\pi]$. Thus, all quasinormal modes $\sigma$ of $\Box_{g_\bhm}$, $\bhm\in(0,\bhm_1]$, satisfy $\Im\sigma\leq h_1^{-1}$. As noted at the end of~\S\ref{SsKBd}, this completes the proof of Theorem~\ref{ThmK}.

\subsection{Quasinormal modes of massive scalar fields}
\label{SsKG}

From \cite[Proposition~2.1]{HintzXieSdS}, we recall the following analogue of Lemma~\ref{LemmaKdSQNM}:

\begin{lemma}[QNMs for massive scalar fields on de~Sitter space]
\label{LemmaKG}
  Let $\nu\in\C$ and $\lambda_\pm=\frac32\pm\sqrt{\frac94-\nu}$ as in Theorem~\usref{ThmIKG}. Then the set $\QNM_\dS(\nu)$ of quasinormal modes of $\Box_{g_\dS}-\nu$ is equal to $\bigcup_\pm(-i\lambda_\pm-i\N_0)$, and the multiplicity of $\sigma\in\QNM_\dS(\nu)$ is
  \begin{equation}
  \label{EqKG}
    m_\dS(\nu;\sigma) = \sum_{\genfrac{}{}{0pt}{}{l\in\N_0}{i\sigma-l\in(\lambda_-+2\N_0)\cup(\lambda_++2\N_0)}} (2 l+1).
  \end{equation}
\end{lemma}

The formula~\eqref{EqKG} reduces to~\eqref{EqKdSQNM} for $\nu=0$; see the proof of Lemma~\ref{LemmaKdSQNM}. Define
\[
  \QNM(\nu;\bhm),\quad m_\bhm(\nu;\sigma),\quad \Res_\bhm(\nu;\sigma),
\]
and $\Res_\dS(\nu;\sigma)$ as in~\S\ref{SsKMain} but now using the operators $\Box_{g_\bhm}-\nu$ and $\Box_{g_\dS}-\nu$. \emph{Then Theorem~\usref{ThmK}, except for part~\eqref{ItK0}, remains valid upon adding the parameter $\nu$ to the notation throughout.} (This also proves Theorem~\usref{ThmIKG}.)

The proof is \emph{the same} as that of Theorem~\ref{ThmK}; indeed, the presence of the scalar field mass term $\nu$ affects neither the principal symbol of $\Box(\cdot+i\sigma_1)-\nu$ nor any of its normal operators, with the exception of
\[
  N_{\mface_{\sigma_0}}(\Box(\cdot+i\sigma_1)-\nu) = \Box_{g_\dS}(\sigma)-\nu.
\]
Thus, the invertibility properties of $\Box_{g_\dS}(\sigma)-\nu$ are what determine the limiting quasinormal mode spectrum of $\Box_{g_\bhm}(\sigma)-\nu$.

\appendix

\section{Geometric and analytic background}
\label{SP}

We begin by recalling some basic notions of b-analysis; see \cite{MelroseDiffOnMwc,GrieserBasics} for detailed accounts. Let $M$ be a smooth $n$-dimensional manifold with corners whose boundary hypersurfaces $H\subset M$ are embedded submanifolds; we write $M_1(M)$ for the collection of all boundary hypersurfaces of $M$. We write $\CI(M)$, resp.\ $\CIdot(M)$ for the space of smooth functions, resp.\ smooth functions vanishing to infinite order at all boundary hypersurfaces. A \emph{defining function} of $H$ is a smooth function $\rho_H\in\CI(M)$ so that $H=\rho_H^{-1}(0)$, and $\dd\rho_H\neq 0$ on $H$; when $M$ is a manifold with boundary, then a defining function of $\pa M$ is called a \emph{boundary defining function}. For a subset $\cH\subset M_1(M)$, a function $\rho\in\CI(M)$ is called a (joint) defining function for $\bigcup_{H\in\cH}H$ if it is the product of defining functions of $\rho_H$, $H\in\cH$. Moreover, we denote by $M^\circ$ the interior of $M$.

A \emph{boundary face} of $M$ is a nonempty intersection of boundary hypersurfaces. A \emph{p-submanifold} $S\subset M$ is a closed submanifold so that around each point in $S$ there exist local coordinates $x^1,\ldots,x^k\geq 0$, $y^1,\ldots,y^{n-k}\in\R$ on $M$ (with $k$ the codimension of the smallest boundary face containing the point under consideration) so that $S$ is given by the vanishing of the subset of these coordinates. The \emph{blow-up of $M$ along $S$}, denoted $[M;S]$, is given as a set by
\[
  [M;S] = (M\setminus S) \sqcup S {}^+N S,
\]
where $S^+ N S={}^+N S/\R_+$ is the inward pointing spherical normal bundle; here ${}^+N S={}^+T_S M/T S$ is the inward pointing normal bundle, with ${}^+T_q M\subset T M$ denoting the closed orthant of inward pointing tangent vectors (i.e.\ $\sum_{j=1}^k v_j\pa_{x^j}+\sum_{j=1}^{n-k} w_j\pa_{y^j}$ with all $v_j$ non-negative). The manifold $S$ is called the \emph{center} of the blow-up. The \emph{front face} of $[M;S]$ is $S{}^+N S$; the \emph{blow-down map} is the map $\beta\colon[M;S]\to M$ which is the identity on $M\setminus S$ and the base projection on the front face. The set $[M;S]$ can be given the structure of a smooth manifold with corners by putting on it the minimal smooth structure in which lifts of elements of $\CI(M)$ as well as polar coordinates around $S$ are smooth down to the front face; the blow-down map is then smooth. If $T\subset M$ is another p-submanifold so that near points of $S\cap T$, both $S$ and $T$ are given by the vanishing of a subset of a single local coordinate system on $M$, then we define the \emph{lift} $\beta^*T$ of $T$ to $[M;S]$ as follows: if $T\subset S$, then $\beta^*T=\beta^{-1}(T)$, and otherwise $\beta^*T$ is the closure of $\beta^{-1}(T\setminus S)$. In either case, $\beta^*T$ is a p-submanifold of $[M;S]$ and can thus be blown up again; we denote by $[M;S;T]=[[M;S];\beta^*T]$ the iterated blow-up, similarly for deeper blow-ups. The lift of a smooth map $f\colon M\to N$ between manifolds with corners to $[M;S]$ is the composition $f\circ\beta\colon[M;S]\to N$. It may happen that $[M;S;T]$ and $[M;T;S]$ are naturally diffeomorphic in the sense that the identity map on $M\setminus(S\cup T)$ extends to a diffeomorphism $[M;S;T]\cong[M;T;S]$. In this case, we shall occasionally write $[M;S,T]=[M;T,S]$. This happens in particular when $S\subset T$ or $T\subset S$, or when $S$ and $T$ are transversal.

By $\Vb(M)$ we denote the Lie algebra of \emph{b-vector fields} on $M$, i.e.\ those smooth vector fields which are tangent to all boundary hypersurfaces; in local coordinates $x^1,\ldots,x^k\geq 0$, $y^1,\ldots,y^{n-k}\in\R$ near a point on $\pa M$, such vector fields are linear combinations of $x^j\pa_{x^j}$ ($j=1,\ldots,k$) and $\pa_{y^j}$ ($j=1,\ldots,n-k$) with smooth coefficients. Thus, $\Vb(M)$ is the space of smooth sections of the \emph{b-tangent bundle} $\Tb M\to M$, a rank $n$ vector bundle equipped with a bundle map $\Tb M\to T M$ which is an isomorphism over the interior $M^\circ$; a local frame of $\Tb M$ in local coordinates is given by the aforementioned vector fields $x^j\pa_{x^j}$, $\pa_{y^j}$. Given $V\in\Vb(M)$ and a boundary hypersurface $H\subset M$, we denote by $N_H(V)\in\Vb(H)$ the restriction of $V$ to $H$, defined as $N_H(V)u=(V\tilde u)|_H$ for $u\in\CIdot(H)$, where $\tilde u\in\CI(M)$ is any smooth function with $\tilde u|_H=u$. By $\Diffb^m(M)\subset\Diff^m(M)$ we denote the space of b-differential operators of order $m$: these are locally finite sums of up to $m$-fold compositions of b-vector fields; here a $0$-fold composition is, by definition, multiplication by an element of $\CI(M)$. We write $\Diffb(M)=\bigoplus_{m\in\N_0}\Diffb^m(M)$.

If $M$ is a manifold with boundary, with boundary defining function $\rho\in\CI(M)$, then $\Vsc(M):=\rho\Vb(M)=\{\rho V\colon V\in\Vb(M)\}$ is the Lie algebra of \emph{scattering vector fields}; we have
\begin{equation}
\label{EqPVsc}
  [\Vsc(M),\Vsc(M)]\subset\rho\Vsc(M).
\end{equation}
The corresponding \emph{scattering tangent bundle} $\Tsc M\to M$ has a local frame $x^2\pa_x$, $x\pa_{y^j}$ ($j=1,\ldots,n-1$) in local coordinates $x\geq 0$, $y^1,\ldots,y^{n-1}\in\R$ near a point on the boundary. By $\Diffsc^m(M)$ we denote the corresponding space of scattering differential operators.

Let $\alpha=(\alpha_H\colon H\in M_1(M))$ be a collection of real numbers, and denote by $\rho_H\in\CI(M)$ a defining function of $H$. Then $\cA^\alpha(M)$ is the space of all smooth functions $u\in\CI(M^\circ)$ so that for all $A\in\Diffb(M)$
\begin{equation}
\label{EqPConormal}
  A u \in \biggl(\prod_{H\in M_1(M)} \rho_H^{\alpha_H}\biggr) L^\infty_\loc(M).
\end{equation}
We say that $u$ is \emph{conormal} (with weights $\alpha_H$). Given $\delta=(\delta_H\colon H\in M_1(M))$ where $\delta_H\in[0,\half)$, we define more generally $\cA^\alpha_\delta(M)$ to consist of all $u\in\CI(M^\circ)$ so that for all $m\in\N_0$ and $A\in\Diffb^m(M)$ we have $A u\in(\prod_{H\in M_1(M)}\rho_H^{\alpha_H-m\delta_H})L^\infty_\loc(M)$. More generally still, if $\cC\subset M_1(M)$ is a collection of boundary hypersurfaces, and weights $\alpha_H\in\R$ and numbers $\delta_H\in[0,\half)$ are given only for $H\in\cC$, then $\cA_\cC^\alpha(M)$ and $\cA_{\cC,\delta}^\alpha(M)$ are defined as before, but only taking products over $H\in\cC$, and allowing $A\in\Diff^m(M)$ to be any locally finite sum of up to $m$-fold compositions of smooth vector fields on $M$ which are tangent to all $H\in\cC$ (but not necessarily the other boundary hypersurfaces). We shall refer to such conormal distributions as \emph{smooth down to the boundary hypersurfaces $M_1(M)\setminus\cC$}.

Spaces of conormal functions with $\delta>0$ arise in particular as follows; for notational simplicity we only discuss the case that $M$ is a manifold with boundary $\pa M$. Suppose $\sfa\in\CI(\pa M)$ is bounded, and let $\sfa_-=\inf\sfa$. Let $\rho\in\CI(M)$ denote a boundary defining function. Then $\cA^\sfa(M)\subset\bigcap_{\delta>0}\cA^{\sfa_-}_\delta(M)$ is the space of all functions of the form $\rho^{\tilde\sfa}u_0$ where $u_0\in\bigcap_{\delta>0}\cA_\delta^0(M)$, with $\tilde\sfa\in\CI(M)$ any smooth extension of $\sfa$. The relevance of $\delta>0$ is that it ensures that $\rho^{\tilde\sfa}$ itself lies in $\cA^\sfa(M)$.

Suppose next that $E\to M$ is a smooth vector bundle over a manifold with corners. By $\bar E\to M$ we denote the \emph{radial compactification} of $E$; this is a closed ball bundle. The fiber bundle $\bar E$ is defined fiber-wise by means of the radial compactification of $\R^k$ (with $k$ the rank of $E$ as a real vector bundle), which is defined as
\[
  \ol{\R^k} := \Bigl(\R^k \sqcup \bigl( [0,\infty)_\rho \times \Sph^{k-1} \bigr) \Bigr) / \sim,
\]
where we identify $0\neq x=r\omega$ (in polar coordinates on $\R^k)$ with $(\rho,\omega)=(r^{-1},\omega)$. A special case is $\ol\R=[-\infty,\infty]=\R\cup\{-\infty,\infty\}$, with the function $\pm(1,\infty)\ni x\mapsto \pm x^{-1}$ extending to a diffeomorphism $\pm(1,\infty]\to[0,1)$; thus, the function $|x|^{-1}$ is smooth on $\ol\R\setminus\{0\}$, and it is a defining function of $\pa\ol\R=\{-\infty,\infty\}$. For $a\in\R\cup\{-\infty\}$, we shall write $[a,\infty]$ for the closure of $[a,\infty)$ inside $\ol\R$; and we put $(a,\infty]=[a,\infty]\setminus\{a\}$. The sets $[-\infty,a]$, and $[-\infty,a)$ are defined analogously. The boundary at fiber infinity of $\bar E$ is a sphere bundle $S\bar E\to M$.

We denote by $P^m(E)\subset\CI(E)$ the space of smooth functions which are polynomials of degree $m\in\N_0$ on each fiber of $E$. Similarly, $S^s(E)\subset\CI(E)$ denotes the space of symbols (of class $1,0$) of order $s\in\R$ on the fibers of $E$; an equivalent definition is $S^s(E)=\cA^{-s}(\bar E)$ (with smoothness down to $\bar E|_{\pa M}$).

Finally, suppose $S\subset M$ is an interior p-submanifold of an $n$-dimensional manifold $M$ with corners, meaning that $S\cap M^\circ\neq\emptyset$. Thus, in suitable local coordinates $x^1,\ldots,x^k\geq 0$, $y=(y',y'')\in\R^p\times\R^{n-k-p}$, the submanifold $S$ is given by $y'=0$ where $p\geq 1$ is the codimension of $S$. For $s\in\R$, we then denote by $I^s(M,S)\subset\CmI(M)=(\CIdot(M;\Omega M))^*$ the space of conormal distributions at $S$ of order $s$; in local coordinates, such a distribution is given as
\[
  u(x,y) = (2\pi)^{-p}\int_{\R^p} e^{i\eta'\cdot y'}a(x,y'',\eta')\,\dd\eta',
\]
where $a\in S^{s+\frac{n}{4}-\frac{p}{2}}([0,\infty)^k\times\R^{n-k-p};\R^p)$. (We follow the order convention of \cite{HormanderFIO1}.) One can also consider symbols which are merely conormal (with some weight) at $x=0$, and allow for the presence of parameters $\delta_j\in[0,\half)$ for $j$ in some subset of $\{1,\ldots,k\}$ (which in particular allows for variable decay orders along $(x^j)^{-1}(0)$ for these $j$). Moreover, by $I^s(M,S;E)$ we denote the space of conormal distributions with values in the vector bundle $E\to M$. See \cite[\S18]{HormanderAnalysisPDE3} for further details.

\subsection{b- and scattering pseudodifferential operators}
\label{SsPbsc}

Let $X$ denote an $n$-dimensional manifold with boundary. The \emph{b-double space} of $X$ is
\[
  X^2_\bop:=[X^2;(\pa X)^2].
\]
We denote by $\diag_\bop\subset X^2_\bop$ the lift of the diagonal $\diag X\subset X^2$, by $\ff_\bop\subset X^2_\bop$ the front face, and by $\lb_\bop$ and $\rb_\bop$ the lift of $\pa X\times X$ and $X\times\pa X$, respectively. See Figure~\ref{FigPb}.

\begin{figure}[!ht]
\centering
\includegraphics{FigPb}
\caption{The b-double space $X^2_\bop$ of $X=[0,1)$.}
\label{FigPb}
\end{figure}

Furthermore, $\pi_R\colon X^2_\bop\to X$ denotes the right projection, and $\Omegab X\to X$ is the b-density bundle (i.e.\ the density bundle associated with $\Tb X$). The space $\Psib^s(X)$ of \emph{b-pseudodifferential operators} (or \emph{b-ps.d.o.s}) then consists of all continuous linear operators on $\CIdot(X)$ whose Schwartz kernels $\kappa\in I^s(X^2_\bop,\diag_\bop;\pi_R^*\Omegab X)$ vanish to infinite order at all boundary hypersurfaces of $X^2_\bop$ except $\ff_\bop$, and which are properly supported when $X$ is non-compact. (See \cite{MelroseAPS} for an extensive discussion.) The principal symbol $\sigmab^s$ fits into a short exact sequence
\[
  0 \to \Psib^{s-1}(X) \hra \Psib^s(X) \xra{\sigmab^s} S^s(\Tb^*X)/S^{s-1}(\Tb^*X) \to 0.
\]
Composition of operators is a continuous bilinear map $\Psib^{s_1}(X)\circ\Psib^{s_2}(X)\subset\Psib^{s_1+s_2}(X)$, and the principal symbol is multiplicative. In local coordinates $x\geq 0$, $y\in\R^{n-1}$ on $X$, lifted along the left, resp.\ right projection to smooth functions $x,y$, resp.\ $x',y'$ on $X^2_\bop$, local coordinates on $X^2_\bop$ near $\diag_\bop$ are $x,y,\frac{x-x'}{x'},y-y'$. For $\chi\in\CIc(\R)$ identically $1$ near $0$ and supported in a small neighborhood of $0$, the operator
\begin{align*}
  (\Opb(a)u)(x,y) &:= (2\pi)^{-n} \iint_{\R\times\R^{n-1}\times[0,\infty)\times\R^{n-1}} \exp\Bigl(i\Bigl(\frac{x-x'}{x}\xi_\bop+(y-y')\cdot\eta_\bop\Bigr)\Bigr) \\
    &\hspace{5em} \times \chi\Bigl(\Bigl|\log\frac{x}{x'}\Bigr|\Bigr)\chi(|y-y'|)a(x,y,\xi_\bop,\eta_\bop)u(x',y')\,\dd\xi_\bop\,\dd\eta_\bop\,\frac{\dd x'}{x'}\,\dd y',
\end{align*}
for $a$ a symbol of order $s$ in $(\xi_\bop,\eta_\bop)$ with support in the local coordinate system, defines a typical element of $\Psib^s(X)$; it is a \emph{quantization} of $a$. (The two factors of $\chi$ localize to a neighborhood of $\diag_\bop$.) Every element of $\Psib^s(X)$ is a locally finite sum (on the level of Schwartz kernels) of such operators, plus an element of $\Psib^{-\infty}(X)$. Spaces of weighted operators are defined by $\Psib^{s,l}(X)=\rho^{-l}\Psib^s(X)$ where $\rho\in\CI(X)$ is a boundary defining function (lifted to the left factor of $X^2_\bop$); one can more generally quantize symbols of order $s$ in the fibers of $\Tb^*X$ which are conormal with weight $\rho^{-l}$ down to $\Tb^*_{\pa X}X$.

Turning to scattering ps.d.o.s, we recall the \emph{scattering double space}
\[
  X^2_\scop=[X^2_\bop;\diag_\bop\cap\,\ff_\bop],
\]
with front face denoted $\ff_\scop$; we write $\diag_\scop\subset X^2_\scop$ for the lift of $\diag_\bop$. See Figure~\ref{FigPsc}.

\begin{figure}[!ht]
\centering
\includegraphics{FigPsc}
\caption{The scattering double space $X^2_\scop$.}
\label{FigPsc}
\end{figure}

Schwartz kernels of elements of the space $\Psisc^{s,r}(X)$ of \emph{scattering ps.d.o.s} of order $s,r$ are then elements of $\rho^{-r}I^s(X_\scop^2,\diag_\scop;\pi_R^*\Omegasc X)$, with $\Omegasc X\to X$ denoting the density bundle associated with $\Tsc X\to X$. Such operators are discussed in \cite{MelroseEuclideanSpectralTheory}. (In the special case $X=\ol{\R^n}$, a thorough discussion, including the case of variable orders, is given in \cite{VasyMinicourse}. We note that $\Vsc(\ol{\R^n})$ is spanned over $\CI(\ol{\R^n})$---which is equal to the space of classical symbols of order $0$ on $\R^n$---by translation-invariant vector fields on $\R^n$, and the space $\Psisc^{s,r}(\ol{\R^n})$ is equal to the space of quantizations $(2\pi)^{-n}\int e^{i z\cdot\zeta}a(z,\zeta)\,\dd\zeta$ of smooth functions $a$ which are symbols of order $r$, resp.\ $s$ in $z$, resp.\ $\zeta$.) In local coordinates on $X$ as above, a typical element of $\Psisc^{s,r}(X)$ is given by
\begin{align*}
  (\Opsc(a)u)(x,y) &= (2\pi)^{-n} \iiiint_{\R\times\R^{n-1}\times[0,\infty)\times\R^{n-1}} \exp\Bigl(i\Bigl(\frac{x-x'}{x^2}\xi_\scop+\frac{y-y'}{x}\cdot\eta_\scop\Bigr)\Bigr) \\
    &\hspace{2em} \times \chi\Bigl(\Bigl|\log\frac{x}{x'}\Bigr|\Bigr)\chi(|y-y'|)a(x,y,\xi_\scop,\eta_\scop)u(x',y')\,\dd\xi_\scop\,\dd\eta_\scop\,\frac{\dd x'}{x'{}^2}\,\frac{\dd y'}{x'{}^{n-1}}.
\end{align*}

The principal symbol $\sigmasc^{s,r}$ fits into the short exact sequence
\[
  0 \to \Psisc^{s-1,r-1}(X) \hra \Psisc^{s,r}(X) \xra{\sigmasc^{s,r}} (S^{s,r}/S^{s-1,r-1})(\Tsc^*X) \to 0
\]
where $S^{s,r}(\Tsc^*X)$ denotes functions which are conormal on $\ol{\Tsc^*}X$ of order $-s$ at $\Ssc^*X$ (fiber infinity of $\ol{\Tsc^*}X)$ and of order $-r$ at $\ol{\Tsc^*_{\pa X}}X$.

We can more generally consider quantizations of symbols $a\in S^{s,\sfr}(\Tsc^*X)$ with \emph{variable scattering decay order} $\sfr\in\CI(\ol{\Tsc^*}X)$ (i.e.\ conormal functions on $\ol{\Tsc^*}X$ with variable order at $\ol{\Tsc^*_{\pa X}}X$). The resulting space of operators is denoted $\Psisc^{s,\sfr}(X)$, and the principal symbol now takes values in $(S^{s,\sfr}/S^{s-1,\sfr-1+2\delta})(\Tsc^*X)$ for any $\delta\in(0,\half)$. See also \cite[\S2]{HintzConicProp}.

We also use the \emph{semiclassical scattering algebra}; this was introduced by Vasy--Zworski \cite{VasyZworskiScl} in the context of high energy estimates for resolvents on asymptotically Euclidean manifolds. We discuss this in a slightly nonstandard way, mirroring the discussion in \cite[\S3.4]{HintzConicProp} for the semiclassical b-algebra. The underlying Lie algebra of vector fields is
\[
  \Vsch(X) := h \CI([0,1]_h;\Vsc(X)),
\]
i.e.\ in terms of $X_\schop:=[0,1]_h\times X$ this is the space of elements of $\rho\Vb(X_\schop)$ which annihilate $h$ and which vanish at $h=0$. Thus, $[\Vsch(X),\Vsch(X)]\subset h\rho\Vsch(X)$. In local coordinates, $\Vsch(X)$ is spanned by $h\rho^2\pa_\rho$ and $h\rho\pa_{y^j}$ ($j=1,\ldots,n-1$); these vector fields form a frame for the semiclassical scattering tangent bundle $\Tsch X\to X_\schop$. By $\Diffsch^m(X)$ we denote the corresponding space of $m$-th order semiclassical scattering differential operators (which are thus families of scattering operators on $X$ which degenerate in a particular manner as $h\searrow 0$); the principal symbol map gives rise to a short exact sequence
\[
  0 \to h\rho\Diffsch^{m-1}(X) \hra \Diffsch^m(X) \xra{\sigmasch^m} P^m(\Tsch^*X)/h\rho P^{m-1}(\Tsch^*X) \to 0.
\]
Define the semiclassical scattering double space by
\[
   X^2_\schop := \bigl[ [0,1]_h \times X^2_\scop; \{0\}\times\diag_\scop \bigr],
\]
with $\diag_\schop\subset X^2_\schop$ denoting the lift of $[0,1]\times\diag_\scop$. See Figure~\ref{FigPsch}.

\begin{figure}[!ht]
\centering
\includegraphics{FigPsch}
\caption{The semiclassical scattering double space $X^2_\schop$.}
\label{FigPsch}
\end{figure}

Then Schwartz kernels of elements of the corresponding space
\[
  \Psisch^s(X)
\]
of semiclassical scattering ps.d.o.s are those elements of $I^{s-\frac14}(X^2_\schop,\diag_\schop;\pi_R^*\,\Omegasch X)$ which vanish to infinite order at all boundary hypersurfaces of $X^2_\schop$ except those which intersect $\diag_\schop$ nontrivially, and which are smooth down to the lift of $h^{-1}(1)$. Here $\pi_R$ is the lift of $[0,1]\times X\times X\ni(h,z,z')\mapsto(h,z')\in[0,1]\times X$, and $\Omegasch X\to X_\schop$ is the density bundle associated with $\Tsch X\to X_\schop$. In local coordinates, a typical example of such an operator is the family $\Op_{\scop,h}(a)$, $h\in(0,1]$, of bounded linear operators defined by
\begin{align*}
  &(\Op_{\scop,h}(a)u)(h,x,y) \\
  &\qquad = (2\pi h)^{-n} \iiiint_{\R\times\R^{n-1}\times[0,\infty)\times\R^{n-1}} \exp\Bigl(i\Bigl(\frac{x-x'}{x^2}\xi_\schop+\frac{y-y'}{x}\cdot\eta_\schop\Bigr)/h\Bigr) \\
  &\qquad\hspace{2em} \times \chi\Bigl(\Bigl|\log\frac{x}{x'}\Bigr|/h\Bigr)\chi(|y-y'|/h)a(h,x,y,\xi_\schop,\eta_\schop)u(x',y')\,\dd\xi_\schop\,\dd\eta_\schop\,\frac{\dd x'}{x'{}^2}\,\frac{\dd y'}{x'{}^{n-1}},
\end{align*}
where $a$ is smooth in $h,x,y$ and a symbol of order $s$ in $(\xi_\schop,\eta_\schop)$. More generally, we can consider symbols $a\in S^{s,r,b}(\Tsch^*X)$ which are conormal functions on $\ol{\Tsch^*}X$ with weight $-r$ at $x=0$ and weight $-b$ at $h=0$; these two orders may be variable, but we shall only consider the case of variable scattering decay orders $\sfr\in\CI(\ol{\Tsch^*_{[0,1]\times\pa X}}X_\schop)$. The resulting space of operators is denoted
\[
  \Psisch^{s,\sfr,b}(X),
\]
and the principal symbol map $\sigmasch^{s,\sfr,b}$ on it takes values in $(S^{s,\sfr,b}/S^{s-1,\sfr-1+2\delta,b-1})(\Tsch^*X)$ for any $\delta>0$.

\begin{rmk}[Compact parameter space]
\label{RmkPCompact}
  Semiclassical operators are usually defined for parameters $h$ lying in an interval $(0,1)$ that is open at $1$ (with $1$ simply being a convenient positive number). In this paper, we include the value $1$ as well and require smoothness of Schwartz kernels all the way up to $h=1$. The reason is that the main pseudodifferential algebra in this paper, the Q-algebra, contains at the same time semiclassical and non-semiclassical algebras (say with parameters $h\in(0,1]$ and $\sigma\in(0,1]$) which fit together smoothly (at $\sigma^{-1}=h=1$).
\end{rmk}

\begin{rmk}[Variable orders]
\label{RmkPVariable}
  Pseudodifferential operators with variable orders were used already by Unterberger \cite{UnterbergerVariable}. For a discussion of variable order b-ps.d.o.s, see \cite[Appendix~A]{BaskinVasyWunschRadMink}. Semiclassical spaces with variable semiclassical orders (powers of $h$) are discussed in \cite{HillairetWunschConic}; see also \cite[Appendix~A]{HintzVasyCauchyHorizon}.
\end{rmk}

\subsection{Semiclassical cone operators}
\label{SsPch}

Consider a compact $n$-dimensional manifold $X$ with connected and embedded boundary $\pa X\neq 0$. (We can allow for $X$ to be non-compact if we require all Schwartz kernels to be properly supported.) We recall elements of \emph{semiclassical cone analysis} on $X$ from \cite{HintzConicPowers,HintzConicProp}. (This is a semiclassical version of a large parameter calculus developed by Loya \cite{LoyaConicResolvent,LoyaConicPower}; see also \cite{GilConicHeat,CoriascoSchroheSeilerCone,GilKrainerMendozaResolvents} for variants based on Schulze's cone calculus \cite{SchulzePsdoSing,SchulzePsdoBVP94,SchulzeBVP98}.) The \emph{semiclassical cone single space} (or \emph{$\chop$-single space}), introduced in \cite[\S3.1.1]{HintzConicPowers}, is the blow-up\footnote{See Remark~\ref{RmkPCompact} for the reason for including $h=1$.}
\[
  X_\chop := \bigl[ [0,1] \times X; \{0\} \times \pa X \bigr],
\]
with boundary hypersurfaces denoted $\cface$ (the lift of $[0,1]\times\pa X$), $\tface$ (the front face), and $\sface$ (the lift of $\{0\}\times X$). Denote by $h\in[0,1]$ the first coordinate on $[0,1]\times X$ (identified with its lift as a smooth function to $X_\chop$). The Lie algebra of \emph{$\chop$-vector fields} is
\[
  \Vch(X) := \{ V\in\rho_\sface\Vb(X_\chop) \colon V h=0 \},
\]
where $\rho_\sface\in\CI(X_\chop)$ is a defining function of $\sface$. In particular, restriction to any positive level set $h=h_0>0$ of $h$ gives a surjective map $\Vch(X)\to\Vb(X)$. In local coordinates $r\geq 0$, $\omega\in\R^{n-1}$ near a point in $\pa X$, we can take $\rho_\sface=\frac{h}{h+r}$, and the space $\Vch(X)$ is locally spanned by
\[
  \tfrac{h}{h+r}r\pa_r,\quad
  \tfrac{h}{h+r}\pa_{\omega^j}\ (j=1,\ldots,n-1),
\]
over the space of smooth functions of $h+r\geq 0$, $\frac{r-h}{r+h}\in[-1,1]$, and $\omega$. These vector fields give a local frame for the $\chop$-tangent bundle\footnote{We write $\Tch X$ here, as it is slightly less cumbersome than the notation $\Tch X_\chop$ used in \cite{HintzConicProp}.} $\Tch X\to X_\chop$. Denote by $\Diffch^m(X)$ the space of locally finite sums of up to $m$-fold compositions of $\chop$-vector fields and multiplication operators by elements of $\CI(X_\chop)$. Since $[\Vch(X),\Vch(X)]\subset\rho_\sface\Vch(X)$, we then have a well-defined principal symbol map $\sigmach^m$ which fits into a short exact sequence
\[
  0 \to \rho_\sface\Diffch^{m-1}(X) \hra \Diffch^m(X) \xra{\sigmach^m} P^m(\Tch^*X)/\rho_\sface P^{m-1}(\Tch^*X) \to 0.
\]

The front face of $X_\chop$ is the closure $\tface=\ol{{}^+N}\pa X$ of the inward pointing normal bundle of $\pa X$; its two boundary hypersurfaces are the zero section (with defining function $\rho_\cface=\frac{r}{h+r}$) and the boundary at fiber infinity (with defining function $\rho_\sface=\frac{h}{h+r}$). We can thus consider the space $\cV_{\bop,\scop}(\tface)=\rho_\sface\Vb(\tface)$ of b-scattering vector fields on $\tface$. By \cite[Lemma~3.5]{HintzConicProp}, the restriction $N_\tface$ of b-vector fields on $X_\chop$ to $\tface$ gives rise to a short exact sequence
\[
  0 \to \rho_\tface\Vch(X) \hra \Vch(X) \xra{N_\tface} \cV_{\bop,\scop}(\tface) \to 0,
\]
and correspondingly to an isomorphism $\Tch_\tface X\cong{}^{\bop,\scop}T\tface$ of tangent bundles and
\begin{equation}
\label{EqPchPhase}
  \Tch^*_\tface X \cong {}^{\bop,\scop}T^*\tface
\end{equation}
of cotangent bundles. The map $N_\tface$ extends to a multiplicative map $N_\tface\colon\Diffch^m(X)\to\Diff_{\bop,\scop}^m(\tface)$.

The $\chop$-double space is defined as\footnote{In \cite[Definition~3.1]{HintzConicPowers}, the $\chop$-double space is defined without the blow-up of $\{0\}\times\lb_\bop$ and $\{0\}\times\rb_\bop$. What we call the $\chop$-double space here is the `extended $\chop$-double space' of \cite[Definition~3.7]{HintzConicPowers}, which is geometrically slightly more complex, but more natural (e.g.\ the left and right projections $X_\chop^2\to X_\chop$ from the extended $\chop$-double space to the $\chop$-single space are b-fibrations).}
\[
  X^2_\chop := \bigl[ [0,1] \times X_\bop^2; \{0\}\times\ff_\bop; \{0\}\times\diag_\bop, \{0\}\times\lb_\bop, \{0\}\times\rb_\bop \bigr].
\]
We denote by $\ff_2$, $\tface_2$, and $\dface_2$ the lifts of $[0,1]\times\ff_\bop$, $\{0\}\times\ff_\bop$, and $\{0\}\times\diag_\bop$, respectively, and by $\tlb_2$, $\trb_2$ and $\sface_2$ the lifts of $\{0\}\times\lb_\bop$, $\{0\}\times\rb_\bop$, and $\{0\}\times X_\bop^2$, respectively. Finally, $\diag_\chop$ denotes the lift of $[0,1]\times\diag_\bop$. See Figure~\ref{FigPch}.

\begin{figure}[!ht]
\centering
\includegraphics{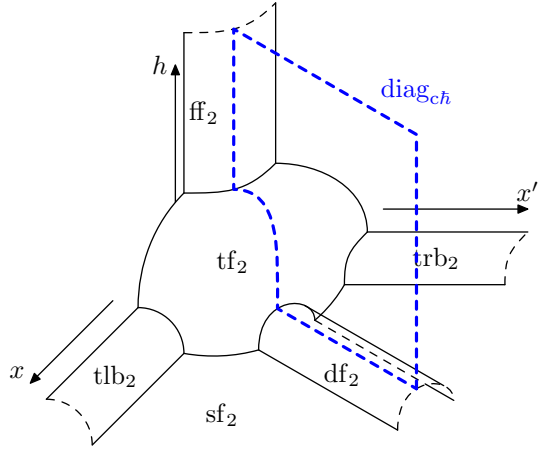}
\caption{The semiclassical cone double space $X^2_\chop$ (called the extended semiclassical cone double space $'X^2_\chop$ in \cite{HintzConicPowers}).}
\label{FigPch}
\end{figure}

The space
\[
  \Psich^s(X)
\]
then consists of smooth (in $h\in(0,1]$) families of continuous linear operators on $\CIdot(X)$ whose Schwartz kernels are elements of $I^{s-\frac14}(X^2_\chop,\diag_\chop;\pi_R^*\Omegach X)$ that vanish to infinite order at all boundary hypersurfaces of $X^2_\chop$ except for $\ff_2,\tface_2,\dface_2$ and the lift of $h^{-1}(1)$. Here $\Omegach X\to X_\chop$ is the density bundle associated with $\Tch X\to X_\chop$, and $\pi_R$ is the lift of the right projection $[0,1]\times X\times X\ni(h,x,x')\mapsto(h,x')\in[0,1]\times X$. In local coordinates as above, a typical element of $\Psich^s(X)$ is the family of operators $\Op_{\cop,h}(a)$ defined by
\begin{align*}
  &(\Op_{\cop,h}(a)u)(r,\omega) \\
  &\quad := (2\pi)^{-n}\iiiint \exp\biggl(i\biggl[\frac{r-r'}{r\frac{h}{h+r}}\xi_\chop + \frac{\omega-\omega'}{\frac{h}{h+r}}\cdot\eta_\chop\biggr]\biggr) \chi\Bigl(\Bigl|\log\frac{r}{r'}\Bigr|\Bigr)\chi(|\omega-\omega'|) \\
  &\quad\qquad\hspace{8em} \times a(h,r,\omega,\xi_\chop,\eta_\chop)u(r',\omega')\,\dd\xi_\chop\,\dd\eta_\chop\,\frac{\dd r'}{r'\frac{h}{h+r'}}\frac{\dd\omega'}{\bigl(\frac{h}{h+r'}\bigr)^{n-1}}.
\end{align*}
Here $a$ is a symbol of order $s$ in $(\xi_\chop,\eta_\chop)$, with smooth dependence on $h+r$, $\frac{r-h}{r+h}$, $\omega$ (i.e.\ on $h$, $r/h$, $\omega$ in $r\lesssim h$ and on $r$, $h/r$, $\omega$ in $h\lesssim r$).

More generally then, we can consider quantizations of symbols $a\in S^{s,l,l',r}(\Tch^*X)$,\footnote{In \cite{HintzConicProp}, the slightly more cumbersome notation $S^{s,l,l',r}(\ol{\Tch^*}X_\chop)$ is used for the same symbol space.} which are conormal functions on $\ol{\Tch^*}X$ of differential order $s$ (i.e.\ have weight $-s$) at fiber infinity $\Sch^*X$, of b-decay order $l$ at $\ol{\Tch^*_\cface}X$, of $\tface$-decay order $l'$ at $\ol{\Tch^*_\tface}X$, and of semiclassical order $r$ at $\ol{\Tch^*_\sface}X$. The resulting space of operators is denoted
\[
  \Psich^{s,l,l',r}(X).
\]
(Restriction of elements of this space to a level set $h=h_0>0$ gives a surjective map to $\Psib^{s,l}(X)$, whereas restriction in both factors of $X$ in $X^2_\chop$ to the interior $X^\circ$ gives a semiclassical ps.d.o. $h^{-r}\Psih^s(X^\circ)$ on $X^\circ$.) The differential and semiclassical orders can moreover be taken to be variable; we only need the case of variable semiclassical orders. Thus, for $\sfr\in\CI(\ol{\Tch^*_\sface}X)$, we denote by
\[
  \Psich^{s,l,l',\sfr}(X)
\]
the corresponding space of operators, defined as the sum of $\Psich^{-\infty,l,l',-\infty}(X)$ and finite sums of quantizations of symbols on $\Tch^*X$ which are conormal on $\ol{\Tch^*}X$ with weights $-s$, $-l$, $-l'$, and $-\sfr$ (thus with an arbitrarily small parameter $\delta_\sface>0$ at $\sface$ in the notation introduced after~\eqref{EqPConormal}). The principal symbol map in this case is
\[
  \sigmach_{s,l,l',\sfr}\colon\Psich^{s,l,l',\sfr}(X) \to (S^{s,l,l',\sfr}/S^{s-1,l,l',\sfr-1+2\delta})(\Tch^*X)
\]
for any $\delta\in(0,\half)$. (See \cite[\S3.2]{HintzConicProp} for further details.)

For those elements of $\Psich^{s,l,0,\sfr}(X)$ which have Schwartz kernels which are smooth down to $\tface_2$ (as distributions conormal to $\diag_\chop$), indicated by a subscript `$\cl$', restriction to $\tface$ gives rise to a surjective map
\[
  N_\tface \colon \Psi_{\chop,\cl}^{s,l,0,\sfr}(X) \to \Psi_{\bop,\scop}^{s,l,\sfr}(\tface)
\]
with kernel $\rho_\tface\Psi_{\chop,\cl}^{s,l,0,\sfr}(X)$; the restriction of the principal symbol of $A\in\Psi_{\chop,\cl}^{s,l,0,\sfr}(X)$ to $\Tch^*_\tface X$ equals the principal symbol of $N_\tface(A)$ under the identification~\eqref{EqPchPhase}. Note here that this identification also relates $\sfr|_{\ol{\Tch^*_\tface}X}$ to a variable scattering decay order $\sfr|_\tface\in\CI(\ol{{}^{\bop,\scop}T^*}\tface)$.

\subsection{Scattering-b-transition algebra}
\label{SsPscbt}

With $X$ denoting a compact $n$-dimensional manifold with connected and embedded boundary $\pa X\neq\emptyset$, the final algebra of families of degenerating ps.d.o.s on $X$ that we recall here was introduced by Guillarmou--Hassell \cite{GuillarmouHassellResI} for the purpose of giving a precise uniform description of the Schwartz kernel of the low energy resolvent on asymptotically conic spaces as one approaches the spectral parameter $0$ from the resolvent set. We only need the small calculus (i.e.\ without boundary terms).

We define the \emph{$\scop$-$\bop$-transition single space} to be
\[
  X_\scbtop := \bigl[ [0,1]\times X; \{0\}\times\pa X \bigr],
\]
with the lift of the first coordinate function denoted $\sigma$. (One can completely analogously study negative $\sigma$, in which case $X_\scbtop=[[-1,0]\times X;\{0\}\times\pa X]$. In the main part of this paper, it will be clear from the context which of the two versions is used.) We denote by $\scface$, $\tface$, and $\zface$ the lift of $[0,1]\times\pa X$, the front face, and the lift of $\{0\}\times X$, respectively. This is the resolved space for low energy spectral theory from \cite[Definition~2.12]{HintzPrice} (and denoted $X_\res^+$ there); a refinement of the corresponding phase space (in the notation introduced below: the blow-up of $\Tscbt^*X$ at the zero section over $\scface$) was previously introduced by Vasy \cite{VasyLowEnergyLag}. With $\rho_H\in\CI(X_\res)$ denoting a defining function of $H$, we set
\[
  \cV_\scbtop(X) := \{ V\in\rho_\scface\Vb(X_\scbtop) \colon V\sigma=0 \}.
\]
In local coordinates $\rho\geq 0$, $\omega\in\R^{n-1}$ near a point on $X$, this space is spanned over $\CI(X_\scbtop)$ by
\[
  \frac{\rho}{\rho+\sigma}\rho\pa_\rho,\qquad
  \frac{\rho}{\rho+\sigma}\pa_{\omega^j}\ (j=1,\ldots,n-1).
\]
The vector bundle which has these vectors as a local frame is the \emph{$\scbtop$-transition tangent bundle} $\Tscbt X\to X_\scbtop$; the dual bundle is denoted $\Tscbt^*X\to X_\scbtop$ as usual, with local frame
\begin{equation}
\label{EqPscbtFrame}
  \frac{\rho+\sigma}{\rho}\frac{\dd\rho}{\rho},\qquad \frac{\rho+\sigma}{\rho}\,\dd\omega^j\ (j=1,\ldots,n-1).
\end{equation}

The space $\Vscbt(X)$ is a Lie algebra, and indeed $[\Vscbt(X),\Vscbt(X)]\subset\rho_\scface\Vscbt(X)$. (Thus, while $X_\scbtop$ is the same as the $\chop$-single space except for renaming $\sigma$ to $h$, the Lie algebras $\Vscbt(X)$ and $\Vch(X)$ are different.) The restriction $N_\tface$ to $\tface$ gives rise to a short exact sequence
\[
  0 \to \rho_\tface\Vscbt(X) \hra \Vscbt(X) \xra{N_\tface} \cV_{\scop,\bop}(\tface) \to 0
\]
and thus to an identification
\begin{equation}
\label{EqPscbtPhase}
  \Tscbt^*_\tface X \cong {}^{\scop,\bop}T^*\tface.
\end{equation}
We also remark that for each $\sigma_0>0$, the restriction of $V\in\Vscbt(X)$ to $\{\sigma=\sigma_0\}\cong X$ defines an element of $\Vsc(X)$, whereas the restriction to the lift $\zface\cong X$ of $\sigma=0$ lies in $\Vb(X)$.

We denote the space of $m$-th order differential operators generated by $\Vscbt(X)$ by $\Diffscbt^m(X)$; the principal symbol map $\sigmascbt^m$ fits into the short exact sequence
\[
  0 \to \rho_\scface\Diffscbt^{m-1}(X) \hra \Diffscbt^m(X) \xra{\sigmascbt^m} P^m(\Tscbt^*X)/\rho_\scface P^{m-1}(\Tscbt^*X) \to 0.
\]

In order to define the microlocalization of $\Diffscbt(X)$, we define the \emph{$\scbtop$-transition double space} as
\[
  X^2_\scbtop := \bigl[ [0,1]\times X^2_\bop; \{0\}\times\ff_\bop, \{0\}\times\lb_\bop, \{0\}\times\rb_\bop; [0,1]\times\pa\diag_\bop \bigr].
\]
(This space is denoted $M^2_{k,sc}$ in \cite{GuillarmouHassellResI}.) We let $\scface_2$, $\tface_2$, and $\zface_2$ denote the lifts of $[0,1]\times\pa\diag_\bop$, $\{0\}\times\ff_\bop$, and $\{0\}\times X^2_\bop$, respectively; and we write $\diag_\scbtop\subset X^2_\scbtop$ for the lift of $[0,1]\times\diag_\bop$. See Figure~\ref{FigPscbt}.

\begin{figure}[!ht]
\centering
\includegraphics{FigPscbt}
\caption{The scattering-b-transition double space $X^2_\scbtop$.}
\label{FigPscbt}
\end{figure}

Then
\[
  \Psiscbt^s(X)
\]
is the space of smooth families (in $\sigma\in(0,1]$) of linear operators on $\CIdot(X)$ whose Schwartz kernels are elements of $I^{s-\frac14}(X^2_\scbtop,\diag_\scbtop;\pi_R^*\Omegascbt X)$ which vanish to infinite order at all boundary hypersurfaces of $X^2_\scbtop$ except those which have nonempty intersection with $\diag_\scbtop$ (these are $\scface_2$, $\tface_2$, and $\zface_2$, as well as the lift of $\sigma^{-1}(1)$). Here $\pi_R$ is the lift of $[0,1]\times X\times X\ni(\sigma,x,x')\mapsto(\sigma,x')$, and $\Omegascbt X\to X_\scbtop$ is the density bundle associated with $\Tscbt X\to X_\scbtop$. (In \cite{GuillarmouHassellResI}, the authors call this space of operators $\Psi_k^s(M)$, and they consider operators acting on b-$\half$-densities instead of scalar functions.)

In local coordinates $\sigma,\rho,\omega$ as above, a typical element of $\Psiscbt^s(X)$ is the family of operators $\Op_{\scbtop,\sigma}(a)$ defined by
\begin{align*}
  &(\Op_{\scbtop,\sigma}(a)u)(\rho,\omega) \\
  &\quad := (2\pi)^{-n}\iiiint \exp\biggl(i\biggl[\frac{\rho-\rho'}{\rho\frac{\rho}{\sigma+\rho}}\xi_\scbtop + \frac{\omega-\omega'}{\frac{\rho}{\sigma+\rho}}\cdot\eta_\scbtop\biggr]\biggr) \chi\Bigl(\Bigl|\log\frac{\rho}{\rho'}\Bigr|\Bigr)\chi(|\omega-\omega'|) \\
  &\quad\qquad\hspace{6em} \times a(\sigma,\rho,\omega,\xi_\scbtop,\eta_\scbtop)u(\rho',\omega')\,\dd\xi_\scbtop\,\dd\eta_\scbtop\,\frac{\dd \rho'}{\rho'\frac{\rho}{\sigma+\rho'}}\frac{\dd\omega'}{\bigl(\frac{\rho}{\sigma+\rho'}\bigr)^{n-1}}.
\end{align*}
Here $a$ is a symbol of order $s$ in $(\xi_\scbtop,\eta_\scbtop)$, with smooth dependence on $\sigma+\rho$, $\frac{\rho-\sigma}{\rho+\sigma}$, $\omega$ (i.e.\ on $\sigma$, $\rho/\sigma$ in $\rho\lesssim\sigma$ and on $\rho$, $\sigma/\rho$ in $\sigma\lesssim\rho$). One can more generally consider quantizations of symbols $a\in S^{s,r,\gamma,l}(\Tscbt^*X)$ which are conormal on $\ol{\Tscbt^*}X$ with differential order $s$ (i.e.\ weight $-s$ at fiber infinity), scattering decay order $r$ (i.e.\ weight $-r$ at $\ol{\Tscbt^*_\scface}X$), $\tface$-decay order $\gamma$ (i.e.\ weight $-\gamma$ at $\ol{\Tscbt^*_\tface}X)$, and $\zface$-order $l$ (i.e.\ weight $-l$ at $\ol{\Tscbt^*_\zface}X$); the resulting space of operators is denoted $\Psiscbt^{s,r,\gamma,l}(X)$. (Restrictions of elements of $\Psiscbt^{s,r,\gamma,l}(X)$ to a level set $\sigma=\sigma_0>0$ lie in $\Psisc^{s,r}(X)$.) Moreover, just as in the scattering calculus, we can generalize this further by allowing $s,r$ to be variable; in this paper we only need to consider variable scattering decay orders $\sfr\in\CI(\ol{\Tscbt^*_\scface}X)$ and the resulting space
\[
  \Psiscbt^{s,\sfr,\gamma,l}(X).
\]

For those elements of $\Psiscbt^{s,r,0,l}(X)$ whose Schwartz kernels are smooth down to $\tface_2$ (as distributions conormal to $\diag_\scbtop$), indicated by the subscript `$\cl$', restriction to $\tface$ gives rise to a surjective map
\[
  N_\tface \colon \Psi_{\scbtop,\cl}^{s,r,0,l}(X) \to \Psi_{\scop,\bop}^{s,r,l}(X)
\]
with kernel $\rho_\tface\Psi_{\scbtop,\cl}^{s,r,0,l}(X)$. The same remains true for variable orders $\sfr$ upon identifying the restriction of $\sfr$ to $\ol{\Tscbt^*_\tface}X$ with an element of $\CI(\ol{{}^{\scop,\bop}T^*}\tface)$ via~\eqref{EqPscbtPhase}.

\subsection{Sobolev spaces}
\label{SsPH}

For all calculi introduced, we can define corresponding $L^2$-based Sobolev spaces and their (possibly parameter-dependent) norms. We assume throughout that the underlying manifold $X$ is compact. Fixing $\alpha_\nu\in\R$, denote by
\[
  \nu = \rho^{\alpha_\nu}\nu_0,\qquad 0<\nu_0\in\CI(X,\Omegab X),
\]
a weighted b-density. All function spaces will be defined relative to the space $L^2(X,\nu)$. When the density is clear from the context (as is the case from here on), we shall omit it from the notation.

Consider first the b-setting; we let $\Hb^0(X,\nu)=L^2(X)$. The weighted space
\[
  \Hb^{s,l}(X) = \rho^l\Hb^s(X)
\]
is then defined for $s\geq 0$ as the space of all $u\in\Hb^{0,l}(X)$ so that $A u\in\Hb^{0,l}(X)$ for any fixed elliptic $A\in\Psib^s(X)$; for negative $s$ we can define $\Hb^{s,l}(X)=(\Hb^{-s,-l}(X))^*$ by duality (with respect to the $\Hb^0(X)$-inner product), or equivalently as the space of all elements of $\cC^{-\infty}(X)=\CIdot(X)^*$ which are of the form $u_0+A u_1$ where $u_0,u_1\in\Hb^{0,l}(X)$, with $A\in\Psib^{-s}(X)$ a fixed elliptic operator.

Weighted scattering Sobolev spaces $\Hsc^{s,r}(X)=\rho^r\Hsc^s(X)$ are defined in a completely analogous manner relative to $L^2(X)$. (For $X=\ol{\R^n}$, the space $\Hsc^{s,r}(X;|\dd x|)$ is equal to the standard weighted Sobolev space $\la x\ra^{-r}H^s(\R^n)$.) If $\sfr\in\CI(\ol{\Tsc^*}X)$ is a variable decay order and $r_0=\inf\sfr$, we define
\[
  \Hsc^{s,\sfr}(X) = \{ u\in \Hsc^{s,r_0}(X) \colon A u\in\Hsc^0(X) \},
\]
where $A\in\Psisc^{s,\sfr}(X)$ is any fixed elliptic operator.

We next consider $\chop$-Sobolev spaces. The base case is again the $L^2$-space $H_{\cop,h}^0(X):=L^2(X)$, with the $h$-independent norm given by the $L^2(X)$-norm. Next, the most general space we shall need is
\[
  H_{\cop,h}^{s,l,l',\sfr}(X),
\]
where $\sfr\in\CI(\ol{\Tch^*_\sface}X)$ is a variable semiclassical order. For each value of $h\in(0,1]$, this space is equal to $\Hb^{s,l}(X)$ as a set, but with a norm that is not uniformly equivalent as $h\searrow 0$. Namely, for $s\geq 0$, we fix an elliptic operator $A\in\Psich^{s,l,l',\sfr}(X)$ and define
\[
  \|u\|_{H_{\cop,h}^{s,l,l',\sfr}(X)}^2 = \|\rho_\cface^{-l}\rho_\tface^{-l'}\rho_\sface^{-\inf\sfr}u\|_{L^2(X)}^2 + \|A u\|_{L^2(X)}^2,
\]
where $\rho_H\in\CI(X_\chop)$ is a defining function of $H$. For $s<0$, we can define $H_{\cop,h}^{s,l,l',\sfr}(X)$ in any one of the two ways explained above for weighted b-Sobolev spaces.

Finally, we define the weighted $\scbtop$-transition Sobolev space
\[
  H_{\scbtop,\sigma}^{s,\sfr,\gamma,l}(X).
\]
This is for any fixed $\sigma>0$ equal to $\Hsc^{s,\sfr}(X)$ as a set; but for $s\geq 0$ it is equipped with the $\sigma$-dependent norm
\[
  \|u\|_{H_{\scbtop,\sigma}^{s,\sfr,\gamma,l}(X)}^2 = \|\rho_\scface^{-\inf\sfr}\rho_\tface^{-\gamma}\rho_\zface^{-l}u\|_{L^2(X)}^2 + \|A u\|_{L^2(X)}^2,
\]
where $A\in\Psiscbt^{s,\sfr,\gamma,l}(X)$ is any fixed elliptic operator. The definition for $s<0$ is analogous to the b-setting explained previously. The norm $\|u\|_{H_{\scbtop,\sigma}^{s,\sfr,\gamma,l}(X)}$ can be related to scattering-b-Sobolev norms near $\tface$ and b-Sobolev norms near $\zface$. Concretely, if we fix as a density on $X$ the scattering density $|\frac{\dd x}{x^2}\frac{\dd\omega}{x^{n-1}}|$ (or any smooth positive multiple thereof), then for $\chi=\chi(\sigma/x)\in\CIc([0,\infty))$, identically $1$ near $0$, we have a uniform (for $\sigma\in[0,1]$) equivalence of norms
\begin{subequations}
\begin{equation}
\label{EqPHEquivzf}
  \|\chi u\|_{H_{\scbtop,\sigma}^{s,\sfr,\gamma,l}(X)} \sim |\sigma|^l \|\chi u\|_{\Hb^{s,\gamma-l}(X)}.
\end{equation}
(That is, there exists constant $C>0$ so that for all $\sigma\in[0,1]$, the left hand side is bounded by $C$ times the right hand side, and vice versa.) Similarly, for a cutoff $\psi=\psi(|\sigma|+x)\in\CIc([0,\eps))$ (for small $\eps>0$, in a collar neighborhood $[0,\eps)\times\pa X$ of $\pa X$), identically $1$ near $0$, we have a uniform equivalence of norms
\begin{equation}
\label{EqPHEquivtf}
  \|\psi u\|_{H_{\scbtop,\sigma}^{s,\sfr,\gamma,l}(X)} \sim |\sigma|^{\frac{n}{2}-\gamma} \|\psi u\|_{H_{\scop,\bop}^{s,\sfr,l-\gamma}(\tface)}
\end{equation}
\end{subequations}
where upon setting $\hat x:=x/\sigma$, we use the density $|\frac{\dd\hat x}{\hat x^2}\frac{\dd\omega}{\hat x^{n-1}}|=|\sigma|^{-n}|\frac{\dd x}{x^2}\frac{\dd\omega}{x^{n-1}}|$ on $\tface$. The equivalences~\eqref{EqPHEquivzf}--\eqref{EqPHEquivtf} are easily checked for $L^2$-spaces ($s=0$) with constant scattering decay order $\sfr$. For general $s,\sfr$, they follow by using the definition of the respective norms using elliptic ps.d.o.s. For~\eqref{EqPHEquivzf}, one notes that the Schwartz kernel of an elliptic b-ps.d.o.\ on $X$ is a distribution on $\zface_2\subset X^2_\scbtop$ and as such can be extended, by $\sigma$-invariance, to the Schwartz kernel of a $\scbtop$-ps.d.o.\ which is elliptic near $\zface$. For~\eqref{EqPHEquivtf}, one uses that the Schwartz kernel of an elliptic scattering-b-ps.d.o.\ on $\tface$ is a distribution on $\tface_2\subset X^2_\scbtop$ which can be extended, by dilation-invariance in $(\sigma,x,x')$, to the Schwartz kernel of a $\scbtop$-ps.d.o.\ which is elliptic near $\tface$. See the proof of Proposition~\ref{PropQHRel} for further details in a similar context.

\section{Very large and extremely large frequency regimes in the Q-setting}
\label{SQSemi}

Here, we relate Q-analysis in the very large and extremely large frequency regimes described in~\S\ref{SsIA} to semiclassical cone analysis and doubly semiclassical cone analysis in the sense of \cite{HintzConicPowers}.

\begin{prop}[Intermediate and fully semiclassical regimes]
\label{PropQPSemi}
  For any fixed $\tilde\sigma_0\in\R\setminus\{0\}$, the level set $X_{\Qop,\tilde\sigma_0}^2:=X_\Qop^2\cap\tilde\sigma^{-1}(\tilde\sigma_0)$ is diffeomorphic, via the coordinates $(\bhm,x,x')\in(0,1]\times X\times X$, to
  \[
    X^2_{\qop\semi} := \bigl[ X^2_\qop; \diag_\qop\cap\,\mface_{\qop,2} \bigr].
  \]
  Moreover, $X^2_{\Qop,\pm,\tilde\semi}:=X^2_\Qop\cap\tilde\sigma^{-1}(\pm[1,\infty])$ is  diffeomorphic to
  \[
    X^2_{\qop\semi\tilde\semi} := \bigl[ [0,1]_{\tilde\semi} \times X^2_{q\semi}; \{0\}\times\diag_{q\semi} \bigr],
  \]
  where $\diag_{\qop\semi}\subset X^2_{q\semi}$ is the lift of $\diag_\qop$.
\end{prop}

\begin{rmk}[Relationship to (doubly) semiclassical cone algebras]
\label{RmkQPSemi}
  The space $X^2_{\qop\semi}$ carries the Schwartz kernels of an algebra $\Psi_{\qop\semi}(X)$ of pseudodifferential operators which microlocalizes the algebra of differential operators based on the Lie algebra $\cV_{\qop\semi}(X):=\{V\in\rho_{\mface_\qop}\Vb(X_\qop)\colon V\bhm=0\}$. Thus, elements of $\cV_{\qop\semi}(X)$ are semiclassical vector fields on $X\setminus\{0\}$, with semiclassical parameter $\bhm$; there is moreover a normal operator at $\zface_\qop$ which is of scattering type at $\pa\zface_\qop=\zface_\qop\cap\mface_\qop$. Note that the space $\cV_{\qop\semi}(X)$ is closely related to the space $\Vch(X)$ of semiclassical cone vector fields with semiclassical parameter $\bhm$, in that the spaces of restrictions of elements of $\cV_{\qop\semi}(X)$ and $\Vch(\dot X)$ to the set $|x|\gtrsim\bhm$ are equal. One can use such an algebra $\Psi_{\qop\semi}(X)$ for uniform analysis as $\bhm\searrow 0$ in the frequency regime $|\sigma|\sim\bhm^{-1}$ (i.e.\ $|\tilde\sigma|\sim 1$). The algebra $\Psi_{\qop\semi}(X)$ is contained in $\PsiQ(X)$ (in the sense that the space of restrictions of elements of $\PsiQ(X)$ to $\tilde\sigma^{-1}(\tilde\sigma_0)$ is equal to $\Psi_{q'}(X)$ for any $\tilde\sigma_0\in\R\setminus\{0\}$), and therefore we do not describe it in detail by itself. When restricting to $\tilde\sigma\in I$ in the case $I=\pm[\tilde\sigma_0,\infty]$ where $\tilde\sigma_0\in(0,\infty)$, Q-ps.d.o.s are semiclassical versions of $\qop\semi$-ps.d.o.s., with $\tilde h=|\tilde\sigma|^{-1}$ being the semiclassical parameter; this regime is thus closely related to (and in $|x|\gtrsim\bhm$ equal to) the doubly semiclassical cone calculus introduced in \cite[\S4]{HintzConicPowers}, with $\bhm$, resp.\ $\tilde h$ being the first, resp.\ second semiclassical parameter. The difference between the double space $X^2_{\qop\semi\tilde\semi}$ defined here and the doubly semiclassical cone double space of \cite[Definition~4.6]{HintzConicPowers} stems from the fact that only in the latter setting there is a cone point at the spatial origin which necessitates a semiclassical cone resolution at $\tilde h=0$.
\end{rmk}

\begin{proof}[Proof of Proposition~\usref{PropQPSemi}]
  Consider first a neighborhood of $\{\infty\}\times(\zface_{\qop,2})^\circ\subset\ol\R\times X^2_\qop$; this has a collar neighborhood $[0,1)_h\times[0,1]_\bhm\times(\hat X^\circ)^2$. Passing to the blow-up of $h=\bhm=0$, we have local coordinates $\tilde h=\frac{h}{\bhm}$, $\bhm$, $\hat x$, $\hat x'$ near the lift of $h=0$. The space resulting from blowing up the lift $\{0\}\times[0,1]_\bhm\times\diag_{\hat X^\circ}$ of $\pa\R\times\diag_\qop$ contains $[0,1]_\bhm\times(\hat X^\circ)^2_{\tilde\semi}$ where $(\hat X^\circ)^2_{\tilde\semi} = \bigl[ [0,1]_{\tilde h}\times(\hat X^\circ)^2; \{0\}\times\diag_{\hat X^\circ} \bigr]$ is the semiclassical double space of $\hat X^\circ$.

  We shall next analyze a neighborhood of the preimage of $[0,1]_h\times\mface_{\qop,2}=[0,1]_h\times\dot X^2_\bop$ in $X^2_\Qop$, and we in fact restrict attention in the second factor to a collar neighborhood $[0,1)_{\rho_{\mface_{\qop,2}}}\times[0,1)_{\dot\rho}\times[0,\infty]_s\times(\pa\dot X)^2$ of $\mface_{\qop,2}$, where $s=\frac{r}{r'}$ and
  \begin{equation}
  \label{EqQPSemirhomf}
    \rho_{\mface_{\qop,2}}=\frac{\bhm}{\dot\rho \dot\rho_L \dot\rho_R},\qquad
    \dot\rho=r+r',\quad \dot\rho_L=\frac{s}{s+1},\quad\dot\rho_R=\frac{1}{s+1}.
  \end{equation}
  We drop the factor $(\pa\dot X)^2$ from the notation; thus we have a coordinate chart
  \[
    [0,1]_h \times [0,1)_{\rho_{\mface_{\qop,2}}} \times [0,1)_{\dot\rho} \times [0,\infty]_s
  \]
  near $[0,1]_h\times\mface_{\qop,2}\subset\ol\R\times X_\qop^2$. In these coordinates, the $6$ submanifolds blown up in~\eqref{EqQPDouble} take the form\footnote{We combine the ones involving $\lb_{\qop,2}$ and $\rb_{\qop,2}$.}
  \begin{align*}
    \{0\}\times\zface_{\qop,2} &= \{0\}\times[0,1)\times\{0\}\times[0,\infty]; \\
    \{0\}\times\diag_\qop &= \{0\}\times[0,1)\times[0,1)\times\{1\}, \\
    \{0\}\times(\diag_\qop\cap\,\mface_{\qop,2}) &= \{0\}\times\{0\}\times[0,1)\times\{1\}, \\
    \{0\}\times(\lb_{\qop,2}\cup\rb_{\qop,2}) &= \{0\}\times[0,1)\times[0,1)\times\{0,\infty\}, \\
    \{0\}\times\mface_{\qop,2} &= \{0\}\times\{0\}\times[0,1)\times[0,\infty].
  \end{align*}
  Upon blowing up $\{0\}\times\zface_{\qop,2}$, a collar neighborhood of the lift of $h=0$ is given by
  \[
    [0,\infty)_{\dot h} \times [0,1)_{\rho_{\mface_{\qop,2}}} \times [0,1)_{\dot\rho} \times [0,\infty]_s,\qquad \dot h:=\frac{h}{\dot\rho}.
  \]
  The lifts of the remaining $4$ submanifolds all involve the factor $[0,1)_{\hat\rho}$, and therefore, upon blowing them up, we obtain the open submanifold with corners of $X^2_\Qop$
  \begin{equation}
  \label{EqQPSemi}
  \begin{split}
    [0,1)_{\dot\rho} \times &\bigl[ [0,\infty)_{\dot h}\times[0,1)_{\rho_{\mface_{\qop,2}}} \times [0,\infty]_s; \{0\}\times[0,1)\times\{1\}, \{0\}\times[0,1)\times\{0,\infty\}; \\
      &\hspace{16em} \{0\}\times\{0\}\times\{1\}; \{0\}\times\{0\}\times[0,\infty] \bigr].
  \end{split}
  \end{equation}
  See Figure~\ref{FigQPSemi}. %
  \begin{figure}[!ht]
  \centering
  \includegraphics{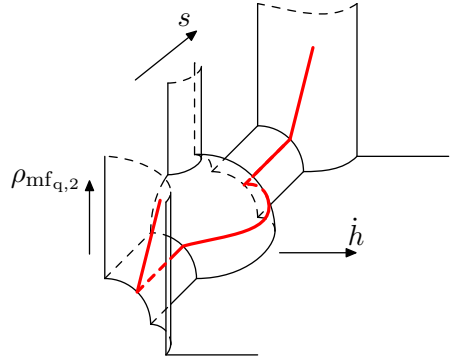}
  \caption{The space~\eqref{EqQPSemi} (with the $\dot\rho$-coordinate suppressed), as a subspace of the Q-double space $X^2_\Qop$. Also shown is the intersection of a level set $\tilde\sigma=\tilde\sigma_0\neq 0$ with $\pa X^2_\Qop$.}
  \label{FigQPSemi}
  \end{figure}%
  In particular, if one does not blow up $\{0\}\times[0,1)\times\{1\}$, then a collar neighborhood of the lift of $\dot h=0$ is given by
  \[
    [0,1)_{\dot\rho} \times [0,\infty)_{\tilde h} \times \bigl[ [0,1)_{\rho_{\mface_{\qop,2}}} \times [0,\infty]_s; \{0\}\times\{1\} \bigr]
  \]
  since $\tilde h=\frac{h}{\bhm}\sim\frac{\dot h}{\rho_{\mface_{\qop,2}}\dot\rho_L\dot\rho_R}$ by~\eqref{EqQPSemirhomf}. Blowing up the lift of $[0,1)\times\{0\}\times[0,1)\times\{1\}$ and reinserting the factor $(\pa\dot X)^2$, the open submanifold~\eqref{EqQPSemi} of $X_\Qop^2$ is thus
  \[
    \bigl[ [0,\infty)_{\tilde h} \times \cN; \{0\}\times\diag_{\qop\semi} \bigr]
  \]
  where $\cN\subset X^2_{\qop\semi}$ is a neighborhood of the preimage of $\mface_{\qop,2}$ under the blow-down map $X^2_{\qop\semi}\to X^2_\qop$. The proof of the Proposition is complete once one performs an analogous analysis of the geometry of $X^2_{\Qop,\tilde\sigma_0}$ and $X^2_{\Qop,+,\semi}$ near the preimages of $[0,1]_h\times\lb_{\qop,2}$ and $[0,1]_h\times\rb_{\qop,2}$; we leave this to the reader.
\end{proof}

\bibliographystyle{alphaurl}
\bibliography{/home/peter/research/bib/math,/home/peter/research/bib/mathcheck,/home/peter/research/bib/phys}

\end{document}